\documentclass[a4paper,11pt,openright,twoside]{book}

\usepackage{comment}
\usepackage{indentfirst}
\usepackage{cite}
\usepackage{palatino} 
\usepackage{mathtools}
\mathtoolsset{showonlyrefs,showmanualtags}
\usepackage{amsfonts}
\usepackage{amssymb}
\usepackage{amsthm,amsbsy}
\usepackage[format=default,margin=7.5pt,labelsep=period,font=small,labelfont=bf,textfont=it]{caption}
\usepackage{multirow}
\usepackage{url}
\usepackage{geometry}
\usepackage{bm,bbm}
\usepackage{fancyhdr}
\usepackage[dvips]{graphicx}
\usepackage[english]{babel}
\usepackage[Lenny]{fncychap}
\usepackage[tight,dotted]{minitoc}
\usepackage{subfigure}
\usepackage{siunitx}
\usepackage{mathtools}
\usepackage{tabularx,array}
\usepackage{perso}
\usepackage{float}
\usepackage{tikz}
\usepackage{ifpdf,epsfig}
\usetikzlibrary{automata,positioning,calc}
\usetikzlibrary{shapes,arrows}
\mathtoolsset{showonlyrefs}

\usepackage{hyphenat}
\usepackage{epigraph}
\usepackage{acro}
\DeclareAcronym{5G}{short = 5G ,long = fifth generation, class = abbrev}
\DeclareAcronym{AIS}{short = AIS ,long = automatic identification system, class = abbrev}
\DeclareAcronym{ACRDA}{short = ACRDA ,long = asynchronous contention resolution diversity ALOHA, class = abbrev}
\DeclareAcronym{ASM}{short = ASM ,long = application-specific messages, class = abbrev}
\DeclareAcronym{AWGN}{short = AWGN ,long = additive white gaussian noise, class = abbrev}
\DeclareAcronym{CDF}{short = CDF ,long = cumulative distribution function, class = abbrev}
\DeclareAcronym{CDMA}{short = CDMA ,long = code division multiple access, class = abbrev}
\DeclareAcronym{CRA-CC}{short = CRA-CC ,long = CRA-convolutional code, class = abbrev}
\DeclareAcronym{CRA-SH}{short = CRA-SH ,long = CRA-shannon bound, class = abbrev}
\DeclareAcronym{CRA}{short = CRA ,long = contention resolution ALOHA, class = abbrev}
\DeclareAcronym{CRC}{short = CRC ,long = cyclic redundancy check, class = abbrev}
\DeclareAcronym{CRDSA}{short = CRDSA ,long = contention resolution diversity slotted ALOHA, class = abbrev}
\DeclareAcronym{CRDSA++}{short = CRDSA++ ,long = contention resolution diversity slotted ALOHA++, class = abbrev}
\DeclareAcronym{CRI}{short = CRI ,long = collision resolution interval, class = abbrev}
\DeclareAcronym{CSA}{short = CSA ,long = coded slotted ALOHA, class = abbrev}
\DeclareAcronym{CSMA}{short = CSMA ,long = carrier sense multiple access, class = abbrev}
\DeclareAcronym{CSMA/CD}{short = CSMA/CD ,long = CSMA collision detection, class = abbrev}
\DeclareAcronym{CSMA/CA}{short = CSMA/CA ,long = CSMA collision avoidance, class = abbrev}
\DeclareAcronym{CSI}{short = CSI ,long = channel state information, class = abbrev}
\DeclareAcronym{DF}{short = D\&F ,long = decode and forward, class = abbrev}
\DeclareAcronym{DAMA}{short = DAMA ,long = demand assigned multiple access, class = abbrev}
\DeclareAcronym{DE}{short = DE ,long = density evolution, class = abbrev}
\DeclareAcronym{DSA}{short = DSA ,long = diversity slotted ALOHA, class = abbrev}
\DeclareAcronym{DVB-RCS2}{short = DVB-RCS2 ,long = Digital Video Broadcasting - Return Channel Satellite $2$nd Generation, class = abbrev}
\DeclareAcronym{E-SSA}{short = E-SSA ,long = enhanced-spread spectrum ALOHA, class = abbrev}
\DeclareAcronym{ECRA}{short = ECRA ,long = enhanced contention resolution ALOHA, class = abbrev}
\DeclareAcronym{ECRA-SC}{short = ECRA-SC ,long = ECRA selection combining, class = abbrev}
\DeclareAcronym{ECRA-MRC}{short = ECRA-MRC ,long = ECRA maximal-ratio combining, class = abbrev}
\DeclareAcronym{EGC}{short = EGC ,long = equal-gain combining, class = abbrev}
\DeclareAcronym{ETSI}{short = ETSI ,long = European Telecommunications Standards Institute, class = abbrev}
\DeclareAcronym{FDMA}{short = FDMA ,long = frequency division multiple access, class = abbrev}
\DeclareAcronym{FEC}{short = FEC ,long = forward error correction, class = abbrev}
\DeclareAcronym{iota_MPR}{short = $\iota$-MPR ,long = $\iota$-multi-packet reception, class = abbrev}
\DeclareAcronym{gateway}{short = GW ,long = gateway, class = abbrev}
\DeclareAcronym{GEO}{short = GEO ,long = geostationary orbit, class = abbrev}
\DeclareAcronym{GPRS}{short = GPRS ,long = general packet radio service, class = abbrev}
\DeclareAcronym{GSM}{short = GSM ,long = global system for mobile communications, class = abbrev}
\DeclareAcronym{i.i.d.}{short = i.i.d. ,long = independent and identically distributed, class = abbrev}
\DeclareAcronym{IC}{short = IC ,long = interference cancellation, class = abbrev}
\DeclareAcronym{IoT}{short = IoT ,long = Internet of things, class = abbrev}
\DeclareAcronym{IP}{short = IP ,long = internet protocol, class = abbrev}
\DeclareAcronym{IRCRA}{short = IRCRA ,long = irregular repetition contention resolution ALOHA, class = abbrev}
\DeclareAcronym{IRSA}{short = IRSA ,long = irregular repetition slotted ALOHA, class = abbrev}
\DeclareAcronym{ITU}{short = ITU ,long = International Telecommunications Union, class = abbrev}
\DeclareAcronym{LDPC}{short = LDPC ,long = low density parity check, class = abbrev}
\DeclareAcronym{LRT}{short = LRT ,long = likelihood ratio test, class = abbrev}%
\DeclareAcronym{LTE}{short = LTE ,long = long-term evolution, class = abbrev}
\DeclareAcronym{LTE-A}{short = LTE-A ,long = LTE-Advanced, class = abbrev}
\DeclareAcronym{M2M}{short = M2M ,long = machine-to-machine, class = abbrev}
\DeclareAcronym{MAC}{short = MAC ,long = medium access, class = abbrev}
\DeclareAcronym{MF}{short = MF ,long = matched filter, class = abbrev}
\DeclareAcronym{MF-TDMA}{short = MF-TDMA ,long = multi-frequency time division multiple access, class = abbrev}
\DeclareAcronym{ML}{short = ML ,long = maximum likelihood, class = abbrev}
\DeclareAcronym{MPR}{short = MPR ,long = multi-packet reception, class = abbrev}
\DeclareAcronym{MRC}{short = MRC ,long = maximal-ratio combining, class = abbrev}
\DeclareAcronym{MUD}{short = MUD ,long = multiuser detection, class = abbrev}
\DeclareAcronym{PDF}{short = p.d.f. ,long = probability density function, class = abbrev}
\DeclareAcronym{PER}{short = PER ,long = packet error rate, class = abbrev}
\DeclareAcronym{PMF}{short = p.m.f. ,long = probability mass function, class = abbrev}
\DeclareAcronym{PLR}{short = PLR ,long = packet loss rate, class = abbrev}
\DeclareAcronym{r.v.}{short = r.v. ,long = random variable, class = abbrev}
\DeclareAcronym{R-ALOHA}{short = R-ALOHA ,long = reservation ALOHA, class = abbrev}
\DeclareAcronym{RA}{short = RA ,long = random access, class = abbrev}
\DeclareAcronym{RACH}{short = RACH ,long = random access channel, class = abbrev}
\DeclareAcronym{RCB}{short = RCB ,long = random coding bound, class = abbrev}
\DeclareAcronym{RFID}{short = RFID ,long = radio frequency identification, class = abbrev}
\DeclareAcronym{RLC}{short = RLC ,long = random linear coding, class = abbrev}
\DeclareAcronym{ROC}{short = ROC, long = receiver operating characteristics, class = abbrev}%
\DeclareAcronym{RTS/CTS}{short = RTS/CTS ,long = request to send / clear to send, class = abbrev}
\DeclareAcronym{RTT}{short = RTT ,long = round trip time, class = abbrev}
\DeclareAcronym{S-MIM}{short = S-MIM , long = S-band mobile interactive multimedia, class = abbrev}
\DeclareAcronym{SA}{short = SA , long = slotted ALOHA, class = abbrev}
\DeclareAcronym{SB}{short = SB ,long = Shannon bound, class = abbrev}
\DeclareAcronym{SC}{short = SC ,long = selection combining, class = abbrev}
\DeclareAcronym{SI}{short = SI ,long = shift invariant, class = abbrev}
\DeclareAcronym{SIC}{short = SIC ,long = successive interference cancellation, class = abbrev}
\DeclareAcronym{SICTA}{short = SICTA ,long = successive interference cancellation tree algorithm, alt = Successive Interference Cancellation Tree Algorithm, class = abbrev}
\DeclareAcronym{SNIR}{short = SNIR ,long = signal-to-noise and interference ratio, class = abbrev}
\DeclareAcronym{SINR}{short = SINR ,long = signal-to-interference and noise ratio, class = abbrev}
\DeclareAcronym{SNR}{short = SNR ,long = signal-to-noise ratio, class = abbrev}
\DeclareAcronym{TDMA}{short = TDMA ,long = time division multiple access, class = abbrev}
\DeclareAcronym{TI}{short = TI ,long = throughput invariant, class = abbrev}
\DeclareAcronym{UCP}{short = $\Code$-UCP ,long = $\Code$-unresolvable collision pattern, class = abbrev}
\DeclareAcronym{V2V}{short = V2V ,long = vehicle-to-vehicle, class = abbrev}
\DeclareAcronym{VDES}{short = VDES ,long = VHF data exchange system, class = abbrev}
\DeclareAcronym{VF}{short = VF ,long = virtual frame, class = abbrev}
\DeclareAcronym{WAVE}{short = WAVE ,long = wireless access in vehicular environments, class = abbrev}

\DeclareCaptionLabelSeparator{perioddoubletilde}{.~~}
\newcommand {\ie}{\hbox{i.e. }}
\newcommand {\eg}{\hbox{e.g. }}

\newcommand{\figw}{0.8\columnwidth}

\newcommand{\bmpg}[1]{\begin{minipage}{#1}}
\newcommand{\empg}{\end{minipage}}

\newlength{\betweenfiguresskip}
\newlength{\minorbetweenfiguresskip}
\renewcommand{\baselinestretch}{1.4}

\theoremstyle{plain}
\newtheorem{defin}{Definition}

\newtheoremstyle{exampstyle}
{\topsep}
{\topsep}
{}
{0pt}
{\bfseries}
{ --- }
{ }
{\thmname{#1}\thmnumber{ #2}\thmnote{ #3}}
\theoremstyle{exampstyle}
\newtheorem{exmp}{Example}[section]

\makeatletter
\def\cleardoublepage{\clearpage\if@twoside \ifodd\c@page\else
\hbox{}
\thispagestyle{empty}
\newpage
\if@twocolumn\hbox{}\newpage\fi\fi\fi}
\makeatother

\begin{document}
\newtheorem{thm}{Theorem}
\newtheorem{lemma}{Lemma}
\newtheorem{coroll}{Corollary}
\newtheorem{prop}{Proposition}
\newtheorem{definition}{Definition}
\newtheorem{example}{Example}

\newcommand{\mc}{\mathcal}
\newcommand{\ve}[1]{\mathbf{#1}}

\newfloat{EqAtPageBottom}{b}{}

\newcommand{\load}{\mathsf{G}}
\newcommand{\tp}{\mathsf{S}}
\newcommand{\tpLT}{\mathsf{S}_3}
\newcommand{\se}{\xi}
\newcommand{\plr}{\mathsf{p}_l}
\newcommand{\plrLT}{\mathsf{p}_{l3}}
\newcommand{\psucc}{\mathsf{p}_s}
\newcommand{\psuccLT}{\mathsf{p}_{s3}}
\newcommand{\psuccLTAppr}{\tilde{\mathsf{p}}_{s3}}
\newcommand{\rate}{\mathsf{R}}
\newcommand{\maxRate}{\mathsf{R}^*}
\newcommand{\mutInf}{\mathsf{I}}
\newcommand{\rvDec}{\mathcal{D}}

\newcommand{\capacity}{\mathsf{C}}
\newcommand{\capacitySC}{\capacity_\mathrm{S}}
\newcommand{\capacityMRC}{\capacity_\mathrm{M}}
\newcommand{\refCap}{\mathsf{C}_g}
\newcommand{\maxCap}{\mathsf{C}^*}
\newcommand{\maxSe}{\se^*}
\newcommand{\normCap}{\eta}

\newcommand{\dg}{\mathsf{d}}
\newcommand{\ratethr}{\mathsf{\bar{b}}}
\newcommand{\snrthr}{\mathsf{b}^{\star}}

\newcommand{\fraLen}{T_f}
\newcommand{\pkLen}{T_p}
\newcommand{\symLen}{T_s}
\newcommand{\slLen}{T_l}
\newcommand{\fraTp}{n_p}
\newcommand{\VFStart}{T}
\newcommand{\RefTime}{t_0}
\newcommand{\Wind}{W}
\newcommand{\WindSh}{\Delta \Wind}

\newcommand{\numUs}{n_u}
\newcommand{\numUsB}{n_b}
\newcommand{\numSym}{n_s}
\newcommand{\sw}{\mathrm{sw}}
\newcommand{\syncSym}{n_\sw}
\newcommand{\numSlot}{m_s}
\newcommand{\numBit}{k}

\newcommand{\rx}{y}
\newcommand{\rxTwo}{\tilde{\rx}}
\newcommand{\rxVec}{\bm{\rx}}
\newcommand{\rxrv}{Y}
\newcommand{\rxrvVec}{\bm{Y}}
\newcommand{\tx}{x}
\newcommand{\txVec}{\bm{\tx}}
\newcommand{\noise}{n}
\newcommand{\noisePr}{\nu}
\newcommand{\noiseVec}{\bm{\noise}}
\newcommand{\noiseVar}{\sigma_{\noise}^2}
\newcommand{\intVal}{z}
\newcommand{\intVec}{\bm{\intVal}}
\newcommand{\intVar}{\sigma_{\intVal}^2}
\newcommand{\noiseInt}{\tilde{n}}
\newcommand{\intNum}{m}
\newcommand{\intRV}{M}
\newcommand{\symVal}{a}
\newcommand{\pulse}{g}
\newcommand{\pulseRoot}{h}

\newcommand{\tm}{t}
\newcommand{\epoch}{\epsilon}
\newcommand{\freq}{f}
\newcommand{\phase}{\varphi}

\newcommand{\user}{u}
\newcommand{\usersRV}{U}
\newcommand{\replica}{r}
\newcommand{\sym}{s}
\newcommand{\symVec}{\bm{\sym}}

\newcommand{\reTx}{\tau}

\newcommand{\usPw}{\mathsf{P}}
\newcommand{\PwAvg}{\mathsf{\bar{P}}}
\newcommand{\PwAg}{\usPw_g}
\newcommand{\usPwTx}{\usPw_t}
\newcommand{\noisePw}{\mathsf{N}}
\newcommand{\noiseSD}{N_0}
\newcommand{\EnSym}{E_s}
\newcommand{\intPw}{\mathsf{Z}}
\newcommand{\intPwVec}{\mathbf{Z}}
\newcommand{\sinr}{\gamma}
\newcommand{\sinrVec}{\underline{\gamma}}

\newcommand{\UCP}{\mathcal{L}}
\newcommand{\UCPcons}{\mathcal{L}^*}
\newcommand{\AllUCP}{\mathcal{L}_\mathrm{S}}
\newcommand{\Code}{\mathcal{C}}
\newcommand{\CollClus}{\mathcal{S}}

\newcommand{\RStart}{\tau}
\newcommand{\VL}{\RStart_l^*}
\newcommand{\VR}{\RStart_r^*}
\newcommand{\Vg}{\RStart^*}
\newcommand{\Vpd}{T_v}
\newcommand{\nVp}{n_v}

\newcommand{\mU}{\alpha_u}
\newcommand{\sL}{\beta_d}
\newcommand{\UL}{\beta_{u-d}}

\newcommand{\intFreeAsy}{\varphi_a}
\newcommand{\intFreeMRC}{\varphi_m}
\newcommand{\intOne}{\mu}
\newcommand{\rateFree}{\rate_f}
\newcommand{\rateInt}{\rate_i}
\newcommand{\rateIntO}{\rate_{i1}}
\newcommand{\rateIntT}{\rate_{i2}}
\newcommand{\off}{\Delta}
\newcommand{\arrRate}{\lambda}
\newcommand{\maxWinSize}{\Delta\tau}
\newcommand{\finWin}{\tau}
\newcommand{\prob}{p}
\newcommand{\CRIL}{l}
\newcommand{\cntL}{c^l}
\newcommand{\cntS}{c^s}
\newcommand{\tsd}{\mathsf{j}}
\newcommand{\mx}{\bm{M}}
\newcommand{\vc}{\bm{v}}
\newcommand{\bs}{\bm{b}}
\newcommand{\bse}{b}
\newcommand{\lbs}{L}
\newcommand{\ts}{\tau}
\newcommand{\df}{d_f}

\newcommand{\TestThrOne}{\Psi^{(1)}}
\newcommand{\TestThrTwo}{\Psi^{(2)}}
\newcommand{\TestThrInt}{\tilde{\Psi}^{(1)}}
\newcommand{\ThrCorr}{\psi}
\newcommand{\ThrCorrInt}{\tilde{\psi}}
\newcommand{\HypAbove}{\mathcal{H}_1}
\newcommand{\HypBelow}{\mathcal{H}_0}
\newcommand{\Hyp}{\mathcal{H}}
\newcommand{\DecAbove}{\mathcal{D}_1}
\newcommand{\DecBelow}{\mathcal{D}_0}
\newcommand{\RStartSet}{\mathcal{R}}
\newcommand{\SlotSize}{\Delta T}
\newcommand{\nSlotsSmall}{n_p}
\newcommand{\probFalse}{P_F}
\newcommand{\probDet}{P_D}
\newcommand{\probCorComb}{P_{CC}}

\newcommand{\LDist}{f_X(x)}
\newcommand{\LDistDisc}{\overline{f}_X(x)}
\newcommand{\FragO}{O_f}
\newcommand{\PadO}{O_p}

\newcommand{\slotind}{c} 
\newcommand{\rateIrsa}{R}
\newcommand{\ld}{\lambda}
\newcommand{\EBdegDist}{\left\{\ld_\dg \right\}}
\newcommand{\Ld}{\Lambda}
\newcommand{\BNdegDist}{\left\{\Ld_\dg \right\}}
\newcommand{\davg}{\bar{\mathrm{d}}}
\newcommand{\dmax}{\mathrm{d}_{\mathrm{max}}}
\newcommand{\rd}{\rho}
\newcommand{\ESdegDist}{\left\{\rd_\slotind \right\}}
\newcommand{\Rd}{\mathrm{P}}
\newcommand{\SNdegDist}{\left\{\Rd_\slotind \right\}}
\newcommand{\thr}{\load^{\star}}
\newcommand{\exitb}{\mathsf{f}_{\mathsf b}}
\newcommand{\exits}{\mathsf{f}_{\mathsf s}}
\newcommand{\p}{\mathtt{p}}
\newcommand{\q}{\mathtt{q}}
\newcommand{\snr}{\mathsf{b}}
\newcommand{\snravg}{\bar{\mathsf{B}}}
\newcommand{\snrrv}{\mathsf{B}}
\newcommand{\power}{\mathsf{p}}
\newcommand{\snrdist}{f_{\snrrv}}
\newcommand{\powerdist}{f_{\usPw}}
\newcommand{\zt}{z_t}
\newcommand{\calB}{\mathcal{B}}
\newcommand{\calS}{\mathcal{S}_l}
\newcommand{\slot}{l}
\newcommand{\PLRtarget}{\bar{\mathsf{p}}_l}
\newcommand{\sDegJ}{\mathsf{r}}
\newcommand{\decStep}{\mathsf{t}}

\newcommand{\ulVal}{ul}
\newcommand{\dlVal}{dl}
\newcommand{\loadMax}{\load^{*}}
\newcommand{\tpMax}{\tp^{*}}
\newcommand{\tpSA}{\tp_{sa}}
\newcommand{\tpUL}{\tp_{\ulVal}}
\newcommand{\tpULMax}{\tpMax_{\ulVal}}
\newcommand{\tpULk}[1]{\tp_{\ulVal,#1}}
\newcommand{\tpULkMax}[1]{\tpMax_{\ulVal,#1}}
\newcommand{\tpDL}{\tp_{\dlVal}}
\newcommand{\tpetoe}{\tp_{e2e}}
\newcommand{\ulrk}[1]{\zeta_{#1}}
\newcommand{\nrx}{\mathsf{K}}
\newcommand{\rxEl}{\mathsf{k}}
\newcommand{\collRV}{\hat{C}}
\newcommand{\coll}{\hat{c}}
\newcommand{\collMax}{\coll^{*}}
\newcommand{\collK}[1]{\coll_{#1}}
\newcommand{\colVec}{\bm{\hat{c}}}
\newcommand{\colVecK}[1]{\colVec_{#1}}
\newcommand{\linll}{\hat{r}}
\newcommand{\linllK}[1]{\linll_{#1}}
\newcommand{\linComb}{\bm{\linll}}
\newcommand{\linCombK}[1]{\linComb_{#1}}
\newcommand{\MxComb}{\bm{M}}
\newcommand{\MxCombTwo}{\tilde{\bm{M}}}
\newcommand{\MxCombK}[1]{\MxComb_{#1}}
\newcommand{\fieldOrd}{\hat{q}}
\newcommand{\peras}{\varepsilon}
\newcommand{\capDL}{\mathsf{C_{\dlVal}}}
\newcommand{\numulslot}{m_{\ulVal}}
\newcommand{\mxNumCol}{\hat{b}}
\newcommand{\rank}{\hat{n}}
\newcommand{\rankK}[1]{\rank_{#1}}
\newcommand{\rankKMax}[1]{\rankK{#1}^{*}}
\newcommand{\mxCol}{\hat{h}}
\newcommand{\mxColK}[1]{\mxCol_{#1}}
\newcommand{\colNumVec}{\bm{\bar{c}}}
\newcommand{\rankNumVec}{\bm{\bar{n}}}
\newcommand{\evSlotSet}{\Omega}
\newcommand{\evSlot}{\omega}
\newcommand{\obsSlot}{X}
\newcommand{\setColPack}{\mc A}
\newcommand{\setRel}{\mc K}
\newcommand{\setCommPack}{\mc I}
\newcommand{\carI}{\hat{\iota}}
\newcommand{\evNotRx}{\mc O}
\newcommand{\dataPart}{\hat{W}}
\newcommand{\dataSet}{\mc{W}}
\newcommand{\obsPk}{\bm{\hat{w}}}
\newcommand{\obsPkTwo}{\bm{\tilde{w}}}
\newcommand{\rateDw}{\hat{\mathsf{R}}}
\newcommand{\erPos}{\hat{\mathsf{E}}}
\newcommand{\entropy}{H}
\newcommand{\Prob}{P}
\newcommand{\setUnique}{\mc P}
\newcommand{\setVar}{\hat{\setRel}}
\newcommand{\setInterest}{\hat{\mc{S}}}
\newcommand{\eventInterest}{\hat{\mc{E}}}
\newcommand{\dropVec}{\hat{\bm{p}}}
\newcommand{\drop}{\hat{p}}
\newcommand{\cardSpace}{Q}
\newcommand{\coeffa}{\hat{a}}
\newcommand{\coeffb}{\hat{b}}
\newcommand{\coeffc}{\hat{c}}
\newcommand{\coeffd}{\hat{d}}
\newcommand{\coeffe}{\hat{e}}
\newcommand{\rateDwSet}{\hat{\mc{R}}}

\definecolor{gl}{rgb}{0.0,0.5,0.8}
\definecolor{fc}{rgb}{0.8,0.5,0}
\definecolor{al}{rgb}{1,0.3,0.3}
\newcommand{\giangio}{\textcolor{gl}}
\newcommand{\fede}{\textcolor{fc}}
\newcommand{\alert}{\textcolor{al}}


\titlepage
\thispagestyle{empty}

\begin{center}
\ifpdf
    \begin{figure}[h!]
    \centering
    \includegraphics[width=3cm]{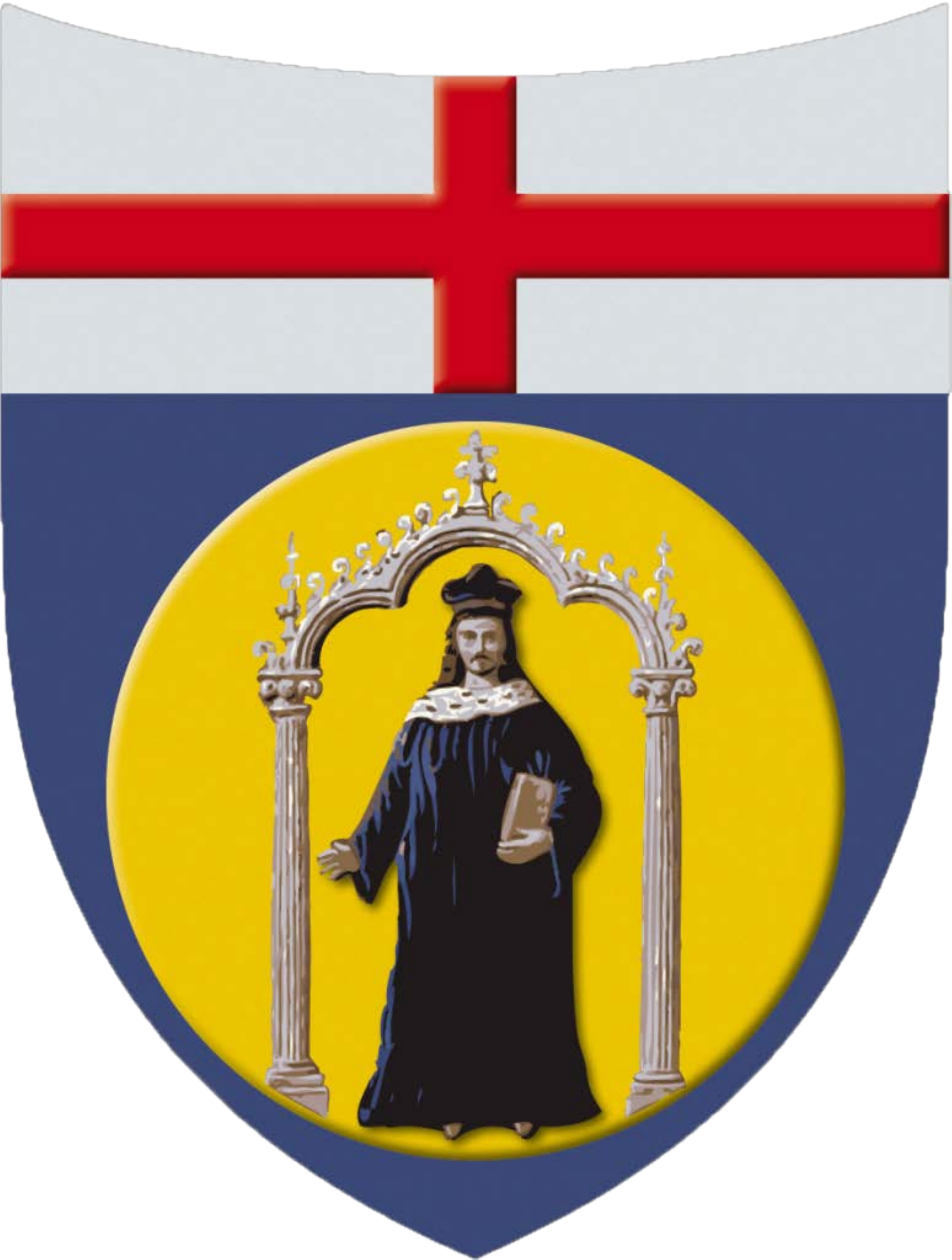}
    \end{figure}
\fi
\huge{\textsc{University of Genoa}}
\end{center}
\hrule

\vspace{0.1cm}

\begin{center}
\large{\textsc{Department of Electrical, Electronic and Telecommunication Engineering and Naval Architecture (DITEN)}}
\end{center}

\vspace{0.1cm}

\begin{center}
\large{\textsc{PhD in Science and Technology for Electronic and Telecommunication engineering}}
\end{center}
\vspace{1.0cm}

\begin{center}
\huge{Modern Random Access \\ for Satellite Communications}
\end{center}

\vspace{1.0cm}

\begin{center}
\large{PhD Thesis}
\end{center}

\begin{center}
\large{Federico Clazzer}
\end{center}

\begin{center}
\large{Genova, April 2017}
\end{center}

\vspace{0.8cm}

\begin{minipage}{\textwidth}
\begin{large}
	\begin{tabular}{@{} l @{\hspace{6.6cm}} r @{}}
		\emph{Tutor:} & \emph{Co-Tutor:} \\
		Prof.~Mario Marchese  & Dr.~Gianluigi Liva\\
         & \emph{Supervisor:} \\
         & Dr.~Andrea Munari \\
	\end{tabular}
\end{large}
\end{minipage}

\vspace{1.5cm}

\hrule
\begin{center}
\large{\emph{Coordinator of the PhD Course: Prof. Mario Marchese}}
\end{center}



\frontmatter

\dominitoc
\dominilof


\tableofcontents

\addcontentsline{toc}{chapter}{List of Figures}
\listoffigures	

\addcontentsline{toc}{chapter}{List of Tables}
\listoftables  

\renewcommand{\chaptermark}[1]{\markboth{LIST OF ABBREVIATIONS}{}}
\chapter{List of Abbreviations}
\printacronyms[include-classes=abbrev] 

\renewcommand{\chaptermark}[1]{\markboth{#1}{}}
\setcounter{secnumdepth}{-1}
\chapter{Acknowledgments - Ringraziamenti}
\thispagestyle{empty}

Dopo tutti questi anni di studio vorrei spendere alcune parole sulla persone che mi sono state vicine, mi hanno incoraggiato, mi hanno aiutato, dato indicazioni e suggerimenti.

Prima di tutto vorrei esprimire profonda gratitudine per il mio relatore, Prof. Mario Marchese, per avermi accettato come suo studente e il suo supporto durante tutto il dottorato. Le sue indicazioni sono stati fondamentali.

Vorrei ringraziare Gianluigi, per il suo instancabile aiuto, per i suoi suggerimenti e per la sua guida. Sono innumerevoli le cose che ho imparato in questi anni, e moltissime le devo a lui. A tutto questo si aggiunge che in lui ho trovato non solo un mentore ma anche un amico.

Vorrei ringraziare Andrea, per le lunghe chiacchierate, sulla tesi, sui nostri articoli e non solo. Se mi sono appassionato alla ricerca e' anche grazie a lui (o per colpa sua, dipende da come la si vede).

I wish to express my gratitude to Prof. Alex Graell i Amat and Prof. Petar Popovski for their competent feedback during the review phase of the thesis. The quality of thesis has profoundly improved thank to them.

I would like to thank all my present and past colleagues \--- or better friends \--- at the Satellite Network department at the German Aerospace Center (DLR). Thanks to Sandro, Matteo and Hermann to let me come to DLR for my Master's thesis and for offering me a position later. Without them, this thesis would not be here. Thanks to Christian, for his patience in the early stage of my research. Thanks to our coffee break and lunch break \emph{gang}, (rigorously in alphabetical order, and hoping not to forget anyone), Alessandro, Andrea, Andreas, Benny, Cristina, Estefania, Fran, Giuliano, Giuseppe, Jan, Javi, Maria-Antonietta, Romain, Stefan, Svilen, Thomas, Tudor and Vincent. All the laughes we had together, helped me to start every day, even to most heavy one, with a smile on the face. A special thanks to my former officemate Tom. And a very special thanks to my current officemate Balazs. He helped me not only with the everyday issues, but also far beyond!

Un grazie particolare agli amici vicini e lontani. Loro sanno chi sono.

La fortuna di trovare un amico sincero, come un vero fratello, \`e senza eguali. Anche se la distanza \`e grande e gli impegni nella vita diventano sempre di pi\`u, lui sar\`a sempre al tuo fianco pronto a farti sorridere della pi\`u piccola stupidaggine. Grazie Michel.

Senza i miei genitori, non sarei diventato ci\`o che sono, e questa tesi non sarebbe qui. Grazie a mia mamma per il suo amore e per avermi insegnato il valore della perseveranza.
Grazie a mio pap\`a per il suo amore e per avermi insegnato cosa vuol dire appassionarsi.

Grazie a mia nonna, che non ha mai smesso di vedermi come un bambino ma anche che \`e sempre orgogliosa di chiamarmi Ingegnere.

Per ultima anche se dovrebbe essere in cima alla lista, l'amore della mia vita, Viviana. Quanta strada abbiamo fatto da quando ci siamo conosciuti. Siamo cresciuti insieme, siamo cambiati e il nostro amore si \`e rafforzato. In te ho trovato non solo un'amica ma anche una confidente.

Mi scuso se per mia mancanza, ho dimenticato qualcuno. Non me ne vogliate.

\begin{flushright}
Federico\\
\vspace{1cm}
Germering\\February 2017
\end{flushright} 
\chapter{Abstract}
\thispagestyle{empty}

The exponential increase in communication-capable devices, requires the development of new and highly efficient communication protocols. One of the main limitation in nowadays wireless communication systems is the scarcity of the frequency spectrum. Moreover, the arise of a new paradigm \--- called in general \ac{M2M} \--- changes the perspective of communication, from human-centric to machine-centric. More and more autonomous devices, not directly influenced by humans, will be connected and will require to exchange information. Such devices can be part of a sensor network monitoring a portion of a smart grid, or can be cars driving on highways. Despite their heterogeneous nature, these devices will be required to share common frequency bands calling for development of efficient \ac{MAC} protocols.

The class of \ac{RA} \ac{MAC} is one possible solution to the heterogeneous nature of \ac{M2M} type of communication systems. Although being quite simple, these protocols allow transmitters to share the medium without coordination possibly accommodating traffic with various characteristics. On the other hand, classical \ac{RA}, as ALOHA or \ac{SA}, are, unfortunately, much less efficient than orthogonal multiple access, where collisions are avoided and the resource is dedicated to one terminal. However, in the recent years, the introduction of advanced signal processing techniques including interference cancellation, allowed to reduce the gap between \ac{RA} and orthogonal multiple access. Undubitable advantages, as very limited signaling required and very simple transmitters, lead to a new wave of interest.

The focus of the thesis is two-folded: on the one hand we analyse the performance of advanced asynchronous random access systems, and compare them with slot synchronous ones. Interference cancellation is not enough, and combining techniques are required to achieve the full potential of asynchronous schemes, as we will see in the thesis. On the other hand, we explore advanced slot synchronous \ac{RA} with fading channels or with multiple receivers, in order to gain insights in how their optimisation or performance are subject to change due to the channel and topology.
\chapter{Sommario}
\thispagestyle{empty}

L'aumento esponenziale di terminali in grado di comunicare, richiama l'attenzione sulla necessit\`a di sviluppare nuovi ed efficienti protocolli di comunicazione. Una delle principali limitazioni negli attuali sistemi di comunicazioni radio \`e la scarsit\`a di banda allocata, che rende di primaria importanza la possibilit\`a di allocare nella stessa frequenza il maggior numero di transmettitori concorrenti. In questo scenario, si presenta la nascita di un nuovo paradigma di comunicazione, chiamato comunemente \ac{M2M} communication, che \`e rivoluzionario in quanto sposta il punto di vista, centrandolo sulla macchina invece che sull'inidividuo. Un numero sempre crescente di oggetti autonomi, quindi non direttamente sotto l'influenza o il comando umano, saranno connessi e avranno la necessit\`a di comunicare con l'esterno o tra loro. Esempi possono essere la rete di sensori che monitora la smart grid, oppure autoveicoli in movimento nelle citt\`a o sulle autostrade. Nonstante la loro natura eterogenea, tutti questi terminali avranno la necessit\`a di condividere le stesse radio frequenze. Alla luce di questo, nuovi ed efficienti protocolli \ac{MAC} dovranno essere sviluppati.

I protocolli \ac{RA} per il \ac{MAC} sono una possibile risposta alla natura eterogenea delle reti \ac{M2M}. Pur essendo relativamente semplici, questa classe di protocolli permette di condividere la frequenza senza la necessit\`a di coordinamento tra i terminali, e permettendo la trasmissione dati di varia lunghezza o con caratteristiche differenti. D'altra parte, schemi classici \ac{RA}, come ALOHA e \ac{SA}, sono molto meno efficienti di schemi che evitano collisioni e che riservano le risorse per un unico transmittitore alla volta, come \ac{TDMA}, ad esempio. Negli ultimi anni, inoltre, l'introduzione di tecniche evolute di processamento del segnale, come la cancellazione di interferenza, ha permesso di ridurre la distanza di prestazioni tra queste due classi di protocolli. Chiari vantaggi dei protocolli \ac{RA} sono la ridottissima necessit\`a di scambio di meta dati e la ridotta complessit\`a dei transmettitori. Per questi motivi, un rinnovato interesse da parte della comunit\`a scientifica si \`e manifestato recentemente.

L'obiettivo della tesi \`e duplice: da un lato ci occuperemo di analizzare il comportamento di protocolli \ac{RA} asincroni e li confronteremo con quelli sincroni a livello di slot. L'utilizzo di tecniche di cancellazione di interference non reiscono da sole a sfruttare l'intera potenzialit\`a di questi protocolli. Per questo motive introdurremo tecniche di combining nel nuovo schema chiamato \ac{ECRA}. Ci occuperemo inoltre, di valutare l'impatto di nuovi modelli di canale con fading o di topoligie con ricevitori multipli su protocolli sincroni a livello di slot. L'analisi teorica sar\`a supportata mediante simulazioni al computer, per verificare i benefici delle scelte effettuate.


\mainmatter
\renewcommand{\chaptermark}[1]{\markboth{#1}{}}
\setcounter{secnumdepth}{-1}
\chapter{Introduction}
\thispagestyle{empty}
\epigraph{The most difficult thing is the decision to act, the rest is merely tenacity}{Amelia Earhart}
\acresetall

The problem of sharing the wireless medium amongst several transmitters \cite{Massey1981} has its roots in the origin of radio communication. The problem can be formulated with the help of a toy example. Let us assume that two transmitters would like to send packets to a common receiver sharing the wireless medium. A first option would be to equally split the time between the two transmitters, dedicating the first half to the transmission of the first user and the second to the second user. Similarly, we could split the bandwidth into two equally spaced carriers and dedicate one to each of the transmitters. A second option, valuable if the transmitters have seldom packets to be sent, is to let them access the wireless medium with their packets whenever they are generated, regardless of the medium activity. If a collision happens, the receiver notifies the transmitters which, after a random time, will retransmit their packets. The two presented options correspond to two \ac{MAC} classes: in the first class a dedicated, non-interfering resource is assigned to each of the transmitters, while in the second class transmitters are allowed to access the medium whenever a data packet is generated, no matter what the medium activity is. In the former class fall paradigms like \ac{TDMA}, \ac{FDMA} and \ac{CDMA} \cite{MAC_1990}, where the orthogonal resources assigned to the transmitters are time slots, frequency carriers or code sequences respectively. To the latter class belong all \ac{RA} solutions, from ALOHA \cite{Abramson1970} and \ac{SA} \cite{Roberts1975} to carrier sensing techniques as \ac{CSMA} \cite{Data_Networks_1987}. While orthogonal schemes are efficient when the scenario is static, \ie the number of transmitters and their datarate requirements are constant, \ac{RA} is particularly indicated for very dynamic situations as well as when a very large transmitter population features a small transmission duty cycle \cite{Data_Networks_1987} (the transmission time is very small compared to the time when the transmitter remains idle). The latter case is when the number of active transmitters compared to the total population is very small and time for transmitting a packet is rather limited compared to the time where the transmitter remains idle. A third class of \ac{MAC} protocols is \ac{DAMA} \cite{Abramson1994}. Users ask for the resources to transmit in terms of time slots or frequency bands to a central entity, which dynamically allocates them. Once the users receive the notification on how to exploit the resources, they are granted interference free access to the channel, as for orthogonal schemes. These protocols are effective when the amount of data to be transmitted largely exceeds the overhead required to assign the resources, being the overhead typically some form handshake procedure. In all the cases where the packets to be transmitted are small compared to the handshake messages, \ac{DAMA} becomes inefficient.

The increasing demand of efficient solutions for addressing the concurrent access to the medium of heterogeneous terminals in wireless communications calls for development of advanced \ac{MAC} solutions. Two main research fields are emerging: on the one hand very large or continuous data transmissions serving streaming services, and on the other hand, small data transmission supporting recent paradigms like \ac{IoT} and \ac{M2M} communications \cite{Wu_2011}. In the latter scenarios, \ac{RA} is an effective solution, but old paradigms like ALOHA or \ac{SA} are not able to meet the requirements in terms of throughput and \ac{PLR}. Nonetheless, in the recent past these protocols have been rediscovered and improved thanks to new signal processing techniques. Interestingly, one of the most recent research waves in this field has come from the typical application scenario of \ac{RA}, \ie satellite communications. This peculiar application scenario, in fact, prevents the use of carrier sensing techniques due to the vast satellite footprint that allows transmitters hundreds of kilometer far apart to transmit concurrently but precludes the sensing of mutual activity. The key advancement of the \ac{RA} solution proposed leverages on \ac{SIC} that can be included in the \ac{MUD} techniques \cite{MUD_1998}. The possibility to iteratively clean the received signal from correctly received packets reduces the occupation of the channel and possibly leads to further correct receptions.

Interestingly, several other application scenarios other than satellite communications have recently shown that \ac{RA} protocols can be beneficial. When for instance, there are tight constraints on the transmitter complexity and cost, \ac{RA} protocols can be one of the most promising choices due to their extremely simple transmitters. Some typical examples are vehicular ad hoc networks \cite{Menouar2006} and \ac{RFID} communication systems \cite{He2013}. \ac{RA} is also beneficial when the \ac{RTT} of the communication system is significantly large, \ie several milliseconds or more as in the case of satellite communication \cite{Peyravi1999} and underwater communication \cite{Pompili2009}. In this way the users are allowed to immediately transmit the information without the need of waiting for the procedure of requesting and getting the resource allocated typical of \ac{DAMA}, which has a minimum delay of $1$ \ac{RTT}.

In this regard, the focus of the thesis is on advanced \ac{RA} solutions targeting satellite communications as well as small data transmissions, \eg \ac{M2M} and \ac{IoT} \cite{Laya2014}. We start with a brief review of the basic ALOHA and \ac{SA} protocols followed by some key definitions in the second Chapter. Chapter~\ref{chapter3} aims at highlighting the beneficial role of \ac{IC} in four different classes of \ac{RA} protocols, \ie asynchronous, slot synchronous, tree-based and without feedback. Interestingly, in all the last three cases it has been demonstrated that under the collision channel (see Chapter~\ref{chapter2} for the definition) the limit
of $1$ [pk/slot] in throughput can be achieved when \ac{SIC} is adopted.\footnote{Nonetheless, not all schemes achieving this limit are practical. More details can be found in Chapter~\ref{chapter3}.} Time slotted \ac{RA} schemes have always been preferred in literature because of two main reasons: they present better performance compared to their asynchronous counterparts and they are easier to treat from an analytical perspective. Bearing this in mind, Chapter~\ref{chapter4} tries to answer the following two key questions:
\begin{itemize}
\item Are there any techniques that allow asynchronous \ac{RA} schemes to compete with the comparable time-slotted ones?
\item Is there a way to analyse systematically such schemes?
\end{itemize}
The first question finds an answer in the first Chapter, making use of \ac{ECRA}, which adopts time diversity at the transmitter and \ac{SIC}, together with combining techniques at the receiver. Under some scenarios, we show that \ac{ECRA} outperforms its slot synchronous counterpart. The second question also finds a partial answer in the same Chapter. An analytical approximation of the \ac{PLR} is derived, which is shown to be tight under low to moderate channel load conditions. Combining techniques require perfect knowledge of all the packet positions of a specific transmitter before decoding. A practical suboptimal solution, which has a very limited performance loss compared to the ideal case, is provided in the second part of Chapter~\ref{chapter4}, .

Recent \ac{RA} protocols including \ac{ECRA} are particularly efficient and some of them have been proposed in recent standard as options also for data transmission. Nevertheless, a comparison between these schemes at higher layers is still missing. To this regard, in Chapter~\ref{chapter5} we aim at addressing the following question:
\begin{itemize}
\item If the packets to be transmitted follow a generic packet duration distribution, possibly exceeding the duration of a time slot, what are the throughput and \ac{PLR} performance of time-slotted compared to asynchronous \ac{RA} schemes?
\end{itemize}
In the case of asynchronous \ac{RA} packets of any duration can be transmitted through the channel, while for time-slotted schemes fragmentation and possbily padding shall take place. For time-slotted scheme we provide an analytical expression of both throughput and \ac{PLR} at higher layers (compared to the \ac{MAC}) so to avoid specific higher layer simulations. An approximation which does not require the perfect knowledge of the packet distribution, but requires just the mean, is also presented.

Along a similar line, in Chapter~\ref{chapter6} we will move to time-slotted \ac{RA} where nodes can send an arbitrary number of replicas of their packets. There the \ac{PMF} is optimised in order to maximise the throughput. The typical channel model considered for the repetition \ac{PMF} optimisation of such schemes is the collision channel, which may not be sufficiently precise in many situations. Therefore, we answer these two key questions:
 \begin{itemize}
 \item How does the analytical optimisation change when considering fading channels and the \emph{capture effect} (see Chapter~\ref{chapter2} for the definition)?
 \item What is the impact of mismatched optimisation over a collision channel, when the system works over a fading channel?
 \end{itemize}
In particular, the Chapter focuses on Rayleigh block fading channel.

\begin{sloppypar}
Finally, in Chapter~\ref{chapter7} we expand the considered topology including multiple non-cooperative receivers and a final single receiver, in the uplink and downlink respectively. In future satellite networks there will be hundreds or even thousands of flying satellites \cite{oneWeb} resulting in more than one satellite at each instant in time possibly receiving the transmitted signal. We let the transmitters use \ac{SA} over a packet erasure channel to transmit to the relays in the uplink. In the downlink, \ie from the relays to the final destination - called also \ac{gateway} - \ac{TDMA} is adopted. Two options are investigated from an analytical perspective: \ac{RLC}, in which every relay creates a given number of linear combinations of the received packets and forwards them to the relay; dropping policies, where packets are enqueued with a certain probability, that may depend on the uplink slot observation (whether a packet suffered interference or not). For the latter, several options are analysed. It is proven that \ac{RLC} is able to achieve the upper bound on the downlink rate, nonetheless it requires an infinite observation window of the uplink. The key novel question answered in this Chapter is:
 \begin{itemize}
 \item What is the impact of finite buffer memory on the performance of \ac{RLC} compared to the dropping policies?
 \end{itemize}
We shall note that this scenario implicitly imposes a finite observation window for the \ac{RLC}. An analytical model employing two Markov chains is developed to compare the \ac{RLC} performance under finite buffer size with the dropping policies undergoing the same limitations. It is shown that for limited downlink resources, dropping policies are able to largely outperform \ac{RLC}, despite their simplicity.
\end{sloppypar}

Finally, in the conclusion Chapter the main results of the thesis are summarised together with the suggestion of some future insightful research directions. 
\setcounter{secnumdepth}{2}
\renewcommand{\chaptermark}[1]{\markboth{\chaptername\ \thechapter.\ #1}{}}
\chapter{Basics of Random Access - ALOHA and Slotted ALOHA}
\label{chapter1}
\acresetall
\ifpdf
    \graphicspath{{chapter1/Images/}}
\fi

\thispagestyle{empty}
\epigraph{Computers are only capable of a certain kind of randomness because computers are finite devices}{Tristan Perich}

The focus of the first Chapter is on the basics of \ac{RA} protocols. An overview of the channel access policies of the two pioneering protocols, ALOHA invented by Abramson \cite{Abramson1970} and \ac{SA} invented by Roberts \cite{Roberts1975} as well as their performance are given. To introduce them there are several possible ways and we would like to mention two: the one presented in the book of Bertsekas and Gallager \cite{Data_Networks_1987} and the one in the book of Kleinrock \cite{Kleinrock1976_book}.

Two versions of the original ALOHA and \ac{SA} protocols can be found in the literature. The first option allows retransmissions upon unsuccessful reception, caused by collisions with other concurrent transmissions. The transmitters enable the receiver to detect a collision via the \ac{CRC} field appended to the packet, while the receiver notifies them via feedback about the detected collision. The second option instead does not consider retransmissions and let the protocols operate in \emph{open loop}. Such configuration relies on possible higher layer for guaranteeing reliability, if necessary. Throughout the description of the two protocols, we will consider both options highlighting the differences.

\section{ALOHA}

The multiple access protocol ALOHA invented by Abramson \cite{Abramson1970} has been primarily developed as a strategy to connect various computers of the University of Hawaii via radio communications \cite{Abramson1985}. In the early 1970s there was a need to provide connectivity to terminals of the University distributed on different islands and ALOHA has been the proposed solution.\footnote{It is of undeniable interest for the curious reader the anecdotal account of the development of the ALOHA system given by Abramson in \cite{Abramson1985}, where system design challenges are pointed out as well as the practical impact that such a system had in the 70s to both radio and satellite systems.}

The main idea is to a population of nodes\footnote{We will use the words node, user and transmitter interchangeably throughout the thesis.} transmit a packet to the single receiver whenever it is generated at the local source, regardless of the medium activity. In particular, each node, upon generation of the packet, transmits it immediately. If a collision occurs, the receiver detects the collision via checking the \ac{CRC} field of the decoded packets and all collided users involved are considered lost. At this point in time, depending whether retransmissions are enabled or not, two different behaviors are followed. In the latter case (no retransmissions), the packets are declared lost and no further action is taken by the transmitters. In the former case, instead (with retransmissions), the receiver feeds back a notification of the occurred collision to the users involved. The collided nodes, after a random interval of time, retransmit the packets. From the time when the node realises that a collision occurred until the retransmission the node is said to be \emph{backlogged}.

\ifpdf
\begin{figure}
\centering
\includegraphics[width=\columnwidth]{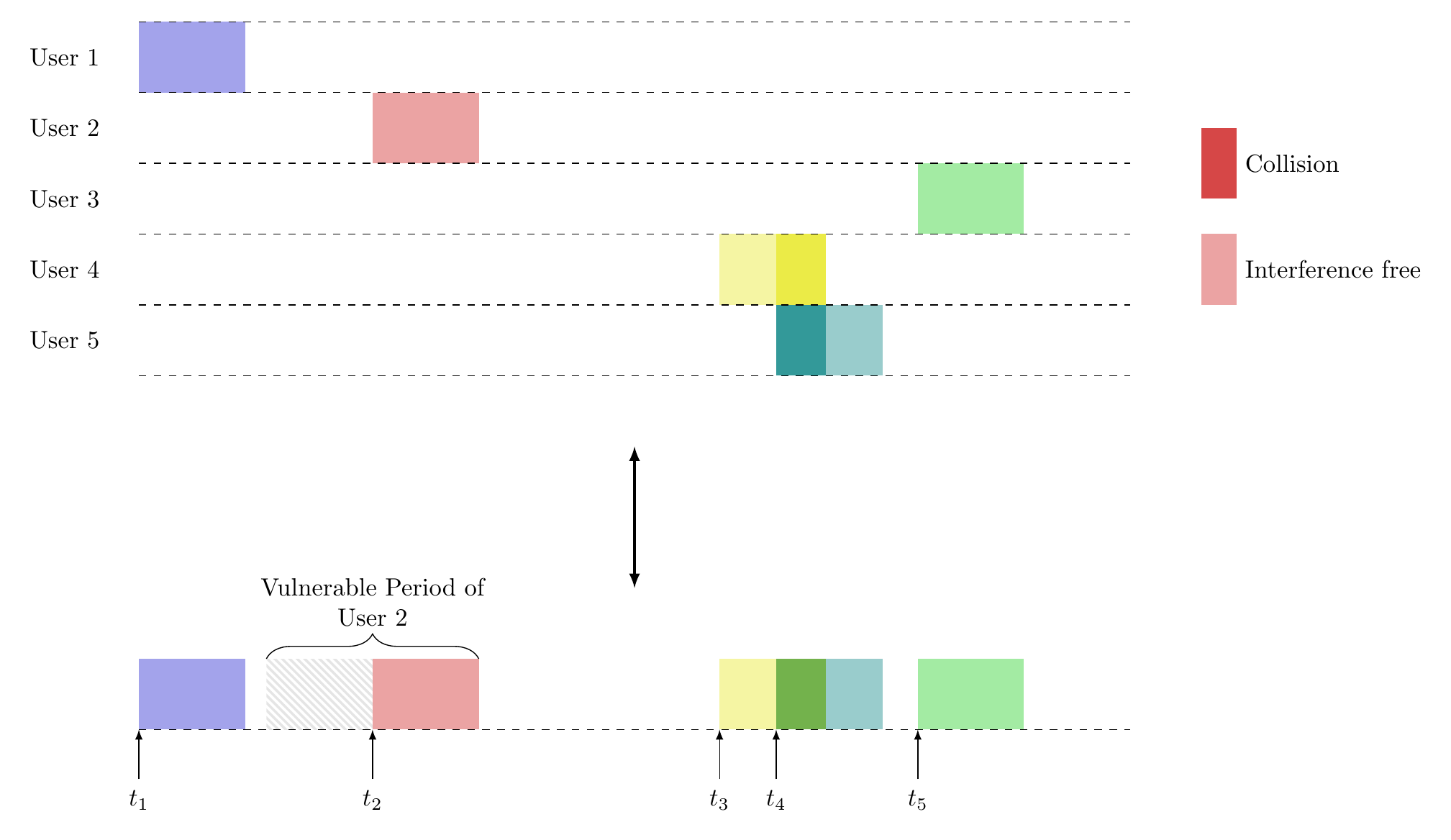}
\caption{A portion of received signal for the ALOHA protocol. We visualize the users both on the same time axes (bottom) as well as separated (top). The latter display mode will be used for many protocols throughout the thesis.}
\label{fig:ALOHA}
\end{figure}
\fi

When a node starts transmitting at time instant $\tm_1$ and assuming that the duration of each transmission is normalised to $1$, every other transmission in the interval $\left[\tm_1-1;\tm_1+1\right]$ will cause a collision (see Figure~\ref{fig:ALOHA}). This interval is also called \emph{vulnerable period} \cite{Kleinrock1976_book}.
\begin{defin}[Vulnerable Period]
For a reference packet, the vulnerable period is the interval of time in which any other start of transmission causes a destructive collision.
\end{defin}
In ALOHA the vulnerable period for any packet equals to two packet durations. A common assumption is that there is an infinite population of nodes and every new arrival is associated to a new node. Such assumption can be casted to any finite number of transmitters setting, associating to each transmitter a set of virtual nodes of the infinite \hyphenation{po-pu-la-tion}. Nonetheless, a different behaviour can be expected in the two scenarios. When a finite population is considered, packets arriving at the same node are forced to be sent in non-overlapping intervals of time, which is not the case for the assumption of infinite population. In that case, in fact, the virtual nodes act independently and multiple packets can be transmitted in overlapping instants of time. Since the typical application scenario of \ac{RA} is a large set of transmitters with low duty cycle, the infinite node population is rather accurate also as approximation for the finite population. The transmission of packets from the entire population is modeled as a Poisson process of intensity $\load$ $\left[\mathrm{packets/packet\, duration}\right]$ or in the following $\left[\mathrm{pk/pk\, duration}\right]$.\footnote{In other words, the binomial distribution can be well approximated with a Poisson under specific conditions.}
\begin{defin}[Channel load $\load$]
The channel load $\load$ is the expected number of transmissions per packet duration.
\end{defin}
In order to correctly receive a packet starting at time $\tm_1$, no other transmission shall occur in the time interval $\left[\tm_1-1;\tm_1+1\right]$. The probability of this event is equivalent to the probability that no further transmission occurs in two packet durations \ie under the Poisson assumption,
\begin{equation}
\label{eq:psucc_ALOHA}
\psucc=e^{-2\load}.
\end{equation}
We can now define the throughput as
\begin{defin}[Throughput]
The expected number of successful transmissions per packet duration is called throughput, and is given by
\begin{equation}
\label{eq:throughput}
\tp(\load)=\load\psucc \quad \left[\mathrm{pk/pk\, duration}\right].
\end{equation}
\end{defin}
Therefore, for ALOHA inserting equation~\eqref{eq:psucc_ALOHA} into equation~\eqref{eq:throughput} yields,
\begin{equation}
\label{eq:t_aloha}
\tp(\load)=\load e^{-2\load} \quad \left[\mathrm{pk/pk\, duration}\right].
\end{equation}

We now derive the channel load for which the peak throughput of ALOHA is found. We compute the derivative of the throughput\footnote{We denote with $f'(x)$ the derivative of $f(x)$.} as a function of the channel load $\load$ as,
\begin{equation}
\label{eq:der_T_ALOHA}
\tp'(\load)=e^{-2\load}\left(1-2\load\right).
\end{equation}
Setting it equal to zero gives
\begin{equation}
\label{eq:peak_G_ALOHA}
\tp'(\load)=0\Rightarrow1-2\load=0\Rightarrow\load=\frac{1}{2}.
\end{equation}
By noting that for $0\leq \load <1/2$ the derivative is positive while for $\load > 1/2$ it is negative, we are ensured that this is the only maximum of the function. Substituting the value of $\load=1/2$ in equation~\eqref{eq:t_aloha} we obtain the peak throughput of
\begin{equation}
\label{eq:peak_T_ALOHA}
\tp(0.5)=\frac{1}{2e}\cong0.18 \quad \left[\mathrm{pk/pk\, duration}\right].
\end{equation}

\section{Slotted ALOHA}

The first and most relevant evolution of ALOHA has been \ac{SA} invented by Roberts \cite{Roberts1975}. It has to be noted that the main impairment to successful transmission in ALOHA is coming from interference. Whenever two packets collide, even partially, they are lost at the receiver.\footnote{This is true under the collision channel model.} In ALOHA, this is particularly detrimental because every transmission starting one packet duration before, till one after, the start of a reference packet can cause a destructive collision. In \ac{SA} this effect is mitigated introducing time slots. A common clock dictates the start of a time slot. Upon local generation of a packet, a node waits until the start of the upcoming slot before transmission. The time slot has a duration equal to the packet length. Although requiring additional delay for a packet transmission w.r.t. ALOHA, \ac{SA} reduces the vulnerable period from two to one packet duration. In fact, only packets starting in the same time slot cause a destructive collision.

\begin{sloppypar}
As previously, an infinite population of nodes is generating traffic modeled as a Poisson process of intensity $\load$ $\left[\mathrm{pk/pk\, duration}\right]$ which in this case can be also measured in $\left[\mathrm{packets/slot}\right]$ or $\left[\mathrm{pk/slot}\right]$ in the following. In contrast with ALOHA, the reference packet transmission starting at time $\tm_1$ (being $\tm_1$ the beginning of a slot) can be correctly received when no other transmission occurs in the time interval $\left[\tm_1;\tm_1+1\right]$ which corresponds to the time slot chosen for transmission. The probability of this event is equivalent to the probability that no further transmission occurs in one packet duration \ie
\begin{equation}
\label{eq:psucc_SA}
\psucc=e^{-\load}.
\end{equation}
Exploiting equation~\eqref{eq:throughput} we can write the throughput expression for \ac{SA} as
\begin{equation}
\label{eq:t_SA}
\tp(\load)=\load e^{-\load} \quad \left[\mathrm{pk/slot}\right].
\end{equation}
\end{sloppypar}

Similarly to ALOHA, we derive here the channel load for which the peak throughput of \ac{SA} can be found. We start from the derivative of the throughput
\begin{equation}
\label{eq:der_T_SA}
\tp'(\load)=e^{-\load}\left(1-\load\right)
\end{equation}
And it holds
\begin{equation}
\label{eq:peak_G_SA}
\tp'(\load)=0\Rightarrow1-\load=0\Rightarrow\load=1.
\end{equation}
Also in this case, the derivative is positive for $0\leq \load < 1$ and negative for $\load > 1$, confirming that $\load=1$ is the global maximum of the throughput function. Computing the throughput for $\load=1$ gives the well-known peak throughput of \ac{SA}
\begin{equation}
\label{eq:peak_T_SA}
\tp(1)=\frac{1}{e}\cong0.36 \quad \left[\mathrm{pk/slot}\right].
\end{equation}

\ifpdf
\begin{figure}
\centering
\includegraphics[width=0.8\columnwidth]{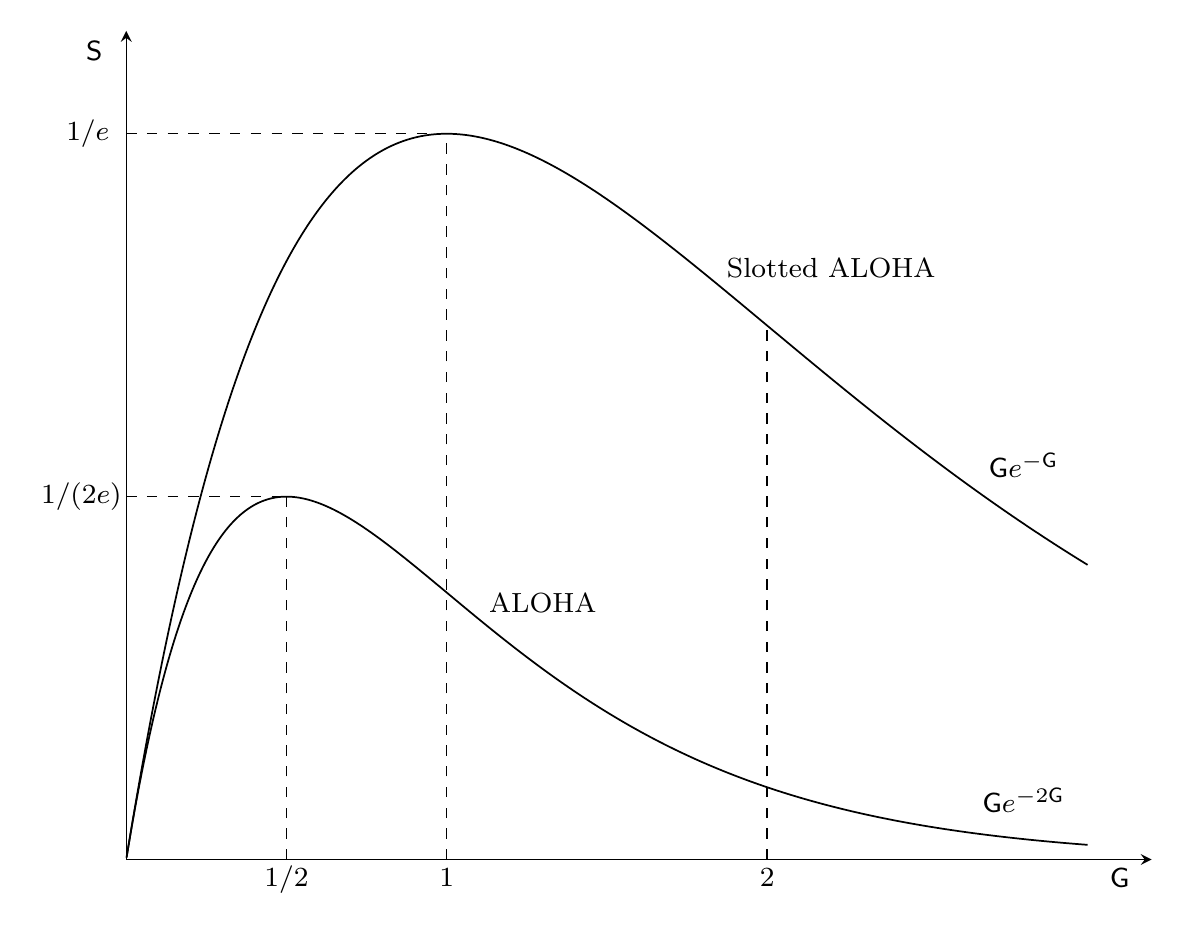}
\caption{Throughput comparison of ALOHA and \ac{SA}. Peak throughput values as well as inflection points and corresponding channel load values are highlighted.}
\label{fig:T_A_SA}
\end{figure}
\fi

\section{Considerations on ALOHA and Slotted ALOHA}

ALOHA-like \ac{RA} protocols have been used since the second half of the seventies in highly successful communications systems and standards, ranging from Ethernet \cite{Metcalfe1976}, to the Marisat system \cite{Marisat1977} that nowadays has become Inmarsat. Most recently they have been employed in mobile networks as in $2$G \ac{GSM} and \ac{GPRS} \cite{Cai1997} for signaling and control purposes, in $3$G using a modification called \ac{R-ALOHA} and $4$G for the \ac{RACH} of \ac{LTE} \cite{Laya2016}.

\Ac{RA} protocols are particularly attractive for all scenarios where the traffic is unpredictable and random, such as satellite return links and ad-hoc networks, just to mention a few. Unfortunately, the throughput performance of both ALOHA and \ac{SA} are quite limited but may allow transmission with lower delay compared to \ac{DAMA} schemes, if no collisions happened. This is particularly relevant for satellite applications, in which the request for resource allocation that precedes the transmission is subject to $1$ \ac{RTT} of delay. In \ac{GEO} satellite systems, this accounts for a delay of at least $500$ ms. Instead, when using \ac{RA} the delay can be drastically reduced.

Applications requiring full reliability (all transmitted packets are successfully received) of the packet delivery will operate ALOHA or \ac{SA} with retransmissions. In these scenarios, the analysis of ALOHA and \ac{SA} has to take into account the dynamism of the channel load due to the variation of backlogged users over time. An insightful analysis can be followed in \cite{Data_Networks_1987}, where a Markov chain model is developed and a drift analysis is given.  In the recent past, thanks to the emerging applications related to \ac{M2M} type of communications of the \ac{IoT} ecosystem the full reliability is not required anymore. In applications like sensor networks, metering applications, etc. in fact, the transmitted data is repetitive so if one transmission is lost, is not particularly dangerous as long as a minimum successful probability can be ensured.

Recent \ac{RA} protocols are able to drastically improve the throughput performance and to guarantee high successful reception probability for a vast range of channel loads. Furthermore, considering satellite communication systems, retransmissions will suffer of at least $1$ \ac{RTT} delay. In view of this, our focus will be on \ac{RA} without retransmissions.

\section{Other Fundamental Protocols in Random Access}

Before moving to the preliminaries, a brief review of the historical milestones is of utmost importance. This Section gives an overview on the different research paths that have been founded starting from the seventies and have developed and evolved until the latest days in the field of \ac{RA}.

As already noted, the approach followed by ALOHA and \ac{SA} poses a number of questions and one of the first to be addressed concerned the stability of the channel. We consider a channel access where retransmissions are allowed and newly generated traffic can also be transmitted over the channel. If the overall channel traffic exceeds a certain rate, more and more collisions will appear, triggering a vicious circle. An increasing amount of traffic is pushed into the network that will lower even more the probability of collision free transmissions leading to a zero throughput equilibrium point. Kleinrock and Lam in \cite{Kleinrock1975} observed that, assuming an infinite population of transmitters, the tradeoffs of the equilibrium throughput-delay are not sufficient for characterising the ALOHA system and also channel stability must be considered. In fact, on the throughput-delay characteristic, each throughput operating point below the channel capacity (which is $0.36$ in case of \ac{SA}) has two equilibrium solutions. This suggests that the assumption of equilibrium conditions may not always be valid. In this work they developed a Markovian model for \ac{SA} allowing the performance evaluation and design of \ac{SA}. They also introduced a new performance metric, the average up time, that represents a measure of the stability. The consequent problem was to define the dynamic control for the unstable \ac{SA}, which has been addressed by the same authors in \cite{Lam1975}. In this second work the authors present three channel control procedures and determine the optimal two-action control policy for each of them. In \cite{Carleial1975} the authors addressed the same problem for the finite user population and for the same three control policies. Thanks to the work of Jenq \cite{Jenq1980}, the question about the number of theoretically possible stable points for \ac{SA} is answered and has been found to be either one or three equilibrium points. The first and last one are stable and the second one is unstable. All the research works mentioned until now assume that the transmitters have no possibility to buffer their packets. Extensions to the work on the stability where this assumption is relaxed and infinite buffers are considered, has been done first in \cite{Tsybakov1979_Stability} where a simple bound for the stability region was obtained. Tsybakov and Mikhailov \cite{Tsybakov1979_Stability} derived the exact stability region for ALOHA in the case of $2$ users in the symmetric case, \ie all input rates and transmission probabilities are the same, applying stability criteria already derived for general cases by Malyshev in \cite{Malyshev1972}. Extensions to higher number of users have been derived by Mensikov in \cite{Mensikov1974} and Malyshev and Mensikov in \cite{Malyshev1981}. In \cite{Rao1988} the stability region for the case of finite user population with infinity buffer size is also addressed and for the case of $2$ users it is exactly determined using a different approach with respect to Tsybakov and Mikhailov, that introduces hypothetical auxiliary systems and simplifies the derivation. Furthermore, the relation and interaction between the two queues is explicitly shown. Applying the same approach for all the other cases, more tight inner bounds are derived\footnote{In the same years a number of other authors investigated the stability of the channel for \ac{SA} with different flavours. Works worth to be mentioned are, for example, \cite{Metcalfe1975, Ferguson1975, Jenq1981, Saadawi1981, Szpankowski1994}}.


Only a few years after the pioneering work on stability, Capetanakis \cite{Capetanakis1979} (work which derives from his PhD thesis \cite{Capetanakis1977}) and at the same time independently, Tsybakov and Mikhailov \cite{Tsybakov1978} opened a new a very productive research area in \ac{RA} protocols, the so-called \emph{collision-resolution algorithms} or splitting algorithms. These protocols are characterized by two operation modes, a normal mode which is normally \ac{SA} and a collision resolution mode. The latter is entered whenever a collision takes place and is exploited by the collided transmitters for retransmissions until the collision is resolved. In the collision mode all other transmitters are prevented from accessing the medium. In this way, the collision-resolution algorithms do not guarantee that newly arrived packets can be immediately transmitted, as was originally allowed by both ALOHA and \ac{SA}. During the collision resolution period the transmitters collided are probabilistically split into a transmitting set and an idle set. The algorithms differ in the rules applied for the split into the two sets during the collision resolution period as well as in the rules for allowing the packets not involved in the collision to transmit after the collision is resolved. The main achievements by the collision-resolution algorithms are a maximum throughput larger than $1/e$ and the proof of stability rather than operating hypothesis. In particular, the maximum stable throughput of the collision-resolution algorithm of Capetanakis, Tsybakov and Mikhailov is $0.429$, that have been extended by Massey \cite{Massey1981} to $0.462$ and by Gallager up to $0.487$ \cite{Gallager1985}. During the same years, there has been a lot of work also in finding upper bounds on the multiple access channel capacity, which has its tightest result in $0.568$ demonstrated by Tsybakov and Likhanov \cite{Tsybakov1980}\footnote{A good list of references on works dealing with collision-resolution algorithms and their maximum throughput as well as upper bounds on the capacity of the multiple access channel can be found in \cite{Chan2012}. A survey on Russian works on the topic, that have been particularly prolific, is given in \cite{Tsybakov1985}.}. Extensions to the case of more than two classes in which the colliding users are split has been investigated by Mathys and Flajolet in \cite{Mathys_1985} showing that ternary splitting (three classes) is optimal for most of the channel access policies. In their analysis they also considered the case in which users are not blocked during the collision resolution but can access the channel. They called this channel access policy \emph{free access protocol}.


Collision-resolution algorithms require a feedback channel to operate in order to actively resolve collisions. An orthogonal research direction has been to investigate what are the performances of \ac{RA} protocols when feedback is not possible. The first insightful investigation of this scenario was done by Massey and Mathys in \cite{Massey1985_RAWithoutFeedback}. The main outcome of their analysis is the fact that, surprisingly, the symmetric capacity, \ie all users adopt the same rate, equals to $1/e$ as in the case of \ac{SA} with feedback. Even more astonishing is that this result is achieved for both time slotted as well as for the asynchronous case. In this way, the simple ALOHA without feedback achieves a symmetric capacity of $1/e$. The approach proposed by Massey and Mathys is to associate to the users different access sequences, \ie slots where the user can send packets in the shared medium. At the same time the users encode their packets via erasure correcting codes. The receiver is able to retrieve a packet if a sufficient number of codeword segments is received without collisions (\ie time slots carry a single packet transmission). Crucial is therefore, a proper access sequence selection for ensuring that enough collision free segments can be received. Independently, having access only to the abstract and the presentation of Massey of the paper \cite{Massey1982_ISIT}, Tsybakov and Likhanov \cite{Tsybakov1983_RAWF} derived the capacity region under the assumption of maximum distance separable codes. Massey's approach, although rather simple, is subject to practical issues especially when the set of users accessing the medium becomes large and varying. Hui \cite{Hui1984} considers a more practical scenario where packets are recoverable only if the collisions in which they are involved are not affecting the vital part of the packet. This is normally the header, where all the physical layer related sections are placed for carrier, phase, time acquisition for example. Under this assumption the capacity attainable is reduced to $0.295$.


Although being already known from the 50s, bandwidth expansion through spreading techniques has been extensively investigated starting from the 70s \cite{Scholtz_1977, Kahn_1978}.\footnote{A precise and comprehensive definition of spread spectrum communication can be found in \cite{Pickholtz_1982}, while a worth reading overview of this class of techniques can be found in \cite{Pursley_1987}.} A first attempt to analyse the throughput-delay characteristic of \ac{RA} spread spectrum systems can be found in \cite{Raychaudhuri_1980}. Pursley \cite{Pursley_1986}, investigates frequency hopping in a satellite network scenario. One of the main outcomes of his work is that time synchronous and asynchronous spread spectrum random access have similar throughput performance, differently from their narrowband counterpart where a factor two is present. In particular, he shows that the throughput of slotted ALOHA spread spectrum is between the lower and upper bounds of the ALOHA spread spectrum throughput performance. Furthermore, the author conjectures that the throughput of the asynchronous system falls between the one of the time slotted system and the lower bound, although no proof is presented in the paper. Motivated by this, the authors in \cite{Polydoros_1987} investigate the time slotted random access spread spectrum scenario only and derive a generic analytical model of the throughput in such networks. Finally a qualitative comparison between \ac{CDMA} and spread spectrum ALOHA is presented in \cite{Abramson1994}, where advantages and disadvantages of the two approaches are highlighted.


During the same years, a conceptual enhancement of the \ac{RA} protocols brought to the idea of \ac{CSMA}. Assuming that the terminals have the possibility to sense the channel before transmitting, the throughput and delay performance can be improved. In \cite{Kleinrock1975a} the throughput and delay analytical performance are derived for three versions of the \ac{CSMA} protocol, non-persistent, $1$-persistent and p-persistent. The first two versions of the protocol have in common the behaviour when the channel is sensed idle, \ie in both cases the terminal transmits the packet with probability $1$. When the channel is sensed busy instead, in the non-persistent \ac{CSMA}, the terminal schedules the retransmission some time later and before transmitting senses again the channel and repeats the procedure, while the 1-persistent \ac{CSMA} persists to sense the channel until it is found idle and then it sends the packet. The p-persistent \ac{CSMA} instead, transmits with probability $p$ when the channel is idle and with probability $1-p$ delays the transmission for one slot. If the channel is busy instead, the terminal keeps sensing the channel until it becomes idle. The investigation assumed line-of-sight for all terminals and that they are all within the range of each other. Relaxation of these assumptions and a deep investigation of a very fundamental problem of the carrier sensing capability, the \emph{hidden terminal} problem, has been given in the subsequent paper of the same authors \cite{Tobagi1975}. The extension of the original \ac{CSMA} to embed collision detection, \ie \ac{CSMA/CD} given by Metcalfe in \cite{Metcalfe1976}, has been adopted for Ethernet. Another extension that permits collision avoidance, \ie \ac{CSMA/CA} given by Colvin in \cite{Colvin1983}, has been adopted for the \ac{MAC} of 802.11\footnote{In \cite{Kleinrock1976_book} a very didactic derivation of the thorughput and delay performance of \ac{CSMA} has been carried out by the author.} known also as Wi-Fi.


All the works presented until now rely on the destructive collision channel model, which can be over pessimistic in some scenarios. Due to the difference in the received power caused for example by difference in relative distance between transmitters and the receiver, packets colliding might be correctly received, \ie \emph{captured}. The capture effect \cite{Roberts1975} has been investigated already by Roberts in its pioneering work on \ac{SA}. Metzener in \cite{Metzener1976} showed that dividing the transmitters into two groups, one transmitting at high power and one at low power, could turn into a double of the maximum achievable throughput. Abramson in \cite{Abramson1977} derived a closed form solution of the throughput under capture in the special case of constant traffic density. Raychaudhuri in \cite{Raychaudhuri1981} presented a modification of \ac{SA} including \ac{CDMA} and as a consequence exploiting multi-packet reception. The outcome of the work is that the performance of \ac{CDMA}-\ac{SA} is similar to the \ac{SA} scheme, but multiaccess coding provides higher capacity and more graceful degradation. Ghez and her co-authors in \cite{Ghez1988}, introduced the \emph{multi-packet reception matrix} which is a very useful representation of the physical layer. Each row of the matrix represents a possible collision size, \ie the number of colliding packets, and each entry $\epsilon_{n,k}$ represents the probability that assuming that $n$ packets are transmitted, $k$ are successfully received. Exploiting this representation, the physical layer can be decoupled from the \ac{MAC} allowing elegant representation of the throughput. Zorzi and Rao in \cite{Zorzi1994} studied the probability of capture in presence of fading and shadowing considering the near-far effect and investigated the stability sufficient conditions in terms of the users spatial distribution. A very insightful overview of cross-layer approaches and multi-user detection techniques for \ac{RA} protocols has been given by Tong in \cite{Tong2004}.


\chapter{Preliminaries}
\label{chapter2}
\thispagestyle{empty}
\epigraph{Mathematics is the art of giving the same name to different things}{Henri Poincare}

\begin{sloppypar}
In this second Chapter the main ingredients that will be relevant in the upcoming Chapters are presented. We will first describe the scenario that will serve most of the \ac{RA} paradigms. The concept of time diversity is then introduced and a set of definitions for the channel models, \ac{SIC} and decoding conditions are presented.
\end{sloppypar}

\section{The Scenario}
\label{sec:scenario}

A population of users, potentially infinitely many, and among those only some are active at the same time, is assumed. They are sharing a common communication channel and want to transmit to one receiver. The users are unable to both coordinate among each other and also to sense the channel.

This is the typical scenario for satellite communications. There in fact, the large footprint of the satellite hinders the effectiveness of channel sensing among transmitters on ground. Coordination, which is typical of orthogonal schemes, like \ac{TDMA} or \ac{FDMA}, as well as of demand assignment protocols, requires the use of a handshake mechanism between the transmitters and a central node. This is very inefficient if the data transmission is small compared to the control messages exchanged during the resources reservation, which is normally the case in messaging applications. For example, a four time handshake to send a single packet produces an overhead of $80\%$ if all packets are of the same size.

\section{Time Diversity}
\label{sec:time_diversity}

The main idea of \emph{diversity} is to counteract fading events by sending different signals carrying the same information. Fading affects these signals independently and the receiver is able to benefit from the signal(s) that are in good channel conditions, \ie not affected by deep fades \cite{TseBook_2005}.

Diversity can be achieved in several ways. Diversity over time, \ie \emph{time diversity}, may be attained via repetition coding. The same signal is repeated (or coded and interleaved) and sent through the channel by spreading the transmission over a time larger than the channel coherence time. Diversity can also be obtained over frequency, \ie \emph{frequency diversity}, if the channel is frequency selective. The same signal is sent at the same time over different frequencies. \emph{Space diversity} is also a possibility when the transmitter and/or receiver are equipped with multiple antennas that are spaced sufficiently apart.

Throughout the entire thesis \emph{time diversity} will be a recurrent concept. Anyhow, it is important to highlight here that we will consider it in a slightly different context with respect to the one of fading channels. In fact, time diversity is used to counteract the effect of interference rather then the fading of the channel. In most of the cases we will assume \ac{AWGN} with interfering packets coming from the random activity typical of ALOHA-like protocols. Two effects pushing against each other arise when time diversity is used in \ac{RA}. On the one hand, multiple replicas may increase the probability that at least one of them can be successfully decoded. On the other hand, an increase of the physical channel load leads to an increasing number of collisions. It turns out that up to a certain channel load, the former effect is predominant, leading to higher throughput, while after this critical point, the latter effect takes over deteriorating the performance compared to protocols without time diversity. The channel load up to which the former beneficial effect takes place can be greatly improved with the presence of \ac{SIC} at the receiver.\footnote{For more details on the \ac{SIC} definition see Section~\ref{sec:sic}.}

\section{Channel Models}
\label{sec:ch_models}

In this Section we define the channel models that are considered in the thesis. The first and simplest channel model is the \emph{collision channel}, typically adopted for investigation of \ac{MAC} protocols, including \ac{RA}. In a collision channel, whenever a packet is received collision free, \ie none of the received symbols are affected by interference, the packet is successfully decoded with probability $1$. Otherwise, whenever a packet is affected by interference (even by one packet symbol only) the packet cannot be decoded.

The collision channel model lacks of accuracy especially when the received power among colliding packets is very different, or the collided portion is very small compared to the packet size. A more accurate channel considers white Gaussian noise as impairment, \ie the \ac{AWGN} channel model. This model is typically adopted in satellite communication systems with fixed terminals where a very strong line of sight is present due to the presence of directive antennas. Depending on the received power and the employed channel code, packets may be lost even without suffering any interference. On the other hand, not all collisions lead to unsuccessful reception. In fact, if the interference power is sufficiently low, or the collided portion of the packet is sufficiently small, \ac{FEC} may be able to counter act it and the receiver can still correctly recover the packet. The latter effect is known as the \emph{capture effect} and has been considered first for \ac{SA} in random access research literature \cite{Roberts1975}.

A third model that takes into consideration multipath fading is described in the following. Communication over the wireless medium may be affected by reflection, distortion and attenuation due to the surrounding environment. In this context, the transmitted signal is split into multiple paths experiencing different levels of attenuation, delay and phase shift when arriving to the receiver. All these signals create interference at the receiver input which can be constructive or destructive. This phenomenon is called \emph{multipath fading}. Since a precise modeling requires to perfectly know attenuation, delay and phase shifts for all the paths, such description becomes easily impractical as the number of paths increases. An alternative is to exploit the central limit theorem assuming that these parameters can be modeled as \acs{r.v.} and resorting to a statistical model. For example, if no predominant path is present and the number of paths is large enough, the envelope of the received signal becomes a lowpass zero mean complex Gaussian process with independent real and imaginary parts. The amplitude is Rayleigh distributed while the phase is uniform in $[0,2\pi)$ and the received power follows an exponential distribution.

\begin{sloppypar}
Depending on the considered scenario, fading channels can be \emph{frequency selective} or/and \emph{time selective}. If the coherence bandwidth is smaller than the transmitted signal bandwidth occupation, then the signal suffers from independent fading on different frequency portions of the signal. In a similar way, if the coherence time is smaller than the transmitted signal duration, then the signal is subject to independent fading in consecutive portions of the transmission. In our case, we will consider only time selective fading, which is a good model for mobile radio communications \cite{Biglieri2005}. Moreover, since we are considering small packet transmissions, the coherence time of the channel is considered to be equal or greater than a packet duration or time slot. Therefore, we resort to a block fading channel model, in which a block corresponds to a time slot and an independent fading coefficient is seen by packets sent in different time slots.\footnote{In reality, some correlation between fading coefficients affecting consecutive time slots is present. Nevertheless, in our case, a given transmitter chooses to transit its replicas in consecutive time slots with probability $1/\numSlot$, where $\numSlot$ is the number of slots per frame. For sufficiently large frames, this probability is vanishing small.} In the \emph{Rayleigh block fading channel model} the received power is drawn from a \ac{PDF} of the form ${f_{\usPw}(\mathsf{p})= \frac{1}{\PwAvg} \exp \left[\frac{\mathsf{p}}{\PwAvg}\right]}$ for $\mathsf{p}\geq0$ with $\PwAvg$ being the average received power. The received power is then constant for an entire block which corresponds to a set of subsequent symbols. Block by block the received power is i.i.d..
\end{sloppypar}


\section{Successive Interference Cancellation}
\label{sec:sic}

The capture effect leads to the possibility of correctly receiving a packet even in presence of underlying interference. But once decoding is successful, the underlying transmissions are left without attempting detection or decoding and the receiver moves forward. This stems from the fact that the interference of the decoded packet is still present and in many cases is predominant over the other transmissions leading to prohibitive conditions for the decoder. Nonetheless, one could think of exploiting the retrieved information coming from the decoded packet and removing its interference contribution on the received signal. In this way, the underlying transmissions will benefit from an increased \ac{SINR} and may be possibly decoded. We can iterate decoding and interference cancellation several times, until underlying packets are discovered and the decoding is successful. This iterative procedure is commonly known as \emph{successive interference cancellation} and is proven to achieve capacity in the multiple access scenario \cite{TseBook_2005} for some specific scenarios.

Indeed, \ac{SIC} can be triggered by the presence of \ac{FEC} that yields more robust transmissions against interference. Another possibility to help \ac{SIC} is to adopt different transmission power levels for different packets so to enable the capture effect and to increase the probability of correct decoding. In uncoordinated access, as \ac{RA}, each user may decide to transmit with a power level independently sampled from a distribution (that can be continuous or discrete) equal for all users. Several works have focused on deriving the best distribution so to achieve the highest throughput \cite{Warrier_1998,Andrews_2003,Agrawal_2005}.

\section{Decoding Conditions}
\label{sec:dec_conditions}

Before going into the details, we will define some important quantities necessary for this section:
\begin{itemize}
\item $\usPw$: received power;
\item $\noisePw$: noise power;
\item $\intPw$: aggregate interference power;
\item $\rate$: transmission rate measured in $\left[\mathrm{bit/symbol}\right]$;
\item $\mutInf$: Instantaneous mutual information.
\end{itemize}
When we will not adopt the collision channel model, more sophisticated decoding conditions will be considered. In order to determine if a received packet can be successfully decoded without deploying a real decoder, we will adopt an abstraction of the decoding condition. A first option is to use the \emph{capacity based decoding condition} which is based on the Shannon-Hartley theorem \cite{Shannon_1948,Hartley_1928},
\begin{defin}[Point-to-Point Capacity-Based Decoding Condition under \ac{AWGN}]
\label{def:point_point_dec}
Assuming an \ac{AWGN} channel, a received packet can be successfully decoded iff
\begin{equation}
\label{eq:dec_cond1}
\rate<\log_2\left(1+\frac{\usPw}{\noisePw}\right).
\end{equation}
\end{defin}

There are some implicit assumptions underlying this model that we highlight: the capacity of the channel can be achieved via a Gaussian codebook and asymptotically long packets. Anyhow this definition holds for point-to-point links or for noise limited scenarios where the interference level is much below the noise power. In \ac{RA} this is normally not the case since the interference power is commonly of the same order of magnitude of the received power and cannot be neglected. Therefore we will resort to the \emph{block interference model} \cite{McEliece1984}. We consider $\numSym$ parallel Gaussian channels \cite{cover2006}, one for each packet symbol. Over these symbols different levels of interference will be observed, depending on the number of interfering packets and their power. Let us consider a very simple example, we assume an \ac{AWGN} channel and a \ac{SA} system, in which transmission are organized into slots of duration equal to the packet size. The interference power seen over each packet symbol will be constant over the entire packet. If instead we consider an ALOHA protocol, where no slots are present, the interference power over each packet symbol may change due to the random nature of packet arrivals, but nevertheless it is correlated across symbols.
\begin{defin}[Capacity-Based Decoding Condition under Block Interference]
\label{def:decoding_capacity}
Assuming an \ac{AWGN} channel and a Gaussian codebook, a received packet can be successfully decoded iff
\begin{equation}
\label{eq:dec_cond2}
\rate< \frac{1}{\numSym} \sum_{i=0}^{\numSym-1} \log_2\left(1+\frac{\usPw}{\noisePw+\intPw_i}\right)
\end{equation}
where $\intPw_i$ is the aggregate interference power over the $i$-th packet symbol.
\end{defin}

In particular, we are taking $\mutInf=\log_2\left(1+\frac{\usPw}{\noisePw+\intPw_i}\right)$, with $\mutInf$ the instantaneous mutual information for each symbol in the packet and averaging over all packet symbols. When considering time slotted systems as \ac{SA}, the interference level is constant over an entire packet duration because all transmissions are synchronized at slot level. In this way we can simplify the capacity-based decoding condition as
\begin{defin}[Capacity-Based Decoding Condition Time Slotted Systems]
\label{def:decoding_capacity_ts}
\begin{sloppypar}
Assuming an \ac{AWGN} channel and a Gaussian codebook, a received packet can be successfully decoded iff
\begin{equation}
\label{eq:dec_cond3}
\rate< \log_2\left(1+\frac{\usPw}{\noisePw+\intPw}\right).
\end{equation}
\end{sloppypar}
\end{defin}

\chapter{The Role of Interference Cancellation in Random Access Protocols}
\label{chapter3}
\thispagestyle{empty}
\ifpdf
    \graphicspath{{chapter3/Images/Figures_Tikz/}}
\fi
\epigraph{All truths are easy to understand once they are discovered; the point is to discover them}{Galileo Galilei}

In this third Chapter the benefits of interference cancellation adopted recently in \ac{RA} are highlighted. The role of interference cancellation is brought to the attention of the reader through the selection of four \ac{RA} schemes. These schemes are the representatives of four classes of \ac{RA} protocols. Their concept and behaviour as well as the design criteria are presented. At the end of the Chapter, we present some applicable scenarios of interest. The Chapter is concluded with a selection of open questions in the area.

\acresetall


\section{Recent Random Access Protocols}

In the early 2000s renewed interest in \ac{RA} driven both by new application scenarios and by the exploitation of advanced signal processing techniques led to an exciting number of new \ac{RA} protocols. In this Section, we review in detail a selection of the most promising ones. Hints and reference to modifications or enhancements of the basic protocols are given to the interested readers. Four classes of protocols are considered:
\begin{enumerate}
\item{slot-synchronous;}
\item{asynchronous (including spread-spectrum protocols);}
\item{tree-splitting algorithms;}
\item{without feedback;}
\end{enumerate}

\subsection{Slot Synchronous Random Access}

The first class of schemes is the slot-synchronous, with \ac{SA} being the pioneer and original root of these protocols. In \ac{SA}, if the number of users accessing the common medium is small, we are in the case of low channel load, the probability that more than one packet is transmitted concurrently in a slot is quite limited.

Exploiting this observation and letting the users transmit multiple times the same packet, increases the probability of receiving correctly at least one of them. Choudhury and Rappaport have been the first observing this feature \cite{Choudhury1983} in the \ac{DSA} protocol. Unfortunately, the replication of packets is beneficial only for low channel load, because as soon as the number of users accessing the medium increases, the collision probability increases as well and the replication of packets is only detrimental. Even more, the replication of packets increases the physical channel load, lowering the throughput compared to the simple \ac{SA} under the same logical load conditions.

Choudry and Rappaport in \cite{Choudhury1983} considered not only time-diversity in \ac{SA} but also fre\-quen\-cy-diversity. In this second variant of the scheme, users are allowed to transmit their packets in multiple frequencies at the same time. Although being conceptually the dual of the time-diversity scheme, this second variant poses some practical requirements on the terminals increasing their complexity and reducing the appeal. If users are allowed to transmit in multiple-frequencies concurrently, more than one transmission chain is required.

The replication of packets alone is beneficial only at low channel loads and does not bring any improvement in terms of maximum achievable throughput. Relevant improvements come when is coupled with more advanced receivers that exploit iterative algorithms. After almost $25$ years from \cite{Choudhury1983}, Casini and his co-authors come out with a very attractive modification of the \ac{DSA} protocol, the \ac{CRDSA} scheme \cite{Casini2007}, that adds at the receiver interference cancellation. Transmissions are organized into frames, where users are allowed to transmit only once. The users replicate their packets two (or more) times and place the replicas in slots selected uniformly at random, providing in all replicas the information on the selected slots. At the receiver, \ac{SIC} exploits the presence of multiple replicas per user for clearing up collisions. Every time a packet is decoded, \ac{SIC} reconstructs the waveform and subtracts it from all the slot locations selected for transmission by the corresponding user, possibly removing the interference contribution with respect to other packets. The performance evaluations in \cite{Casini2007} have shown that the maximum throughput of \ac{CRDSA} can be impressively extended from $\tp\cong 0.36$ (the peak throughput of \ac{SA} measured in average number of successful transmissions per transmission period \cite{Kleinrock1976_book} or packets per slot\footnote{Following the definition of \cite{Kleinrock1976_book} we assume that a transmission period is equal to $\pkLen$ seconds, which coincides with the physical layer packet duration and also coincides with the slot duration. Therefore, for slotted protocols the throughput can be measured also in packets/slot.}), up to $\tp\cong 0.55$. Further throughput improvements can be achieved when, 1) more than two replicas per user and per frame are sent, 2) difference in received power due to induced power unbalance or fading and capture effect are considered \cite{Herrero2014}. The stability of \ac{CRDSA} has been investigated in \cite{Kissling2011}, while more recently an analytical framework for slotted \ac{RA} protocols embracing \ac{SA}, \ac{DSA} and \ac{CRDSA} has been presented in \cite{Herrero2014}.

We are going to explain in detail the operation of \ac{CRDSA} from both the transmitter and receiver perspectives in the following subsections.

\subsubsection{CRDSA - Transmitter Side}

We assume that $m$ users share the medium and are synchronized to a common clock that determines the start of each slot. Transmission of packets can start only at the beginning of a slot. A group of $n$ slots is called frame and each user has the possibility to transmit at most once per frame. The single transmission per \ac{MAC} frame may correspond to a new packet or to the retransmission of a previously lost one. The channel traffic is $\load = m/n$ and is measured in packets per slot, \ie [pk/slot].

Each user is allowed to access the medium $\dg$ times for each frame, with $\dg$ called \emph{degree}. Since only a single transmission is allowed, the $\dg$ accesses are all performed with the same identical packet. The physical layer packets are called also \emph{replicas} and the $\dg$ unique slots where the replicas are sent are selected uniformly at random. Each replica contains information for localizing all the $\dg$ replicas, \ie the slot number of all replicas is stored in the header.\footnote{The seed of a pseudo-random sequence is normally stored instead of the slots number, in order to decrease the overhead impact of the replicas location information. At the transmitter a pseudo-random algorithm is run for choosing the slots for transmission. At the receiver, the seed is inserted in the same pseudo-random algorithm for retrieving the selected slots.} This information is used at the receiver to remove the interference contribution of a correctly decoded packet from all the $\dg$ slots selected for transmission, possibly resolving collisions. Please note that the physical layer traffic injected by the users is $\load \cdot \dg$ but it considers multiple times the same information, since the $\dg$ replicas of each user carry all the same content. In this way, $\load$ is the measure of the information handled by the scheme.

\subsubsection{CRDSA - Receiver Side}

After the receiver stores an incoming \ac{MAC} frame, the \ac{SIC} iterative procedure starts. Under the collision channel model, only the replicas not experiencing collisions in their slot can be decoded. Every time a replica is correctly decoded, the location of its $\dg-1$ copies is retrieved from the header. Assuming perfect knowledge of frequency, timing and carrier offsets for all replicas, their interference contribution can be perfectly removed in all the $\dg$ slots via interference cancellation. The procedure can then be iterated possibly cleaning further collisions. In order to give a deeper insight on the \ac{SIC} process of \ac{CRDSA} we make the use of the exemplary \ac{MAC} frame of Figure~\ref{fig:CRDSA_frame}. There, it is important to mention that the packets are placed in different levels only for simplifying the visualization, but they are assumed to be sent in the same frequency and therefore they may overlap in time. The SIC procedure starts from the first replica free from interference which is replica $C_2$. Once it is successfully decoded its twin replica, packet $C_1$ can be removed from the frame. In this way, replica $B_2$ can be correctly decoded and its twin is also removed from the frame. Similarly, replica $D_2$ can be correctly decoded, since it has not been subject to collisions and therefore its twin $D_1$ can be also removed from the frame, releasing replica $E_1$ from the collision. Finally also replica $A_2$ can be decoded since its collision with replica $E_2$ has been released. In this specific scenario, all transmitted packets could be decoded, but there are collisions configurations that \ac{SIC} is not able to resolve \cite{Herrero2014}. These configurations lead to an error floor in the packet loss probability at low load, as it has been highlighted in \cite{Herrero2014} and \cite{Ivanov2015}.

\begin{figure}
\centering
\includegraphics[width=\textwidth]{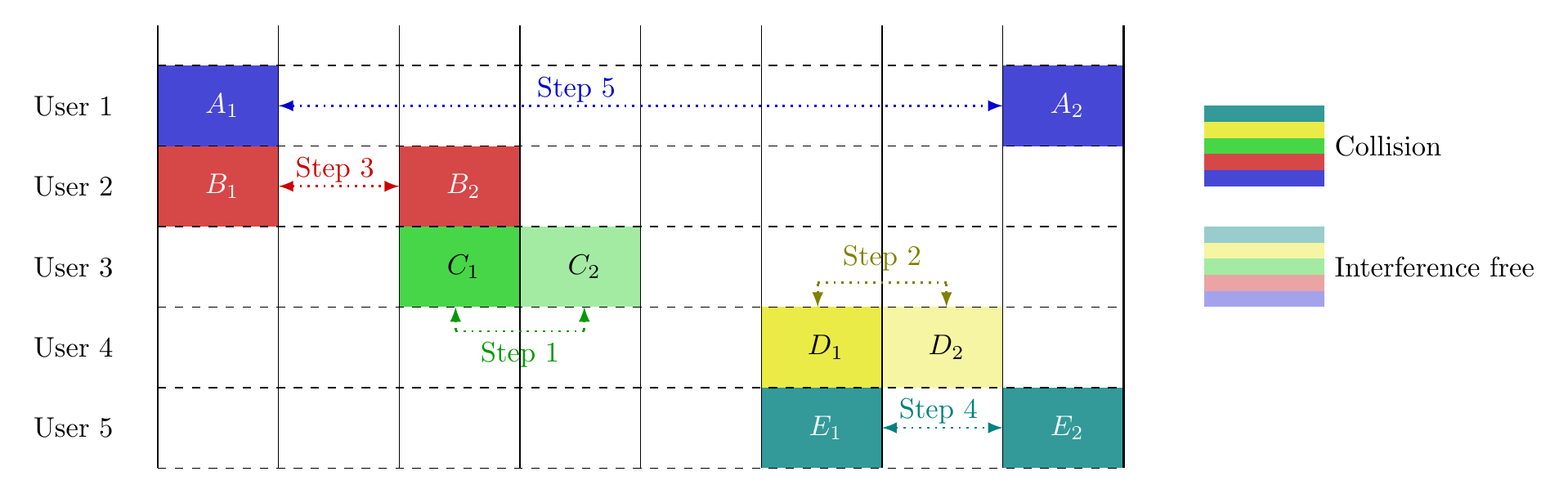}
\caption{Example of a received \ac{MAC} frame for \ac{CRDSA} and the \ac{SIC} iterative procedure.}
\label{fig:CRDSA_frame}
\end{figure}

\subsubsection{CRDSA Evolutions and Some Results}


The exploitation of interference cancellation pioneered by \ac{CRDSA} and its outstanding performance compared to \ac{SA} and \ac{DSA}, lead to a new wave of interest for \ac{RA} protocols in the research community. Among the evolutions of \ac{CRDSA}, we would like to mention three of them, \ie \ac{IRSA} \cite{Liva2011}, \ac{CSA} \cite{Paolini2011,Paolini2015} and Frameless ALOHA \cite{Stefanovic2013_J}. In \ac{IRSA} the link between \ac{SIC} and iterative decoding on bipartite graphs has been for the first time established. This intuition directly leads to the exploitation of tools used typically for the analysis and optimisation of iteratively-decodable codes to the optimisation of \ac{CRDSA}. The similitude holds replacing the check nodes with time slots and variable nodes with the users. In this context, each replica is represented as an edge between a variable node and a check node. The observation that irregular bipartite graph constructions (variable number of edges emanating from the variable nodes) lead to better iterative decoding threshold with respect to regular graphs (fixed number of edges emanating from the variable nodes) \cite{Liva2011}, suggests that allowing a variable number of replicas per user improves the performance of \ac{CRDSA}. Moreover, the probability distribution of the degree can be optimised in order to achieve the maximum throughput.\footnote{Although the argument of the optimisation holds for asymptotically large frames, simulation results have proven that good degree distributions in the asymptotic domain perform also well in the finite frame scenario.} Liva \cite{Liva2011} used density evolution for the optimisation and showed a consistent improvement in the throughput performance (as can be observed in Figure~\ref{fig:T_SL_Comp}). A further evolution of \ac{CRDSA} is \ac{CSA} \cite{Paolini2011,Paolini2015} where the replicas of each user are not simple repetition of the same information, but are instead an encoded sequence that is defined independently by each user. By means of density evolution, the degree distribution is optimised in order to obtain the maximum throughput. Recently, it has been demonstrated that a proper design of the degree probability mass function in \ac{IRSA} and \ac{CSA} achieves $1$ pk/slot throughput \cite{Narayanan2012}, which is the maximum achievable under the collision channel.

Along a similar line of research, the work of Stefanovic and his co-authors in \cite{Stefanovic2013_J}, investigates the behaviour of \ac{RA} protocols with repetitions where the frame dimension is not set a-priori but it is dynamically adapted for maximising the throughput. The scheme resembles to rateless codes (called also fountain codes) \cite{Luby1998}, a forward error correction construction where no code rate is set a-priori and the sources keeps sending encoded symbols until the message is correctly received. In a similar way, an increasing number of users sharing the channel are allowed to send their data, until a target throughput is reached. At this point the frame is stopped and unresolved users are notified.

A first attempt to allow a higher degree of freedom to the user has been done in \cite{Meloni2012}. There, the frame synchronization among user is released and the first replica of a user can be transmitted in the immediate next time slot, without the need to wait for the new frame to start. The remaining $\dg - 1$ replicas are sent in slots selected uniformly at random within a maximum number from the first replica. The receiver adopts a sliding window which comprises a set of time slots over which \ac{SIC} is performed. Upon correct reception of all replicas in the decoding window or when a maximum number of \ac{SIC} iterations is reached, the window moves forward and \ac{SIC} starts again. Similarly, the authors in \cite{Sandgren2016} extended the concept also for the more general \ac{CSA}. Substantial analytical work is presented in their paper where a packet error rate bound tight for the low channel loads is derived. Very interestingly, they show that not only the delay but also the packet error rate can be reduced.

When the capture effect is considered, not all collisions are destructive for the packets and eventually the one received with the highest power can be recovered, triggering the \ac{SIC} process. In this regard, many investigations have shown that power unbalance between received signal boosts \ac{SIC} performance \cite{Andrews2005} and therefore, even higher throughput performance can be achieved \cite{Herrero2014}.

\begin{figure}
\centering
\includegraphics[width=0.8\textwidth]{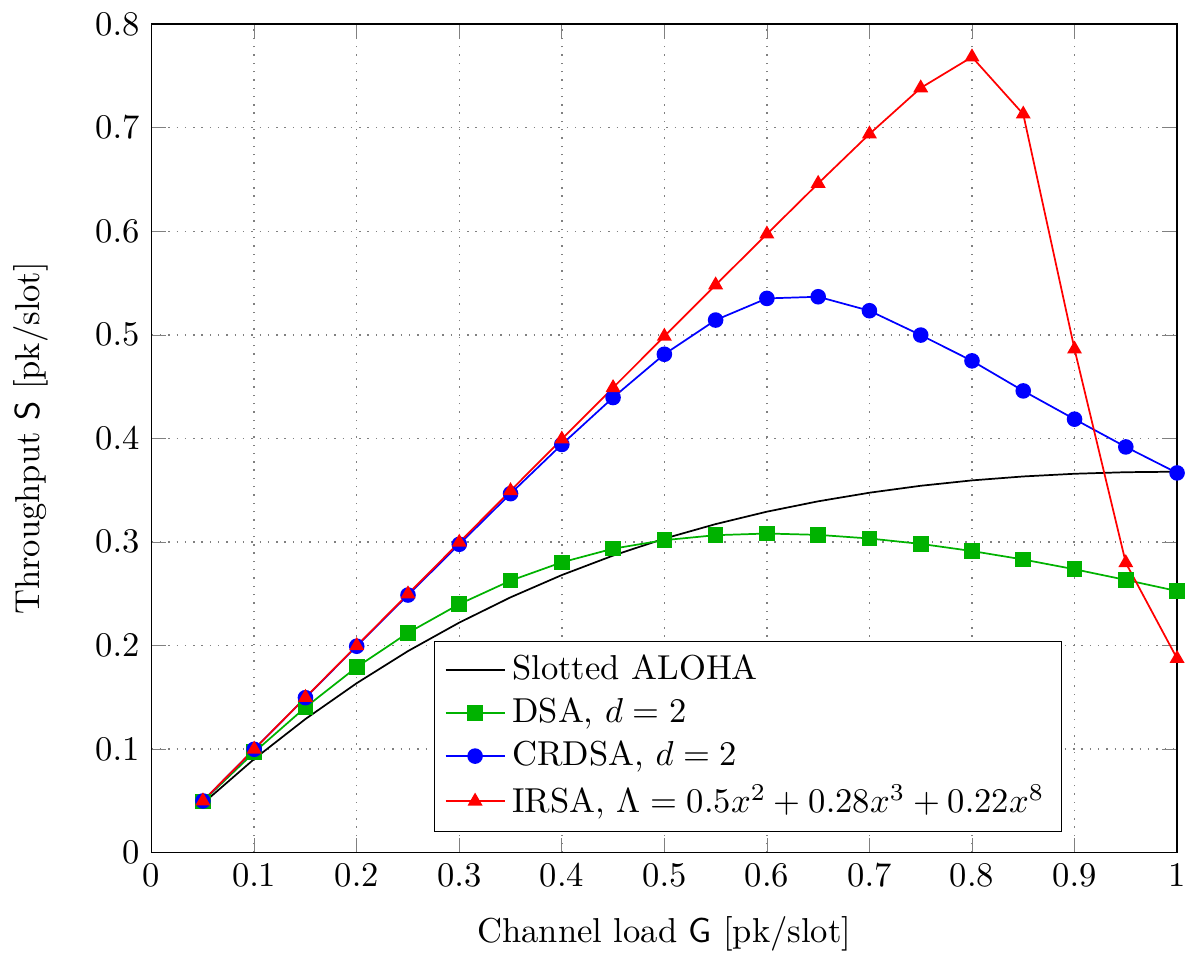}
\caption{Throughput comparison of \ac{SA}, \ac{DSA}, \ac{CRDSA} and \ac{IRSA} under the collision channel. Both \ac{DSA} and \ac{CRDSA} send two copies for each transmitted packet and \ac{CRDSA} employs \ac{SIC} at the receiver side. In \ac{IRSA} each user picks a degree $\dg$ following the \ac{PMF} $\Lambda$. Specifically, with probability $1/2$ two copies will be sent, with probability $0.28$ three and with probability $0.22$ eight. The use of variable number of replicas per user greatly improves the throughput compared to \ac{CRDSA} although for very high channel load values, the degradation is more severe.}
\label{fig:T_SL_Comp}
\end{figure}

In \cite{Gamb_Schlegel2013} it has been shown that joint decoding of the collided packets can be attempted, resorting to \ac{MUD} techniques. Further evolutions of \ac{RA} include the extension to multiple receiver scenarios \cite{Munari2013,Jakovetic2015}, to all-to-all broadcast transmission \cite{Ivanonv2015} and to combining techniques \cite{Bui2015}.

\subsection{Asynchronous Random Access}
\begin{sloppypar}
Although ALOHA has been originally designed for asynchronous transmissions by Abramson \cite{Abramson1970}, research has been focusing mainly on its time slotted enhancement due to its better performance (and also to its simpler mathematical tractability). Indeed, also in recent years, most of the research has been focusing on the time slotted scenario. Nevertheless, the concept introduced in \ac{CRDSA} has been brought also to the asynchronous scenario firstly in \cite{Kissling2011a} with the \ac{CRA} scheme. The time slots boundaries are abandoned there, but the frame structure is maintained. The decoding procedure takes the advantage of multiple replicas and exploits \ac{SIC} as in \ac{CRDSA}. Trying to abandon any synchronization requirement leads to the \ac{ACRDA} scheme \cite{deGaudenzi2014_ACRDA}, where the decoder employs a window based procedure where \ac{SIC} is performed.
\end{sloppypar}

\begin{sloppypar}
A similar approach can be followed also when considering asynchronous spread spectrum random access as in \cite{delRioHerrero_2012}, with the \ac{E-SSA}. Prior transmission on the channel, each packet symbol is multiplied with a spreading sequence, possibly different user by user. It is important to stress that here no replicas are used, \ie each terminal sends only one packet per transmission. At the receiver side, \ac{SIC} is employed for removing interference once packets are correctly decoded. In this way the throughput performance can be drastically improved \cite{Gallinaro_2014}. Thanks to its remarkable performance, \ac{E-SSA} has been selected by \ac{ETSI} as one of the two options for the return link \ac{MAC} layer of the \ac{S-MIM} standard \cite{Scalise_2013}. The aim of \ac{S-MIM} is to standardise messaging services over S-band via \ac{GEO} satellites.
\end{sloppypar}

Being a very interesting solution, the extension of \ac{CRDSA} to asynchronous schemes has minor differences to its original. For this reason, we will focus here on a slightly different issue that is typical in wireless \ac{MAC} and also in \ac{RA}: the hidden terminal problem and a new way of combating it employing \ac{IC}. Let us assume that the terminals are employing \ac{CSMA/CA} as medium access as in IEEE 802.11 standard and are able to sense the channel. Unfortunately, if two terminals that are not in reciprocal radio reception range want to communicate with a third node as in Figure~\ref{fig:Hidden}, there is a probability that the two terminal transmissions collide at the receiver resulting in the hidden terminal problem.
\begin{figure}
\centering
\includegraphics[width=0.4\textwidth]{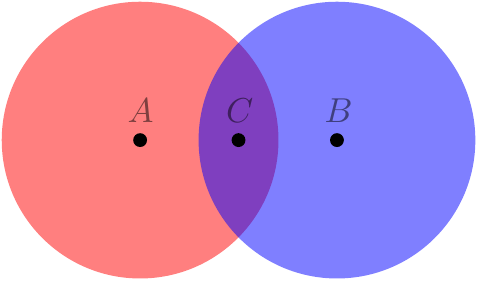}
\caption{Hidden terminal scenario. Node $A$ and node $B$ want to communicate to node $C$ and they are not able to sense each other since they are out of the reception range.}
\label{fig:Hidden}
\end{figure}
In order to overcome possible collisions, the \ac{RTS/CTS} procedure is suggested by the standard, but it is disabled by default from the access points manufacturers due to its significant impact on the overall throughput \cite{ZigZag}. A new way of counteracting the hidden terminal problem without the need of \ac{RTS/CTS} has been firstly proposed by Gollakota and Katabi in \cite{ZigZag} with the \emph{ZigZag decoding}. The core idea is to exploit asynchrony between successive collisions. Terminals that collide once, are with high probability susceptible to collide again in the retransmission phases. On the other hand, successive collisions are likely to have different interference-free stretches at their start, allowing ZigZag decoding to operate. In the following subsection we will describe the ZigZag decoder operations.

\subsubsection{ZigZag Decoder}

We make the use of an example in order to describe how the ZigZag decoder operates. Let us assume that we are in the hidden terminal scenario of Figure~\ref{fig:Hidden}, where the two transmitters are denoted as node $A$ and node $B$ respectively. The two nodes want to communicate to node $C$ using \ac{CSMA/CA} but they are not able to sense each other because they are out of reciprocal reception range. We further assume that nodes $A$ and $B$ transmit simultaneously to node $C$ causing a collision. In order to resolve the collision a second transmission attempt is made by both nodes which is likely to collide again since the sensing phase is useless. We further observe that the random jitter in the two collisions will be different, \ie the relative times offset between the two packets in the two collisions are different $\off_1\neq\off_2$, as shown in Figure~\ref{fig:zigzag}. If node $C$ is able to compute $\off_1$ and $\off_2$, it can identify sections of the packets being interfered in only one of the two collisions, such as $s_1$. These sections can be then used to trigger the ZigZag decoder \ac{IC} operation to resolve the collisions. In fact, the packet section $s_1$ can be demodulated by node $C$ using a standard demodulator. Node $C$ is able to remove the interference contribution of $s_1$ from the second collision freeing from interference section $s_2$. In this way, $s_2$ can also be decoded by a standard demodulator and \ac{IC} can now operate on the first collision. Letting the \ac{IC} iterate between the two collisions allow to demodulate both collided packets from nodes $A$ and $B$, and finally, when the entire packets are successfully retrieved, decoded.
\begin{figure}
\centering
\includegraphics[width=0.8\textwidth]{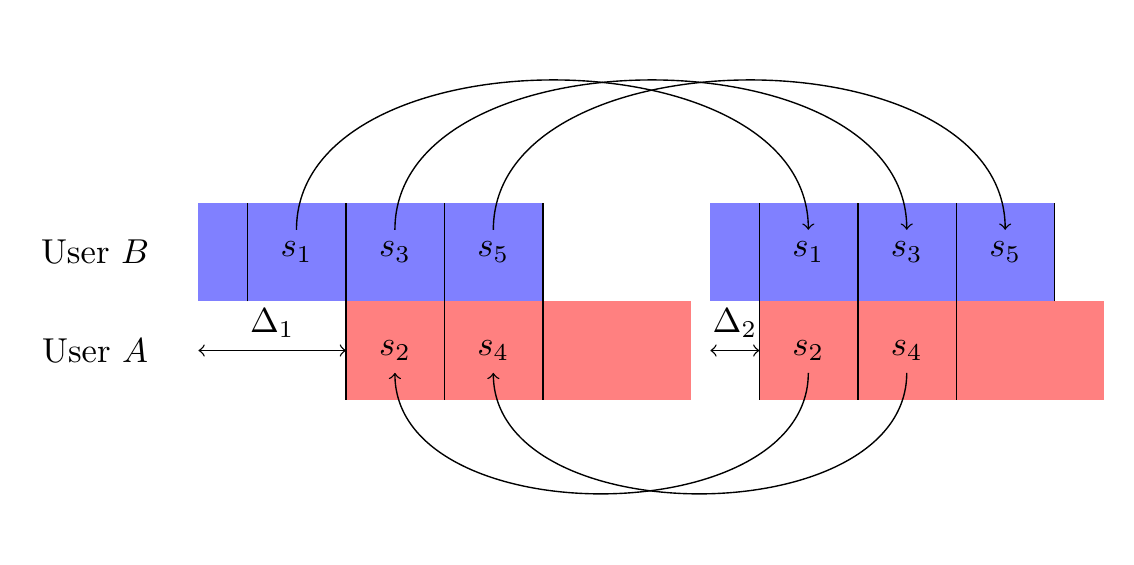}
\caption{ZigZag decoding procedure. The interference free portion $s_1$ of user $B$ first transmission is removed from the second transmission and reveals $s_2$, \ie the first portion of user $A$'s transmission. This portion is then used back in the first collision to remove interference and revealing portion $s_3$. Proceeding iteratively between the two collisions portion by portion, the packets can be successfully decoded.}
\label{fig:zigzag}
\end{figure}

In general, once the ZigZag decoder detects a packet, it tries to decode it with a common decoder, assuming that no collision has occurred. If the decoding is not successful, the ZigZag receiver seeks for possible collisions. Since every packet of IEEE 802.11 starts with a known preamble equal for all packets and transmitters, correlation can be performed over the received packet symbols in order to identify the start of other possible packets. Correlation can be subject to significant losses if the frequency offset between the receiver and the transmitter is not compensated. To counteract this effect, the ZigZag receiver keeps a coarse frequency offset estimate of all active terminals. Once a collision is detected and the relative time offset $\off$ is estimated, the receiver looks for other \emph{matching collisions}, \ie other collisions involving the same packets. To do so, the ZigZag receiver is required to store recent unmatched collisions, specifically the received complex samples. Matching the detected collision with the stored ones implies using once more correlation. It is important to underline that in this case the correlation operates on the entire packets, instead of only on the preamble, increasing the probability of correctness in the correlation search. The received packet symbols are correlated with all the unmatched collision stored and the highest correlation peak can be associated to the matching collision. Since a matching collision is found and the two relative offsets $\off$ and $\off'$ are now known, the ZigZag receiver is able to isolate interference free symbols and start the demodulation process. The demodulator will then iteratively employ \ac{IC} to succeedingly free from interference received packet sections as described beforehand (see also Figure~\ref{fig:zigzag}). The decoder used in ZigZag can be any standard, since it operates on packet that are freed from interference. While the ZigZag re-modulator is the novel part added to the receiver and it is responsible
\begin{enumerate}
\item to use the pre-knowledge of a packet section in order to estimate system parameters on the second packet sent by the same transmitter;
\item to remove its interference contribution.
\end{enumerate}
Let us focus first on the system parameters estimation, which involves the estimation of the channel coefficient, the frequency offset, the sampling offset and the inter-symbol interference.

The channel coefficient is found using the correlation with the known preamble and inverting the equation of the correlation (see \cite{ZigZag} for more details). A coarse frequency offset estimate is kept in memory for each active terminal by the receiver, but is not precise enough. A fine frequency offset estimate is obtained, with the comparison of the packets sections before and after \ac{IC}. In a similar manner, the sampling offset can also be tracked. Finally the inter-symbol interference is removed, applying to the inverse filter introduced by the standard decoder before re-encoding is applied. It has to be underlined that for this operation the filter applied by the decoder must be known. The ZigZag receiver operations aforementioned are also summarised via flow diagram in Figure~\ref{fig:zigzag_fc}. Additionally, ZigZag decoding does not require to be started only at the beginning or at the end of a packet, but a parallel start on both sides is also possible.
\begin{figure}
\centering
\includegraphics[width=0.8\textwidth]{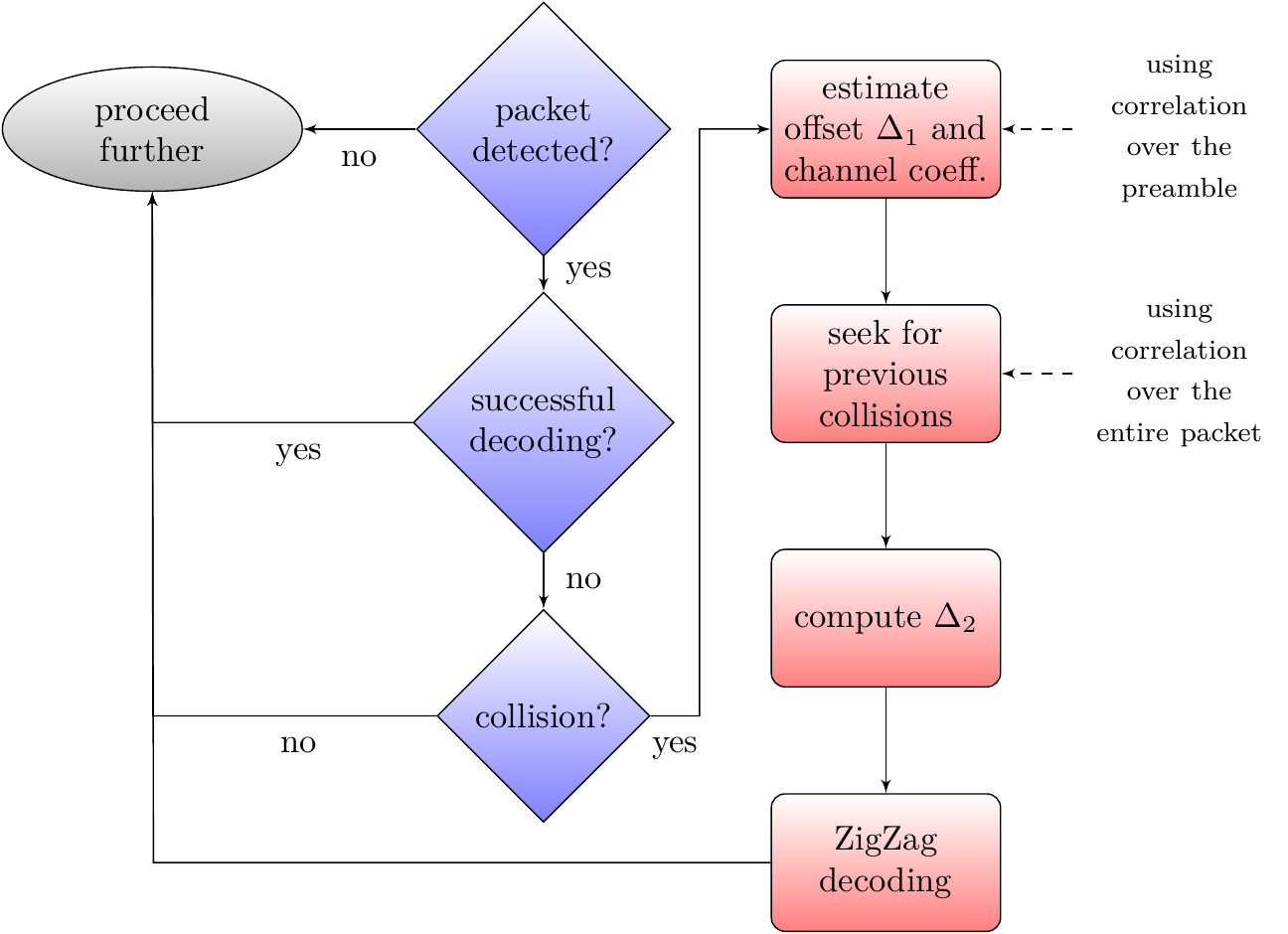}
\caption{ZigZag receiver operations depicted via functional flow chart.}
\label{fig:zigzag_fc}
\end{figure}


Although viable, ZigZag decoding comes with some practical challenges. In the first place, the correlation with the single preamble involves many different values of frequency offsets, increasing the complexity of the correlator. Storing of unmatched collisions may be limited by the receivers storage capability, especially in large networks. There are also some algorithmic limitations in ZigZag. In case of identical collision patterns among users, when the same random time for retransmission is chosen by the terminals, ZigZag is not able to recover the packets, regardless the \ac{SNR}.\footnote{Unless the \ac{SNR}s of the users is so different that the capture effect can be exploited.} Finally, ZigZag propagates errors due to its hard decision in the sequential decoding.

\subsubsection{Extension of ZigZag, The SigSag Decoder}

\begin{sloppypar}
An extension of ZigZag, called SigSag decoding, has been proposed in \cite{SigSag}. The authors define a procedure where soft-information of each symbol is exploited. In fact, ZigZag decoding can be seen as a special instance of belief propagation where only the back-substitution is applied. This corresponds to belief propagation in the high-\ac{SNR} regime. SigSag decoding uses an iterative soft message passing algorithm running over the factor graph of the linear equations that represent packet collisions. In particular, a factor graph representation of the collisions is associated to the consecutive collisions, as we can see from Figure~\ref{fig:sigsag}. There two variations of SigSag can be applied, the sum-product algorithm or the max-product algorithm \cite{Kschischang_2001, Richardson_2001}. The former aims at minimising the bit-error rate via computing the marginals of each bit. The latter aims at minimising the block-error rate computing jointly the most likely codeword.
\begin{figure}
\centering
\includegraphics[width=0.8\textwidth]{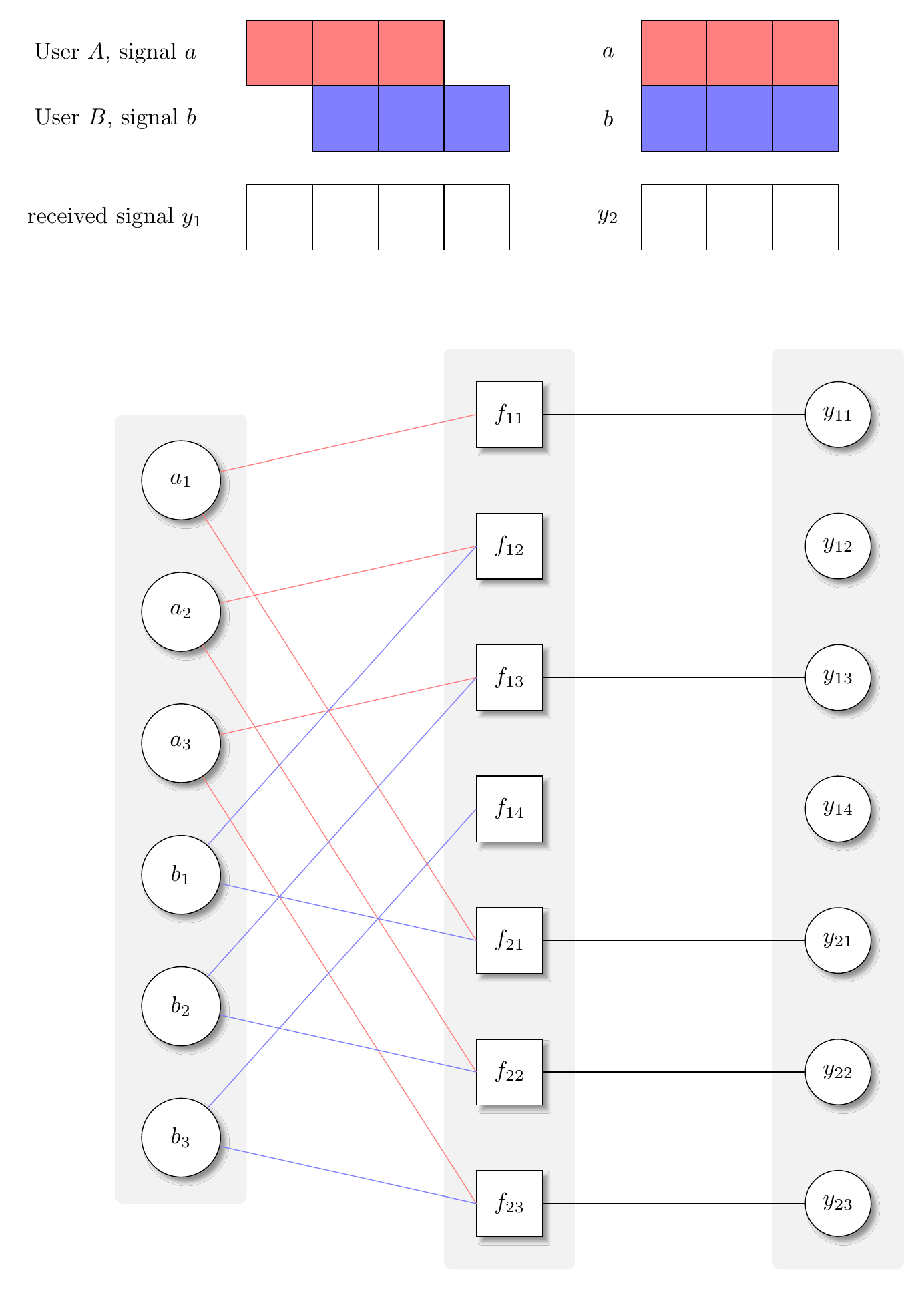}
\caption{Two consecutive collisions of two packets $a$ and $b$ composed by $3$ symbols each. Both packets are transmitted twice and the receiver stores $y_1$ and $y_2$. The corresponding factor graph of the two collisions shows how the symbols of individual packets are connected to the received symbols. The check nodes denoted with $f_{1x}$ and $f_{2x}$ are responsible for computing the \ac{PDF} or log-lokelihood ratios of received bits, given the noisy observations $y_1x$ and $y_2x$. The two options correspond to the sum-product or max-product algorithms respectively.}
\label{fig:sigsag}
\end{figure}
\end{sloppypar}

SigSag counteracts succesfully the error propagation introduced by ZigZag, and is able to profit from larger packet sizes. This effect comes from the message passing algorithm that profits from larger graphs. Nevertheless, SigSag is not able to resolve collisions where the relative distance between the users is the same, as in ZigZag. On the other hand, for the two-user cases and with at least one bit difference in the collision pattern, the authors in \cite{SigSag} have shown that SigSag is optimal, so it performs as \ac{ML} decoding.

\subsection{Tree-Splitting Algorithms}

Carefully looking at the tree-splitting collision resolution algorithm, Yu and Giannakis \cite{Giannakis_2007} observed that \ac{SIC} can be particularly useful in order to reduce the \ac{CRI}, \ie the protocol phase in which collisions among multiple transmitters are resolved through an algorithm.

Let us start with the system model assumptions. In the following, an infinite population of users sharing a common channel and transmitting to a single receiver is considered. Time is slotted and slots have a duration equal to the physical layer packets. Poisson arrivals with a total rate $\arrRate$ distributed over the users is assumed. The collision channel model is adopted and instantaneous $0/k/e$ feedback is provided to the users. At the end of each slot the users are informed errorlessly about the outcome of the last slot. Idle $(0)$ is returned when no transmission took place. The number of correctly decoded packets $(k)$ is returned in case of successful decoding (more than one packet per slot can be decoded when \ac{SIC} is employed, as we will see later). Finally, if none of the previous two cases happened, an erroneous reception $(e)$ due to a collision appears in the last slot.

Tree-splitting \ac{RA} protocols entails two modes, the channel access mode and the collision resolution mode which is entered when a collision during the channel access mode happened.

\subsubsection{Channel Access Modes}

The channel access mode determines when new packets can join the system. There are three main channel access modes:
\begin{enumerate}
\item \textbf{Gated Access:} New packets can join the system only once all previous conflicts are resolved. If the system is operating in the collision resolution mode, all arriving packets are buffered until the collision is fully resolved (all packets involved are correctly decoded). From the slot after the end of the collision resolution, packets are allowed to join the system again. The time span between the beginning and the end of the collision resolution is commonly called \ac{CRI}. This has been the first channel access mode introduced by Capetanakis \cite{Capetanakis1979} and Tsybakov and Mikhailov \cite{Tsybakov1978}.
\item \textbf{Window Access:} Once the \ac{CRI} is finished, only the packets arrived in a specific time interval, called window, are allowed to join the new \ac{CRI}. The access mode determines the next window size based on a time counter and a maximum window size. The former, denoted with $\finWin$, measures the time elapsed between the end of the last window and the end of the \ac{CRI}, while the latter is denoted with $\maxWinSize$. The next window size is then selected adopting the rule: $\min\{\maxWinSize,\finWin\}$.
\item \textbf{Free Access:} Every time a packet is generated, it is transmitted at the beginning of the upcoming time slot, without any delay.
\end{enumerate}

The authors of \cite{Giannakis_2007} focused on the first two access modes due to their better performance and ease of implementation.

\subsubsection{Collision Resolution Modes}

There are two standard approaches for the collision resolution modes, the so-called \emph{standard tree algorithm} and the \emph{Massey's modified tree algorithm}. We review these two modes through an example, shown in Figure~\ref{fig:sta_mta}.\footnote{A binary tree is used for simplicity, but extensions to higher order trees is straightforward.} Every time the receiver feeds back a collision, each user involved in the collision tosses a two-sided coin. With probability $\prob$ the user joins the subset on the right while with probability $1-\prob$ it joins the subset on the left. The subset on the right is then allowed to transmit first, while the left one is forced to wait until all users in the subset on the right are correctly decoded. The collision resolution ends when all users are correctly received, \ie the receiver feeds back idle or success.

\begin{figure}
\centering
\subfigure[Standard tree algorithm example.]{
    \includegraphics[width=0.4\textwidth]{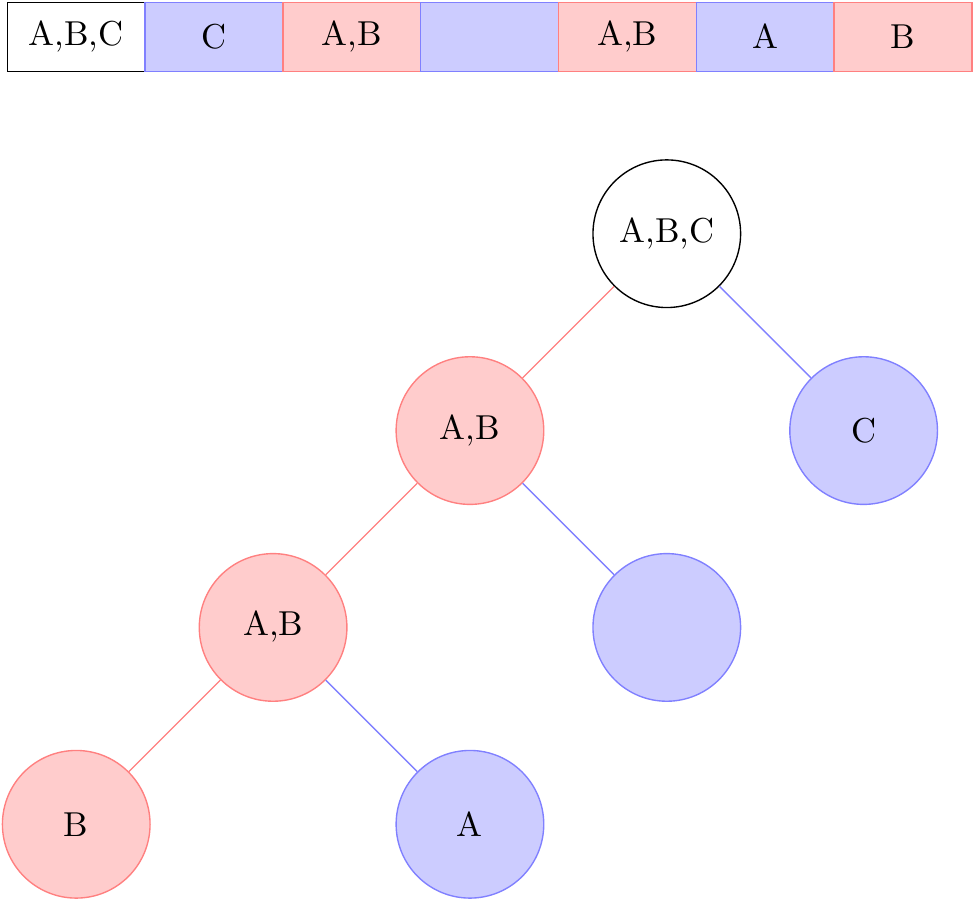}
    \label{fig:sta}
}\hspace{0 mm}
\subfigure[Massey's modified tree algorithm for the example of Figure~\ref{fig:sta}. A collision followed by an idle slot deterministically implies a further collision. Jumping directly to the next level of the tree helps in reducing the \ac{CRI} duration.]{
    \includegraphics[width=0.3858\textwidth]{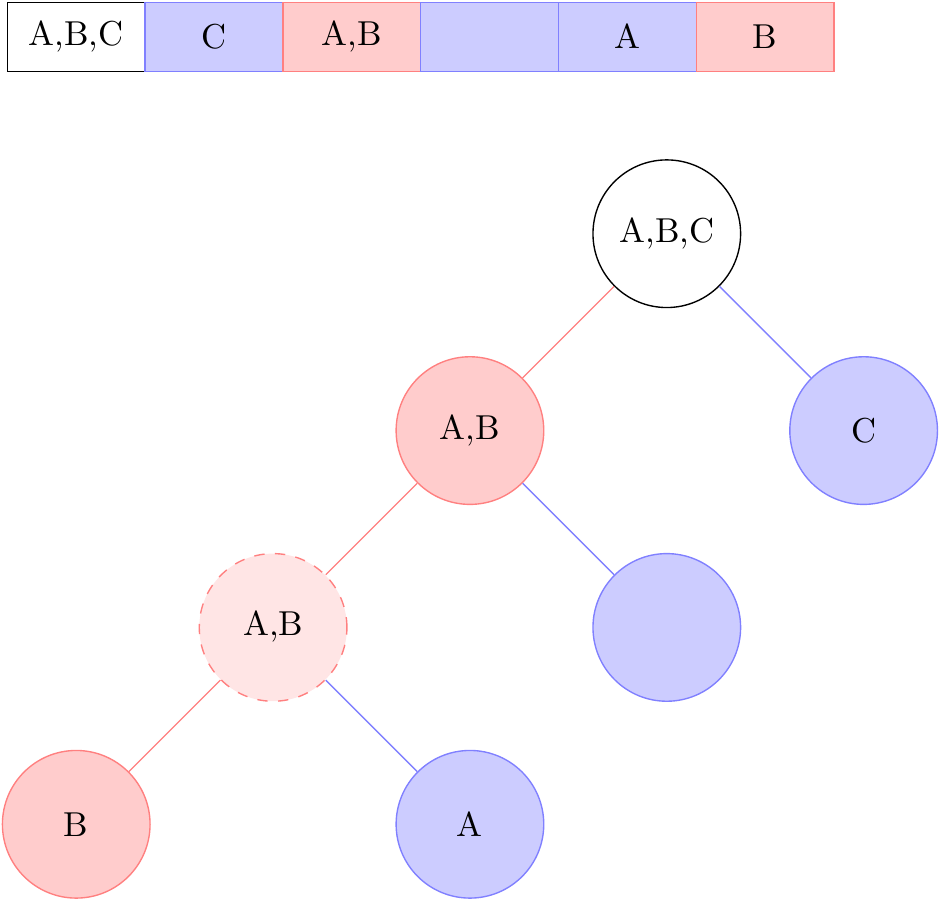}
    \label{fig:mta}
    }
\caption{Standard and modified tree algorithms for the same collision example.}
\label{fig:sta_mta}
\end{figure}

An improvement of the standard tree algorithm is achieved thanks to Massey's modification \cite{Massey1981}. Observing the example in Figure~\ref{fig:sta} the collision in the third slot is followed by an idle slot which foresees an inevitable collision in the fifth slot. Instead of allowing the deterministic collision in slot $5$, it is possible to skip this level of the tree and let the users toss the coin again. We can observe that in this example, we reduce the \ac{CRI} from $7$ slots of the standard tree algorithm to $6$ slots in the modified tree algorithm. Denoting with $\CRIL_m$ the \ac{CRI} duration in number of time slots, given that $m$ packets collide, for the standard tree algorithm it holds
\begin{equation}
\label{eq:cri_sta}
\CRIL_m=1+\CRIL_i+\CRIL_{m-i},\qquad m\geq2
\end{equation}
where $i$ denotes the number of users joining the right subset. From literature, it is known that for the gated access channel access mode fair splitting ($\prob=0.5$) is optimal, and the maximum stable throughput is $0.347$ \cite{Massey1981}. When the window access is considered, the maximum stable throughput can be extended to $0.429$ \cite{Massey1981}. When the Massey's modified tree algorithm is considered, the \ac{CRI} duration becomes
\begin{equation}
\label{eq:cri_mta}
\CRIL_m=
\begin{cases}
1+\CRIL_i+\CRIL_{m-i},\qquad \text{if}\quad 1\geq i\geq m\\
\CRIL_0+\CRIL_{m},\qquad\quad\enspace\, \text{if}\quad i=0
\end{cases}
\quad m\geq2.
\end{equation}
The modified tree algorithm reaches a maximum stable throughput of $0.375$ \cite{Massey1981} for the gated access channel access mode with fair splitting. Nevertheless the optimal splitting probability is $\prob=0.582$ which extends the stable throughput up to $0.381$ \cite{Mathys_1985}. When the window access is considered instead, the maximum stable throughput reaches $0.462$ \cite{Massey1981}.

\subsubsection{\aca{SICTA} (SICTA)}

Differently from the standard and modified tree algorithms collision resolution modes, in \ac{SICTA} collided packets during the collision resolution are not discarded but are kept in memory for further processing. Indeed they become useful when \ac{SIC} takes place. Let us consider a case in which users $A$ and $B$ collide in the first slot of the \ac{CRI} and then user $B$ is correctly received in the second time slot. The standard and modified tree algorithms would require a third time slot dedicated to the transmission of user $A$. \Ac{SICTA}, instead, employs interference cancellation on the signal received in the first slot, removing the contribution of user $B$ from the collision and being therefore able to extract user $A$ as well. It is important to note here that in the second slot the receiver will feedback $k=2$ since both users can be resolved. Let us give a look on how the example of Figure~\ref{fig:sta_mta} would change with the use of \ac{SICTA}.
\begin{figure}
\centering
\includegraphics[width=0.8\textwidth]{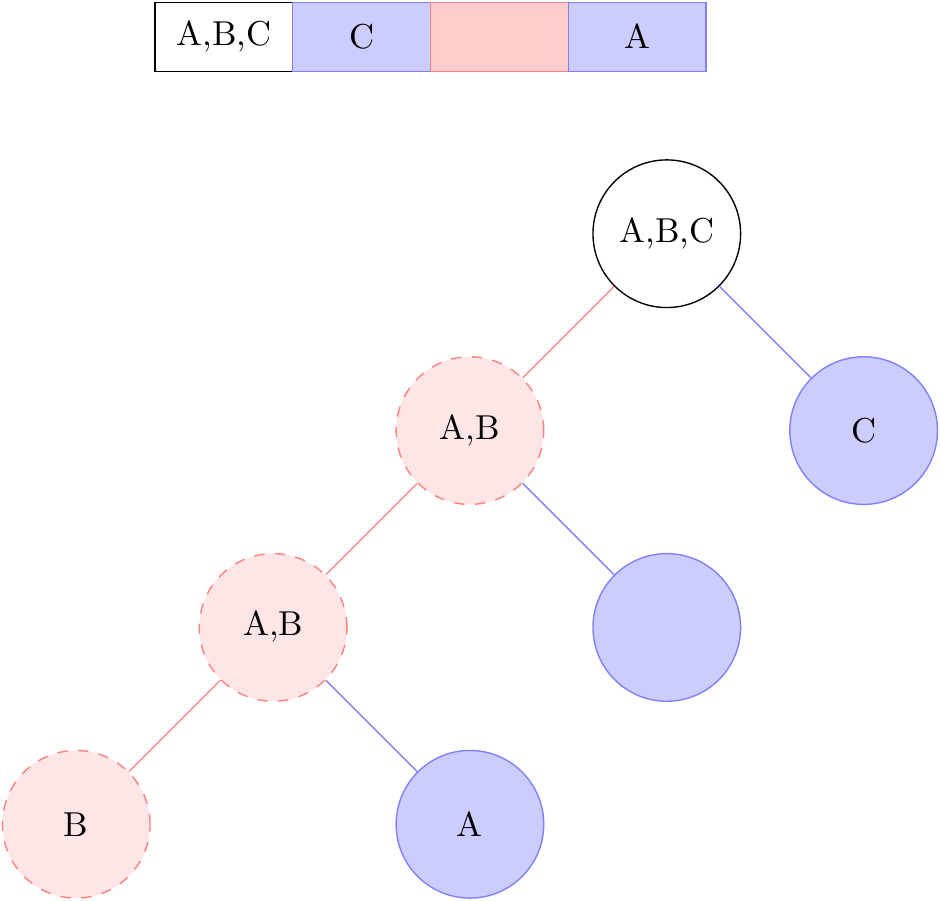}
\caption{\Ac{SICTA} collision resolution algorithm for the collision example of Figure~\ref{fig:sta_mta}.}
\label{fig:sicta}
\end{figure}
At the end of the second slot, \ac{SICTA} feeds back the information that user $C$ could be correctly decoded and maintains the signal of the first slot, where user $A$, $B$ and $C$ originally collided. Since only user $C$ is now available to the receiver, none of user $A$ and $B$ can yet be retrieved from the collision in the first slot, even with \ac{SIC}. Now, the users $A$ and $B$ defer the transmission in the third slot, knowing that a collision is deterministically going to happen. So slot three will be idle, and after a new coin throwing, only user $A$ transmits in the fourth slot. Thanks to \ac{SIC} also the packet of user $B$ can be correctly decoded in the very same slot, since the signal contribution of users $A$ and $C$ can be removed from the collision in the first slot. In this way, the receiver feeds back $k=2$. In general, it can be observed that the first slot in the left subtree can be omitted. This is a direct consequence of the \ac{SIC} process triggered by the decoding of users on the right subtree branches. In this way the \ac{CRI} length for \ac{SICTA} becomes
\begin{equation}
\label{eq:cri_sicta}
\CRIL_m=1+\CRIL_i+\CRIL_{m-i}-1=\CRIL_i+\CRIL_{m-i}\qquad m\geq2.
\end{equation}
As opposite to the previous cases, it holds for $i=0$ as well.

In the following an algorithmic description of \ac{SICTA} is presented. Two counters are kept active by each user. The first is the local counter $\cntL_t$ which is updated based on the received feedback and is used for selecting the action followed by the user in the next slot. The second is the system counter $\cntS_t$ and is employed for determining the \ac{CRI} boundaries. At the beginning of a new \ac{CRI}, both counters are reset to $0$. The update of the local counter $\cntL_t$ depends on the feedback as follows:
\begin{itemize}
\item If the feedback is $e$ and $\cntL_t>0$ then
\begin{equation}
\label{eq:sicta_case1}
\cntL_{t+1}=\cntL_t+1
\end{equation}
\item If the feedback is $e$ and $\cntL_t=0$ then
\begin{equation}
\label{eq:sicta_case2}
\cntL_{t+1}=
\begin{cases}
0,\quad \text{with probability $\prob$}\\
1,\quad \text{with probability $1-\prob$}
\end{cases}
\end{equation}
\item If the feedback is $0$ and $\cntL_t>1$ then
\begin{equation}
\label{eq:sicta_case3}
\cntL_{t+1}=\cntL_t
\end{equation}
\item If the feedback is $0$ and $\cntL_t=1$ then
\begin{equation}
\label{eq:sicta_case4}
\cntL_{t+1}=
\begin{cases}
0,\quad \text{with probability $\prob$}\\
1,\quad \text{with probability $1-\prob$}
\end{cases}
\end{equation}
\item If the feedback is $k\geq1$ then
\begin{equation}
\label{eq:sicta_case5}
\cntL_{t+1}=\cntL_t-(k-1).
\end{equation}
If $\cntL_{t+1}\leq0$, the packet of the user is successfully decoded. The user leaves the collision resolution. If instead $\cntL_{t+1}=1$ the user cannot be resolved by \ac{SIC} since other competitive users transmitted in the same slot. Then,
\begin{equation}
\label{eq:sicta_case6}
\cntL_{t+1}=
\begin{cases}
0,\quad \text{with probability $\prob$}\\
1,\quad \text{with probability $1-\prob$}
\end{cases}
\end{equation}
\end{itemize}
A user retransmits the packet in slot $t$ every time $\cntL_t=0$. The system counter is updated as follows
\begin{itemize}
\item If feedback is $e$
\begin{equation}
\label{eq:sicta_case7}
\cntS_{t+1}=\cntS_t+1
\end{equation}
\item If feedback is $0$
\begin{equation}
\label{eq:sicta_case8}
\cntS_{t+1}=\cntS_t
\end{equation}
\item If feedback is $k$
\begin{equation}
\label{eq:sicta_case9}
\cntS_{t+1}=\cntS_t-(k-1).
\end{equation}
\end{itemize}
The current collision resolution algorithm stops when the system counter reaches $0$.

In \cite{Giannakis_2007} the authors show that for the gated access the maximum stable throughput of \ac{SICTA} is $0.6931$. Interestingly it is achieved for binary $\tsd=2$ fair splitting, unlike the standard tree algorithm where the optimum was achieved with ternary $\tsd=3$ splitting (instead of dividing the users into two classes depending on the outcome of the coin flipping procedure, they are divided into three classes). For \ac{SICTA}, the maximum stable throughput $\tp$ follows\footnote{Details on the stable throughput expression are here omitted and can be found in \cite{Giannakis_2007}.}
\begin{equation}
\label{eq:sicta_T}
\tp\approx \frac{\ln \tsd}{\tsd-1}.
\end{equation}
We can observe that,
\begin{itemize}
\item For $\tsd \rightarrow \infty$ the maximum stable throughput of \ac{SICTA} degenerates rapidly towards the standard tree algorithm;
\item \ac{SICTA} exceeds the upper bound derived by Tsybakov and Likhanov of $0.568$ \cite{Tsybakov1980}. We shall note here that the upper bound by Tsybakov and Likhanov is valid for ternary feedback (idle, collision, success). Indeed, as pointed out in a seminal work of Massey \cite{Massey1988}, a more proper bound for random access with multiplicity feedback, as \ac{SICTA} is, is due to Pippinger \cite{Pippinger1981} and it is $1$.
\end{itemize}
Moving to the window access, the maximum stable throughput $\tp$ is as well $0.693$ but it is reached for a maximum window size tending to infinite, \ie $\maxWinSize \rightarrow \infty$, which corresponds to the gated access. For an average arrival rate of $1$ packet per window duration $\maxWinSize$, the maximum stable throughput is reduced to $0.6$ for \ac{SICTA}.

Interestingly, \ac{SICTA} performs better for gated access, which is easier to implement w.r.t. the window access. The window size has to be optimised for each arrival rate and the maximum stable throughput of the modified tree algorithm of Gallager $0.487$ can be guaranteed only for Poisson arrivals. In \ac{SICTA}, instead, Poisson arrivals are not necessary for demonstrating the maximum stable throughput of $0.693$ \cite{Giannakis_2007}.

\subsubsection{An Extension of \ac{SICTA}}

Recently, the authors in \cite{Peeters2015} extended the work of \ac{SICTA} and showed that their protocol is able to achieve a maximum stable throughput of $1$ under gated access and the collision channel (which closes the gap to the upper bound derived by Pippinger \cite{Pippinger1981}). Nevertheless, this is possible only allowing unbounded computational complexity and therefore is of limited practical value. The authors show also an hybrid approach between their algorithm and \ac{SICTA} able to reach a maximum stable throughput of $1-\epsilon$ with $\epsilon$ chosen arbitrarily, given sufficient computational power.

The key of the algorithm is the collision resolution mode. At the start of the \ac{CRI} all users are required to transmit their packet. From the second slot of the \ac{CRI} onwards each user transmits its packet with probability $\prob=1/2$. At the receiver side these operations are performed:
\begin{itemize}
\item A matrix $\mx\in\mathbb{B}^{m\times m}$ is created, such that, each column represents a time slot, and each row represents an user. A $1$-entry in position $i,j$ represents the transmission of a user $i$ in slot $j$. The dimension $m$ represents both the collision size and the collision resolution duration in time slots.
\item The matrix $\mx$ contains only linearly independent columns so it is invertible. Taking the inverse of $\mx$, $\mx^{-1}$ gives us the solution for all users involved in the collision. In particular, the solution for the first user is given by the first column. Each of the coefficients of $\mx^{-1}$ in the first column applied to the time slot signals \--- received in the $N$ time slots of the \ac{CRI} \--- gives the transmitted packet of the first user.
\end{itemize}
The problem of computational complexity arises since $m$ is unknown at the receiver. The authors propose a brute force approach \cite{Peeters2015} where at each received slot, all possible values of each unknown in $\mathbf{M}$ is tested. Since $\mathbf{M}$ is binary, the search is bounded for every value of $m$.

The authors are able to upper bound the average time of the \ac{CRI} $\CRIL_m$ as
\begin{equation}
\label{eq:cri_sicta_ev}
\CRIL_m\leq m+ \sum_{k=1}^{m-1}\frac{1}{2^k-1}\approx m \qquad \text{for sufficiently large $m$}.
\end{equation}

\subsection{Random Access Without Feedback}

If feedback is not possible due to the lack of dedicated return channel or due to hardware limitations, we speak about \ac{RA} without feedback. A way to approach channel capacity in such scenarios is to let users adopt \emph{protocol sequences} which dictates the time slots where the users are active. Protocol sequences are deterministic, but each user picks a sequence at random from the pool and the user will become active in a time slot which is independent from the other users, preserving the random nature of the channel access.

Protocol sequences for the \ac{RA} without feedback channel have been mainly studied for the collision channel. Only recently authors in \cite{Zhang2016} investigated the protocol sequences' design and performance for the \ac{iota_MPR} channel. The \ac{iota_MPR} channel model ensures that collisions involving up to $\iota$ users can be successfully resolved, \ie all the $\iota$ users involved in the collision can be successfully decoded by means of \ac{MPR} techniques.\footnote{The \ac{iota_MPR} channel model can emulate the behaviour of a receiver able to do interference cancellation, when adopting the proper $\iota$. More in general, the adoption of the \ac{MPR} matrix representation \cite{Ghez1988} can lead to a precise characterization of the \ac{SIC} behaviour but it is not considered by the authors in the paper \cite{Zhang2016}.}

The sequences are periodic and in general the throughput depends on the relative shifts among the users. A first class is the so-called \ac{SI} sequences. Their peculiar property is that their generalised Hamming cross-correlation functions are independent of relative shifts.\footnote{The generalised Hamming cross-correlation function can be seen as the number of slot indices where the two sequences have both an entry equal to one. A formal definition can be found in \cite{Zhang2016a}. For \ac{SI} sequences in particular, the generalised Hamming cross-correlation function it is invariant to sequence shifts.} A second class are the \ac{TI} sequences where instead the throughput is shift invariant. The advantages of \ac{TI} sequences are strictly positive and maximal worst case throughput among all protocol sequence based access schemes and ease to use with \ac{FEC} due to the shift invariant property.

Let us have a look at the system model. We consider a slot synchronous system where all $\numUs$ users' transmissions are synchronised at slot level. Each user's packet occupies exactly one time slot. The user $i$'s protocol sequence is a deterministic and periodic binary sequence $\bs^{(i)}:=\left[\bse_0^{(i)}, \bse_1^{(i)},...,\bse_{\lbs-1}^{(i)}\right]$ that determines the instants in time where user $i$ will transmit, \ie user $i$ transmits in slot $t$ iff $\bse_{t-\ts_i}^{(i)}=1$. The protocol sequence duration or sequence period in time slots is $\lbs$. The subtraction is therefore a modulo-$\lbs$ subtraction. User $i$ has a relative shift of $\ts_i$ time slots w.r.t. the receiver's common clock time indexed by $t$. The relative shifts $\ts_i,i=1,...,\numUs$ of the users are independently chosen and are not known to the common receiver. The duty factor $\df^{(i)}$ of user $i$'s protocol sequence is defined as the fraction of ones in $\bs^{(i)}$ over its length, or more formally
\begin{equation}
\label{eq:def_duty_fac}
\df^{(i)}:=\frac{1}{\lbs}\sum_{t=0}^{\lbs-1} \bse_t^{(i)}.
\end{equation}
It is assumed that each sequence has a non-zero duty factor and that there are at most $\iota-1$ all one sequences.

\subsubsection{Throughput Analysis}

The analysis of \cite{Zhang2016} aims at deriving the throughput of \ac{TI} sequences and shows that it is dependent on the duty factors of the sequences only. The full derivation including the case of non symmetric duty factors for the users is presented in \cite{Zhang2016}, while here only the final throughput result for the symmetric duty cycle case is recalled:
\begin{equation}
\label{eq:without_F_T}
\tp=\numUs \sum_{k=1}^{\iota-1}\binom{\numUs-1}{k}\df^{k+1}(1-\df)^{\numUs-1-k}.
\end{equation}
The channel load $\load$ is then $\load=\numUs \df$. As mentioned by Zhang, the throughput $\tp$ of the system depends only on the duty factor, on the maximum collision size such that all packets can be still correctly decoded the parameter $\iota$ and the number $\numUs$ of users in the system. For every couple of total active users $\numUs$ and $\iota$, there is an optimal duty factor that maximises the throughput. This can be found following the approach of \cite{Bae2014}, where the transmission probability is the equivalent of the duty factor. Noticeably, no feedback is required in this approach with respect to the slotted ALOHA system considered in \cite{Bae2014}.
\begin{figure}
\centering
\includegraphics[width=0.8\textwidth]{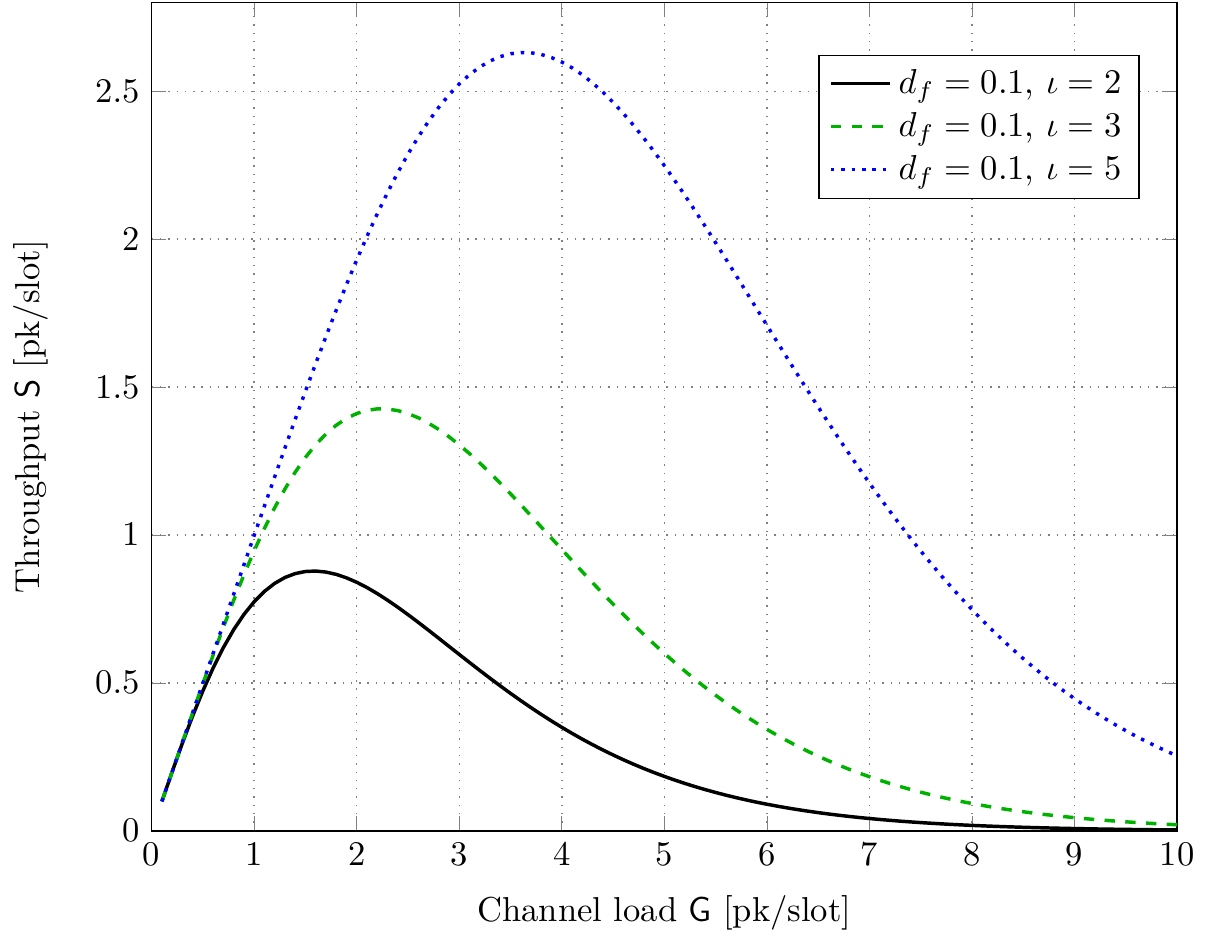}
\caption{Throughput as a function of the channel load for \ac{TI} sequences under the \ac{iota_MPR} channel for different values of $\iota$.}
\label{fig:without_f}
\end{figure}

\subsubsection{Construction of Minimum Period TI Sequences}

In this Section we present the construction of minimum period \ac{TI} sequences. The importance of minimum period sequences relies in the higher robustness against performance variability over short time periods given the same target performance. Without going into theoretical details, Zhang and his co-authors demonstrated that minimum period \ac{TI} sequences can be built from \ac{SI} sequences. A possible construction can be found in \cite{Shum2009} and is recalled here through an example.

\begin{sloppypar}
We specify the $\numUs$ duty factors as $u_1/v_1=\df^{(1)},u_2/v_2=\df^{(2)},..., u_{\numUs}/v_{\numUs}=\df^{(\numUs)}$. For each user we build the matrix $\mx^{(i)}=\left[\vc_1^{(i)}, \vc_2^{(i)},...,\vc_{v_i}^{(i)} \right]$. The matrix $\mx^{(i)}$ is binary with the constraint that in each row there are exactly $u_i$ ones and it has dimension ${(\prod_{k=1}^{i-1}v_k)\times v_i}$, \ie $\mx^{(i)} \in \mathbb{B}^{(\prod_{k=1}^{i-1}v_k)\times v_i}$.\footnote{$\prod_{k=1}^0 v_k=1$  by convention.} The sequence $\bs^{(i)}$ is generated by repeating the row vector $\left[\vc_{1}^{(i)T},\vc_{2}^{(i)T}, ... ,\vc_{\df^{(i)}}^{(i)T}\right]$ for $\frac{\prod_{k=1}^{\numUs} v_k}{\prod_{h=1}^{i} v_h}$ times,
\begin{equation}
\label{eq:sequence}
\bs^{(i)} = \underbrace{\left[\vc_{1}^{(i)T}, ... ,\vc_{\df^{(i)}}^{(i)T},  \vc_{1}^{(i)T}, ... ,\vc_{\df^{(i)}}^{(i)T}, ... , \vc_{1}^{(i)T}, ... ,\vc_{\df^{(i)}}^{(i)T}\right]}_\text{$\prod_{k=1}^{\numUs} v_k$ sequence entries}
\end{equation}
\begin{exmp}
We focus on a three users symmetric case with ${\df^{(1)}=\df^{(2)}=\df^{(3)}=\df=1/2}$. The three matrices built are
\begin{equation}
\label{eq:mx_example_without_F}
\mx^{(1)}=
\begin{bmatrix}
1 & 0
\end{bmatrix} \quad
\mx^{(2)}=
\begin{bmatrix}
1 & 0\\
0 & 1
\end{bmatrix} \quad
\mx^{(3)}=
\begin{bmatrix}
1 & 0\\
0 & 1\\
1 & 0\\
0 & 1
\end{bmatrix}.
\end{equation}
The following three bit sequences can be obtained:
\begin{align}
\label{eq:example_sequence}
\bs^{(1)} &= \left[1 0 1 0 1 0 1 0 \right]\\
\bs^{(2)} &= \left[1 0 0 1 1 0 0 1 \right]\\
\bs^{(3)} &= \left[1 0 1 0 0 1 0 1 \right].
\end{align}
\end{exmp}
\end{sloppypar}

\subsubsection{Some Recent Results of SI Sequences with Interference Cancellation}

Recently, the research group published a work in which the protocol sequences used in the collision channel without feedback are investigated together with \ac{SIC} \cite{Zhang2016a}. The main result achieved is the demonstration that, with proper sequence selection and when \ac{SIC} is employed at the receiver, the limit of $1$ packet per slot is reachable. Some important consequences are
\begin{itemize}
\item The zero-error capacity region of a \ac{RA} system without feedback employing \ac{SIC} is coincident with the one of a system guaranteeing orthogonal multiple access, \eg \ac{TDMA}. In this way, the lack of orthogonality, which does not guarantee interference free transmissions to the user, does not harm the performance as long as \ac{SIC} can be performed at the receiver. It is here important to underline that ideal \ac{SIC} is assumed, \ie perfect interference cancellation of correctly decoded packets is performed.
\item The symmetric case (typical for many systems) in which all users transmit with the same rate is on the outer boundary of the capacity region, \ie also the symmetric case achieves $1$ packet per slot capacity. This result is particularly interesting because it is in contrast with the case without \ac{SIC} \cite{Massey1985_RAWithoutFeedback}. In that case in fact, the symmetric case minimises the capacity.
\item The outer boundary of the zero-error capacity region with rational components can be achieved (which means that it is possible to reach zero-error capacity with rate points that lie on the outer boundary) only with the help of \ac{SI} sequences. Minimum period \ac{SI} sequences are therefore the sequences with smallest period possible reaching the outer boundary of the zero-error capacity region. Construction of these follows the procedure presented in the previous Section.
\item For the symmetric case, the minimum period is $\numUs!$, which is consistently smaller than $\numUs^{\numUs}$, the minimum period of the case without \ac{SIC} \cite{Zhang2016}.
\end{itemize}

\section{Applicable Scenarios}



The improvements in the efficiency of recent \ac{RA} protocols, observable from all the previously presented schemes and not limited to them, widen the scenarios where \ac{RA} schemes can be used. A selection of four possible areas in which recent \ac{RA} schemes are identified as suitable candidates for future communication systems or standards is the following:
\begin{itemize}
\item satellite communication systems and \ac{M2M};
\item \ac{LTE} and the upcoming \ac{5G};
\item \ac{V2V} communication systems;
\item Underwater communication systems.
\end{itemize}
For each of the four areas we give a brief overview of the typical challenges as well as some hints on how RA protocols manage to overcome them.

\subsubsection{Satellite Communication Systems and \ac{M2M}}

From the origin of ALOHA, the typical scenario of such \ac{RA} protocols has been the satellite communication system. The distance between transmitting terminals prevents the use of sensing techniques deployed in terrestrial \ac{RA} based systems, such as \ac{CSMA}, while the peculiar very long round trip time limits the efficiency of on-demand systems such as \ac{DAMA}. The former is due to the wide footprint of \ac{GEO} satellites guarantee connectivity to terminals and that can be thousands of kilometers far apart. The latter, instead, is due to the extreme transmission delay (up to $500$ ms round trip time delay) that a \ac{GEO} satellite link suffers which affects any communication and may restrict the usage of handshake mechanisms necessary for \ac{DAMA}-like protocols to allocate resources to the terminals.

This is the reason why a number of different satellite standards foresee the use of \ac{RA} protocols. The \ac{DVB-RCS2} \cite{dvb_rcs2} adopts both \ac{CRDSA} \cite{Casini2007} and \ac{IRSA} \cite{Liva2011} as options, not only for the user terminal logon at the beginning of a session, but also foresees the possibility to use these protocols for data communication through reserved slots in the \ac{MF-TDMA} frames. The emerging market of mobile satellite terminals calls for development of suitable satellite systems able to cope with this extremely challenging scenario. A first answer comes from the novel \ac{S-MIM} standard \cite{Scalise_2013}. It comprises forward and return links for bi-directional communication and for the return link (from gateway, through satellite to mobile terminal) one of the two options is to use \ac{RA} for communication. In particular \ac{E-SSA} \cite{delRioHerrero_2012} is foreseen as the \ac{RA} \ac{MAC} option.

Nowadays, the use of satellites for communication is the only option for ships that are in oceanic travel. In order to incorporate the \ac{AIS} and \ac{ASM} \cite{Ais_std} messaging systems in a more wide communication system, the novel \ac{VDES} has been recently introduced \cite{Vdes_std} and is currently under standardization by \ac{ITU}. Safety critical messages for collision avoidance, as provided by \ac{AIS}, are sharing the channel with other types of messages. \ac{VDES} aims at releasing the \ac{AIS} channels from non-safety critical traffic which is reverted to \ac{ASM} and other dedicated channels. Satellite and terrestrial components are both considered for non-safety critical messaging. For the former, \ac{RA} with both narrowband and spread spectrum communication systems is allowed. In this regard both \ac{CRA} \cite{Kissling2011a} \ac{ECRA} \cite{Clazzer2012} or \ac{E-SSA} \cite{delRioHerrero_2012} appear to be good candidates.

\subsubsection{\ac{LTE} and \ac{5G}}

The \ac{M2M} application in terrestrial mobile networks is a big driver as well. Many recent works focus on the investigation of \ac{RA} for \ac{M2M} applications in \ac{LTE} or \ac{LTE-A} \cite{Laya2014}. In today's \ac{LTE} the \ac{RA} procedure comprises a four messages handshake in order to establish a connection. Precisely, \ac{RA} is used in \ac{LTE} for the network logon or for request of resources for transmission or re-establish a connection after failure. Alternatives that avoid the need of a four message handshake providing the same reliability can be foreseen from the discussed protocols of the previous Sections.

One of the consortia looking at the upcoming \ac{5G} is proposing asynchronous \ac{RA} for some \ac{M2M} traffic classes, especially the ones of \ac{IoT} where sporadic and low duty cycles prevail over typical traffic \cite{5GNow}. Foreseen advantages are drastic reduction in signaling overhead as well as reduction of device energy consumption due to relaxation of synchronization requirements.

\subsubsection{Vehicle-to-Vehicle (V2V) Communication Systems}

A dedicated working group of IEEE is working on a modification of the 802.11 standard in order to satisfy the peculiar features of a \ac{V2V} scenario \cite{Menouar2006}. The \ac{WAVE} standard, \ie 802.11p, proposes the use of \ac{CSMA/CA} as \ac{MAC} protocol. Unfortunately, under several aspects \ac{CSMA/CA} appears to be not the optimal choice. In fact, it provides large throughput when the user population is rather limited, the traffic generated by each terminal is large and there are limited delay constraints. There are many situations in which \ac{V2V} faces different conditions from the one better suited for \ac{CSMA/CA}. In this regard alternatives are required. Ivanov and his co-authors propose a modification of the \ac{CSA} protocol \cite{Paolini2015}, the so-called \emph{all to all \ac{CSA}} \cite{Ivanov2015}, in which \ac{CSA} operates in a scenario where terminals are both transmitters and receivers. When a terminal is not transmitting, it has the receiver circuit activated and is able to possibly decode incoming messages. Many of the limitations of \ac{CSMA/CA} are shown to be overcome by this protocol \cite{Ivanov2015}.

\subsubsection{Underwater Communication Systems}

In recent years, driven by the development of underwater sensors and underwater manned and unmanned vehicles, underwater ad-hoc networks have been deeply investigated. Due to their peculiar features imposed by the use of acoustic waves as communication medium like high propagation delay, limited bandwidth and data rate, noise, energy consumption and high bit error rates \cite{Chen2014}, the design of efficient \ac{MAC} protocol is of upmost importance. Due to the similarity with the satellite channel, ALOHA-like protocols have been proposed as suitable alternatives \cite{Cui2006,Pompili2009}. Interestingly, advanced signal processing techniques that include \ac{SIC}, have not been studied in the context of underwater communication systems yet.

\section{Open Questions}

Although \ac{RA} has been well investigated in past years, it still present open research fields areas. In this Section we give a non-comprehensive list of unanswered questions. It is important to mention here that this is a partial and biased list and does not have the aim to cover all important open questions and issues of the field.

\begin{itemize}
\item \textbf{Comparison between time synchronous and asynchronous \ac{RA} with time diversity and \ac{SIC}.} Since the introduction of time diversity with the presence of replicas, e.g. \ac{CRDSA} or \ac{CRA}, and advanced signal processing at the receiver with \ac{SIC}, a fair and comprehensive comparison between time synchronous and asynchronous schemes is missing. Furthermore, the presence of combining, \eg \ac{ECRA}, provides a further dimension for the comparison that shall be taken into account as well. Finally, spread spectrum techniques with \ac{SIC} shall also be included in the picture. A first attempt was done in \cite{ESA_survey_2016} where a first comparison between \ac{CRDSA}, \ac{ACRDA} and \ac{E-SSA} was presented. Nevertheless, the focus of this Capter is different and therefore this comparison has been done for a single specific setting. Another work worth to be mentioned is \cite{Biason2016} where \ac{ACRDA} and \ac{E-SSA} are compared.
\item \textbf{Energy efficiency.} Traditionally, \ac{RA} schemes without sensing capabilities have been used in scenarios where the energy efficiency per successful transmission has never been an issue, like in satellite networks where the terminals are normally connected to the power line. Nevertheless, these schemes are gaining momentum for other type of applications, like sensor networks, were instead the energy consumption is one of the key design factors. In this regard, the energy efficiency of such advanced schemes is only partially considered as side effect in works as \cite{Liva2011} and requires much more attention. Not only the transmitter side is of upmost interest but also the receiver side. All schemes adopting \ac{SIC} will have a burden in complexity at the receiver that will definitively have a big impact on the energy consumption. Battery-based receivers are required to save energy as much as possible and such investigation has to be carried out for letting these schemes become appealing.
\item \textbf{Delay performance.} The use of \ac{SIC} brings undoubtable advantages. In many situations it has also a positive impact in the delay performance, \eg in ZigZag or \ac{SICTA} the possibility to use \ac{SIC} reduces the probability of retransmission in the first case and the \ac{CRI} average duration in the second case, such that in both cases lower average delay is expected. Unfortunately this is not the case when we consider \ac{CRDSA}  or \ac{CRA}. There, in fact, packets prior decoding need to wait for the reception of entire frames and the delay performance can become worse than in \ac{SA} or ALOHA.
\item \textbf{Robustness against traffic and channel conditions.} The more the protocol performance is pushed to the limit, the more the sensibility to change in the traffic conditions increases. In this way, the channel load, the channel conditions, or received \ac{SNR}, may impact negatively the scheme behaviour as soon as they differ from the values considered in the scheme optimisation. It would be important to assess robustness of such schemes to variations in the traffic conditions as well as in the channel conditions.
\end{itemize} 
\section{Conclusions}

In this Chapter the role of interference cancellation applied to recent \ac{RA} protocols was highlighted. Starting from the original ideas in ALOHA and \ac{SA}, several recent schemes were proposed. Many of them shared the use of \ac{SIC} at the receiver to improve the performance. Four schemes were selected, each one representing a specific class of \ac{RA} schemes and their basics were recalled. For some schemes, performance figures were presented, while for others the design basics were described. The Chapter gave an overview of scenarios that can particularly benefit from such advanced protocols, ranging from satellite networks, to terrestrial (as \ac{5G}) and going towards vehicular and underwater communication systems. Finally a selection of open questions was given.
 
\chapter[Asynchronous RA with Time Diversity and Combining: ECRA]{Asynchronous Random Access with Time Diversity and Combining: the ECRA Approach}
\label{chapter4}
\thispagestyle{empty}
\ifpdf
    \graphicspath{{chapter4/Images/}}
\fi
\epigraph{Some great things are born from laziness and meditation}{Elliott Erwitt}

The fourth Chapter presents the novel \ac{ECRA} decoding procedure for asynchronous random access protocols. We first focus on the system model detailing the transmitter and receiver operations while highlighting the use of \ac{SIC} in combination with combining techniques, as the key aspect of the \ac{ECRA} decoding procedure. We present an analytical approximation for the \ac{PLR} performance of asynchronous random access tight for moderate channel load conditions which includes \ac{ECRA} as a special case. Numerical results evaluate the goodness of such approximation and rise the attention to scenarios where \ac{ECRA} is able to outperform slot synchronous schemes as \ac{CRDSA}. In order to benefit from combining techniques, the perfect knowledge of replicas position has to be guaranteed prior decoding. To achieve this stringent requirement, the two last Sections present a modification of \ac{ECRA}. A two-phase approach that avoids the need to signal the replicas position within the header, but relies only on non-coherent soft-correlation for the detection and replicas matching is introduced. The Chapter finishes with concluding remarks.

\section{Introduction}

In asynchronous \ac{RA} protocols, \ac{CRA} \cite{Kissling2011a} has been the first attempt to mimic the improvements given by \ac{CRDSA}. Time slots are removed, but frames are kept and users are allowed to transmit their replicas within the frame without any constraint except avoiding self-interference. At the receiver, \ac{SIC} is employed to improve the performance, similarly to the slotted counterpart \ac{CRDSA}. Recently, \ac{ACRDA} \cite{deGaudenzi2014_ACRDA} has removed also the frame structure still present in \ac{CRA}, reducing once more the transmitter complexity. However, both \ac{CRA} and \ac{ACRDA} do not exploit the inherent time diversity of the interference among replicas which naturally arises due to the asynchronous nature of the protocol, \ie different portions of replicas of a given user might be interfered.

Driven by this observation, the present Chapter introduces the \ac{ECRA} asynchronous \ac{RA} scheme. It employs combining techniques in order to resolve collision patterns where \ac{SIC} alone is unable to succeed. The main contributions can be summarised as:
 \begin{itemize}
 \item{\emph{Extension of asynchronous \ac{RA} protocols towards combining techniques} such as \ac{SC}, \ac{EGC} and \ac{MRC} \cite{Brennan1959,Jakes1974}. The novel \ac{ECRA} exploits time diversity of the interference pattern suffered by the replicas for creating a combined observation at the receiver on which decoding is attempted.
 \item{\emph{Development of an analytical approximation of the \ac{PLR} performance for asynchronous \ac{RA} schemes} particularly tight for low channel load. The approximation focuses on a subset of collision patterns unresolvable with \ac{SIC}.}
 \item{\emph{Comparison of \ac{ECRA} with asynchronous and slot synchronous protocols under several metrics} as throughput, spectral efficiency and normalised capacity.}}
 \end{itemize}


\section{System Model}
\label{sec:sys_ov}
We assume an infinite user population generating traffic following a Poisson process of intensity $\load$. The channel load\footnote{The channel load corresponds to the \emph{logical load} $\load$ since it takes into consideration the net information transmitted, depurated from the number of replicas per user $\dg$.} $\load$ is measured in packet arrivals per packet duration $\pkLen$. Upon arrival, each user replicates his packet $\dg$ times, with $\dg$ the repetition degree of the system. The first replica is transmitted immediately while the remaining $\dg-1$ are sent within a \ac{VF} of duration $\fraLen$ starting at the beginning of the first replica.\footnote{It is important to underline that the concept of \ac{VF} has been firstly introduced in \ac{ACRDA} \cite{deGaudenzi2014_ACRDA} and was not present neither in \ac{CRA} nor in the first statement of \ac{ECRA} \cite{Clazzer2012}.} As a consequence, virtual frames are asynchronous among users. Replicas are sent such that self-interference is avoided. The time location within the \ac{VF} of each replica is stored in a dedicated portion of the packet header. Each replica is composed by $\numBit$ information bits. In order to protect the packets against channel impairments and interference, we adopt a channel code $\Code$ with Gaussian codebook. We define the coding rate $\rate=\numBit/\numSym$, where $\numSym$ is the number of symbols within each packet after channel encoding and modulation. We denote with $\symLen$ the duration of a symbol so that $\pkLen=\symLen \numSym$. Replicas are then transmitted through an \ac{AWGN} channel.

Let us consider the transmitted signal $\tx^{(\user)}$ of the $\user$-th user,
\begin{equation*}
\label{eq:tx_sig}
\tx^{(\user)}(\tm)= \sum_{i=0}^{\numSym-1} \symVal_i^{(\user)} \pulse(\tm-i \symLen).
\end{equation*}
Where $\bm{\symVal}^{(\user)}=\left( \symVal_0^{(\user)}, \symVal_1^{(\user)}, \dots, \symVal_{\numSym-1}^{(\user)} \right)$ is the codeword of user $\user$ and $\pulse(\tm)=\mathcal{F}^{-1}\left\{\sqrt{\mathtt{CR}(f)}\right\}$ is the pulse shape, being $\mathtt{CR}(f)$ the frequency response of the raised cosine filter. The signal of the generic user $\user$ is affected by a frequency offset modeled as an uniformly distributed random variable $\freq^{(\user)} \sim \mathcal{U}\left[-f_{\mathrm{max}};f_{\mathrm{max}}\right]$ and an epoch, also modeled as an uniformly distributed random variable $\epoch^{(u)} \sim \mathcal{U}\left[0;T_s\right)$. Both frequency offset and epoch are common to each replica of the same user, but independent user by user. The phase offset is modeled as a random variable uniformly distributed between $0$ and $2\pi$, \ie $\phase^{(\user,\replica)} \sim \mathcal{U}\left[0;2\pi\right)$, and it is assumed to be independent replica by replica. Assuming that $f_{\mathrm{max}}T_s\ll 1$, the received signal $\rx(\tm)$ after matched filtering can be approximated as
\begin{equation}
\label{eq:rx_sig}
\rx(\tm) \cong \sum_{\user} \sum_{\replica=0}^{\dg-1} \hat{\tx}^{(\user)}(\tm - \epoch^{(u)} - \VFStart^{(\user,\replica)} - \RefTime^{(\user)}) e^{j\left(2\pi \freq^{(\user)} + \phase^{(\user,\replica)}\right)} + \noise(\tm).
\end{equation}
With $\hat{\tx}^{(\user)}=\sum_{i=0}^{\numSym-1} \symVal_i^{(\user)} \hat{\pulse}(\tm-i \symLen)$, where $\hat{\pulse}(\tm) = \mathcal{F}^{-1}\left\{\mathtt{CR}(f)\right\}$. In equation~\eqref{eq:rx_sig}, $\VFStart^{(\user,\replica)}$ is the delay w.r.t. the \ac{VF} frame start for user $\user$ and replica $\replica$, while $\RefTime^{(\user)}$ is the $\user$-th user delay w.r.t. the common reference time. The noise term $\noise(\tm)$ is given by $\noise(\tm) \triangleq \noisePr(\tm) \ast \pulseRoot(\tm)$, where $\noisePr(\tm)$ is a white Gaussian process with single-sided power spectral density $\noiseSD$ and $\pulseRoot(\tm)$ is the \ac{MF} impulse response of the root raised cosine filter, \ie $\pulseRoot(\tm) = \mathcal{F}^{-1} \left\{ \sqrt{\mathtt{CR}(f)} \right\}$.

For the $\user$-th user, $\replica$-th replica, assuming an ideal estimate of the epoch $\epoch^{(u)}$, the frequency offset $\freq^{(\user)}$ and the phase offset $\phase^{(\user,\replica)}$, the discrete-time version of the received signal ${\rxVec^{(\user,\replica)} = (\rx_0^{(\user,\replica)}, ..., \rx_{\numSym-1}^{(\user,\replica)})}$ is given by
\begin{equation*}
\rxVec^{(\user,\replica)} = \txVec^{(\user)} + \intVec^{(\user,\replica)} + \noiseVec.
\end{equation*}
Here $\txVec^{(\user)} = \bm{\symVal}^{(\user)}$, and $\intVec^{(\user,\replica)}$ is the interference contribution over the user-$\user$ replica-$\replica$ signal and $\noiseVec=(\noise_0,...,\noise_{\numSym-1})$ are the samples of a complex discrete white Gaussian process with ${\noise_i \sim \mathcal{CN}(0,2\noiseVar)}$.

The instantaneous \ac{SINR} $\sinr$ for the $i$-th sample of the $\user$-th user $\replica$-th replica is
\begin{equation*}
\sinr_{i}^{(\user, \replica)} = \frac{\usPw_{i}^{(\user)}}{\noisePw + \intPw_{i}^{(\user, \replica)}}
\end{equation*}
with $\usPw_{i}^{(\user)} \triangleq \mathbb{E}\left[|\symVal_{i}^{(\user)}|^2 \right]$, $\noisePw=2\noiseVar$ and $\intPw_{i}^{(\user, \replica)} \triangleq \mathbb{E}\left[ |\intVal_{i}^{(\user,\replica)}|^2 \right]$, which is the aggregate interference power contribution on the $i$-th sample of the considered replica. Throughout the Chapter, we assume that all users are received with the same power, \ie perfect power control is adopted. Hence, $\usPw_{i}^{(\user)} = \usPw$ and $\intPw_{i}^{(\user, \replica)} = \intNum_i^{(\user, \replica)} \usPw$, where $\intNum_i^{(\user, \replica)}$ denotes the number of active interferers over the $i$-th symbol of the $\user$-th user $\replica$-th replica. The aggregate interference is modeled as a memoryless discrete Gaussian process with $\intVal_{i} \sim \mathcal{CN}\left(0,\intNum_i^{(\user, \replica)} \usPw\right)$ and the \ac{SINR} thus becomes
\begin{equation*}
\sinr_{i}^{(\user, \replica)} = \frac{\usPw}{\noisePw + \intNum_i^{(\user, \replica)} \usPw}.
\end{equation*}

\begin{sloppypar}
The \ac{SINR} vector over the $\numSym$ symbols of the considered replica is denoted with ${\sinrVec^{(\user, \replica)} = (\sinr_{0}^{(\user, \replica)}, \sinr_{1}^{(\user, \replica)}, ..., \sinr_{\numSym-1}^{(\user, \replica)})}$.
\end{sloppypar}

\subsection{Modeling of the Decoding Process}
\label{sec:int_model}

Typically, the destructive collision channel model is adopted \cite{Tong2004} in the analysis of the \ac{MAC} layer of \ac{RA} protocols. This physical layer abstraction assumes that only packets received collision-free can be correctly decoded, while all packets involved in collisions are lost. For asynchronous schemes, where packets are protected with a channel code, this assumption is particularly pessimistic. In fact, low levels of interference can be counteracted by the error correction code and some collisions can be resolved.

Motivated by this, we resort to a \emph{block interference model} \cite{McEliece1984} given by $\numSym$ parallel Gaussian channels \cite{cover2006} (one for each replica symbol), where the $i$-th channel is characterised by a \ac{SNR} $\sinr_{i}$.\footnote{We are omitting here the superscript $(\user, \replica)$ for ease of notation.} The instantaneous mutual information over the $i$-th channel $\mutInf(\sinr_{i})$ is
\begin{equation*}
\mutInf(\sinr_{i}) = \log_2( 1 + \sinr_{i} ).
\end{equation*}
Differently from the classical parallel Gaussian channel problem of finding the best power allocation per channel in order to maximise capacity (cf. Chapter~10.4 of \cite{cover2006}), here \ac{CSI} is not present at the transmitter since the interference contribution cannot be predicted due to the uncoordinated user transmissions. Therefore, the power allocation over the channels, \ie symbols of the replica, is kept constant and is not subject to optimisation. The instantaneous mutual information averaged over the $\numSym$ parallel channels is
\begin{equation}
\label{eq:mi_dec}
\mutInf(\sinrVec) = \frac{1}{\numSym} \sum_{i=0}^{\numSym-1} \mutInf(\sinr_{i}) = \frac{1}{\numSym} \sum_{i=0}^{\numSym-1} \log_2( 1 + \sinr_{i} ).
\end{equation}
Interference has been modeled similarly in \cite{Thomas_2000}. We introduce a binary variable $\rvDec$ modelling the decoding process, such that
\begin{equation*}
\begin{split}
\rvDec &= 1 \qquad \text{if decoding succeeds}\\
\rvDec &= 0 \qquad \text{otherwise}.
\end{split}
\end{equation*}
We have
\begin{equation}
\label{eq:dec_res}
\rvDec = \mathbb{I}\left\{\rate \leq \mutInf(\sinrVec)\right\}
\end{equation}
where $\mathbb{I}\{X\}$ denotes the indicator (Inverson) function.\footnote{This model allows to take into account features like channel coding, multi-packet reception and capture effect \cite{Ghez1988, Zorzi1994}.} Observe that, the destructive collision model is a special case, where the rate $\rate$ is chosen such that only packets collision-free can be succesfully decoded, \ie $\rate = \log_2\left(1+\frac{\usPw}{\noisePw}\right)$. The decoding process model based on the threshold induced by the selected rate, has some non-negligible effect on the performance with respect to more accurate models that take into account the specific channel code and block length. Nevertheless, it is a good first approximation for highlighting the improvements given by the proposed scheme.

\subsection{Enhanced Contention Resolution ALOHA Decoding Algorithm}

\begin{figure*}
\centering
\includegraphics[width=\textwidth]{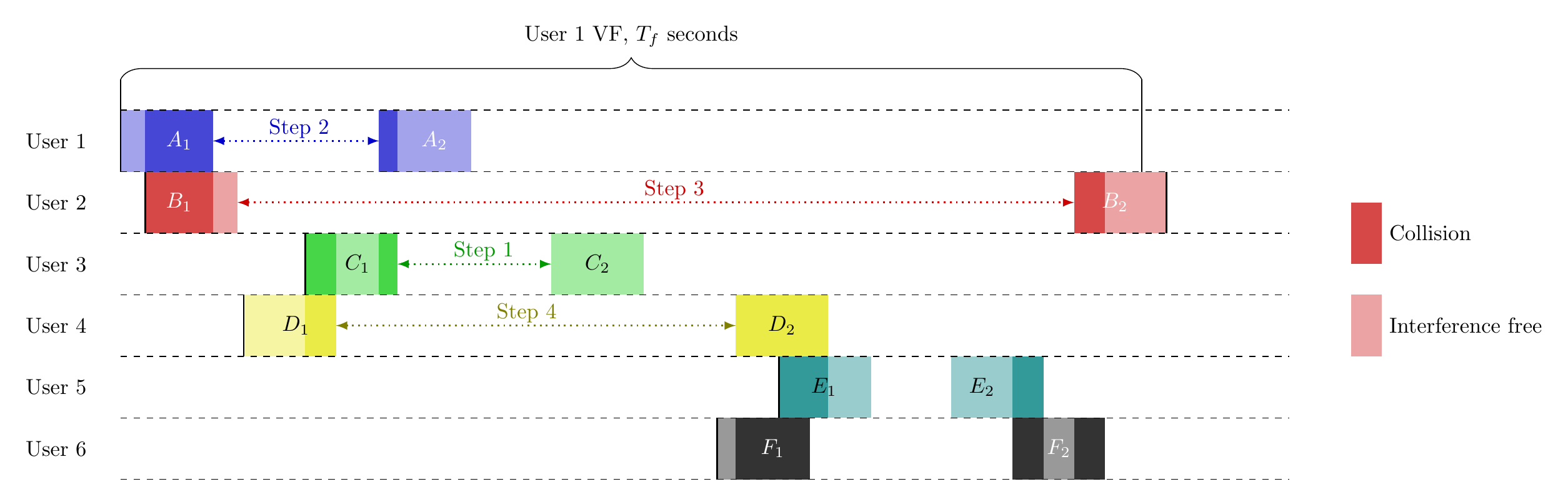}
\caption{\ac{SIC} procedure in \ac{ECRA}, first phase. The decoder starts looking for replicas that can be successfully decoded. The first to be found is replica $C_2$ which is collision-free. After successfully decoding, the information on the location of replica $C_1$ is retrieved from the header. So, the data carried by $C_2$ can be re-encoded, re-modulated, frequency offset and epoch are superimposed on the signal and its interference contribution is removed from both locations within the received signal. The interference caused on replica $A_2$ is now removed. The decoder can successfully decode also replica $A_2$ and - applying the same procedure - remove its interference together with the one of replica $A_1$. Now, replica $B_1$ is collision-free, can be successfully decoded and its interference contribution, together with the one of replica $B_2$, can be removed. Finally, replica $D_1$, also collision-free, is correctly decoded and removed from the received signal, together with its twin $D_2$. Unfortunately user $5$ and $6$ replicas are in a collision pattern that cannot be resolved by \ac{SIC} only, and still remain in the received signal after the end of the first phase.}
\label{fig:CRA_SIC}
\end{figure*}

At the receiver, \ac{ECRA} follows a two-phase procedure in order to decode the received packets. The receiver will operate with a sliding window, similarly to \cite{Meloni2012, deGaudenzi2014_ACRDA}. The decoder starts operating on the first $\Wind$ samples, with $\Wind$ the designed window size.
\subsubsection{SIC phase}
During the first phase, the decoder seeks for replicas that can be successfully decoded. Making the use of the example shown in Figure~\ref{fig:CRA_SIC} where a degree $\dg=2$ has been selected, we describe the \ac{SIC} procedure. The first replica that can be decoded is $C_2$. Thanks to the pointer to the position of all replicas of this user in the header, the decoder can retrieve the position of replica $C_1$ as well. In this way, replica $C_2$ can be re-encoded, re-modulated, frequency offset and epoch are superimposed on the signal and its interference contribution is removed from both locations within the received signal. In the following we assume ideal \ac{SIC}, \ie the entire interference contribution is removed from the received signal. Replica $A_2$ is now released from the interference and can also be correctly decoded. In this scenario, the \ac{SIC} procedure is iterated until none of the replicas can be successfully decoded anymore. At the end of \ac{SIC}, users $1, 2, 3, 4$ can be correctly decoded, while users $5$ and $6$ remain still unresolved, due to the presence of reciprocal interference that cannot be counteracted by the channel code.

\subsubsection{Combining phase}
\acreset{SC,EGC,MRC,ECRA-SC,ECRA-MRC}
In the second phase of \ac{ECRA}, combining techniques are applied on the received packets unable to be decoded in the first phase, and on these \emph{combined observations} decoding is attempted. The formal definition of a combined observation is as follows:
\begin{defin}[Combined observation]
\begin{sloppypar}
Consider the $\dg$ observations of the $\user$-th packet, $\rxVec^{(\user,1)}, \rxVec^{(\user,2)}, ..., \rxVec^{(\user,\dg)}$ with $\rxVec^{(\user,\replica)}=\left(\rx_0^{(\user,\replica)}, \rx_1^{(\user,\replica)}, ..., \rx_{\numSym-1}^{(\user,\replica)}\right)$. We define the \emph{combined observation} the vector
\[
\rxVec^{(\user)}=\left(\rx_0^{(\user)}, \rx_1^{(\user)}, ..., \rx_{\numSym-1}^{(\user)}\right)
\]
with $\rx_i^{(\user)}$ being a suitable function of the individual observation samples $\rx_i^{(\user,1)},\rx_i^{(\user,2)},...,\rx_i^{(\user,\dg)}$, \ie
\[
\rx_i^{(\user)}:=f \left(\rx_i^{(\user,1)},\rx_i^{(\user,2)},...,\rx_i^{(\user,\dg)}\right).
\]
\end{sloppypar}
\end{defin}
Any of \ac{SC}, \ac{EGC} or \ac{MRC} \cite{Brennan1959,Jakes1974} can be applied in the second phase of \ac{ECRA}, although our focus will be on \ac{SC} and \ac{MRC}. If \ac{SC} is adopted, the combined observation is composed by the replica sections with the highest \ac{SINR}, \ie for each observed symbol, the selection combiner chooses the replica with the highest \ac{SINR}. Hence, the instantaneous mutual information of the $\user$-th user combined observation, $i$-th symbol after \ac{SC} is
\begin{equation*}
\mutInf \left(\sinr_{i}^{\mathrm{S}}\right) = \log_2\left( 1 + \sinr_{i}^{\mathrm{S}} \right) = \log_2\left( 1 + \max_\replica \left[\sinr_{i}^{(\user, \replica)}\right] \right).
\end{equation*}
\begin{figure}
\centering
\includegraphics[width=0.6\textwidth]{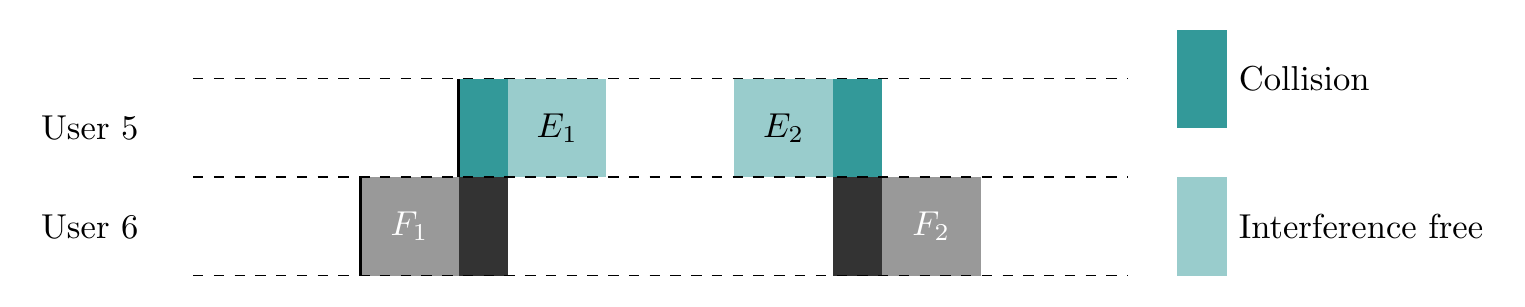}
\caption{Example of collision pattern blocking \ac{SIC}. Different portions of replicas $F_1$ and $F_2$ are collision-free. When \ac{SC} is applied, \ac{ECRA} selects these portions, creates a combined observation and attempts decoding on it.}
\label{fig:CRA_SIC_2}
\end{figure}
Figure~\ref{fig:CRA_SIC_2} depicts the situation at the beginning of the second phase for the example presented in Figure~\ref{fig:CRA_SIC}. The selection combiner selects the first part of replica $F_1$ and the second part of $F_2$ creating a combined observation free from interference.

In \ac{ECRA-MRC} instead, each replica's observed symbol of a given user is weighted proportionally to its squared root mean received signal level \cite{Brennan1959}. In this way, the \ac{SINR} at the output of the combiner is the sum of all replicas \ac{SINR}s. It is also known from literature that \ac{MRC} is optimal if the interference on each replicas is independent \cite{Winters1984}. The instantaneous mutual information of the $\user$-th user combined observation $i$-th symbol after \ac{MRC} is
\begin{equation*}
\mutInf \left(\sinr_{i}^{\mathrm{M}}\right) = \log_2\left(1+\sinr_{i}^{\mathrm{M}}\right) = \log_2\left(1+\sum_{\replica=1}^{\dg} \sinr_{i}^{(\user, \replica)}\right).
\end{equation*}
The decoder outcome after \ac{SC} or \ac{MRC} is modeled substituting the expression of $\mutInf(\sinr_{i})$ with $\mutInf \left(\sinr_{i}^{\mathrm{S}}\right)$ or $\mutInf \left(\sinr_{i}^{\mathrm{M}}\right)$ in equation~\eqref{eq:mi_dec} and adopting the same condition as in equation~\eqref{eq:dec_res}. When decoding is successful, the packet is re-encoded, re-modulated and its interference contribution is removed in all the positions within the frame where the replicas of the decoded user are placed. Combining and \ac{SIC} are iterated until either all users are correctly decoded, or no more packets are present in the receiver window $\Wind$. The receiver window is then shifted forward by $\WindSh$ samples and the procedure starts again.

\subsection{Summary and Comments}

The second step of \ac{ECRA} needs complete knowledge of the replicas position of the remaining users in the frame. Although stringent, this requirement can be addressed in two practical ways: either adopting dedicated pointers to the replicas locations in the header, or exploiting correlation techniques for detection and combining of the replicas prior to decoding. The former adopts a pseudo-random seed that is used at the receiver for retrieving the information on all replicas position of the decoded user. This option was proposed first in \cite{Casini2007} for slotted protocols, but can be extended also to \ac{ECRA}. The pseudo-random seed is used for generating the relative time offset between replicas and together with the replica sequence number allows to identify the replicas locations. In \cite{Clazzer2013} it is shown that in the low to moderate channel traffic regions, low probability of interference in the header can be found. In the high channel traffic region, instead, replicating the header twice is beneficial. Moreover, if a dedicated channel code is introduced for protecting the header, lower header loss probability are expected.\footnote{Dedicated channel code applied to the headers can allow retrieving the information about replica locations although the packet itself is not decodable due to collisions.}

When correlation techniques are adopted, no overhead due to a dedicated field in the header is necessary, and replicas are detected and combined before decoding \cite{Clazzer2017}. A two-phase procedure is proposed in \cite{Clazzer2017}. First the detection of the replicas on the channel is carried out and second, the matching between replicas belonging to the same user is performed. In both phases a simple non-coherent soft-correlation metric is applied. Details on the procedure as well as numerical results can be found in Sections~\ref{sec:sys_mod_ECRA_dis} and \ref{sec:num_res_ECRA_dis}.

The \ac{MRC} combining technique requires the knowledge of the \ac{SINR} symbol-by-symbol, in order to choose the optimal weights \cite{Brennan1959} beforehand the combination is done. In case this information cannot be retrieved, combining can be applied with equal weights for all the symbols, \ie \ac{EGC}.

The scenarios under consideration in the work of \cite{ZigZag} and its extension \cite{SigSag}, are similar to the one that can block the \ac{SIC} procedure (see Figure~\ref{fig:CRA_SIC_2}), although some differences in the solutions between their work and \ac{ECRA} can be identified. \ac{ECRA} creates the combined observation and tries decoding on it, while \cite{ZigZag} requires an iterative demodulation procedure within packet portions, that may increases the overall packet decoding delay. Furthermore, in \cite{ZigZag} an error in one decoded bit propagates to the entire packet unless compensated by further errors. This is due to the iterative procedure applied which subtracts the uncorrect bit from the same packet in the second collision, while in \ac{ECRA} an error in one decoded bit will not affect any other portion of the packet. 
\section{Packet Loss Rate Analysis at Low Channel Load}
\label{sec:PLR}
In this Section a \ac{PLR} approximation tight for low channel load conditions is derived. Packet losses are caused by particular interference patterns that \ac{SIC} is not able to resolve. In the slotted synchronous \ac{RA} protocols, these patterns are analogous to the stopping sets present in the \ac{LDPC} codes \cite{Ivanonv_2015_Letter} and can be analysed exploiting tools from coding theory and graph theory. In the asynchronous \ac{RA} schemes, a graph representation is not straightforward since no discrete objects as slots are present anymore. Therefore, we resort to investigate the collision patterns that involve two users only, with a generic degree $\dg$ and conjecture that these are the patterns driving the \ac{PLR}, especially at low channel loads. In the next Section the approximation of the \ac{PLR} is compared with Monte Carlo simulations in order to verify its tightness. A set of definitions are required for the analysis.
\begin{defin}[Collision cluster $\CollClus$]
Consider a subset $\CollClus$ of users. Assume that packets of all users in $\CollClus^{c}$ (complementary of the subset $\CollClus$) have been successfully decoded. The subset $\CollClus$ is referred to as \emph{collision cluster} iff no packet replicas for the users in $\CollClus$ is collision-free.
\end{defin}
Under the assumption of collision channel, none of the users in the collision cluster can be successfully decoded. Conversely, when a channel code $\Code$ is employed by each transmitted packet, the collision cluster might be resolvable, leading to the following definition.
\begin{defin}[$\Code$-unresolvable collision pattern]
\acreset{UCP}
Given each packet encoded with a channel code $\Code$, a \emph{\ac{UCP}} $\UCP$ is a collision cluster where no user in the set can be successfully decoded.
\end{defin}
Every \ac{UCP} is also a collision cluster, but not viceversa. In order to evaluate the probability of \ac{UCP} involving two users only, a generalization of the definition of vulnerable period \cite{Kleinrock1976_book} is required.
\begin{defin}[$\Code$-vulnerable period for $|\CollClus|=2$]
Consider the transmission of a packet protected with a channel code $\Code$ between time $\RStart$ and $\RStart + \pkLen$. The packet's \emph{$\Code$-vulnerable period} is the interval of time $\left[\RStart - \VL, \RStart + \VR\right]$ in which the presence of a single interferer leads to a failure in the decoding.
\end{defin}
Hence, the vulnerable period duration $\Vpd$ is defined as
\[
\Vpd = \VL + \VR.
\]
In slotted synchronous schemes under the collision channel model, $\VL=0$ and $\VR = \pkLen$ so $\Vpd = \pkLen$. For asynchronous schemes in general and therefore for \ac{ECRA}, it holds ${\VL = \VR \triangleq \Vg}$. The vulnerable period duration for asynchronous schemes is $\Vpd = 2 \Vg$. Considering the collision channel model, the vulnerable period duration is then $\Vpd = 2 \pkLen$. So, the duration of packets' vulnerable period is doubled in asynchronous schemes w.r.t. comparable synchronous ones \cite{Kleinrock1976_book}.

\subsection{Packet Loss Rate Approximation}
\label{sec:Approx_PER}
In this Section we derive an approximation of the \ac{PLR}, denoted as $\plr$. The approach follows \cite{Sandgren2016} extending the investigation to asynchronous schemes. Let us consider the user $\user$. We denote with $\AllUCP$ the set of all possible \ac{UCP} that cause the loss of user $\user$ packets and with $\UCPcons$ the unique type of \ac{UCP} that we consider to drive the \ac{PLR} performance $\plr$. Let $\fraTp = \fraLen/\pkLen$ denote the \ac{VF} length measured in packet durations and $\nVp=\lfloor \fraLen/\Vpd \rfloor$ denote the number of disjoint vulnerable periods per \ac{VF}. Clearly $\fraTp \geq \dg$. The \ac{PLR} can be approximated with
\begin{equation}
\begin{split}
\label{eq:plr_approx}
\plr &= \Pr\left\{\bigcup_{\UCP \in \AllUCP} \user \in \UCP \right\} \leq \sum_{\UCP \in \AllUCP} \Pr\left\{\user \in \UCP \right\} \approx \Pr\left\{\user \in \UCPcons \right\} \\
&= \sum_{m=2}^{\infty} \frac{e^{-\fraTp \load}\left( \fraTp \load \right)^m}{m!}\Pr\left\{ \user \in \UCPcons | m \right\}.
\end{split}
\end{equation}
The probability $\plr$ is first upper bounded with the union bound and then approximated considering only one type of \ac{UCP}, \ie $\UCPcons$. Finally we take the expectation of active users number over the $\user$-th user \ac{VF}. The \ac{UCP} $\UCPcons$ considered in the analysis is formed by two users only with a generic degree $\dg$. Let $\mU(\UCPcons,m)$ be the number of combinations of $m$ users taken two by two and $\sL(\UCPcons)$ be the number of ways in which $\dg$ replicas can be placed into a \ac{VF}. Finally, $\UL(\UCPcons)$ denotes the number of possible ways that the two users can place their replicas. Making the assumption that $\UCPcons$ spans at most $\fraLen$ seconds in order to simplify the derivation, the probability that the user $\user$ belongs to the \ac{UCP} $\UCPcons$ is
\begin{equation}
\label{eq:prob_UCP}
\Pr\left\{ \user \in \UCPcons | m \right\} \approx \frac{\mU(\UCPcons,m) \sL(\UCPcons)} {\UL(\UCPcons)} \frac{2}{m}
\end{equation}
The quantity $\mU$ is $\mU(\UCPcons,m) = \binom{m}{2}$. The considered user $\user$ sends its first replica immediately, while the remaining $\dg-1$ replicas start times are selected uniformly at random within $\fraLen$ seconds, so $\sL(\UCPcons) \approx \binom{\nVp-1}{\dg-1}$. Finally, the number of ways in which two users can select their position for the replicas follows $\UL(\UCPcons) \approx \nVp \binom{\nVp-1}{\dg-1}^2$. Substituting into equation~\eqref{eq:prob_UCP} these values
\begin{equation}
\label{eq:prob_UCP2}
\Pr\left\{ \user \in \UCPcons | m \right\} \approx \frac{\binom{m}{2}}{\nVp \binom{\nVp-1}{\dg-1}} \frac{2}{m} = \frac{\binom{m}{2}}{\dg \binom{\nVp}{\dg}} \frac{2}{m}.
\end{equation}
Finally, inserting the result of equation~\eqref{eq:prob_UCP2} in equation~\eqref{eq:plr_approx}, we can approximate the \ac{PLR} $ \plr$ as
\begin{equation}
\label{eq:plr_approx_final}
\plr \approx \sum_{m=2}^{\infty} \frac{e^{-\fraTp \load}\left( \fraTp \load \right)^m}{m!} \frac{\binom{m}{2}}{\dg \binom{\nVp}{\dg}} \frac{2}{m}.
\end{equation}

The \ac{PLR} approximation directly depends on the vulnerable period duration via $\nVp$. In the next Sections the vulnerable period duration is computed for two scenarios of interest, including the \ac{MRC} case.

\subsection{Vulnerable Period Duration for Asynchronous RA with FEC}

In this scenario, replicas are protected by a channel code so that not all collisions are destructive. The only \ac{UCP} to be considered is the one involving two users and their replicas. We recall that perfect power control is assumed so that both users are received with the same power $\usPw$. Without loss of generality, we focus on replica $\replica$ involved in an \ac{UCP} of type $\UCPcons$ which has a first section free of interference and a second part interfered. The selected rate $\rate$ determines what it is the minimum fraction of interference-free replica $\intFreeAsy$ that still allows correct decoding, \ie
\begin{equation}
\label{eq:app_a_3}
\intFreeAsy \log_2\left(1 + \frac{\usPw}{\noisePw}\right)+( 1 - \intFreeAsy )\log_2\left(1 + \frac{\usPw}{\noisePw + \usPw}\right) = \rate.
\end{equation}
For the sake of simplicity we denote as
\begin{equation*}
\begin{split}
\rateFree&= \log_2\left(1 + \frac{\usPw}{\noisePw}\right)\\
\rateInt&= \log_2\left(1 + \frac{\usPw}{\noisePw + \usPw}\right)
\end{split}
\end{equation*}
and we solve equation~\eqref{eq:app_a_3} for $\intFreeAsy$
\begin{equation}
\label{eq:app_a_4}
\intFreeAsy = \frac{\rate-\rateInt}{\rateFree-\rateInt}.
\end{equation}
Equation~\eqref{eq:app_a_4} is valid for $\rate \geq \log_2\left(1 + \frac{\usPw}{\noisePw + \usPw}\right)$. In fact, for $\rate < \rateInt$ no \ac{UCP}s involving only two users can be observed, and regardless the level of interference, packets involved in collisions with only one other packet can always be decoded. In this way,
\begin{equation*}
\label{eq:app_a_5}
\intFreeAsy = \left\{
\begin{array}{rl}
\frac{\rate - \rateInt}{\rateFree - \rateInt} & \mbox{for $\rate \geq \rateInt$}\\
0 & \mbox{for $\rate < \rateInt$}
\end{array}
\right .
\end{equation*}
It is worth noticing that $\intFreeAsy$ is constrained to $0\leq \intFreeAsy \leq1$, since the selectable rate $\rate$ is $\rate \leq \log_2\left(1 + \frac{\usPw}{\noisePw}\right) = \rateFree$ for reliable communication.

\begin{sloppypar}
In this way, $\Vg = \intFreeAsy \pkLen$ and therefore the vulnerable period is reduced to ${\Vpd = 2 \Vg = 2 \intFreeAsy \pkLen}$. And finally $\nVp=\lfloor \fraLen/\Vpd \rfloor=\lfloor \fraLen/(2 \intFreeAsy \pkLen) \rfloor$. Inserting the value of $\nVp$ in equation~\eqref{eq:plr_approx_final} we obtain the final expression of the \ac{PLR} approximation for asynchronous \ac{RA} schemes. Note that for $\intFreeAsy \rightarrow 0$, $\nVp \rightarrow + \infty$ and therefore the \ac{PLR} approximation in equation~\eqref{eq:plr_approx_final} tends to $0$.
\end{sloppypar}

\subsection{Vulnerable Period Duration for Asynchronous RA with MRC and $\dg=2$}
\label{sec:low_bound_ECRA_MRC}

Similarly to the previous Section, $\UCPcons$ is the considered \ac{UCP} where two users are interfering each other and they are received with the same power $\usPw$. In this scenario the degree is fixed to $\dg=2$. The focus is on the combined observation after \ac{MRC}. Without loss of generality, it is assumed that the first section of both replicas is free of interference, while there is a second part where just one replica is interfered and in the end there is the last part where both replicas are interfered. We aim at computing the minimum combined observation portion interference free $\intFreeMRC$ that is required for correctly decoding the user after \ac{MRC}. It holds
\begin{equation}
\label{eq:mrc1}
\intFreeMRC \log_2\left(1+2 \frac{\usPw}{\noisePw}\right) + \intOne \log_2\left(1+\frac{\usPw}{\noisePw} + \frac{\usPw}{\noisePw + \usPw}\right) + (1-\intFreeMRC -\intOne )\log_2\left(1+ 2 \frac{\usPw}{\noisePw + \usPw}\right) = \rate.
\end{equation}
For the sake of simplicity we use the following notation
\begin{equation*}
\begin{split}
\rateFree &= \log_2\left(1+2 \frac{\usPw}{\noisePw}\right)\\
\rateIntO &= \log_2\left(1+\frac{\usPw}{\noisePw} + \frac{\usPw}{\noisePw + \usPw}\right)\\
\rateIntT &= \log_2\left(1+ 2 \frac{\usPw}{\noisePw + \usPw}\right).
\end{split}
\end{equation*}
So that equation~\eqref{eq:mrc1} becomes
\begin{equation}
\label{eq:app_a_6}
\intFreeMRC \rateFree + \intOne \rateIntO + (1 - \intFreeMRC - \intOne ) \rateIntT = \rate.
\end{equation}

In order to solve equation~\eqref{eq:app_a_6}, $\intOne$ is expressed as a function of $\intFreeMRC$, as $\intOne = \alpha \intFreeMRC$, where $0\leq\alpha\leq (1-\intFreeMRC)/\intFreeMRC$. When $\alpha=0$, there are no portions where only one out of the two replicas is interfered, while $\alpha=(1-\intFreeMRC)/\intFreeMRC$ represents the case when there are no portions where both replicas are interfered. Solving \eqref{eq:app_a_6} for $\intFreeMRC$ gives
\begin{equation*}
\label{eq:app_a_7}
\intFreeMRC = \frac{\rate-\rateIntT}{\rateFree-\rateIntT+\alpha(\rateIntO-\rateIntT)}.
\end{equation*}
Also in this case, for $\rate<\rateIntT$ we have $\intFreeMRC=0$ which means that no \ac{UCP} involving two replicas can be found,
\begin{equation*}
\label{eq:app_a_8}
\intFreeMRC= \left\{
\begin{array}{rl}
\frac{\rate-\rateIntT}{\rateFree-\rateIntT+\alpha(\rateIntO-\rateIntT)} & \mbox{for $\rate\geq \rateIntT$}\\
0 & \mbox{for $\rate<\rateIntT$}
\end{array}
\right .
\end{equation*}
The average vulnerable period duration over the two replicas is $\Vpd = 2 \Vg = 2\left( \intFreeMRC + \frac{\intOne}{2} \right) \pkLen = 2\intFreeMRC\left(1 + \frac{\alpha}{2}\right) \pkLen$.\footnote{It is important to underline that the expression of the average vulnerable period duration presented is valid no matter how the two replicas are interfered, \ie also when the portions interfered are not both at the beginning of the packets.} And finally $\nVp=\lfloor \fraLen/\Vpd \rfloor=\lfloor \fraLen/\left(2\intFreeMRC\left(1 + \frac{\alpha}{2}\right) \pkLen \right) \rfloor$. 
\section{Performance Analysis}
\label{sec:simulations_all}

In this Section \acs{ECRA-SC} and \ac{ECRA-MRC} are compared to the reference \ac{CRA} protocol as well as with ALOHA. For this first comparison two metrics are considered, the \ac{PLR} and the throughput. The throughput $\tp$ is defined as the expected number of successfully decoded packets per packet duration $\pkLen$,
\[
\tp=(1-\plr)\, \load.
\]
The \ac{ECRA} algorithm is also compared to slotted synchronous \ac{RA} protocols as \ac{CRDSA}. Since a channel code $\Code$ is adopted by the proposed scheme, the throughput is not anymore a sufficient metric. While for slotted synchronous protocols packets are decoded only if they are received collision-free,\footnote{This holds assuming that no power unbalance is present between the received packets and the channel code cannot counteract any collision, \ie $\log_2 \left( 1 + \frac{\usPw}{\noisePw + \usPw}\right) < \rate \leq \log_2\left( 1 + \frac{\usPw}{\noisePw} \right)$.} in the asynchronous case up to a certain level of interference collisions can still be resolved. The level of interference that can be sustained depends on the selected rate $\rate$. In the former case, regardless of the selected rate, the throughput performance remains the same, while in the latter lower rates lead to higher throughput. Nevertheless, lowering the rate decreases the information carried by each packet. This tradeoff is captured by the spectral efficiency $\se$,
\begin{equation*}
\label{eq:thr2}
\se=(1-\plr) \, \load \, \rate \qquad \mathrm{[b/s/Hz]}.
\end{equation*}
Although \ac{ECRA} can outperform considerably the ALOHA protocol, on average it requires more power. In fact, this scheme assumes to replicate each packet sent in the frame $\dg$ times. In order to take into account the increase in average power, we follow the approach of \cite{Abramson1977}, that was extended for slotted synchronous protocols as \ac{CRDSA} and \ac{IRSA} in \cite{Liva2011}. The \emph{normalised capacity} $\normCap$ is defined as the ratio between the maximum achievable spectral efficiency of one of the \ac{RA} schemes and the channel capacity of multiple access Gaussian channel under the same average power constraint. The idea is to compute the maximum spectral efficiency of the asynchronous \ac{MAC} schemes (\ac{ECRA-SC} or \ac{ECRA-MRC}) and normalise it to the sum rate capacity of the multiple access Gaussian channel $\refCap =\log_2(1+\PwAg/\noisePw)$. This is achieved by fixing the average aggregate received signal power $\PwAg$ equal in all the schemes. In this way, for the \ac{RA} protocols the user transmission power $\usPwTx$ takes into account the fact that the channel is used intermittently but $\dg$ times w.r.t. ALOHA, \ie $\usPwTx=\frac{\PwAg}{\load\cdot \dg}$. The ultimate performance of the asynchronous \ac{RA} schemes is given by the maximum spectral efficiency $\maxSe$ defined as
\begin{equation}
\label{eq:max_se}
\maxSe=\max_{\rate\in \left[0,..,\maxRate\right]}\tp(\load) \, \rate
\end{equation}
where for each channel traffic value, the rate $\rate$ which maximises the spectral efficiency is chosen.\footnote{The maximum possible rate for reliable communication $\maxRate$ is $\maxRate=\log_2(1+\usPwTx/\noisePw)$ and depends upon the selected channel load $\load$.} Unfortunately, the throughput expression $\tp(\load)$ is not available in closed form for \ac{ECRA-SC} and \ac{ECRA-MRC}, so only a numerical evaluation of equation~\eqref{eq:max_se} is possible. The normalised capacity $\normCap$ is defined as
\begin{equation*}
\label{eq:efficiency}
\normCap = \frac{\maxSe}{\refCap},
\end{equation*}
where, depending on the \ac{RA}, a different expression of $\maxSe$ will be used.

\subsection{Numerical Results}
\label{sec:simulations}

In the following, numerical results for \ac{ECRA-SC} and \ac{ECRA-MRC} schemes are presented.
The packets sent by the users are composed by $\numBit=1000$ bits, which translate into ${\numSym=\left(\numBit/\rate\right)}$ symbols. The transmission period is then $\pkLen = \symLen \numSym$. The \ac{VF} duration $\fraLen$ is selected to be equal to $200$ packet durations, \ie $\fraLen=200\, \pkLen$. We recall that the number of users generating traffic follows a Poisson distribution with mean $\load$ measured in packets per $\pkLen$ durations, and each of the users transmits $\dg=2$ replicas per generated packet. The decoder operates on a window of $\Wind = 3\, \fraLen$ symbols and once either the maximum number of \ac{SIC} iterations is expired or no more packets can be successfully decoded, it is shifted forward by $\WindSh = 20\, \pkLen$. Ideal interference cancellation is assumed and the block interference model introduced in Section~\ref{sec:int_model} is used for determining the successful decoding of a packet.

\begin{sloppypar}
We present first the simulations of the throughput and \ac{PLR} for both \ac{ECRA-SC} and \ac{ECRA-MRC}. For reference purposes also \ac{CRA} and the ALOHA protocols are depicted in the figures. The assumptions are $\usPw/\noisePw=6$ dB and $\rate=1.5$ equal for all users.
\begin{figure}
\centering
\includegraphics[width=0.8\textwidth]{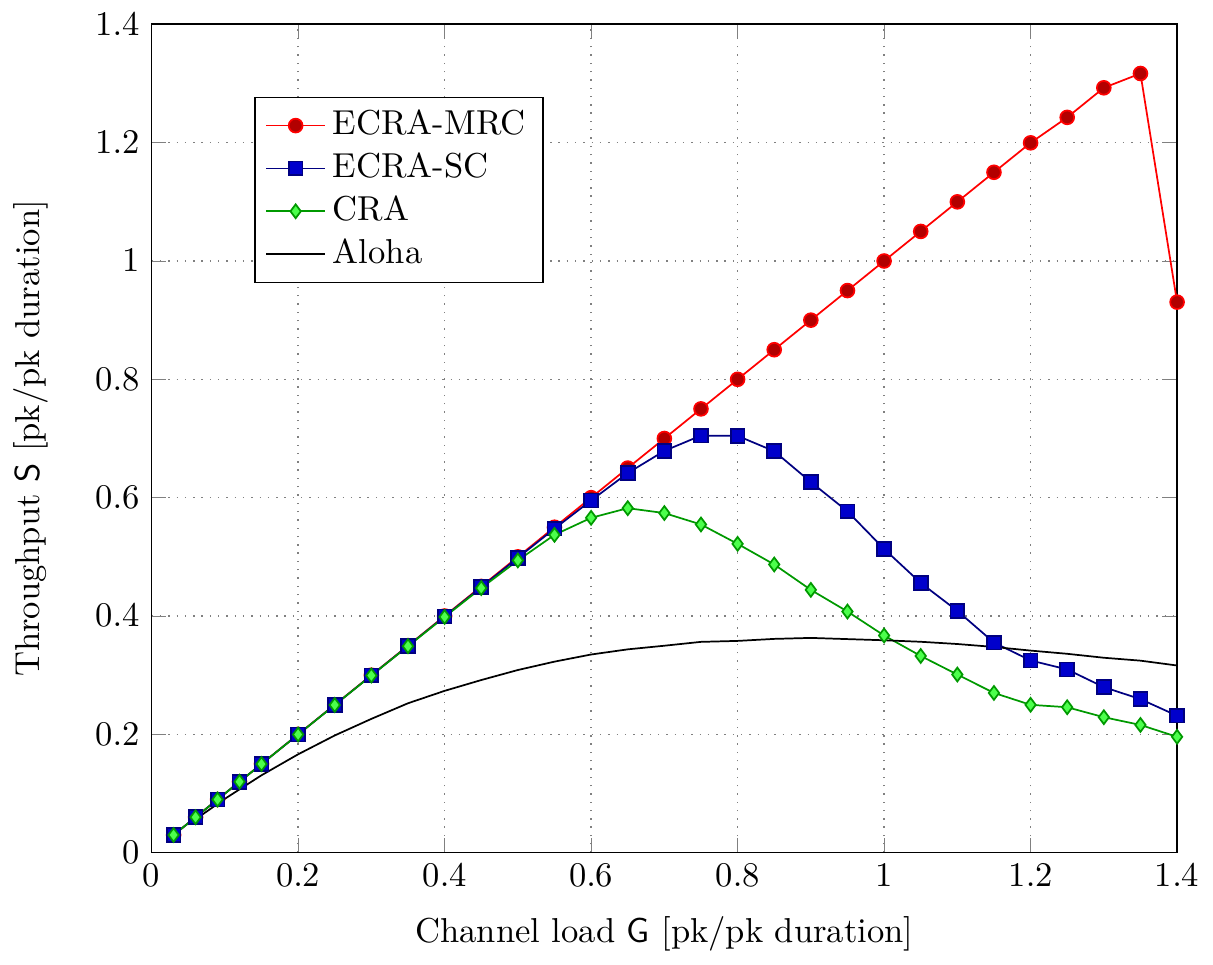}
\caption{Throughput $\tp$ vs. channel load $\load$ for ALOHA, \ac{CRA},
\ac{ECRA-SC} and \ac{ECRA-MRC}, {$\usPw/\noisePw=6$~dB} and $\rate=1.5$.}
\label{fig:T}
\end{figure}
In Figure~\ref{fig:T} we present the throughput $\tp$ vs. the channel load $\load$. \ac{ECRA-MRC} largely outperforms both \ac{ECRA-SC} and \ac{CRA}: reaching a maximum throughput of $\tp=1.32$ at $\load=1.35$, which is more than twice the one of \ac{CRA} ($\tp=0.58$) and it increases by $89$\% with respect to the one of \ac{ECRA-SC} ($\tp=0.70$). Furthermore, the throughput of \ac{ECRA-MRC} follows linearly the channel load up to $1.3$ packets per $\pkLen$, implying very limited \ac{PLR}. In fact, looking at the \ac{PLR} performance in Figure~\ref{fig:PER},
\ac{ECRA-MRC} is able to maintain the \ac{PLR} below $10^{-3}$ for channel load below $1.2$ packets per $\pkLen$.
\begin{figure}
\centering
\includegraphics[width=0.8\textwidth]{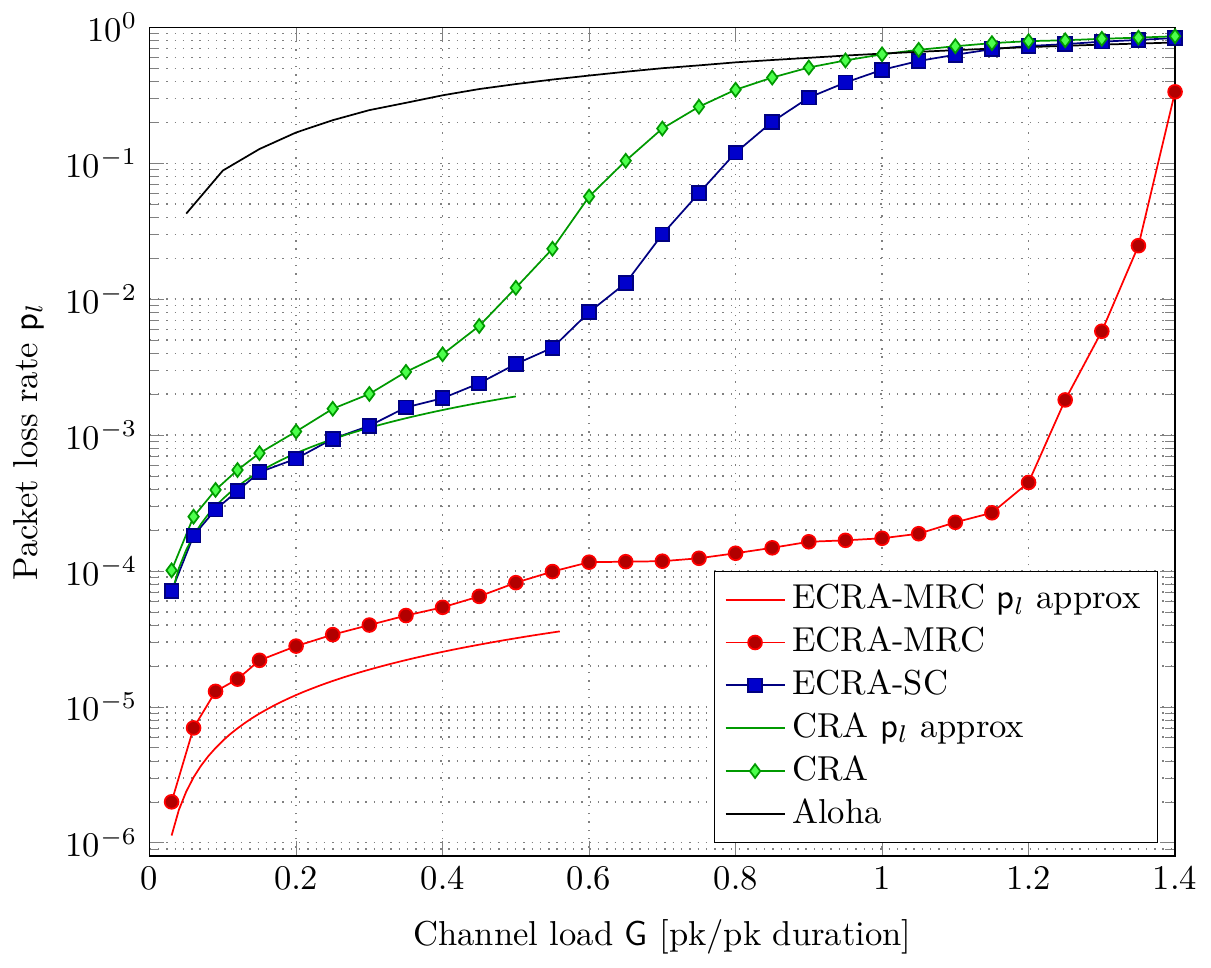}
\caption{Packet error rate $\plr$ vs. channel load $\load$ for ALOHA,
\ac{CRA}, \ac{ECRA-SC} and \ac{ECRA-MRC}, $\usPw/\noisePw=6$ dB and $\rate=1.5$. The
average value of $\alpha$ used in the approximation of $\plr$ for \ac{ECRA-MRC}
is derived through Monte Carlo simulations.}
\label{fig:PER}
\end{figure}
In other words for a target \ac{PLR} of $\plr=10^{-3}$, \ac{ECRA-MRC} can be operated up to $\load=1.2$, while both \ac{ECRA-SC} and \ac{CRA} only up to $\load\cong0.3$ and $\load\cong0.2$ respectively. The gain of \ac{ECRA-MRC} with respect to both \ac{ECRA-SC} and \ac{CRA} in terms of \ac{PLR} is of at least one order of magnitude, except in the very high channel load region, where it largely exceeds this value. It is also shown in the figure, that this protocol is the only one that can maintain $\plr\leq10^{-4}$ for channel load values up to $\load=0.6$. Very low \ac{PLR} are particularly appealing in specific scenarios as satellite applications or control channels where reliability can be as important as efficiency.
\end{sloppypar}

In Figure~\ref{fig:PER}, the approximation on the packet loss rate $\plr$ for both \ac{CRA} and \ac{ECRA-MRC}, derived in Section~\ref{sec:Approx_PER}, is also shown. This approximation takes into account only the errors coming from \ac{UCP}s involving two users, and for very limited channel load values is very close to the simulated $\plr$. For \ac{CRA}, when $\load\leq0.3$, the approximation approaches the $\plr$ simulated performance, while for increasing $\load$ the probability of having \ac{UCP}s involving more than two users starts to have an impact on $\plr$ and therefore the approximation starts to become loose. Although a similar behaviour can be found for the approximation of \ac{ECRA-MRC}, interestingly the relative distance between the approximation and the simulations remains almost constant for a large range of channel load values.

\begin{sloppypar}
In the second set of simulations, performance comparison of the slot synchronous scheme \ac{CRDSA} with the asynchronous schemes \ac{CRA}, \ac{ECRA-SC} and \ac{ECRA-MRC} is presented. The metric used for the comparison is the spectral efficiency $\se$. For the sake of completeness, we recall that \ac{CRDSA} has the same throughput performance $\tp$ for ${\log_2 \left( 1 + \frac{\usPw}{\noisePw + \usPw}\right) < \rate \leq \log_2\left( 1 + \frac{\usPw}{\noisePw} \right)}$, under the assumption of equal received power for all users and no multi-packet reception. Therefore, we select the rate $\rate^{\mathrm{s}}=\log_2\left( 1 + \frac{\usPw}{\noisePw} \right)$  for \ac{CRDSA} in both simulations. For asynchronous schemes (\ac{CRA}, \ac{ECRA-SC} and \ac{ECRA-MRC}) instead, a rate $\rate^{\mathrm{a}}$ with $\rate^{\mathrm{a}}<\log_2\left( 1 + \frac{\usPw}{\noisePw} \right)$ is chosen.
\begin{figure}
\centering
\subfigure[$\usPw/\noisePw=6$ dB and for the asynchronous schemes
$\rate^{\mathrm{a}}=1.5$.]{
    \includegraphics[width=0.8\textwidth]{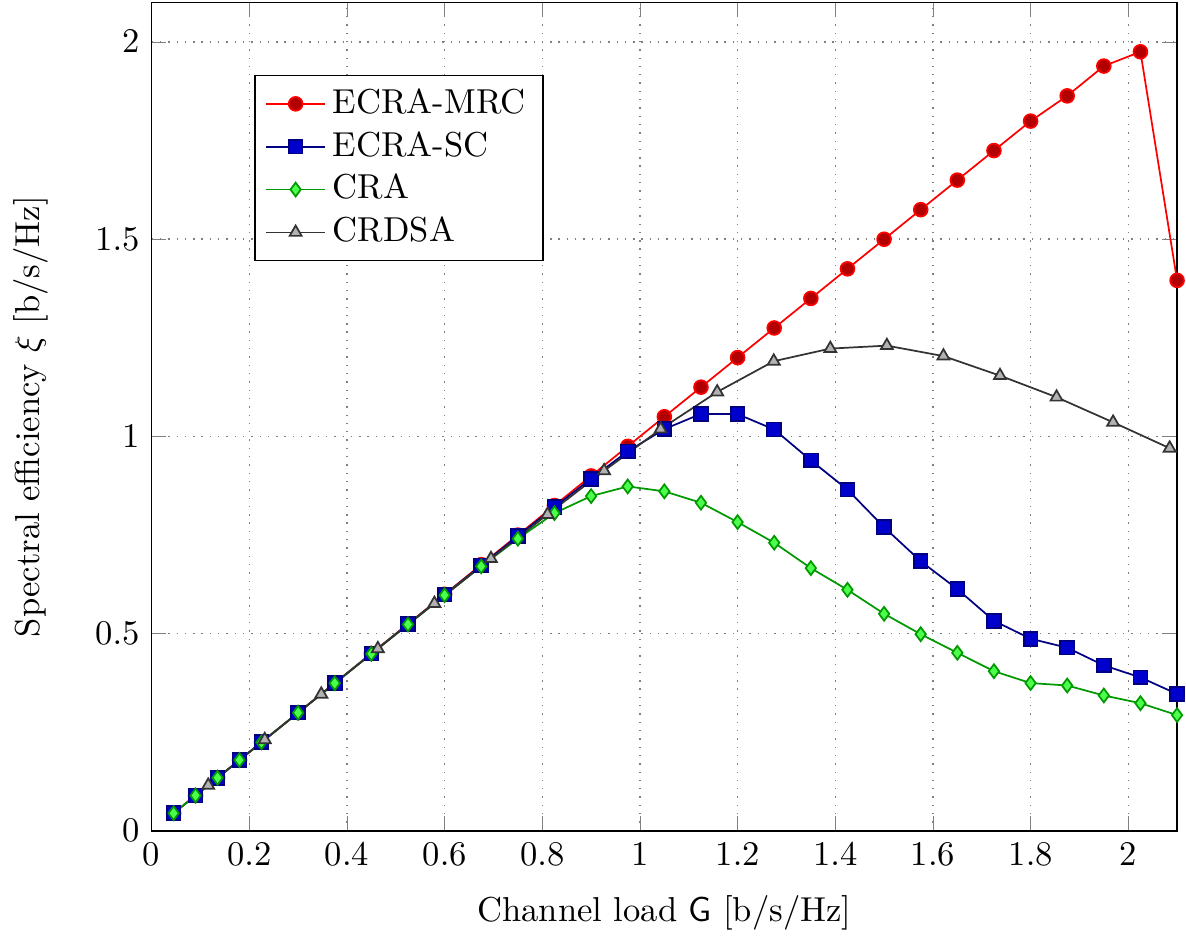}
    \label{fig:sp_eff_a}
}\\
\subfigure[$\usPw/\noisePw=6$ dB and for the asynchronous schemes
$\rate^{\mathrm{a}}=2$.]{
    \includegraphics[width=0.8\textwidth]{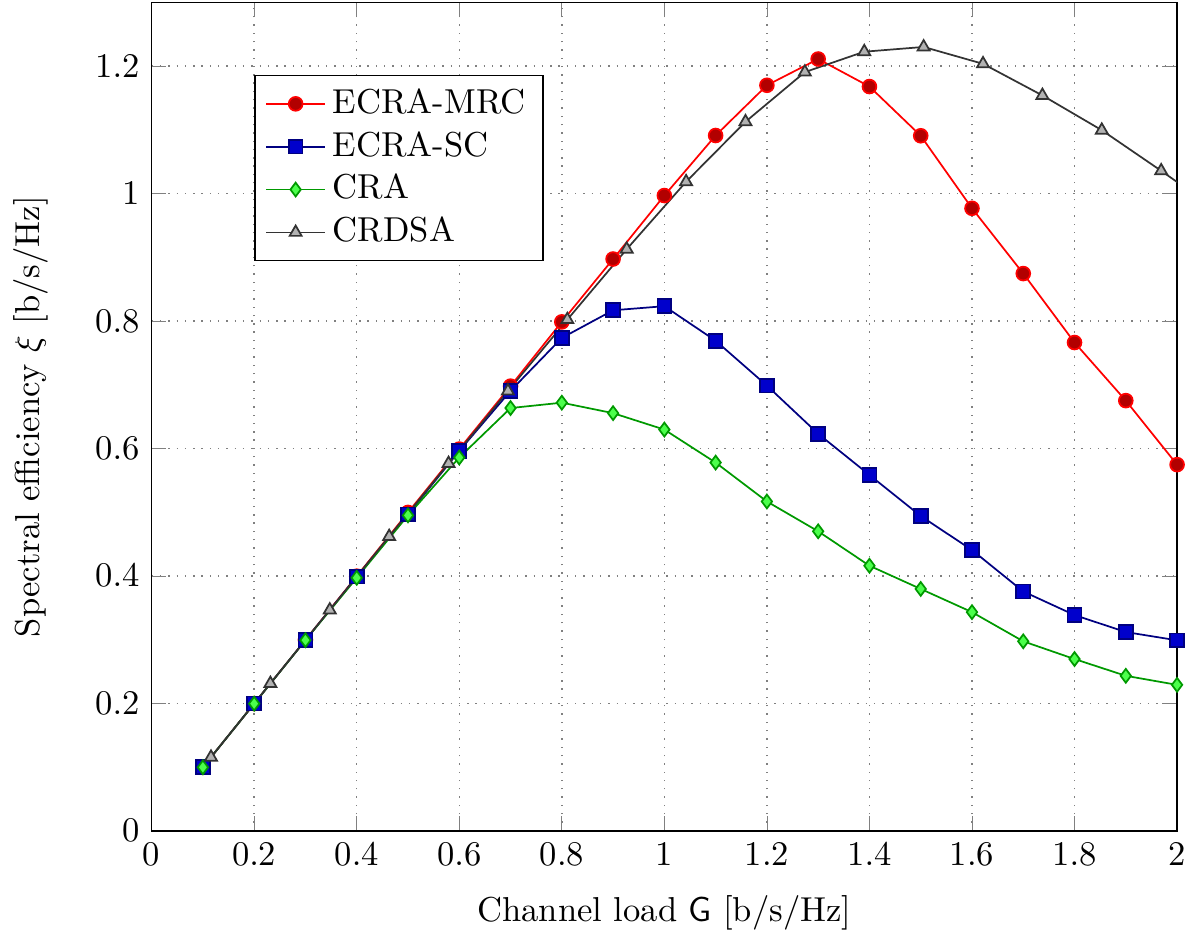}
    \label{fig:sp_eff_b}
}
\caption{Spectral efficiency $\se$ vs. channel load $\load$ for \ac{ECRA-SC}, \ac{ECRA-MRC}, \ac{CRA} and \ac{CRDSA}.}
\label{fig:sp_eff}
\end{figure}
In Figure~\ref{fig:sp_eff} we present the aforementioned comparison. We select $\usPw/\noisePw=6$ dB and in the first scenario the rate for the asynchronous schemes is $\rate^{\mathrm{a}}=1.5$, (Figure~\ref{fig:sp_eff_a}), while in the second one it is $\rate^{\mathrm{a}}=2$, (Figure~\ref{fig:sp_eff_b}). In the former scenario, \ac{ECRA-MRC} is able to largely outperform \ac{CRDSA} as well as all other asynchronous schemes. Outstandingly, \ac{ECRA-MRC} reaches spectral efficiencies up to $2$ b/s/Hz while under similar conditions \ac{CRDSA} can only exceed $1.2$ b/s/Hz. The result is a gain of $60\%$ in the maximum spectral efficiency of \ac{ECRA-MRC} over \ac{CRDSA}. On the other hand, \ac{ECRA-MRC} is particularly sensible to the channel load working point. In fact, once the maximum spectral efficiency is reached, the performance degrade drastically. This is not the case for \ac{CRDSA} and the other asynchronous schemes, which are subject to a more graceful degradation after the maximum spectral efficiency. Finally, observe that \ac{ECRA-SC} is able to reach spectral efficiencies of $1.05$ b/s/Hz which is $85\%$ of the one of \ac{CRDSA}.
\end{sloppypar}

In the second scenario, instead, \ac{CRDSA} is able to slightly outperform \ac{ECRA-MRC}. Note that \ac{CRDSA} uses the same rate in both scenarios, while the asynchronous schemes are using a rate of $\rate^{\mathrm{a}}=2$, greater than the one of the first scenario. \ac{ECRA-MRC} benefits from lower \ac{PLR} in the increasing slope of spectral efficiency thanks to a rate able to counteract small levels of interference. Nonetheless, after the maximum spectral efficiency of $1.21$ b/s/Hz is reached, the performance degrade rapidly. \ac{CRDSA} instead benefits from a more gradual increase in spectral efficiency towards the maximum and also a more graceful degradation. As a final remark, \ac{ECRA-MRC} \--- as well as asynchronous schemes in general \--- is very sensitive to the selected rate. More conservative rates, \ie lower rates lead to better performance with respect to more aggressive ones. Although we can expect great improvements adopting \ac{ECRA-MRC}, less robustness in terms of channel load conditions around the maximum spectral efficiency has to be taken into account w.r.t. \ac{CRDSA} or the other asynchronous schemes, including \ac{ECRA-SC}.

The last set of simulations shows the comparison among ALOHA, \ac{CRA}, \ac{ECRA-SC} and \ac{ECRA-MRC}, in terms of the normalised capacity $\normCap$. $\PwAg/\noisePw=6$ dB is selected and the results are presented in Figure~\ref{fig:norm_cap_tot}.
\begin{figure}
\centering
\subfigure[Normalised capacity $\normCap$ for ALOHA, \ac{CRA}, \ac{ECRA-SC} and \ac{ECRA-MRC} with $\PwAg/\noisePw=6$ dB.]{
    \includegraphics[width=0.8\textwidth]{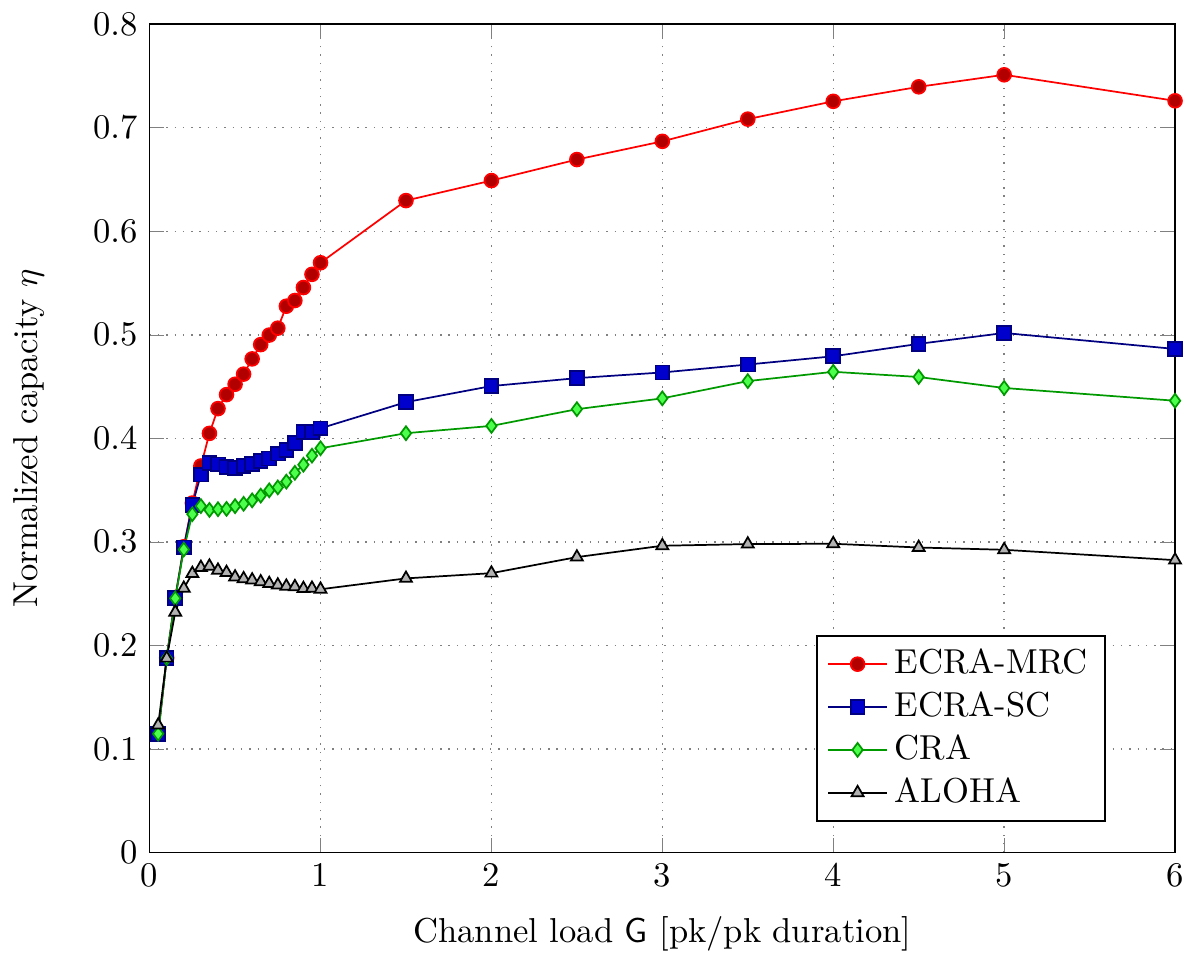}
    \label{fig:norm_cap}
}\hspace{0 mm}
\\
\subfigure[Rate maximizing the spectral efficiency for \ac{CRA}, \ac{ECRA-SC} and \ac{ECRA-MRC} with $\PwAg/\noisePw=6$ dB.]{
    \includegraphics[width=0.8\textwidth]{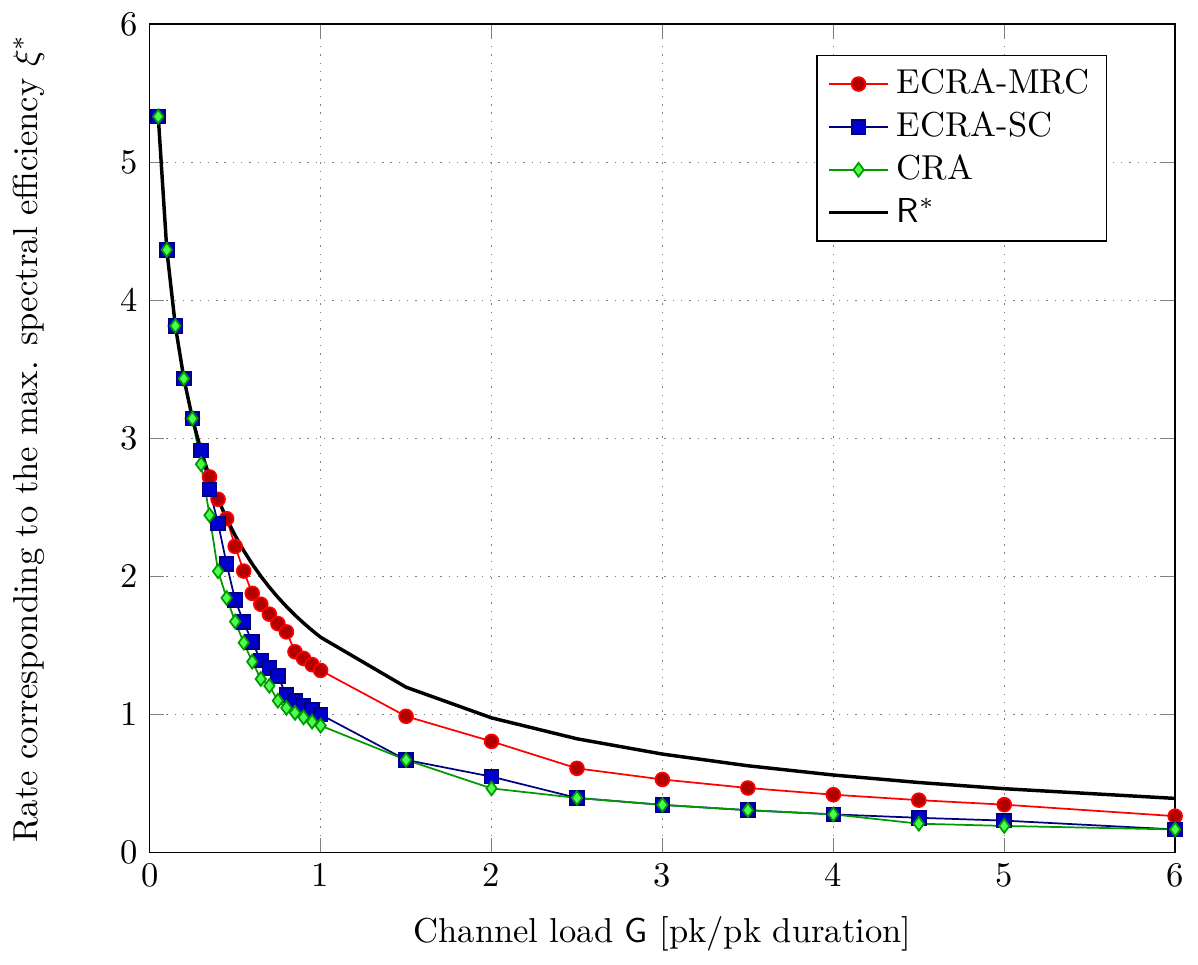}
    \label{fig:rate_max}
    }
\caption{Normalised capacity $\normCap$ for ALOHA, \ac{CRA}, \ac{ECRA-SC} and \ac{ECRA-MRC} with $\PwAg/\noisePw=6$ dB and corresponding rate.}
\label{fig:norm_cap_tot}
\end{figure}
The normalised capacity for \ac{ECRA-MRC} can reach up to 75\% of the \ac{MAC} channel capacity, for a channel load $\load=5$ with rate $\rate\cong0.35$; see Figure~\ref{fig:rate_max}. At this channel load, the gain is 50\% with respect to \ac{ECRA-SC} and 67\% with respect to \ac{CRA}. Interestingly, the normalised capacity for \ac{ECRA-MRC} as well as for both \ac{ECRA-SC} and \ac{CRA} is relatively constant for heavy channel load \ie $\load>3$. In this way, the schemes appear to be robust against channel load fluctuations. On the other hand, the rate for which the maximum spectral efficiency $\maxSe$ (and so the normalised capacity) of the schemes is achieved decreases as the channel load increases, see Figure~\ref{fig:rate_max}. Therefore the system would be required to adapt the rate in order to reach the best performance in terms of normalised capacity. Nevertheless, the adaptation of the rate remains quite limited in this channel load region, ranging from a maximum of 0.53 at $\load=3$ to a minimum of 0.27 at $\load=6$ for \ac{ECRA-MRC}. For limited channel load, all the schemes perform very close, with ALOHA being slightly the best option. This is due to the low collision probability and the benefit of double transmit power of ALOHA compared to \ac{CRA} or \ac{ECRA} since no replicas are sent.

In Figure~\ref{fig:rate_max}, the rate corresponding to the maximum spectral efficiency for \ac{ECRA-MRC}, \ac{ECRA-SC} and \ac{CRA} is shown. The maximum possible rate under this scenario is also depicted with a solid line in the figure. For limited channel load, the maximum spectral efficiency is achieved when using the maximum rate allowed, supporting the fact that collisions of received packets are seldom, and the spectral efficiency can be maximised pushing the rate as much as it is allowed. On the other hand, as soon as the channel load exceeds $\load=0.3-0.4$, the maximum spectral efficiency is reached for rate values below the maximum one. In this way, the maximum spectral efficiency under moderate to high channel load conditions can be maximised taking a margin with respect to the maximum rate. This margin is helpful to counteract part of the collisions and at the same time does not reduces heavily the spectral efficiency.\footnote{Please note that the rate for ALOHA is not depicted in Figure~\ref{fig:rate_max} because it has a different degree $\dg$, and therefore the results are not comparable.} 

\section{Detection, Combining and Decoding - A Two-Phase Procedure}
\label{sec:sys_mod_ECRA_dis}

In this Section we present the procedure that can be adopted for allowing \ac{ECRA} to do combining without the need of decoding the packets beforehand. The key idea is to do split the detection, combining and decoding into a two phase procedure:
\begin{enumerate}
\item \textbf{Detection Phase:} All candidate replicas are detected exploiting the sync word common to every user which is concatenated to the packet, in the first phase.
\item \textbf{Replica Matching Phase:} A data-aided non-coherent soft-correlation on candidates surviving the matching criterion is carried out and the first $\dg-1$ matches are declared as replicas belonging to the same user. On them we apply combining and decoding is attempted on the combined observation.
\end{enumerate}

Differently to the original \ac{ECRA} transmission procedure, here we introduce (optional) time slots. We shall note that the time slots are not necessary for the correct behaviour of the two-phase procedure, but are easing the second phase reducing the number of possible replicas' matches, as we will see later. In the same way as the original \ac{ECRA} transmission procedure, each user arranges its transmission within a \ac{VF}. The \ac{VF} are asynchronous among users. Differently from the original \ac{ECRA}, each \ac{VF} is divided in $\numSlot$ slots of duration $\SlotSize$, so that $\fraLen=\numSlot\, \SlotSize$, see also Figure~\ref{fig:MAC_frame_tx}. Users transmit $\dg$ replicas of duration $\pkLen$ seconds within the \ac{VF}. Each replica is transmitted over $\nSlotsSmall$ consecutive slots within the \ac{VF} and we have that a replica duration is a multiple of the slot duration, $\pkLen=\nSlotsSmall\, \SlotSize$. Each replica is composed by $\numSym$ modulated symbols and the symbol duration is $\symLen$, so ${\numSym\, \symLen=\nSlotsSmall\, \SlotSize=\pkLen}$. Each replica is transmitted starting from a slot index chosen uniformly at random in ${[0,\numSlot-\nSlotsSmall-1]}$, rejecting starting slot indexes which lead to self-interference among replicas of a user's packet.

\begin{figure*}
    \centering
    \includegraphics[width=\textwidth]{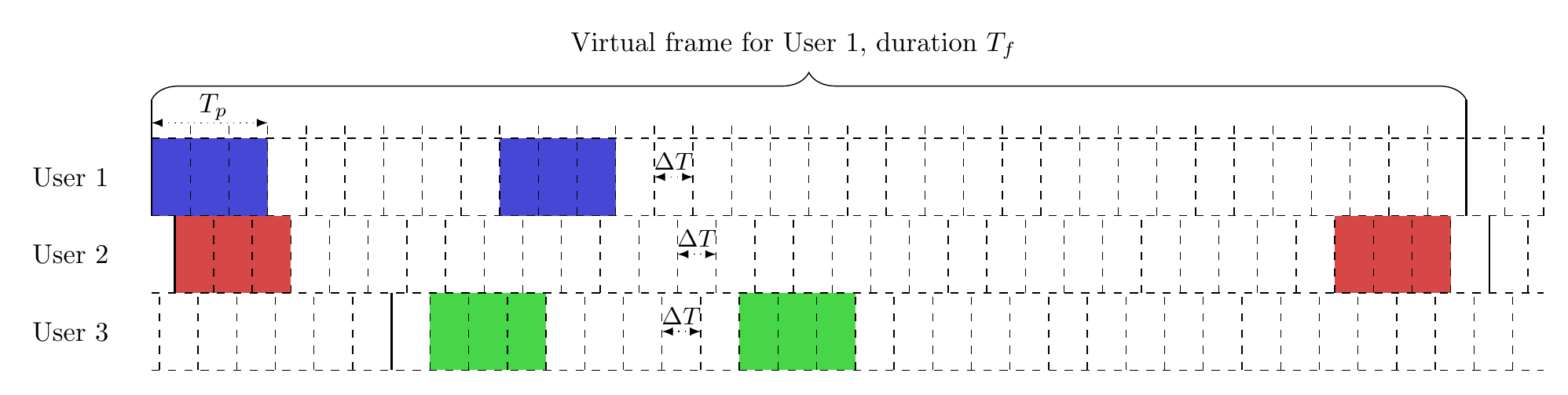}
    \caption{Transmitted signals. Each user sends two replicas of duration $\pkLen$ seconds that occupy $3$ time slots in the example.}
    \label{fig:MAC_frame_tx}
\end{figure*}
In contrast to \ac{CRA} \cite{Kissling2011a} and the first version of \ac{ECRA} \cite{Clazzer2012}, no pointer field is required in the header for localizing the replicas position. The first section of each replica is a sync word composed by $\syncSym$ binary symbols $\symVec=(\sym_0,...,\sym_{\syncSym-1})$ common to all users, with ${\sym_i\in \{-1,+1\}}$ for $i=0,...,\syncSym-1$. The sync word is concatenated with the BPSK modulated data part and sent through an \ac{AWGN} channel. The data part carries the actual information and the redundancy introduced by the channel code.

\subsection{Detection and Decoding}
\label{sec:rx}
At the receiver side, the incoming signal $\rx(\tm)$ is sampled and input to the frame start detector. The receiver will operate with a sliding window, similarly to \cite{Meloni2012, deGaudenzi2014_ACRDA}. The decoder starts operating on the first $\Wind$ samples, with $\Wind$ the designed window size. First it detects candidate replicas.
\begin{figure*}
\centering
 \subfigure[Non-coherent soft-correlator used for the detection of candidates replicas.]
   {\includegraphics[width=\textwidth]{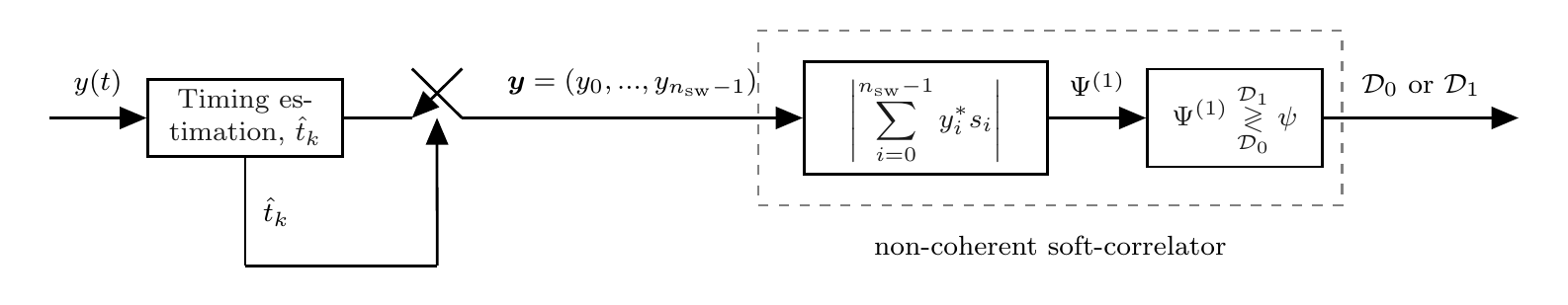}
   \label{fig:rule1}}
 \subfigure[Interference-aware soft-correlator used for the detection of candidates replicas.]
   {\includegraphics[width=\textwidth]{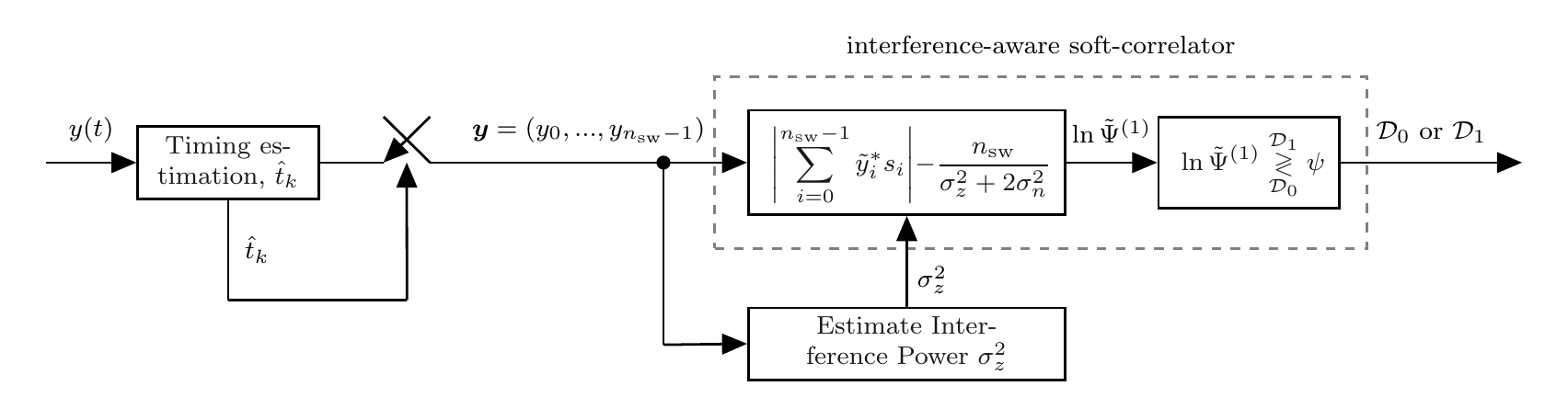}
   \label{fig:rule2}}
 \caption{The non-coherent soft-correlator and the interference-aware soft-correlator used for the detection of candidate replicas.}
 \label{fig:rules}
\end{figure*}

\subsubsection{Detection Phase}
In the first phase the non-coherent soft-correlation metric \cite{Chiani_2010} is used for identifying candidates replicas; see Figure~\ref{fig:rule1}. Within a receiver window, a threshold-based test is applied to each of the $\Wind-\syncSym$ sequences of $\syncSym$ consecutive samples (referred in the following as test intervals) to detect the presence of a sync word. We denote with
\begin{equation*}
\rxVec=(\rx_0,...,\rx_{\syncSym-1})
\end{equation*}
the sequence of $\syncSym$ samples on which the threshold test is applied. Here, we are implicitly assuming that the epoch is estimated prior to frame synchronization. Under the hypothesis that the test interval is aligned to a sync word, the epoch estimation can be reliably performed using pilot-aided\footnote{Observe that the sync word can be effectively used as pilot field for timing estimation.} techniques mutated from code synchronization algorithms used in spread-spectrum communications, see \eg \cite{Polydoros1984} and references therein. If the test window is not aligned with the sync word of any user, we assume the epoch estimator returning a random sampling offset, uniformly-distributed in $(0,\symLen]$. For each test interval \--- similarly to \cite{Chiani_2010} \--- the frame synchroniser has to decide among two hypothesis, \ie
\[
\begin{aligned}
\HypBelow &:\rxVec=\intVec+\noiseVec\\
\HypAbove &:\rxVec=\symVec\,e^{j\phase^{(\user,\replica)}}+\intVec+\noiseVec
\end{aligned}
\]
where the first hypothesis refers to the case of no sync word, while the second one refers to the case of sync word present. Here $\noiseVec=(\noise_0,...,\noise_{\sw-1})$ are samples of a discrete white Gaussian process with $\noise_i \sim \mathcal{CN}(0,2\noiseVar)$ and $\intVec$ is the interference contribution over the $\syncSym$ observed samples.

We adopt the threshold test
\begin{equation}
\label{eq:corr}
\TestThrOne(\rxVec)=\left|\sum_{i=0}^{\syncSym-1}\rx_{i}^* \sym_i\right|\underset{\DecBelow}{\overset{\DecAbove}{\gtrless}}\ThrCorr.
\end{equation}
Where decision $\DecAbove$ corresponds to hypothesis $\HypAbove$ and decision $\DecBelow$ corresponds to hypothesis $\HypBelow$ and the threshold $\ThrCorr$ is the discriminant between the two decision regions. We call ${\RStartSet=\left\{\RStart_1 ,\RStart_2,...\right\}}$ the set of candidate replicas starting positions, \ie the set containing the positions within the receiver window for which the test of equation~\eqref{eq:corr} outputs $\DecAbove$. The set of candidate replica positions is the outcome of the first phase.

\subsubsection{Replica Matching Phase}
Let us consider the first candidate replica identified in the first phase. We denote its starting position as $\RStart_1$, with $\RStart_1 \in \RStartSet$. The focus is on finding a subset $\RStartSet_1\subseteq \RStartSet$ containing the initial positions of bursts that are likely replicas of the (hypothetical) burst starting in position $\RStart_1$. To do so, we define the following compatibility criterion:
\begin{defin}[Compatibility Criterion] \label{def:comp_crit}
	A start position $\RStart_i \in \RStartSet$ is said to be compatible with $\RStart_1$ iff
\begin{equation}
\RStart_i=\RStart_1+k\, \SlotSize
\end{equation}
for some positive integer $k$, $\RStart_1<\RStart_i<\Wind\, \symLen-\SlotSize$.
\end{defin}
The set $\RStartSet_1$ is hence formally defined as
\begin{equation}
\RStartSet_1\triangleq \left\{\RStart_i \in \RStartSet| \RStart_i=\RStart_1+k\, \SlotSize, k \in \mathbb{Z}^+\right\}.
\end{equation}
The subset $\RStartSet_1$ contains the starting positions that are compatible (given the \ac{VF} structure) with $\RStart_1$, \ie their associated burst are likely replicas of the burst starting at position $\RStart_1$.

Denote with $\rxVec^{(i)}=(\rx_0^{(i)},...,\rx_{\numSym-1}^{(i)})$ the $\numSym$ samples of the received signal starting in position $\RStart_i$ within the window. For each $\RStart_i \in \RStartSet_1$, we compute the non-coherent correlation
\begin{equation}
\label{eq:corr_comb}
\TestThrTwo_{1,i}(\rxVec)\triangleq\left|\sum_{j=0}^{\numSym-1}\rx_{j}^{(1)}\left[\rx_j^{(i)} \right]^*\right|.
\end{equation}
We order the $\TestThrTwo_{1,i}$ in descending order and we mark the first $\dg-1$ as replicas of the same user.

On these replicas we apply combining techniques as \acs{SC}, \ac{MRC} or \ac{EGC}. If decoding is successful, all the replicas are removed from the received signal. The set $\RStartSet$ is updated accordingly by removing the starting positions of the cancelled replicas. The process is iterated until $\RStartSet$ is empty, or if decoding fails for all remaining candidates in $\RStartSet$. The channel decoder is assumed to be capable of identifying unrecoverable errors with high probability.\footnote{Error detection can be implemented either by using an incomplete channel decoder or by concatenating an outer error detection code with the inner channel code.} Once no more packets can be decoded within the window, the receiver's window is shifted forward by $\WindSh$ samples and the procedure starts again.

\subsection{Hypothesis Testing, Interference-Aware Rule}
\label{sec:hyp}
We derive here an advanced correlation rule, named $\TestThrInt$, which takes into consideration the presence of interference. We resort to a Gaussian approximation of the interference contribution. The interference term $\intVal_i$ is modeled as $\intVal_i\sim \mathcal{CN}(0,\intVar)$. Furthermore, we assume $\intVar$ to be constant for the entire test interval. The joint noise and interference contribution is given by $\noiseInt_i =\intVal_i +\noise_i$, so that $\noiseInt_i\sim \mathcal{CN}(0,\intVar+2\noiseVar)$.
The approximated \ac{LRT} is then obtained by evaluating,
\begin{align}
\label{eq:adv_corr1}
\TestThrInt(\rxVec) &= \frac{f_{\rxrvVec|\HypAbove}(\rxVec|\HypAbove)} {f_{\rxrvVec|\HypBelow}(\rxVec|\HypBelow} \underset{\DecBelow}{\overset{\DecAbove}{\gtrless}} \ThrCorrInt\\
\intertext{
where $f_{\rxrvVec|\Hyp_i}(\rxVec|\Hyp_i)$ is the approximated distribution of the random vector $\rxrvVec=(\rxrv_0,...,\rxrv_{\syncSym-1})$ under the hypothesis $\Hyp_i$. For the $\HypBelow$ hypothesis we can write}
\label{eq:H0}
f_{\rxrvVec|\HypBelow}\left(\rxVec|\HypBelow\right) &= \prod_{i=0}^{\syncSym-1}\frac{1}{\pi \left(\intVar+2\noiseVar\right)} e^{-\frac{|\rx_i|^2} {\intVar+2\noiseVar}}.\\
\intertext{For the $\HypAbove$ hypothesis we can write}
\label{eq:H1}
f_{\rxrvVec|\HypAbove,\phase}(\rxVec|\HypAbove,\phase) &= \prod_{i=0}^{\syncSym-1}\frac{1}{\pi \left(\intVar+2\noiseVar\right)} e^{-\frac{|\rx_i-\symVal_ie^{j\phase}|^2} {\intVar+2\noiseVar}}\\
\intertext{We define $\rxTwo_i=\rx_i/\left(\intVar+2\noiseVar\right)$. Averaging \eqref{eq:H1} over $\phase$ we find,}
\label{eq:H1_2}
f_{\rxrvVec|\HypAbove}(\rxVec|\HypAbove) &= \left[\prod_{i=0}^{\syncSym-1}\frac{1}{\pi \left(\intVar+2\noiseVar\right)} e^{-\frac{|\rx_i|^2+1} {\intVar+2\noiseVar}}\right] \cdot I_0\left(\left|\sum_{i=0}^{\syncSym-1}\rxTwo_i^*\symVal_i\right|\right).\\
\intertext{Substituting equations~\eqref{eq:H1_2} and \eqref{eq:H0} in the expression of equation~\eqref{eq:adv_corr1} we get}
\label{eq:adv_corr2}
\TestThrInt(\rxVec) &= e^{-\frac{\syncSym}{\intVar+2\noiseVar}} I_0\left(\left|\sum_{i=0}^{\syncSym-1}\rxTwo_i^*\symVal_i\right|\right) \underset{\DecBelow}{\overset{\DecAbove}{\gtrless}} \ThrCorrInt.\\
\intertext{\begin{sloppypar}Applying the natural logarithm of both sides and making the use of the approximation ${\ln (I_0(x))\cong |x|-\ln\sqrt{2\pi |x|}\cong |x|}$ \cite{Chiani_2010}, we can rework equation~\eqref{eq:adv_corr2} as\end{sloppypar}}
\label{eq:adv_corr3}
\ln \TestThrInt(\rxVec) &\cong \left|\sum_{i=0}^{\syncSym-1}\rxTwo_i^*\symVal_i\right| - \frac{\syncSym}{\intVar+2\noiseVar} \underset{\DecBelow}{\overset{\DecAbove}{\gtrless}} \ThrCorr
\end{align}
where $\ThrCorr=\ln\left(\ThrCorrInt\right)$. With respect to the non-coherent soft-correlation rule of equation~\eqref{eq:corr}, we can observe that in \eqref{eq:adv_corr3} the correlation term is followed by a correction term that depends on the sync word length and on the interference level. The latter is required to be estimated (See Figure~\ref{fig:rule2}). 
\section{Two-Phase Procedure Numerical Results}
\label{sec:num_res_ECRA_dis}

We first compare the two non-coherent soft-correlation rules derived in Sections \ref{sec:rx} and \ref{sec:hyp} in terms of \ac{ROC}. In the second part we show the performance of the \ac{ECRA} receiver in terms of the probability of correct detection of the replicas and the probability of correct combining of replicas from the same user.

\subsection{ROC Comparison}
\begin{sloppypar}
We compare the performance of the two correlation rules $\TestThrOne$ and $\TestThrInt$ via Monte Carlo simulations. The comparison is done in terms of \ac{ROC}. The false alarm probability $\probFalse$ is defined as ${\probFalse=\Pr\{\Lambda>\lambda|\DecBelow\}}$. The detection probability $\probDet$ is defined as ${\probDet=\Pr\{\Lambda>\lambda|\DecAbove\}}$. We set $f_{\mathrm{max}} = 0.01/ \symLen$. The aggregate signal is then summed with Gaussian noise. The selected $E_s/\noiseSD$ is $\EnSym/\noiseSD=10$~dB. A sync word of $32$ bits of hexadecimal representation $\{1ACFFC1D\}$ has been adopted, which results in $\syncSym=32$ symbols.
\end{sloppypar}

\begin{figure}
\centering
 \subfigure[\ac{ROC} for $\TestThrOne$, $\TestThrInt$, $\load=0.5$.]
   {\includegraphics[width=0.48\textwidth]{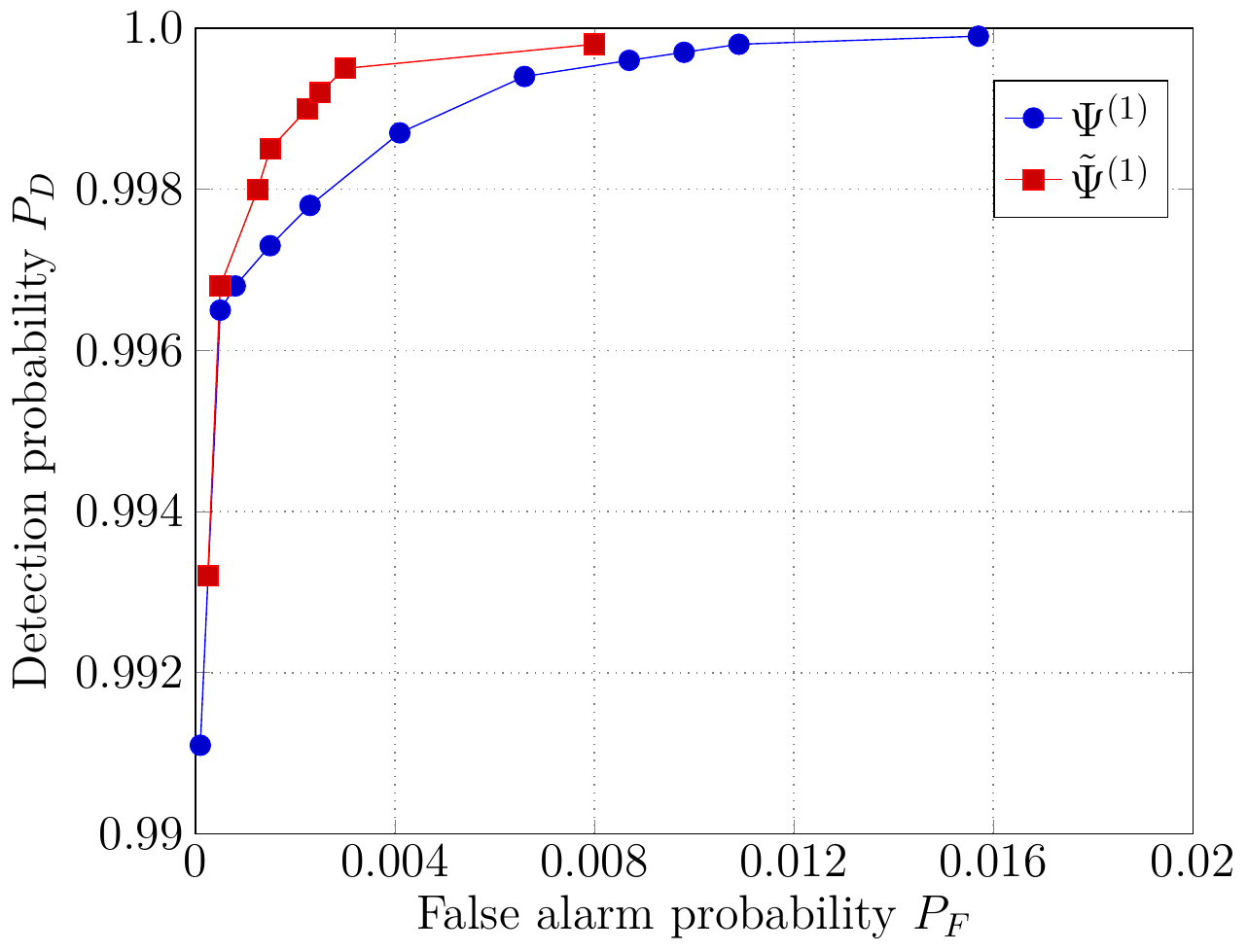}
   \label{fig:Roc_0_5}}
 \subfigure[\ac{ROC} for $\TestThrOne$, $\TestThrInt$, $\load=1.5$.]
   {\includegraphics[width=0.48\textwidth]{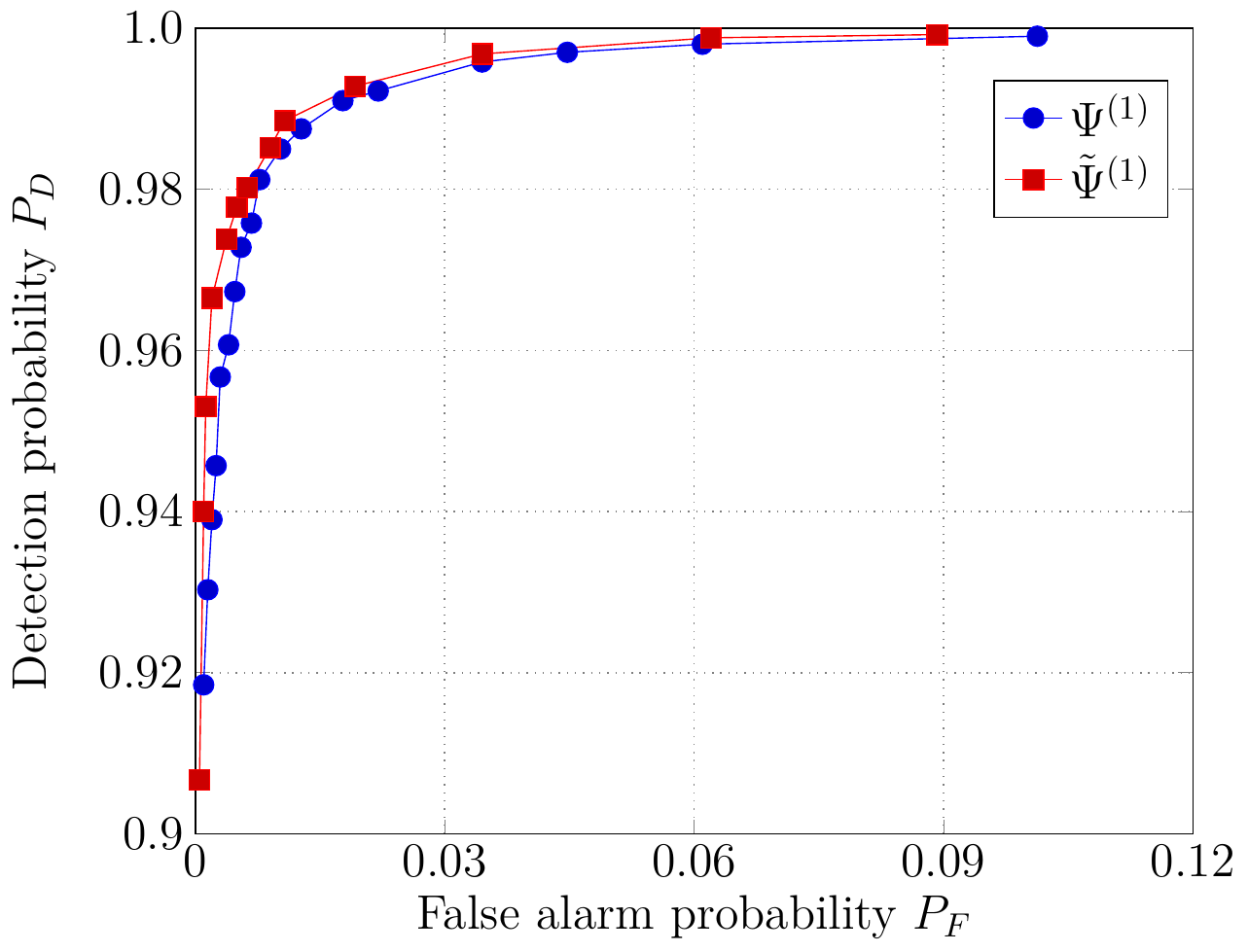}
   \label{fig:Roc_1_5}}
 \caption{\Ac{ROC} for the non-coherent and interference-aware soft-correlation synchronization rules, with $\load=\{0.5,1.5\}$, equal received power, $\EnSym/\noiseSD=10$~dB and $\syncSym=32$ symbols.}
 \label{fig:Roc}
\end{figure}

Results for channel traffic values $\load=\{0.5,1.5\}$ are presented in Figure~\ref{fig:Roc}. As expected, the knowledge on the interference level exploited in rule $\TestThrInt$ leads to a better \ac{ROC} performance, regardless of the channel traffic conditions. Nevertheless, the gain compared to the non-coherent correlation rule $\TestThrOne$ is rather limited. In general, both rules show good performances, reaching a detection probability $\probDet>0.99$ for $\probFalse>0.02$ in the worst case (channel traffic $\load=1.5$).

\subsection{ECRA Detection and Replicas Coupling Performance}
We present here the results for the detection and the correct combining probabilities. We focus on the particular setting where $\dg=2$ (\ie users transmit $2$ replicas of their packets). The detection probability $\probDet$ has been defined in the previous subsection. We define the correct combining probability $\probCorComb$ as the probability that two replicas of a burst are correctly selected for combining after the two-phase procedure. Obviously, $\probCorComb\leq \probDet^2$, \ie a necessary condition for correct combination is the actual detection of the sync words associated with the two replicas, during the first phase. We select a fixed threshold $\ThrCorr^*$ equal for all the channel traffic values and we use the non-coherent soft-correlation rule $\TestThrOne$. The threshold $\ThrCorr^*$ has been selected through numerical simulations. We show the results in Figure~\ref{fig:Coup_P}, for a \ac{SNR} of $\EnSym/\noiseSD=10$~dB. The discretization interval equals to one physical layer packet duration, \ie $\SlotSize = \pkLen$. Each packet is composed by a sync word of $\syncSym=32$ symbols (as the one already presented) and a total of $\symLen=1000$ BPSK antipodal modulated symbols (including the sync word symbols), the \ac{VF} duration as well as the window duration $\Wind\, \symLen$ are $100$ times the packet duration, $\fraLen=\Wind\, \symLen=100\, \pkLen$.

\begin{figure}
    \centering
    \includegraphics[width=0.8\textwidth]{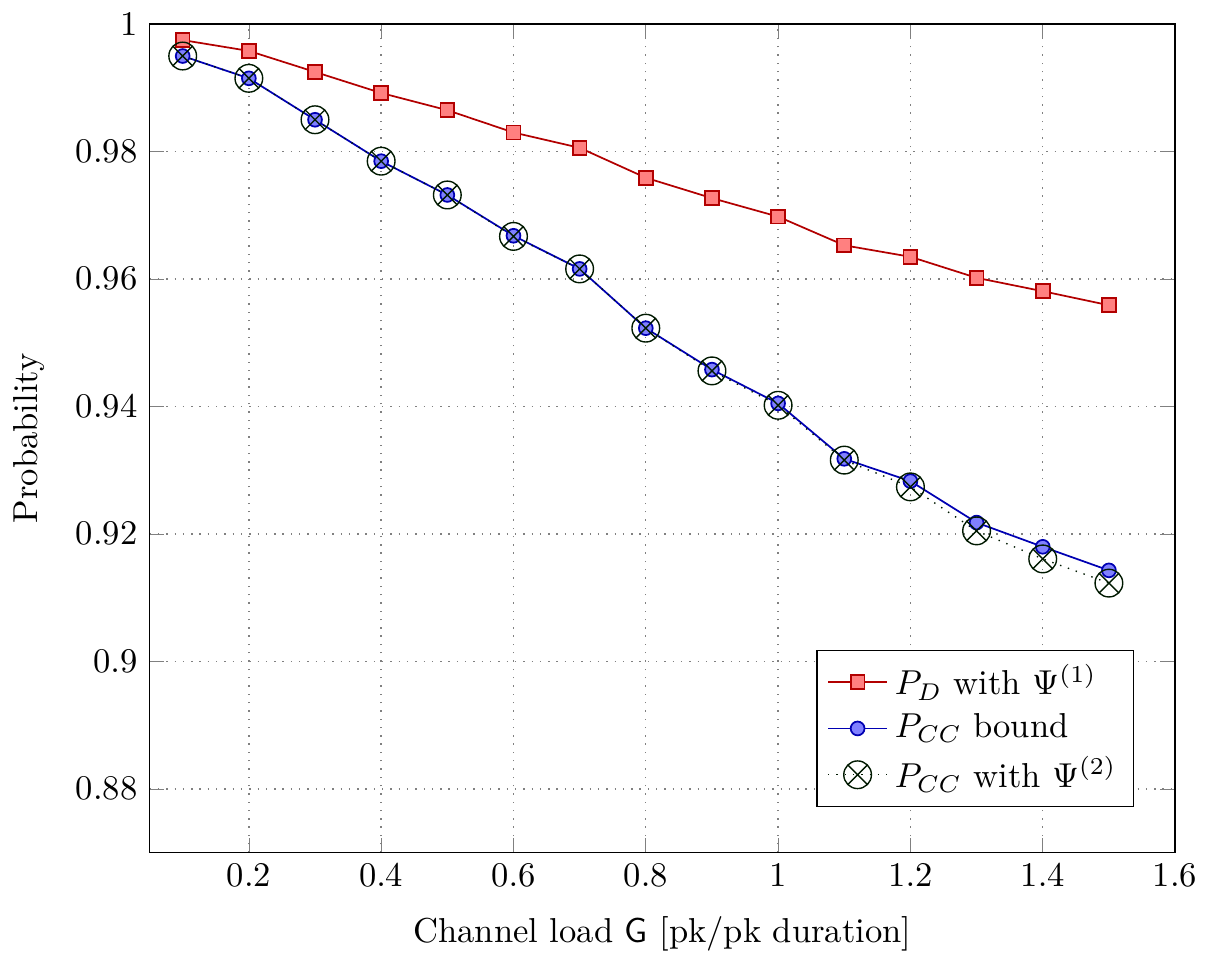}
    \caption{Detection probability $\probDet$ for a fixed threshold $\ThrCorr^*$ independent from the channel traffic using $\TestThrOne$ and correct combining probability $\probCorComb$ with $\TestThrTwo$.}
    \label{fig:Coup_P}
\end{figure}

Observe that the detection probability remains above $95\%$ for all the channel traffic values $\load$, up to $\load=1.5$. The non-coherent soft-correlation rule $\TestThrOne$ is particularly robust to variations in the channel traffic, since the presented results are obtained for a single threshold value $\ThrCorr^*$ which has been kept constant for all the channel traffic values. For all values of channel traffic simulated, the correct combining probability is very close to the bound $\probDet^2$.

\subsection{Spectral Efficiency}
We compare the simulation results in terms of spectral efficiency achieved by \ac{ECRA} with \ac{MRC}, after the two-phase detection process described in Section~\ref{sec:rx} and under ideal replicas matching. The proposed technique is compared to the idealised case in which all replicas positions are known to the receiver prior to decoding. We select $\SlotSize=\pkLen$ and again the window duration is $\fraLen = \Wind\, \symLen=100\, \pkLen$. Perfect \ac{CSI} at the receiver is assumed for enabling \ac{MRC}.

We adopt the non-coherent soft-correlation rule $\TestThrOne$ and a fixed threshold kept constant, regardless the channel load $\load$. All replicas are received with equal power $\EnSym/\noiseSD=2$~dB. We assume a capacity-achieving code which adopts a Gaussian codebook with rate $\rate=1$, so that if the mutual information at the output of the combiner exceeds the rate $\rate$, then the packet is considered to be successfully decoded. Further refinements of the decoding model can be adopted following a realistic \ac{PLR} performance of a specific code for example. Nonetheless, for the present work such a model is sufficient to show the goodness of the detection and identification approach. The maximum number of \ac{SIC} iterations is set to $10$, and \ac{SIC} is assumed ideal. That is, if the position of both replicas of one user is known at the receiver, \ac{MRC} is applied and if the packet can be decoded its interference contribution is fully removed from the received signal.
\begin{figure}
    \centering
    \includegraphics[width=0.8\textwidth]{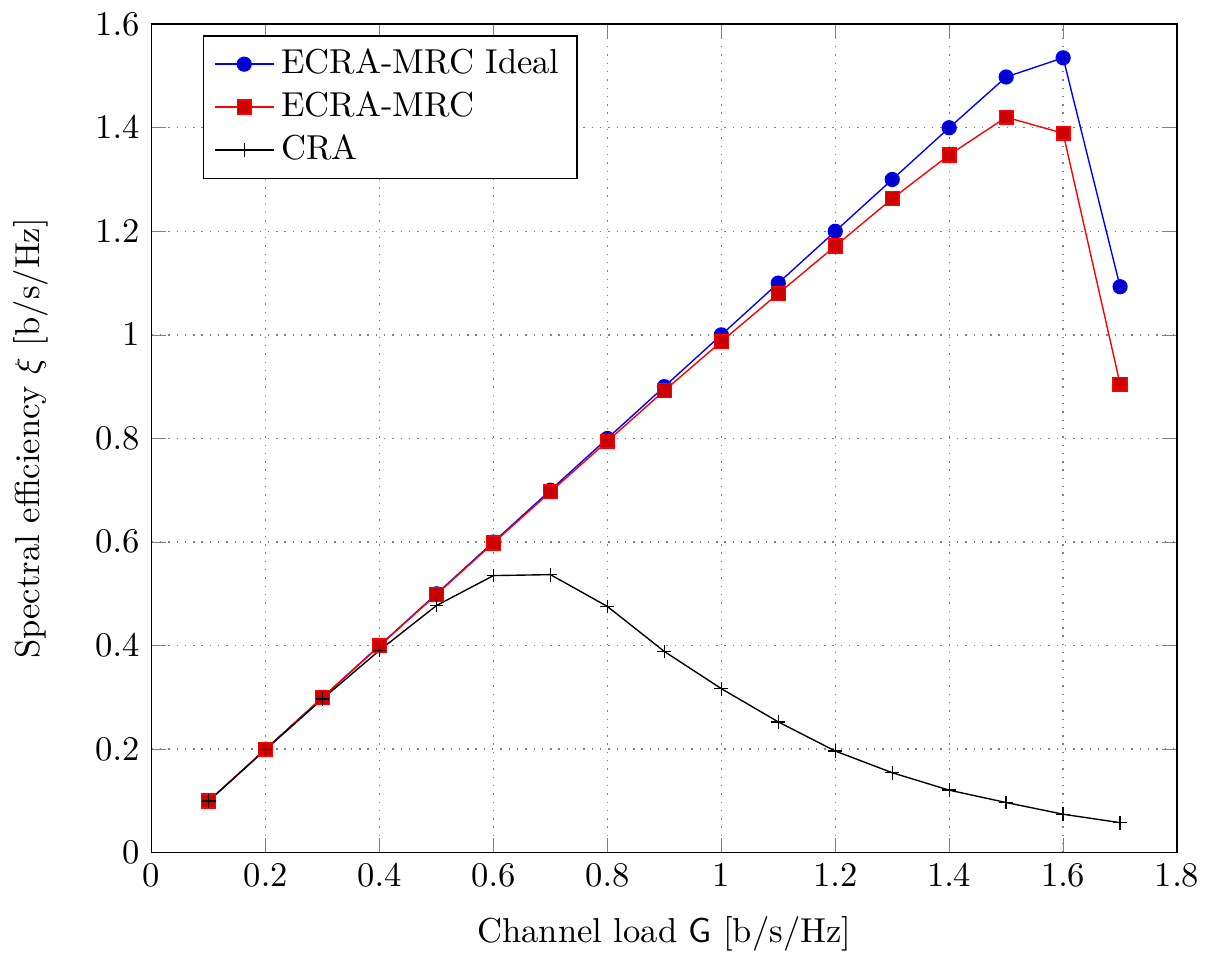}
    \caption{Spectral efficiency of \ac{ECRA}-\ac{MRC} with the proposed two phase detection and combining technique compared to the ideal \ac{ECRA}-\ac{MRC}.}
    \label{fig:Thr}
\end{figure}
In Figure~\ref{fig:Thr}, we present the spectral efficiency results for the proposed two phase detection and combining technique (called \ac{ECRA}-\ac{MRC} in the legend) the ideal \ac{ECRA}-\ac{MRC} where all the replica positions are known at the receiver. The proposed technique is close to the performance of the ideal case. The maximum spectral efficiency exceeds $1.4$ b/s/Hz, which is only $8\%$ less than the maximum spectral efficiency of the ideal case. For reference, also the performance without \ac{MRC} is depicted in the figure (\ac{CRA} in the legend). 
\section{Conclusions}
\label{sec:conclusions}
\acreset{UCP}

In this Chapter we presented a novel \ac{RA} decoding algorithm named \ac{ECRA}. Motivated by the presence of \acsp{UCP} within frames, \ac{ECRA} tried to reduce their detrimental impact on the receiver's \ac{SIC} procedure applying combing techniques. \ac{ECRA} exploited the presence of multiple instances of the same packet in order to trigger a \ac{SIC} procedure. In addition, \ac{ECRA} tried to reduce interference of asynchronous protocols attempting to resolve partial collisions among packets, with the creation of a combined observation. The combined observation was either generated from the lowest interfered parts of the replicas sent within the frame, \ie \ac{SC}, or from the weighted combination of the replicas symbols of each user, \ie \ac{MRC}. An analytical framework for evaluating a lower bound on the \ac{UCP} for \ac{CRA} with and without channel coding as well as \ac{ECRA} was developed. The lower bound was then used for computing an approximation of the \ac{PLR} which was found to be tight in the low channel load region. A comprehensive framework with several metrics was also derived, in order to compare both asynchronous and slot synchronous schemes with channel coding. An investigation on the performance of \ac{ECRA} under average power constraint was also performed. Numerical simulations showed that \ac{ECRA} in both its variants, largely outperformed \ac{CRA} in all the considered scenarios for both throughput and \ac{PLR}. Furthermore \ac{ECRA-MRC} was able to double the maximum throughput w.r.t. \ac{CRA} while \ac{ECRA-SC} had an improvement of 21\% in the maximum throughput w.r.t. \ac{CRA}. For a properly selected rate, \ac{ECRA-MRC} was also able to outperform \ac{CRDSA} with the same number of replicas. Finally, \ac{ECRA-MRC} showed remarkable performance gains in terms of normalised capacity w.r.t. the other asynchronous \ac{RA} schemes reaching up to 75\% the \ac{MAC} channel capacity. In order to present a practical way to exploit the \ac{ECRA} concept, a two-phase procedure was also proposed which avoided the need of pointers to the replicas. Candidate replicas were identified thanks to the presence of a common sync word using non-coherent soft-correlation. In the second phase, data-aided non-coherent soft-correlation was applied in order to rank the candidate replicas and the first $\dg-1$ were declared as replicas of the same users, on top of which decoding was attempted. We showed that in a very simple setup, this approach was very close to the ideal performance of \ac{ECRA}.
 
\chapter[Layer~$3$ Throughput and PLR Analysis for Advanced RA]{Layer~$3$ Throughput and Packet Loss Rate Analysis for Advanced Random Access}
\label{chapter5}
\thispagestyle{empty}
\ifpdf
    \graphicspath{{chapter5/Images/}}
\fi
\epigraph{The most exciting phrase to hear in science, the one that heralds the most discoveries, is not ``Eureka!" (I found it!) but ``That's funny..."}{Isaac Asimov}

In the previous Chapter we focused on the investigation of \ac{ECRA}. We showed that in specific scenarios (code rate, received \ac{SNR}) asynchronous \ac{RA} was able to outperform also slot synchronous \ac{RA} as \ac{CRDSA}. In this Chapter we further elaborate on this comparison, moving one layer above the \ac{MAC}. We want to verify if the identified benefits of asynchronous \ac{RA} are still present when higher layers generate packets with variable data content. While with asynchronous \ac{RA}, packets can be transmitted whichever size they have, for slot synchronous systems packets might not fit into one \ac{MAC} time slot, requiring fragmentation and eventually padding. Nonetheless, it is known from the literature that for ALOHA under the collision channel model, the packet size distribution maximising the throughput is the one having all packets with the same size \cite{Ferguson1977,Abramson1977,Bellini1980}.

The degradation due to the variable packet size for asynchronous \ac{RA} schemes, and the degradation due to fragmentation and padding overhead for the time synchronous \ac{RA} schemes, lead to interesting \--- yet unexplored \--- research which option leads to better results in terms of layer~$3$ throughput and \ac{PLR}.

\section{Introduction}


For new \ac{RA} protocols the investigation of higher layer throughput is a very important and still a missing element in literature. The possibility to exploit gains in terms of throughput and \ac{PLR} also at higher layers can open to \ac{RA} protocols new application scenarios still considered unsuitable for such protocols.

In this Chapter we analyse and investigate the layer~$3$ throughput and \ac{PLR} of framed asynchronous and slot synchronous \ac{RA} protocols, such as ALOHA, \ac{SA}, \ac{CRA}, \ac{CRDSA}, \ac{ECRA} and \ac{IRSA}. A generic layer~$3$ packet length distribution $\LDist$ is assumed. We propose an analytical expression of the layer~$3$ throughput and \ac{PLR} for slot synchronous \ac{RA} so to avoid numerical simulations dedicated to such layer. It will require knowledge of the packet successful probability at layer~$2$ and the distribution $\LDist$. For the cases in which the latter is not available we elaborate a bound (lower bound for the layer~$3$ throughput and upper bound for the layer~$3$ \ac{PLR}) that requires the knowledge only on the average of the distribution $\LDist$. We provide finally numerical results for the performance of all the considered schemes.

\section{System Model}



We consider a population of $\numUs$ users sharing the medium. Each user generates packets of layer~$3$ where the duration is subject to a generic probability measure $\LDist$.\footnote{The probability measure $\LDist$ corresponds to the \ac{PDF} in case of continuous random variable, or corresponds to the \ac{PMF} in case of discrete random variable.} The probability measure captures the behaviour of the traffic profile followed by the users. We assume that the probability measure $\LDist$ is equal for all the users and can be either continuous or discrete. The user cannot forward the next layer~$3$ packet to the layer~$2$ until the all the previous layer~$2$ packet segments have been sent. Furthermore, for the investigation carried out in this Chapter, the protocol is assumed to work in open loop, \ie no feedback and packet acknowledgment procedure is here considered.

All the layer~$2$ \ac{MAC} protocols evaluated in the following are \ac{RA}, either asynchronous or slot synchronous. We consider framed ALOHA \cite{Okada1978,Wieselthier1989}, \ac{SA}, \ac{CRA}, \ac{CRDSA}, \ac{ECRA} and \ac{IRSA}. For all protocols, a frame structure is included, in which appear all packet transmissions of the active users. In this way, every time a layer~$3$ packet is generated, it is forwarded to layer~$2$ and it is sent in the upcoming frame. For slot synchronous systems, before it can be sent, it is eventually fragmented and sent in consecutive frames. This requires also the asynchronous schemes to have a (possibly loose) time synchronization to the frame start. The analysis performed in Section~\ref{sec:L3_T_analys} is independent of the specific protocol and can be applied also to other \ac{RA} schemes. The time is composed by frames of duration $\fraLen$. Each user can try to transmit only once for each \ac{MAC} frame and we assume that all users have always packets to be sent. Upon the generation of more than one packet per frame, the users need to store or discard the exceeding packets (depending on the protocol implementation). The physical layer packets are supposed to have a duration of $\pkLen$. Concerning slot synchronous protocols, the frames are further subdivided in slots of a fixed duration $\slLen$. In slot synchronous schemes, the physical layer packet duration has to be equal to the slot duration in order to fit in the time slot, \ie $\pkLen=\slLen$, while in asynchronous schemes this constraint is relaxed.

We define the offered traffic load $\load$ as the fraction of time occupied by transmissions. In formulas, $\load=\left(\numUs \, \pkLen\right)/\,\fraLen$. In the slot synchronous schemes we can write also $\load=\numUs / \numSlot$, where $\numSlot=\fraLen/\,\pkLen=\fraLen/\,\slLen$ is the number of slots that constitute each \ac{MAC} frame. 
\section{L3 Throughput and Packet Loss Rate Analysis}
\label{sec:L3_T_analys}

The throughput is one of the mostly used performance metrics for \ac{RA} protocols and it is defined as the average number of decoded packets per packet duration. We recall that $\psucc$ is the probability of success which is in general function of both the channel load $\load$ but also of the packet duration distribution $\LDist$,
\begin{equation}
\label{eq:T_L3T}
\tp = \load \, \psucc(\load,\LDist).
\end{equation}
The packet loss rate $\plr$ is the probability of packet loss or unsuccessful decoding, \ie
\begin{equation}
\label{eq:PLR_L3T}
\plr = 1- \psucc(\load,\LDist).
\end{equation}
For simplicity of notation these dependencies are neglected in the following.

\subsection{Asynchronous RA Protocols}

When employing asynchronous framed protocols, as framed ALOHA or \ac{CRA}, all layer~$3$ packet durations not exceeding the frame duration can be accommodated in the \ac{MAC} frame. In this way, no fragmentation of layer~$3$ packets is required and the successful decoding probability of layer~$3$ packets $\psuccLT$ coincides with the successful decoding probability of layer~$2$ packets $\psucc$. Therefore, the layer~$3$ throughput $\tpLT$ of asynchronous \ac{RA} protocols is
\begin{equation}
\label{eq:T_L3_unslotted}
\tpLT = \load \, \psuccLT = \load \, \psucc = \tp.
\end{equation}
And the layer~$3$ \ac{PLR} is
\begin{equation}
\label{eq:PLR_L3_unslotted}
\plrLT = 1- \psuccLT = 1- \psucc = \plr.
\end{equation}

\subsection{Slot Synchronous RA Protocols}

Although users at layer~$3$ generate packets whose durations are ruled by $\LDist$ probability measure, slot synchronous schemes cannot accommodate packets with generic duration at \ac{MAC} layer, and fragmentation is needed. When a layer~$3$ packet has to be fragmented, additional overhead is added to the fragmented layer~$2$ packets. Two types of overhead are added to the layer~$2$ packets: 1) fragmentation overhead and 2) padding overhead. The former is due to the need to have a layer~$3$ header for each fragment, where the fields to both route the packet in the network and to recompose the packet at the destination are contained. Header fields dedicated to fragmentation must include: original-not-fragmented packet identifier, fragment identifier, and either overall number of fragments or indicator about last/not last fragment. For example, concerning \ac{IP} layer~$3$ protocol, the fields identification, offset and flag M are used, respectively, for the aims listed above. The padding overhead is needed in order to fulfill layer~$2$ slot duration requirements.

Layer $3$ throughput $\tpLT$ for slotted schemes is
\begin{equation}
\label{eq:T_L3_slotted}
\tpLT = \load \, \psuccLT \,(1-\FragO) \,(1-\PadO),
\quad 0\leq \FragO\leq1, \quad 0\leq \PadO\leq1
\end{equation}
where $\FragO$ is the fragmentation overhead and $\PadO$ is the padding overhead. Layer $3$ \ac{PLR} $\plrLT$ for slotted schemes is
\begin{equation}
\label{eq:PLR_L3_slotted}
\plrLT = 1- \psuccLT.
\end{equation}
The successful decoding probability of layer~$3$ packets $\psuccLT$ in the slot synchronous schemes is the probability that all the layer~$2$ packets composing the layer~$3$ packet are received correctly weighted with the probability that this specific layer~$3$ duration has been selected,
\begin{equation}
\label{eq:P_succ_L3}
\psuccLT = \prob_1 \psucc + \prob_2 \psucc^2 + ... + \prob_n \psucc^n + ... = \sum_{i=1}^{\infty} \prob_i \psucc^i, \quad \sum_{i=1}^{\infty} \prob_i = 1.
\end{equation}

\begin{exmp} Case $\LDist = \frac{1}{2}\delta(x-1) + \frac{1}{2}\delta(x-2)$:
\label{ex:pk_dist_two_pk}

We can suppose for example, that the layer~$3$ packet duration distribution \ac{PMF} $\LDist$ results in $50$\% of layer~$3$ packets with a duration of one layer~$2$ packet and $50$\% of layer~$3$ packets with a duration of two layer~$2$ packets, \ie $\LDist = \frac{1}{2} \delta(x-1) + \frac{1}{2} \delta(x-2)$. The layer~$3$ successful decoding probability for the first packet duration (one layer~$2$ packet) is simply $\psucc$. Actually, we need only that the layer~$2$ received packet is successfully decoded. While for the second packet duration (two layer~$2$ packets) the probability is $\psucc^2$ because two consecutive layer~$2$ packets must be successfully decoded. In this way, equation~\eqref{eq:P_succ_L3} becomes for this example
\begin{equation}
\label{eq:exemp_1_Psucc_L3}
\psuccLT = \frac{1}{2}\psucc + \frac{1}{2}\psucc^2.
\end{equation}
Note that we can write the closed form expression of $\psucc$ when considering \ac{SA} protocol. In this case $\psucc=e^{-G}$ assuming an infinite user population \cite{Abramson1977} and we can further write
\begin{equation}
\label{eq:exemp_1_Psucc_SA}
\psuccLT = \frac{1}{2}e^{-G} + \frac{1}{2}e^{-2G}.
\end{equation}
\end{exmp}

\subsubsection{Computation of $p_i$}

In general, $\LDist$ can be continuous, \ie the layer~$3$ packet durations may not be constrained to be a multiple of layer~$2$ packet durations. In these scenarios the computation of $p_i$ is of paramount importance to determine the layer~$3$ successful decoding probability $\psuccLT$. The generic $p_i$ value is computed as the definite integral of the \ac{PDF} $\LDist$ in the interval from the layer~$2$ duration $i-1$ to the layer~$2$ duration $i$, where $i=1,...,\infty$ are discrete layer~$2$ durations ($i=1$ means one layer~$2$ packet duration),

\begin{equation}
\label{eq:p_i}
p_i=\int_{i-1}^{i}\LDist dx \quad i=1,...,\infty.
\end{equation}

\begin{exmp} Case $\LDist = e^{-x}$:
\label{ex:pk_dist_exp}

As second example, we assume $\LDist = e^{-x}$, the exponential distribution with mean of one layer~$2$ packet duration. In this example, according to equation~\eqref{eq:p_i}, $p_i$ are
\begin{equation}
\label{eq:exemp_2_pi}
p_i = \int_{i-1}^{i}\LDist dx = \int_{i-1}^{i}e^{-x} dx = \left[-e^{-x} \right]_{i-1}^{i} = e^{-i}\left(e-1\right) \quad i=1,...,\infty.
\end{equation}
After the discretization, the mean of the exponential distribution will be greater than one layer~$2$ packet duration. This is due to the discretization process, which is a ceiling operation. In general, the expected value of $\LDist$ after discretization, $\mathbb{E}[\LDistDisc]$ can be written as
\begin{equation}
\label{eq:exp_val}
\mathbb{E}[\LDistDisc]=\sum_{i=1}^{\infty}i\, p_i,
\end{equation}
where the result is expressed in layer~$2$ packet durations. The values of $p_i$ are computed according to equation~\eqref{eq:p_i}. In the case of the exponential distribution, equation~\eqref{eq:exp_val} becomes
\[
\mathbb{E}[\LDistDisc]=\frac{e}{e-1}\cong 1.58.
\]
\end{exmp}

\subsubsection{Distribution $\LDist$ Unknown}

In some situations $\LDist$ or its expression after discretization are not known. This prevents the analytical computation of $p_i$, but we would like to be able to analytically evaluate the successful decoding probability of layer~$3$ packets $\psuccLT$ also in these situations. We assume that the expected value of the probability measure after discretization $\mathbb{E}[\LDistDisc]$ is available or can be measured at the receiver.\footnote{The packets correctly decoded will contain the sequence number needed for layer~$3$ packet reconstruction that can be used to numerically evaluate $\mathbb{E}[\LDistDisc]$.} We can then compute a lower bound to the successful decoding probability of layer~$3$ packets exploiting Jensen's inequality \cite{Jensen} as
\begin{equation}
\label{eq:P_succ_L3_appr}
\psuccLT \geq \psuccLTAppr = \left(\psucc\right)^{\mathbb{E}[\LDistDisc]}.
\end{equation}
In this way the throughput in equation~\eqref{eq:T_L3_slotted} can be bounded from below by
\begin{equation}
\label{eq:T_L3_slotted_2}
\tpLT \geq \load \, \psuccLTAppr \,(1-\FragO) \,(1-\PadO), \quad 0\leq\FragO\leq1, \quad 0\leq \PadO\leq1.
\end{equation}
While the \ac{PLR} in equation~\eqref{eq:PLR_L3_slotted} can be bounded from above by
\begin{equation}
\label{eq:PLR_L3_slotted_2}
\plrLT \leq 1-\psuccLTAppr.
\end{equation}
Numerical results showing the difference between $\psuccLT$ and $\psuccLTAppr$ in terms of layer~$3$ throughput and \ac{PLR} are given in the next Section.

\subsubsection{Additional Remarks}

The analytical framework developed provides the possibility to evaluate the layer~$3$ throughput $\tpLT$ and \ac{PLR} $\plrLT$ in case of slot synchronous schemes without the need of numerical simulation of this layer. Two possibilities can be followed:
\begin{enumerate}
\item{When $\psucc$ and $\LDist$ are available, one can use equations~\eqref{eq:T_L3_slotted}, \eqref{eq:PLR_L3_slotted}, \eqref{eq:P_succ_L3}, \eqref{eq:p_i}};
\item{If instead $\psucc$ and $\mathbb{E}[\LDistDisc]$ are available, one can use equations~\eqref{eq:T_L3_slotted_2}, \eqref{eq:PLR_L3_slotted_2}, \eqref{eq:P_succ_L3_appr}}.
\end{enumerate}
In both cases the probability $\psucc$ of successful decoding at layer~$2$ is required. It is normally available or can be easily derived either in close form or through \ac{MAC} layer simulations.

\section{Numerical Results}

In this Section a comparison between asynchronous and time synchronous \ac{RA} protocols is provided. We consider two distributions for the layer~$3$ packet durations:
 \begin{enumerate}
 \item equal probability of one layer~$2$ and two layer~$2$ packet durations (see example \ref{ex:pk_dist_two_pk}), \ie $\LDist = \frac{1}{2}\delta(x-1)+ \frac{1}{2}\delta(x-2)$;
 \item exponential packet length distribution with mean one layer~$2$ packet duration (see example \ref{ex:pk_dist_two_pk}), \ie $\LDist = e^{-x}$.
 \end{enumerate}
We first present the throughput and packet loss rate performance at layer~$3$ with a comparison between all the considered schemes, \ie ALOHA, \ac{SA}, \ac{CRA}, \ac{CRDSA}, \ac{ECRA} and \ac{IRSA}. In order to see the degradation when looking at the layer~$3$ performance for time slotted systems, a comparison with the layer~$2$ throughput and \ac{PLR} is provided. Finally, we consider the time slotted approximation of the layer~$3$ throughput and \ac{PLR}, using only the average layer~$3$ packet duration.

Before discussing the numerical results, we briefly explain the simulation setup. Upon the generation of a layer~$3$ packet, its duration is drawn following one of the two aforementioned probability measures. For time synchronous schemes, when required, fragmentation in multiple layer~$2$ packets is carried out. This happens when a layer~$3$ packets exceed one layer~$2$ packet duration, which is set to $1000$ information bits. If multiple segments have to be sent, the user terminal is marked as busy and no further layer~$3$ packet is queued until all segments are sent. In this way, for slotted schemes, multiple frames might be required for transmitting a single layer~$3$ packet, while for asynchronous schemes the layer~$3$ packet can be accommodated in one single frame. For \ac{CRDSA}, \ac{CRA} and \ac{ECRA} at \ac{MAC} layer, every user sends two replicas within the frame, since a system degree $\dg=2$ has been selected. For \ac{IRSA}, instead, the \ac{PMF} from which every user selects the number of replicas follows $\Ld(x) = 0.5x^2+0.28x^3+0.22x^8$ \cite{Liva2011}. In this case, on average $50$\% of the transmitters will send two replicas, $28$\% three replicas and the rest $22$\% eight replicas. User transmission are organised into frames of duration $100\, \pkLen$, which corresponds, for the time-slotted systems, to frames composed of $100$ time slots. Please note that here no sliding window decoder is adopted for none of the schemes, differently from the previous Chapter. Packets are encoded with a rate $\rate=1$ Gaussian codebook channel code and are sent through an \ac{AWGN} channel. The received \ac{SNR} is set to $2$ dB and it is equal for all users. At the receiver we adopt the capacity based decoding condition under block interference (see Definition~\ref{def:decoding_capacity} in Chapter~\ref{chapter2}) for all schemes including ALOHA and \ac{SA}. In ALOHA, the capture effect is possible, while for \ac{SA} this is not the case.\footnote{Since in this scenario, for the slots in which two packets are sent it holds $\rate \geq \log_2(1+\frac{\usPw}{\noisePw+\usPw})$, no capture effect is possible for \ac{SA} or any other time synchronous scheme.} For any of \ac{CRA}, \ac{CRDSA}, \ac{IRSA} and \ac{ECRA} the receiver performs interference cancellation for a maximum of $10$ iterations. Additionally, in \ac{ECRA-MRC} combining is also adopted when packets are still present in the frame and \ac{SIC} is further performed on the decoded packet observation built after \ac{MRC}. Additional $10$ \ac{SIC} iterations are here enabled. For slot synchronous schemes, a layer~$3$ packet is declared as not correctly decoded if any of the fragments is not correctly received. No fragmentation or padding overhead is here considered for these schemes. We finally recall that the layer~$3$ performance of the time synchronous schemes are analytically derived from the layer~$2$ performance.

\subsection{Asynchronous and Slot Synchronous RA Layer~$3$ Performance Comparison}

In Figure~\ref{fig:T_L3_equal} the layer~$3$ throughput results for the considered schemes are presented with layer~$3$ packet duration distribution $\LDist = \frac{1}{2}\delta(x-1) + \frac{1}{2}\delta(x-2)$. The first observation is that \ac{ECRA} with \ac{MRC} is able to largely outperform all other asynchronous and slot synchronous schemes. Its peak throughput is the only one exceeding one packet per packet duration (or slot in the time synchronous case) and more then doubles the peak throughput of \ac{IRSA}. Furthermore \ac{IRSA} would require on average more power, since the its average degree is $3.6$ packets, while \ac{ECRA} transmits only two replicas per layer~$2$ packet. Nonetheless, \ac{ECRA} requires a decoding algorithm by far more complex than \ac{IRSA}, due to the need of performing \ac{MRC}. Finally, asynchronous schemes can benefit from the capture effect thanks to the selected rate, which is not the case for any of the slot synchronous ones. Similarly, \ac{CRA} with two replicas is able to outperform \ac{CRDSA} in terms of layer~$3$ throughput for the entire simulated channel load range. In this second case, no complexity increase is found comparing the two schemes in terms of the receiver algorithm. Finally, for reference purposes, also ALOHA and \ac{SA} performance are depicted in the figure. Interestingly, we can observe that also ALOHA is able to outperform \ac{SA} in this scenario, thanks to the capture effect that triggers correct decoding of packets even in presence of a collision.

\begin{figure}
\centering
\includegraphics[width=0.8\textwidth]{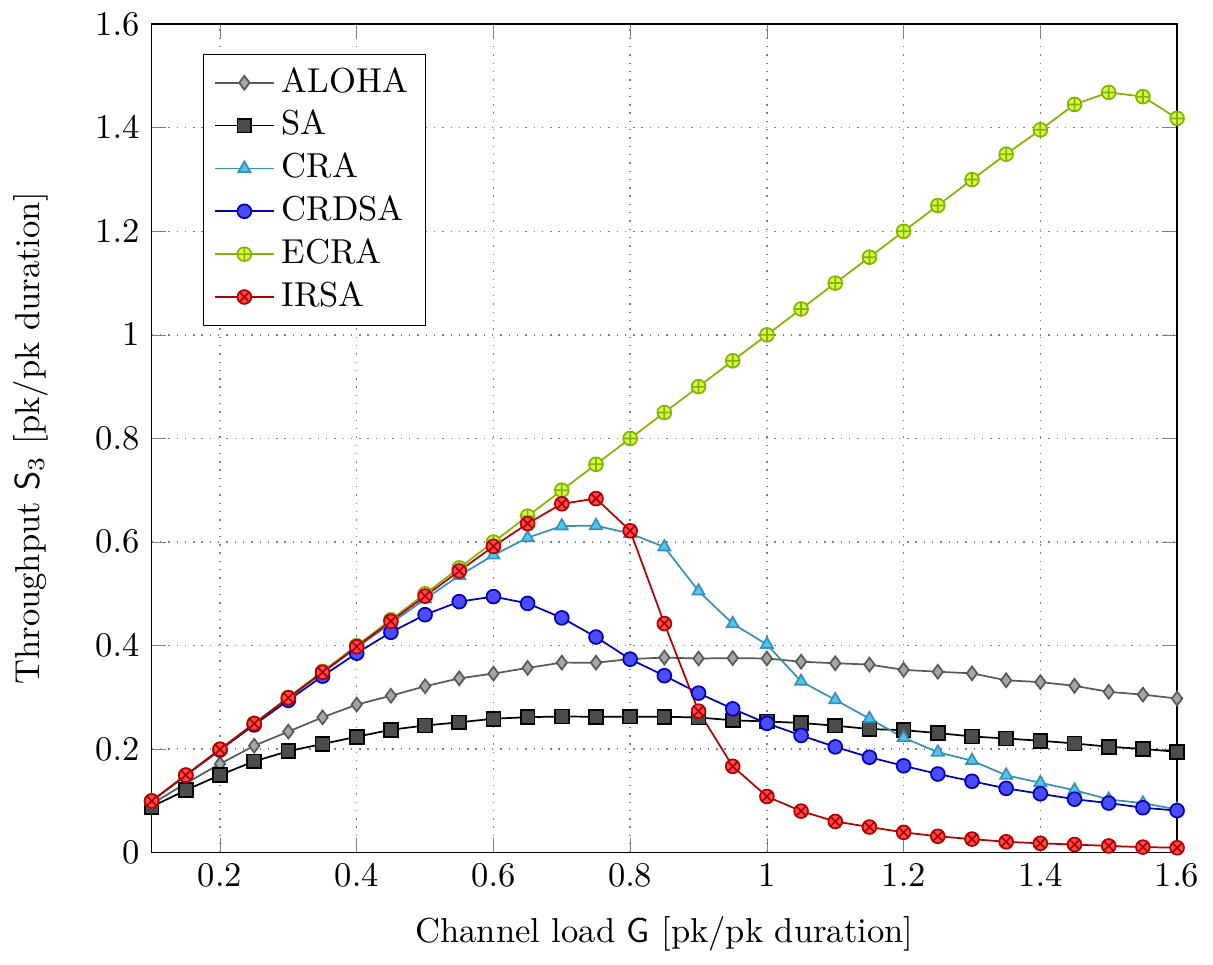}
\caption{Layer $3$ throughput comparison with $\LDist = \frac{1}{2}\delta(x-1) + \frac{1}{2}\delta(x-2)$.}
\label{fig:T_L3_equal}
\end{figure}

We present in Figure~\ref{fig:PER_L3_equal} the layer~$3$ \ac{PLR} results for the considered schemes under the same probability measure of Figure~\ref{fig:T_L3_equal}. As expected, \ac{ECRA} is the scheme that shows the best performance providing a \ac{PLR} lower than $10^{-3}$ for channel loads up to $\load \leq 1.35$ and exhibiting \ac{PLR} below $10^{-4}$. Fixing a fictitious target \ac{PLR} at $10^{-2}$, \ac{IRSA} can guarantee channel load values up to $\load \leq 0.55$, \ac{CRA} up to $\load \leq 0.4$ and \ac{CRDSA} up to $\load \leq 0.2$. None of ALOHA or \ac{SA} is able to reach such \ac{PLR} values for the considered channel load values.

\begin{figure}
\centering
\includegraphics[width=0.8\textwidth]{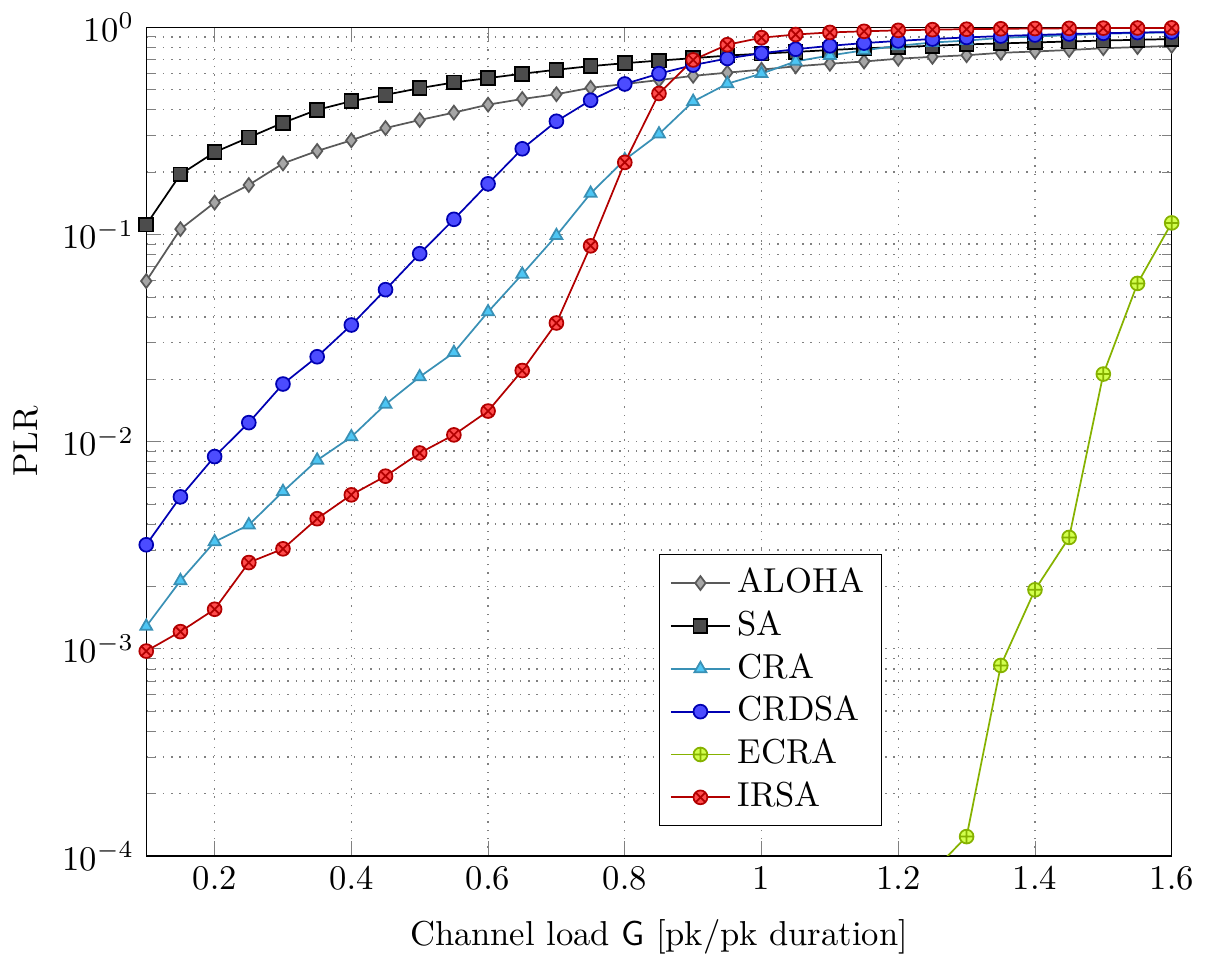}
\caption{Layer $3$ \ac{PLR} comparison with $\LDist = \frac{1}{2}\delta(x-1) + \frac{1}{2}\delta(x-2)$.}
\label{fig:PER_L3_equal}
\end{figure}

In Figure~\ref{fig:T_L3_exp} we focus on another layer~$3$ packet duration distribution $\LDist = e^{-x}$ and we show the layer~$3$ throughput results. The major difference w.r.t. the results considering the other probability measure can be found observing \ac{ECRA}. There, the layer~$3$ peak throughput of $1.35$ packet is reached which is a $8$\% degradation w.r.t. the results with the other probability measure.

\begin{figure}
\centering
\includegraphics[width=0.8\textwidth]{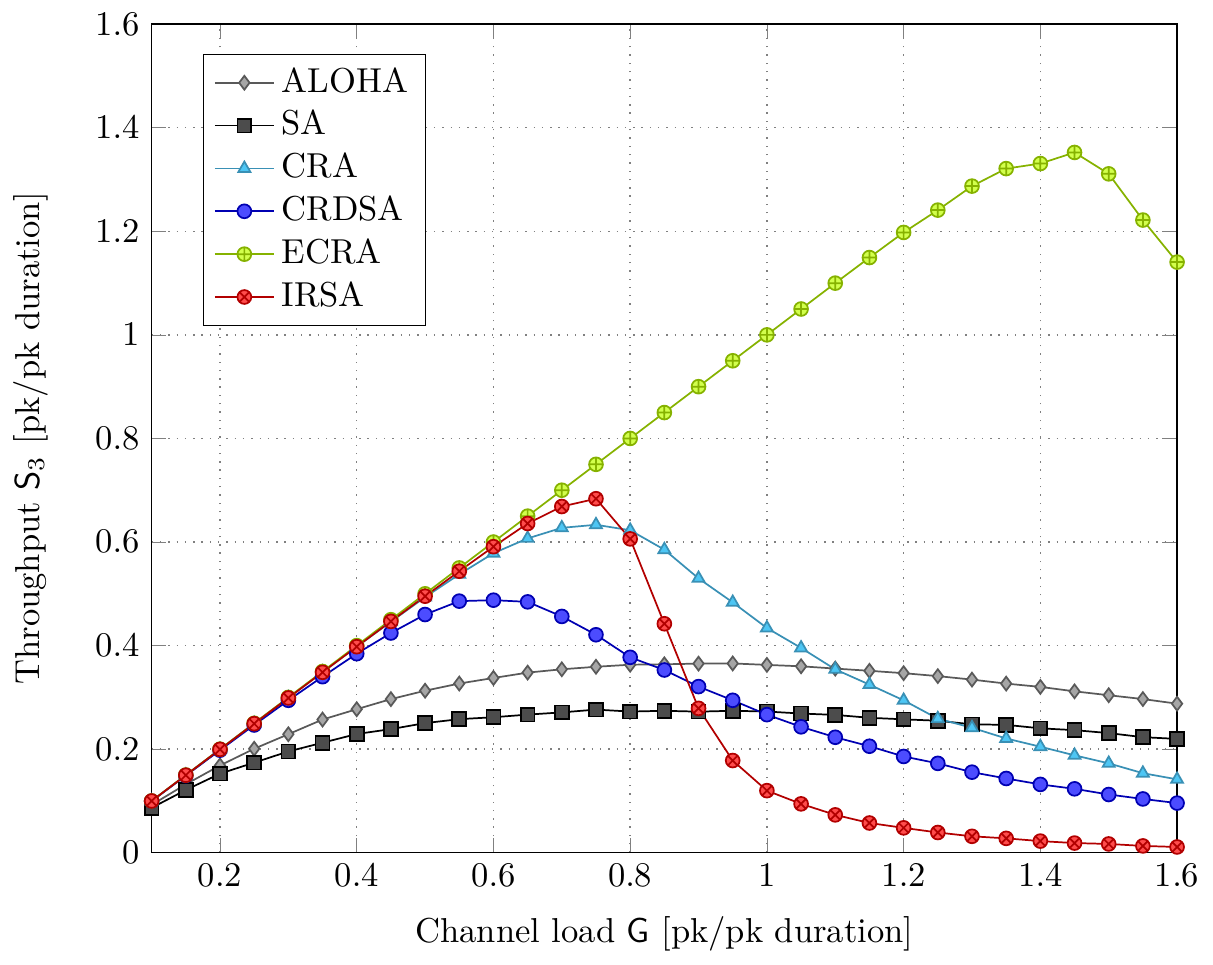}
\caption{Layer $3$ throughput simulations with $\LDist = e^{-x}$.}
\label{fig:T_L3_exp}
\end{figure}

In Figure~\ref{fig:PER_L3_exp} we focus on the layer~$3$ \ac{PLR}. For \ac{ECRA}, a visible degradation of the \ac{PLR} performance brings the channel load to $1.15$ for guaranteeing a \ac{PLR} lower than $10^{-3}$. Another interesting observation comes from the comparison between \ac{CRA} and \ac{IRSA}. The two \ac{PLR} curves intersect two times. At low channel load, \ac{CRA} slightly outperforms \ac{IRSA} (for $\load \leq 0.4$), while in the moderate channel load regime \ac{IRSA} is able to give better \ac{PLR} performance (up to $\load \leq 0.8$). At high load, instead, \ac{CRA} again performs better than \ac{IRSA}, although this is a region where none of the two protocols shall be operated.

\begin{figure}
\centering
\includegraphics[width=0.8\textwidth]{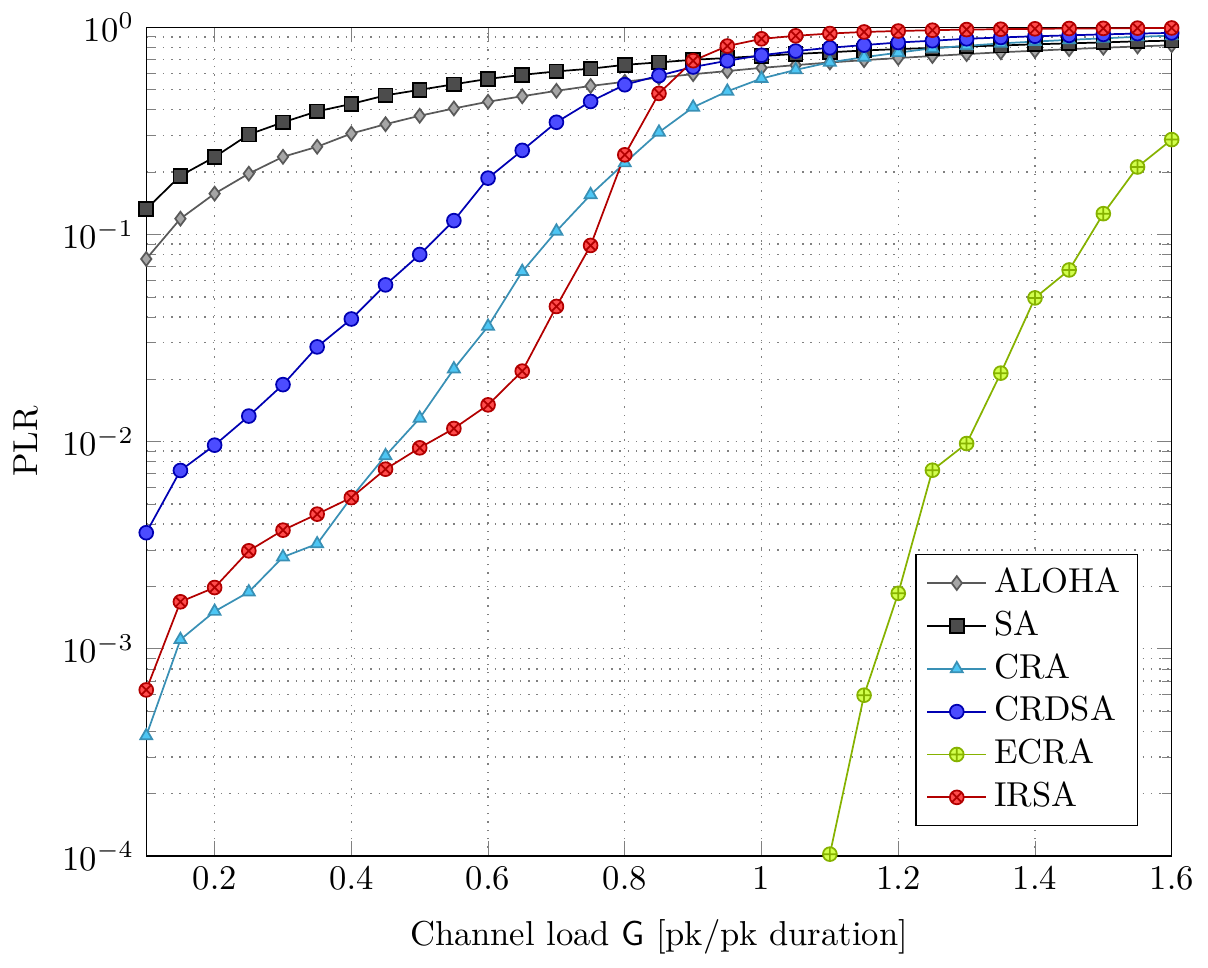}
\caption{Layer $3$ \ac{PLR} simulations with $\LDist = e^{-x}$.}
\label{fig:PER_L3_exp}
\end{figure}

The presented results show that asynchronous \ac{RA} schemes can be even more competitive against slot synchronous ones when considering layer~$3$ performance. At physical layer, nevertheless, more complex encoding and decoding algorithms are required. These stem from different block length sizes that depend on the transmitted packet and possibly require calling to more than one packet detector and/or decoder.

\subsection{Slot Synchronous Layer~$2$ and Layer~$3$ Comparison}

In this Section we present the layer~$2$ and layer~$3$ performance comparison for the slot synchronous schemes. Since the results are very similar for both probability measures of the packet duration distribution at layer~$3$, we are presenting only the one for the case ${\LDist = \frac{1}{2}\delta(x-1) + \frac{1}{2}\delta(x-2)}$. In Figure~\ref{fig:T_L3_L2_equal_Slotted} we show the layer~$3$ and $2$ throughput results. When considering layer~$3$ with respect to layer~$2$, the degradation for \ac{IRSA} is quite limited, while for both \ac{CRDSA} and \ac{SA} this is not the case. In \ac{IRSA} the peak throughput of layer~$3$ is $3$\% lower than the peak layer~$2$ throughput, while for \ac{CRDSA} it is $7$\% lower and for \ac{SA} even $29$\% lower. Furthermore, in both \ac{CRDSA} and \ac{SA} the peak throughputs at layer~$3$ and $2$ are achieved for different channel load values.
\begin{figure}
\centering
\includegraphics[width=0.8\textwidth]{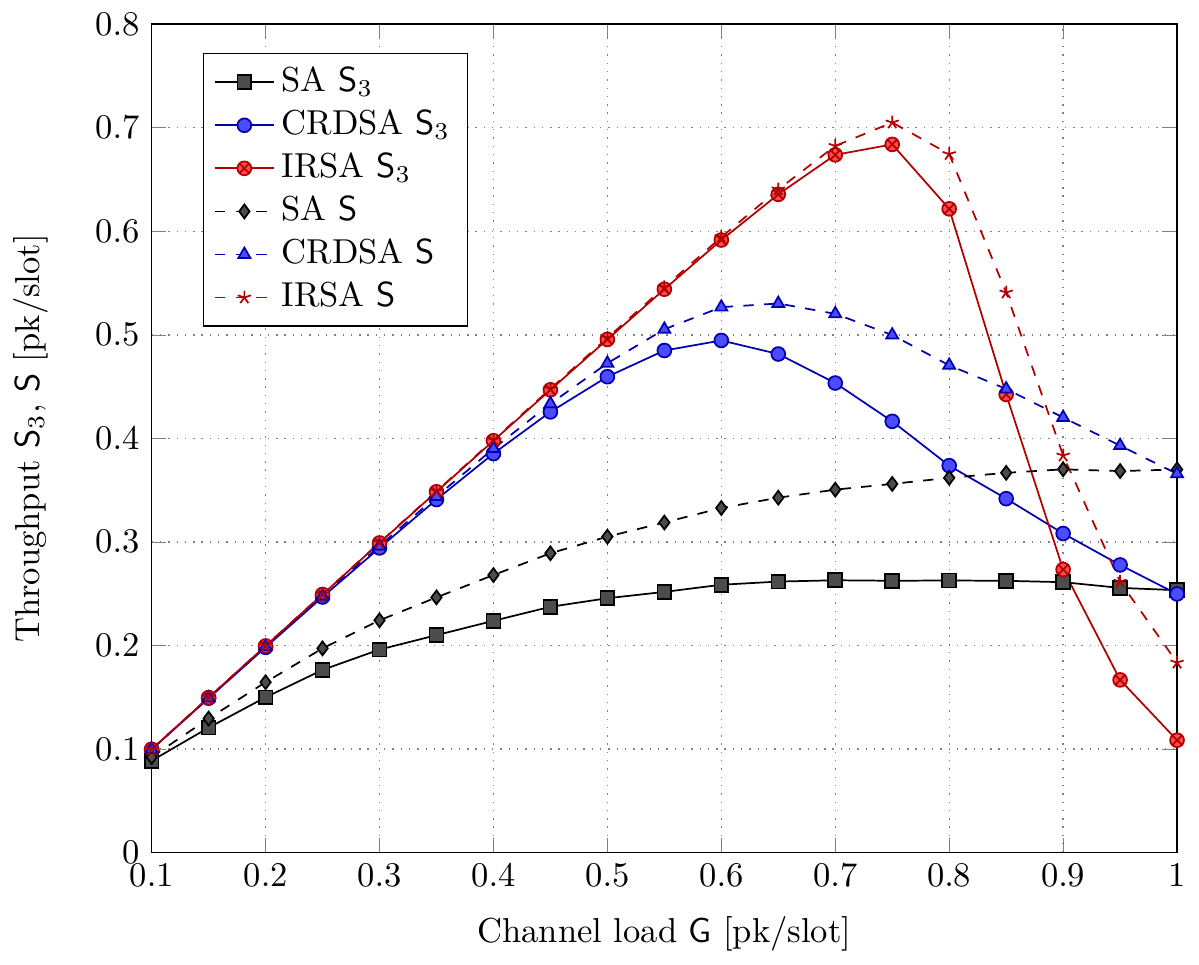}
\caption{Layer $2$ and $3$ throughput simulations with $\LDist = \frac{1}{2}\delta(x-1) + \frac{1}{2}\delta(x-2)$.}
\label{fig:T_L3_L2_equal_Slotted}
\end{figure}
The main reason of such difference between the schemes can be seen observing the \ac{PLR} performance in Figure~\ref{fig:PER_L3_L2_equal_Slotted}. The relative distance between the curves of layer~$3$ and $2$ is similar for all the schemes, nevertheless the absolute values are very different, ranging among three orders of magnitude depending on the considered scheme. This effect, coupled with the different steepness of the curves explains the throughput behaviour.

\begin{figure}
\centering
\includegraphics[width=0.8\textwidth]{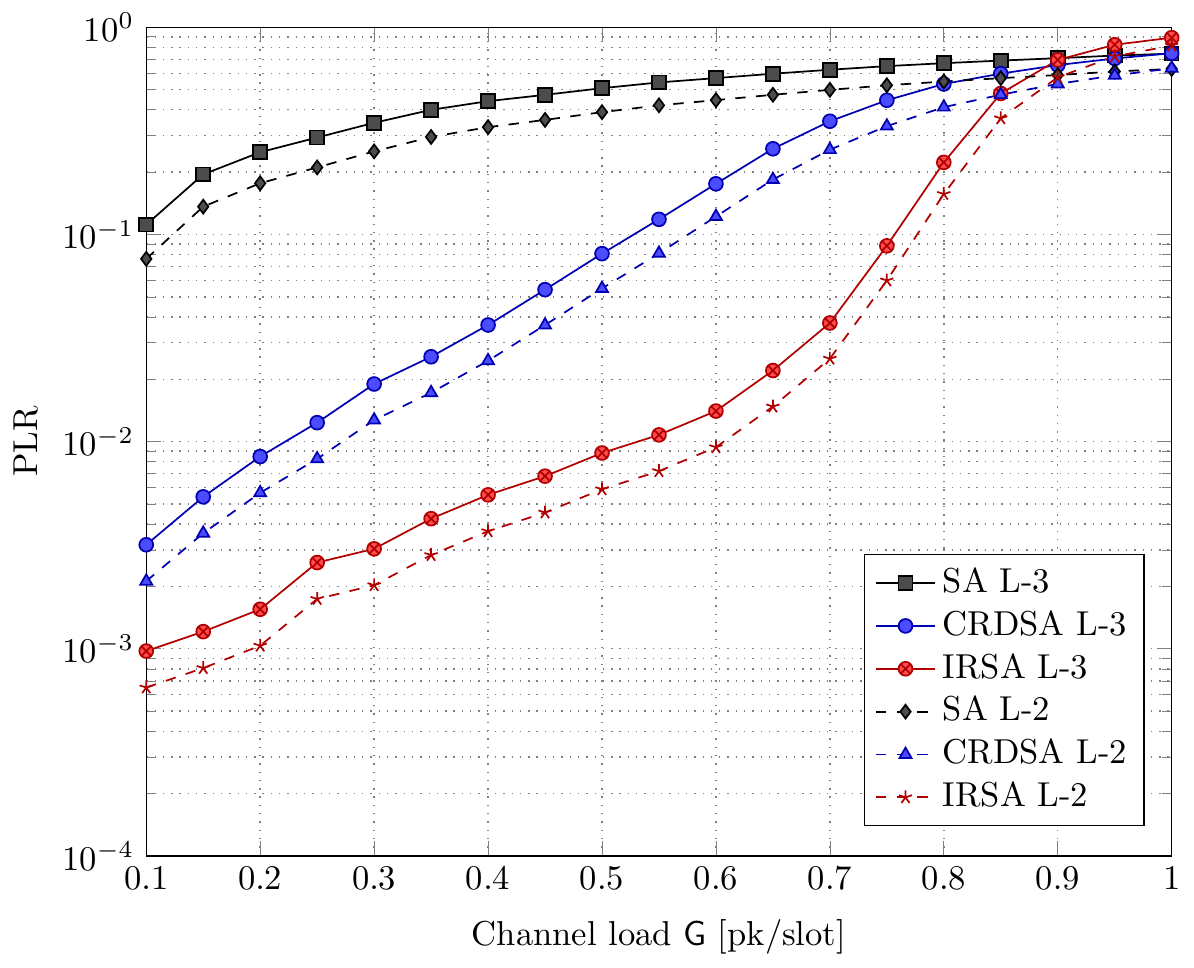}
\caption{Layer $2$ and $3$ \ac{PLR} simulations with $\LDist = \frac{1}{2}\delta(x-1) + \frac{1}{2}\delta(x-2)$.}
\label{fig:PER_L3_L2_equal_Slotted}
\end{figure}

If we consider advanced \ac{RA} schemes, as \ac{IRSA} or \ac{CRDSA}, the performance degradation when looking at the layer~$3$ throughput is quite limited and below $10$\% in terms of peak throughput. It has to be noted that no padding or fragmentation overhead is here considered, which will further degrade the layer~$3$ performance of slot synchronous schemes.

\subsection{Layer~$3$ Slot Synchronous Bounds}

In this Section we present the layer~$3$ comparison between the exact layer~$3$ performance and the one using the successful probability $\psuccLTAppr$, which turns into a layer~$3$ throughput lower bound and \ac{PLR} upper bound. As for the previous Section, since the results are very similar for both probability measures of the packet duration distribution at layer~$3$, we are presenting only the one for the case ${\LDist = \frac{1}{2}\delta(x-1) + \frac{1}{2}\delta(x-2)}$. In Figure~\ref{fig:T_L3_L2_equal_Slotted_Appr} the throughput results are presented. The dashed curves represent the lower bounds to the throughput. As we can observe, it is particularly tight for all channel load values of interest for all schemes.

\begin{figure}
\centering
\includegraphics[width=0.8\textwidth]{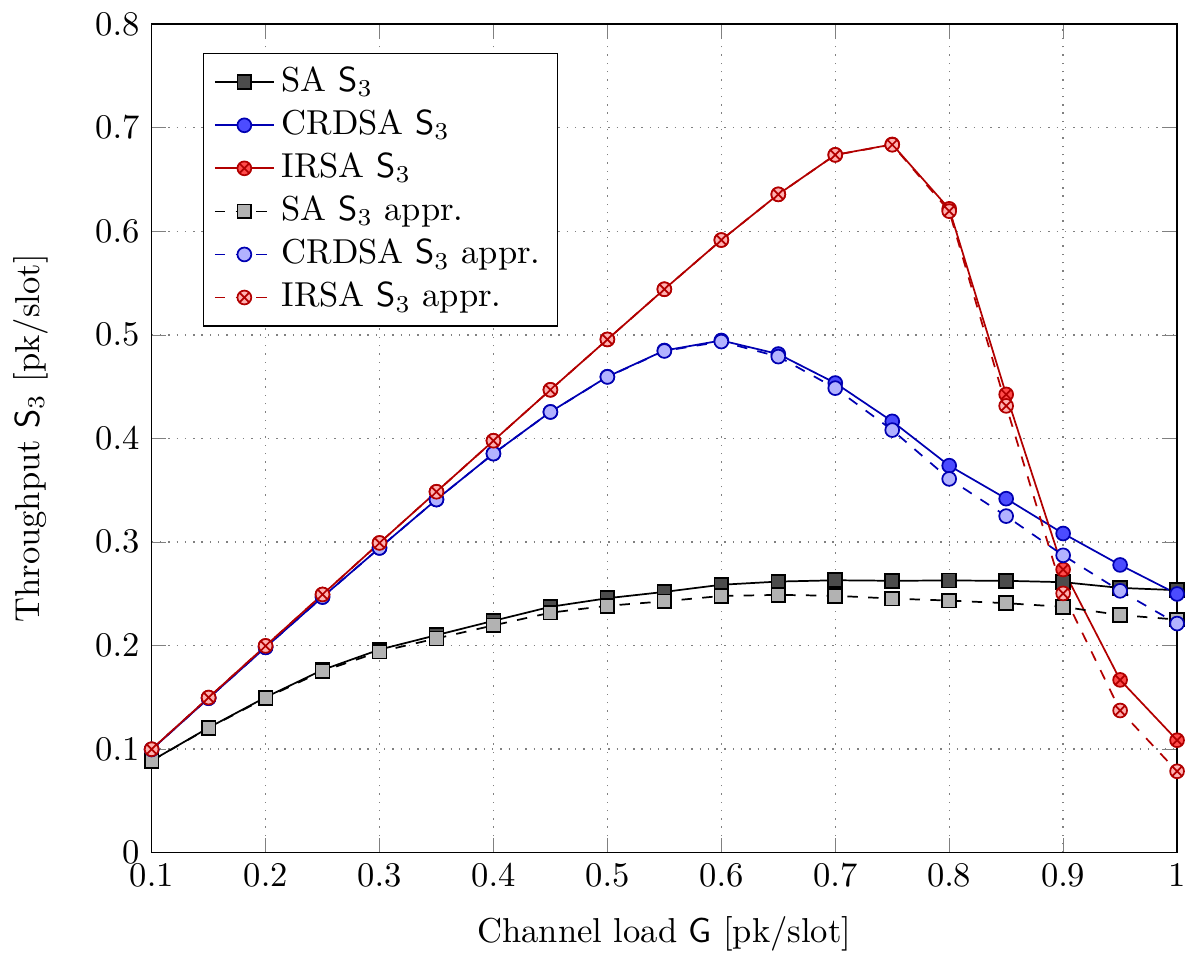}
\caption{Layer $3$ throughput comparison between exact and approximated performance considering $\LDist = \frac{1}{2}\delta(x-1) + \frac{1}{2}\delta(x-2)$.}
\label{fig:T_L3_L2_equal_Slotted_Appr}
\end{figure}

In Figure~\ref{fig:PER_L3_L2_equal_Slotted_Appr} we present the \ac{PLR} results. Here the dashed lines represent the upper bounds to the \ac{PLR}. As for the throughput, also here the upper bound is very tight for all schemes, proving the suitability of the presented approach.

\begin{figure}
\centering
\includegraphics[width=0.8\textwidth]{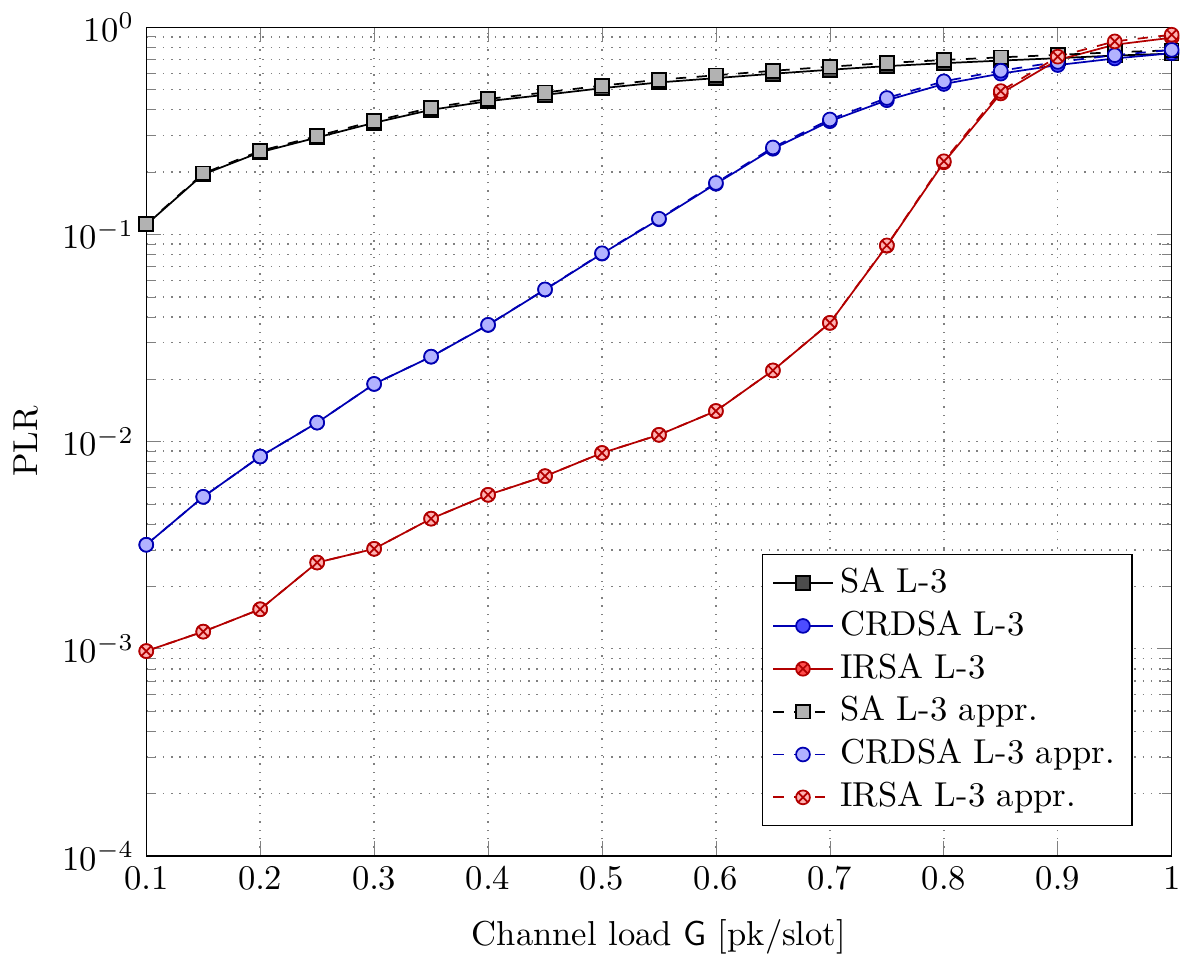}
\caption{Layer $3$ \ac{PLR} comparison between exact and approximated performance considering ${\LDist = \frac{1}{2}\delta(x-1) + \frac{1}{2}\delta(x-2)}$.}
\label{fig:PER_L3_L2_equal_Slotted_Appr}
\end{figure}

\section{Conclusions}

The aim of this Chapter was to investigate the behaviour of the recently presented asynchronous and slot synchronous \ac{RA} schemes as \ac{CRA}, \ac{ECRA}, \ac{CRDSA} and \ac{IRSA} in terms of layer~$3$ performance. We derived an analytical expression of both layer~$3$ throughput and \ac{PLR} for slot synchronous schemes so to avoid unnecessary layer~$3$ simulations. A layer~$3$ successful probability lower bound that could be derived knowing only the average layer~$3$ packet duration was also provided, so to address also all cases in which the packet duration distribution was not known. The expression required the knowledge of layer~$2$ successful probability and the probability measure defining the layer~$3$ packet durations. Two probability measures defining the distribution of layer~$3$ packet duration were considered, \ie $\LDist = \frac{1}{2}\delta(x-1)+ \frac{1}{2}\delta(x-2)$ and $\LDist = e^{-x}$. In all cases, \ac{ECRA} with \ac{MRC} was able to largely outperform all other schemes, under the specific scenario selected. Nonetheless it was also the scheme that was more sensible to the distribution of layer~$3$~packet durations. We presented also a comparison between layer~$3$ and layer~$2$ performance for slot synchronous schemes, and we showed that limited performance loss were achieved by schemes like \ac{CRDSA} and \ac{IRSA} (below $10$\% in terms of peak throughput loss). Finally, we compared the derived bound with the exact performance for the slot synchronous schemes, and we showed that very minor differences could be found. Moreover, the bounds were particularly tight for all channel load values of interest.  
\chapter[IRSA over the Rayleigh Block Fading Channel]{Irregular Repetition Slotted ALOHA over the Rayleigh Block Fading Channel}
\label{chapter6}
\thispagestyle{empty}
\ifpdf
    \graphicspath{{chapter6/figures/}}
\fi
\epigraph{The littlest thing can cause a ripple effect that changes your life}{Ted Mosby}

Until now we have shown the performance of some advanced asynchronous and slot synchronous \ac{RA} schemes without taking into consideration the possibility of fading. The recently proposed slot synchronous \ac{IRSA} \cite{Liva2011} scheme is subject to optimisation of the user degree distribution for the collision channel model. In this Chapter we aim at extending the analysis for the Rayleigh block fading channel. The result is two-folded, on the one hand we extend the optimisation procedure for fading channels, showing that for a certain average \ac{SNR} and decoding threshold, throughput exceeding $1$ $\mathrm{[pk/slot]}$ can be achieved. On the other hand, the benefit of such optimisation compared to a mismatched one that uses the collision channel model in spite of the correct Rayleigh block fading one.


\section{Introduction}\label{sec:Intro}



The collision channel model is a rather simple one, which assumes that 1) noise can be neglected, such that a transmission can be decoded from a singleton slot by default, and 2) no transmission can be decoded from a collision slot. This model has a limited practical applicability and does not describe adequately the wireless transmission scenarios where the impact of fading and noise cannot be neglected. In particular, fading may incur power variations among signals observed in collisions slot, allowing the \emph{capture effect} to occur, when sufficiently strong signals may be decoded.
In the context of slotted ALOHA, numerous works assessed the performance of the scheme for different capture effect models \cite{Roberts1975,N1984,Wieselthier1989,Zorzi1994,Z1997,ZZ2012}.
One of the standardly used models is the threshold-based one, in which a packet is captured, \ie decoded, if its \ac{SINR} is higher then a predefined threshold, c.f. \cite{Zorzi1994,Z1997,NEW2007,ZZ2012}.

A brief treatment of the capture effect in \ac{IRSA} framework was made in \cite{Liva2011}, pointing out the implications related to the asymptotic analysis. In \cite{SMP2014}, the method for the computation of capture probabilities for the threshold-based model in single-user detection systems with Rayleigh fading was presented and instantiated for the frameless ALOHA framework \cite{Stefanovic2012}.

In this Chapter, we extend the treatment of the threshold-based capture effect for \ac{IRSA} framework. First, we derive the exact expressions of capture probabilities for the threshold-based model and Rayleigh block-fading channel. Next we formulate the asymptotic performance analysis. We then optimise the scheme, in terms of deriving the optimal repetition strategies that maximise throughput given a target \ac{PLR}. Finally, the obtained distributions are investigated in the finite frame length scenario via simulations. We show that \ac{IRSA} exhibits a remarkable throughput performance that is well over $1$ $\mathrm{[pk/slot]}$, for target \ac{PLR}, SNR and threshold values that are valid in practical scenarios. This is demonstrated both for asymptotic and finite frame length cases, showing also that the finite-length performance indeed tends to the asymptotic one as the frame length increases.



\section{System Model}\label{sec:Preliminaries}

\subsection{Access Protocol}
\begin{sloppypar}
For the sake of simplicity, we focus on a single batch arrival of $\numUs$ users having a single packet (or \emph{burst}) each, and contending for the access to the common receiver. The link time is organised in a \ac{MAC} frame of duration $\fraLen$, divided into $\numSlot$ slots of equal duration ${\slLen =\fraLen / \numSlot}$, indexed by $j \in \{1, 2, \dots,  \numSlot\}$. The transmission time of each packet equals the slot duration. The system load $\load$ is defined as
\begin{align}
\load = \frac{\numUs}{\numSlot}\,\,\,\, [\mathrm{pk}/\mathrm{slot}] \, .
\end{align}
\end{sloppypar}

According to the \ac{IRSA} protocol, each user selects a repetition degree $\dg$ by sampling a \ac{PMF} $\BNdegDist_{\dg=2}^{\dmax}$ and transmits $\dg$ identical replicas of its burst in $\dg$ randomly chosen slots of the frame. It is assumed that the header of each burst replica carries information about the locations (\ie slot indexes) of all $\dg$ replicas. The \ac{PMF} $\BNdegDist$ is the same for all users and is sampled independently by different users, in an uncoordinated fashion.
The average burst repetition degree is $\davg =\sum_{\dg=2}^{\dmax} \dg \, \Ld_\dg$, and its inverse
\begin{align}\label{eq:rate}
\rateIrsa = \frac{1}{\davg}
\end{align}
is called the \emph{rate} of the IRSA scheme. Each user is then unaware of the repetition degree employed by the other users contending for the access.
The number of burst replicas colliding in slot $j$ is denoted by $c_j \in \{1,2,\dots,\numUs\}$. Burst replicas colliding in slot $j$ are indexed by $i \in \{1,2,\dots,c_j\}$.

\subsection{Received Power and Fading Models}

We consider a Rayleigh block fading channel model, \ie fading is Rayleigh distributed, constant and frequency flat in each block, while it is \ac{i.i.d.} on different blocks. Independent fading between different burst replicas is also assumed. In this way, the power of a burst replica $i \in \{1,2,\dots,c_j\}$ received in slot $j$, denoted as $\usPw_{ij}$, is modeled as a \ac{r.v.} with negative exponential distribution
\[
\powerdist\left(\power\right)=
\begin{cases}
\frac{1}{\PwAvg} \exp \left[-\frac{\power}{\PwAvg} \right] , & \power\geq 0 \\
0, & \text{otherwise}
\end{cases}
\]
where $\PwAvg$ is the average received power. This is assumed to be the same for all burst replicas received in the \ac{MAC} frame by using, \eg a long-term power control. The \acp{r.v.} $\usPw_{ij}$ are \ac{i.i.d.}  for all pairs $(i,j)$. If we denote by $\noisePw$ the noise power, the \ac{SNR} \ac{r.v.} $\snrrv_{ij} = \usPw_{ij}/\noisePw$ is also exponentially distributed as
\[
\snrdist\left(\snr\right)=
\begin{cases}
\frac{1}{\snravg}\exp \left[ -\frac{\snr}{\snravg} \right], & \snr \geq 0 \\
0, & \text{otherwise}
\end{cases}
\]
where the average \ac{SNR} is given by
\[
\snravg=\frac{\PwAvg}{\noisePw}.
\]

\subsection{Graph Representation}\label{subsec:graph}

\begin{figure}[tb]
\begin{center}
\includegraphics[width=0.45\columnwidth,draft=false]{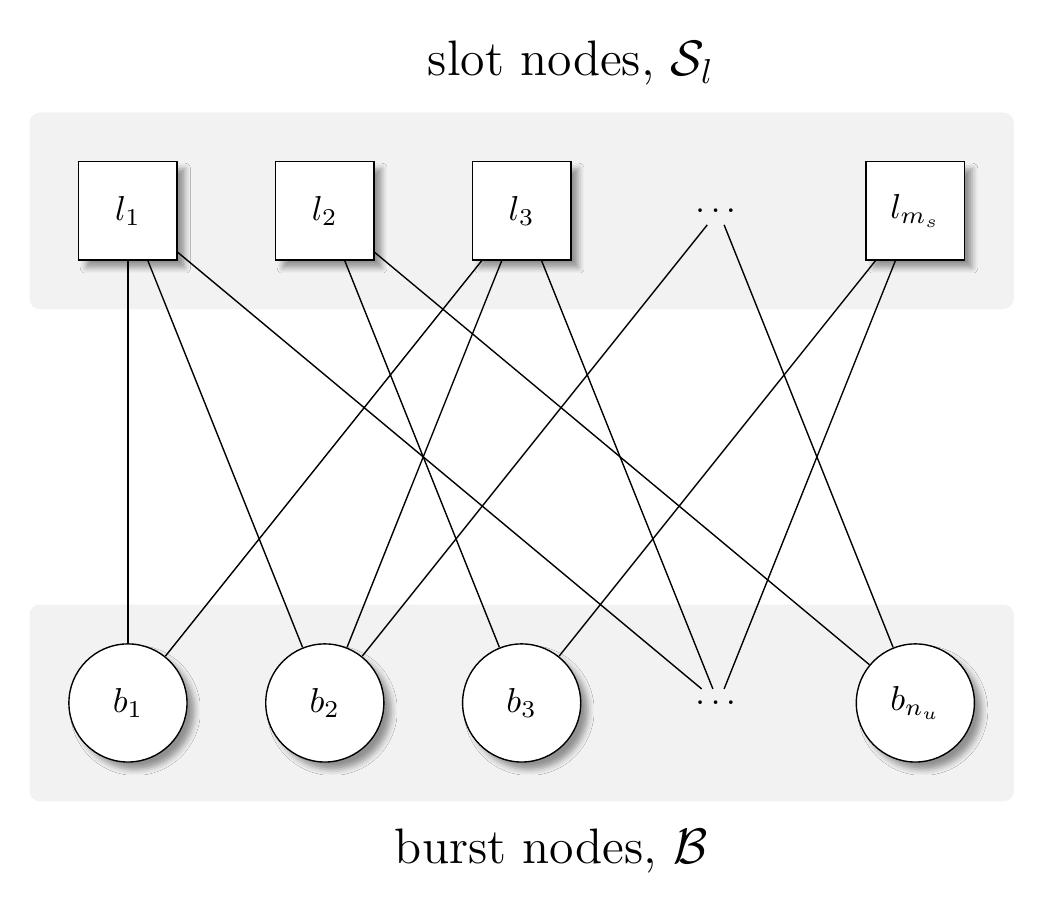}
\end{center}
\caption{Graph representation of MAC frame.}\label{fig:graph}
\end{figure}

\begin{sloppypar}
In order to analyse the \ac{SIC} process, we introduce the graph representation of a \ac{MAC} frame\cite{Liva2011}. As depicted in Figure~\ref{fig:graph}, a \ac{MAC} frame is represented as a bipartite graph ${\mathcal{G}=(\calB,\calS,E)}$ consisting of a set $\calB$ of $\numUs$ burst nodes (or user nodes), one for each user, a set $\calS$ of $\numSlot$ slot nodes, one per slot, and a set $E$ of edges, one per transmitted burst replica. A burst node $b_k\in \calB$ is connected to a slot node $\slot_j\in \calS$ if and only if user $k$ has a burst replica sent in the $j$-th slot of the frame. The \emph{node degree} represents the number of edges emanating from a node.
\end{sloppypar}

For the upcoming analysis it is convenient to resort to the concept of \emph{node-} and \emph{edge-perspective degree distributions}. The burst node degree distribution from a node perspective is identified by the above-defined \ac{PMF} $\BNdegDist_{\dg=2}^{\dmax}$. Similarly, the slot node degree distribution from a node perspective is defined as $\SNdegDist_{\slotind=0}^{\numUs}$, where $\Rd_\slotind$ is the probability that a slot node has $\slotind$ connections (\ie that $\slotind$ burst replicas have been received in the corresponding slot). The probability $\Rd_\slotind$ may be easily calculated by observing that $(\load / \rateIrsa) / \numUs$ is the probability that the generic user transmits a burst replica in a specific slot. Since users behave independently of each other, we obtain
\begin{align}
\Rd_\slotind={\numUs \choose \slotind}
\left(\frac{\load/\rateIrsa}{\numUs}\right)^{\slotind} \left(1-\frac{\load/\rateIrsa}{\numUs}\right)^{\numUs - \slotind}\, .
\end{align}

The polynomial representations for both node-perspective degree distributions are given by
\begin{align}
\Ld(x) = \sum_{\dg=2}^{\dmax} \Ld_\dg \, x^\dg \quad \mathrm{and} \quad \Rd(x) = \sum_{\slotind=0}^{\numUs} \Rd_\slotind \, x^\slotind
\end{align}
where, for $\numUs \rightarrow \infty$ and constant $ \load / \rateIrsa$, ${\Rd(x) = \exp \left\{ -\frac{\load}{\rateIrsa}(1-x) \right\}}$.
Degree distributions can also be defined from an edge-perspective. Adopting a notation similar to the one used for the node-perspective distributions, we define the edge-perspective burst node degree distribution as the \ac{PMF} $\EBdegDist_{\dg=2}^{\dmax}$, where $\ld_\dg$ is the probability that a given edge is connected to a burst node of degree $\dg$. Likewise, we define the edge-perspective slot node degree distribution as the \ac{PMF} $\ESdegDist_{\slotind=0}^{\numUs}$, where $\rd_\slotind$ is the probability that an edge is connected to a slot node of degree $\slotind$. From the definitions we have $\ld_\dg = \dg \, \Ld_\dg / \left(\sum_t t \, \Ld_t \right)$ and $\rd_\slotind = \slotind \, \Rd_\slotind / \left(\sum_t t \, \Rd_t \right)$;
it can be shown that, for $\numUs \rightarrow \infty$ and constant $ \load / \rateIrsa$,  $ \rd_\slotind  = \exp\{-\load / \rateIrsa\} (\load / \rateIrsa)^{\slotind-1} / (\slotind-1)! $.
The corresponding polynomial representation are
$\ld(x) = \sum_{\dg=2}^{\dmax} \ld_\dg \, x^{\dg-1}$ and $\rd(x) = \sum_{\slotind=0}^{\numUs} \rd_\slotind \, x^{\slotind-1}$.
Note that $\ld(x)=\Ld'(x)/\Ld'(1)$ and $\rd(x)=\Rd'(x)/\Rd'(1)$.\footnote{We recall that, notation $f'(x)$ denotes the derivative of $f(x)$.}

\subsection{Receiver Operation}
\label{subsec:receiver}

In our model, the receiver is always able to detect burst replicas received in a slot, \ie to discriminate between an empty slot where only noise samples are present and a slot in which at least one burst replica has been received. Moreover, a threshold-based capture model for the receiver is assumed, by which the generic burst replica $i$ is successfully decoded (\ie captured) in slot $j$ if the \ac{SINR} exceeds a certain threshold $\snrthr$, namely,
\begin{align}\label{eq:capture_model}
\Pr \{ \text{burst replica $i$ decoded} \} =
\begin{cases}
1, & \frac{\usPw_{ij}}{\noisePw + \intPw_{ij}} \geq \snrthr \\
0, & \text{otherwise}.
\end{cases}
\end{align}
The quantity $\intPw_{ij}$ in \eqref{eq:capture_model} denotes the power of the interference impairing replica $i$ in slot $j$. In our system model the threshold $\snrthr$ fulfills $\snrthr \geq 1$, which corresponds to a conventional narrowband single-antenna system. As we are considering a \ac{SIC}-based receiver, under the assumption of perfect \ac{IC} the quantity $\intPw_{ij}$ is equal to sum of the powers of those bursts that have not yet been cancelled from slot $j$ in previous iterations (apart form burst $i$).
Specifically,
\begin{align}\label{eq:interference_expression}
\intPw_{ij} = \sum_{u \in \mathcal{R}_j \setminus \{ i \} } \usPw_{uj} 
\end{align}
where $\usPw_{uj}$ is the power of burst replica $u$ not yet cancelled in slot $j$ and $\mathcal{R}_j$ denotes the set of remaining burst replicas in slot $j$.
Exploiting \eqref{eq:interference_expression}, after simple manipulation we obtain
\begin{align}
\frac{\usPw_{ij}}{\noisePw + \intPw_{ij}}
= \frac{\snrrv_{ij}}{ 1 + \sum_{u  \in \mathcal{R}_j \setminus \{i \} } \snrrv_{uj}} \, .
\end{align}
Hence, in the adopted threshold-based capture model the condition $\frac{\usPw_{ij}}{\noisePw + \intPw_{ij}} \geq \snrthr$ in \eqref{eq:capture_model} may be recast as
\begin{align}\label{eq:thr}
\frac{\snrrv_{ij}}{ 1 + \sum_{u  \in \mathcal{R}_j \setminus \{ i \} } \snrrv_{uj}} \geq \snrthr \, .
\end{align}

When processing the signal received in some slot $j$, if burst replica $i$ is successfully decoded due to fulfillment of \eqref{eq:thr}, then 1) its contribution of interference is cancelled from slot $j$, and 2) the contributions of interference of all replicas of the same burst are removed from the corresponding slots.\footnote{We assume that the receiver is able to estimate the channel coefficients required for the removal of the replicas.} Hereafter, we refer to the former part of the \ac{IC} procedure as \emph{intra-slot} \ac{IC} and to the latter as \emph{inter-slot} \ac{IC}. Unlike \ac{SIC} in \ac{IRSA} protocols over a collision channel, which only rely on inter-slot \ac{IC}, \ac{SIC} over a block fading channel with capture takes advantage of intra-slot \ac{IC} to potentially decode burst replicas interfering with each other in the same slot. In this respect, it effectively enables \emph{multi-user decoding} in the slot.

Upon reception of a new \ac{MAC} frame, slots are processed sequentially by the receiver. By definition, one \emph{\ac{SIC} iteration} consists of the sequential processing of all $\numSlot$ slots. In each slot, intra-slot \ac{IC} is performed repeatedly, until no burst replicas exist for which \eqref{eq:thr} is fulfilled. When all burst replicas in slot $j$ have been successfully decoded, or when intra-slot \ac{IC} in slot $j$ stops prematurely, inter-slot \ac{IC} is performed for all burst replicas successfully decoded in slot $j$ and the receiver proceeds to process slot $j+1$. When all $\numSlot$ slots in the \ac{MAC} frame have been processed, there are three possible cases: 1) a success is declared if all user packets have been successfully received; 2) a new iteration is started if at least one user packet has been recovered during the last iteration, its replicas removed via inter-slot IC, and there still are slots with interfering burst replicas; 3) a failure is declared if no user packets have been recovered during the last iteration and there are still slots with interfering burst replicas, or if a maximum number of \ac{SIC} iterations has been reached and there are still slots with interfering burst replicas.

Exploiting the graphical representation reviewed in Section~\ref{subsec:graph}, the \ac{SIC} procedure performed at the receiver may be described as a successive removal of graph edges. Whenever a burst replica is successfully decoded in a slot, the corresponding edge is removed from the bipartite graph as well as all edges connected to the same burst node, due to inter-slot \ac{IC}. A success in decoding the \ac{MAC} frame occurs when all edges are removed from the bipartite graph.
We should remark two important features pertaining to the receiver operation, when casted into the graph terms. The first one is that, due to the capture effect, an edge may be removed from the graph when it is connected to a slot node with residual degree larger than one. The second one is that an edge connected to a slot node with residual degree one may not be removed due to poor \ac{SNR}, when  \eqref{eq:thr}, with $\mathcal{R}_{j} \setminus \{i\} = \emptyset$, is not fulfilled. 

\section{Decoding Probabilities}\label{sec:decoding_probs}

Consider the generic slot node $j$ at some point during the decoding of the \ac{MAC} frame and assume it has degree $\sDegJ$ under the current graph state. This means that $\sDegJ$ could be the original slot node degree $\slotind_j$ or the residual degree after some inter-slot and intra-slot \ac{IC} processing. Note that, as we assume perfect \ac{IC}, the two cases $\sDegJ=c_j$ and $\sDegJ<c_j$ are indistinguishable.

Among the $\sDegJ$ burst replicas not yet decoded in slot $j$, we randomly choose one and call it the \emph{reference burst replica}. Moreover, we denote by $D(\sDegJ)$ the probability that the reference burst replica is decoded starting from the current slot setting and only running intra-slot \ac{IC} within the slot. As we are considering system with $ \snrthr \geq 1$, the threshold based criterion \eqref{eq:thr} can be satisfied only for one single burst replica at a time. Therefore there may potentially be $\sDegJ$ decoding steps (and $\sDegJ-1$ intra-slot \ac{IC} steps), in order to decode the reference burst replica. Letting $D (\sDegJ,\decStep)$ be the probability that the reference burst replica is successfully decoded in step $\decStep$ and not in any step prior to step $\decStep$, we may write
\begin{align}
D( \sDegJ ) = \sum_{\decStep=1}^{\sDegJ} D (\sDegJ,\decStep) \, . \label{eq:D(d,r)}
\end{align}

Now, with a slight abuse of the notation, label the $r$ burst replicas in the slot from $1$ to $r$, arranged such that: (i) the first $t-1$ are arranged by their \acp{SNR} in the descending order (\ie $\snrrv_{1} \geq \snrrv_{2} \geq \ldots \snrrv_{t-1} ) $, (ii) the rest have \ac{SNR} lower than $\snrrv_{t-1}$ but do not feature any particular \ac{SNR} arrangement among them, (iii) the reference burst is labeled by $t$, \ie its \ac{SNR} by $\snrrv_{t}$, and (iv) the remaining $ r - t$ bursts are labeled arbitrarily.
The probability of having at least $t$ successful burst decodings through successive intra-slot \ac{IC} for such an arrangement is
\begin{align}
& \Pr \left\{\frac{\snrrv_{1}}{1 + \sum_{i=2}^{r} \snrrv_{i} } \geq \snrthr, \dots, \frac{\snrrv_{t}}{1 + \sum_{i=t+1}^{r} \snrrv_{i} } \geq \snrthr \right\} \\
&= \frac{1}{ \snravg^r} \int_{0}^{\infty} \mathrm{d} b_{r} \cdots \int_{0}^{\infty}  \mathrm{d} b_{t + 1} \notag \\
& \times \int_{ \snrthr ( 1 + \sum_{i = t+1}^{r} b_i ) }^{\infty} \mathrm{d} b_{t} \, \cdots  \nonumber  \int_{ \snrthr ( 1 + \sum_{i = 2}^{r} b_i )}^{\infty} \mathrm{d} b_{1} e^{- \frac{b_{r}}{\snravg}} \, \cdots \, e^{- \frac{b_{1}}{\snravg}} \\
&=  \frac{e^{-\frac{\snrthr}{\snravg} \sum_{i=0}^{t-1} (1+ \snrthr)^i}}{(1+ \snrthr )^{t ( r - \frac{t+1}{2})} }  = \frac{e^{- \frac{1}{\snravg} ( ( 1 + \snrthr )^t - 1 ) }}{(1+ \snrthr )^{t \left( r - \frac{t+1}{2} \right)}} .
\label{eq:ordered_prob}
\end{align}
Further, the number of arrangements in which the power of the reference replica is not among the first $\decStep -1$ largest is $\frac{(\sDegJ-1)!}{(\sDegJ-\decStep)!}$, which combined with \eqref{eq:ordered_prob} yields the probability that the reference burst is decoded (\ie captured) in the $\decStep$-th step
\begin{align}
\label{eq:D(r|t)}
D (\sDegJ,\decStep) = \frac{(\sDegJ-1)!}{(\sDegJ-\decStep)!}   \frac{e^{ - \frac{1}{\snravg} ((1 + \snrthr)^\decStep - 1) }}{(1+ \snrthr )^{\decStep \left( \sDegJ - \frac{\decStep+1}{2} \right)}} , \; 1 \leq \decStep \leq \sDegJ.
\end{align}

We conclude this Section by noting that  $D ( 1 ) =  e^{ - \frac{\snrthr}{\snravg} } \leq 1$, \ie a slot of degree 1 is decodable with probability that may be less than 1 and that depends on the ratio of the capture threshold and the expected SNR.
Again, this holds both for slots whose original degree was 1 and for slots whose degree was reduced to 1 via IC, as these two cases are indistinguishable when the IC is perfect.


\section{Density Evolution Analysis and Decoding Threshold Definition}\label{sec:DE}

\begin{figure}[tb]
\begin{center}
\includegraphics[width=\columnwidth,draft=false]{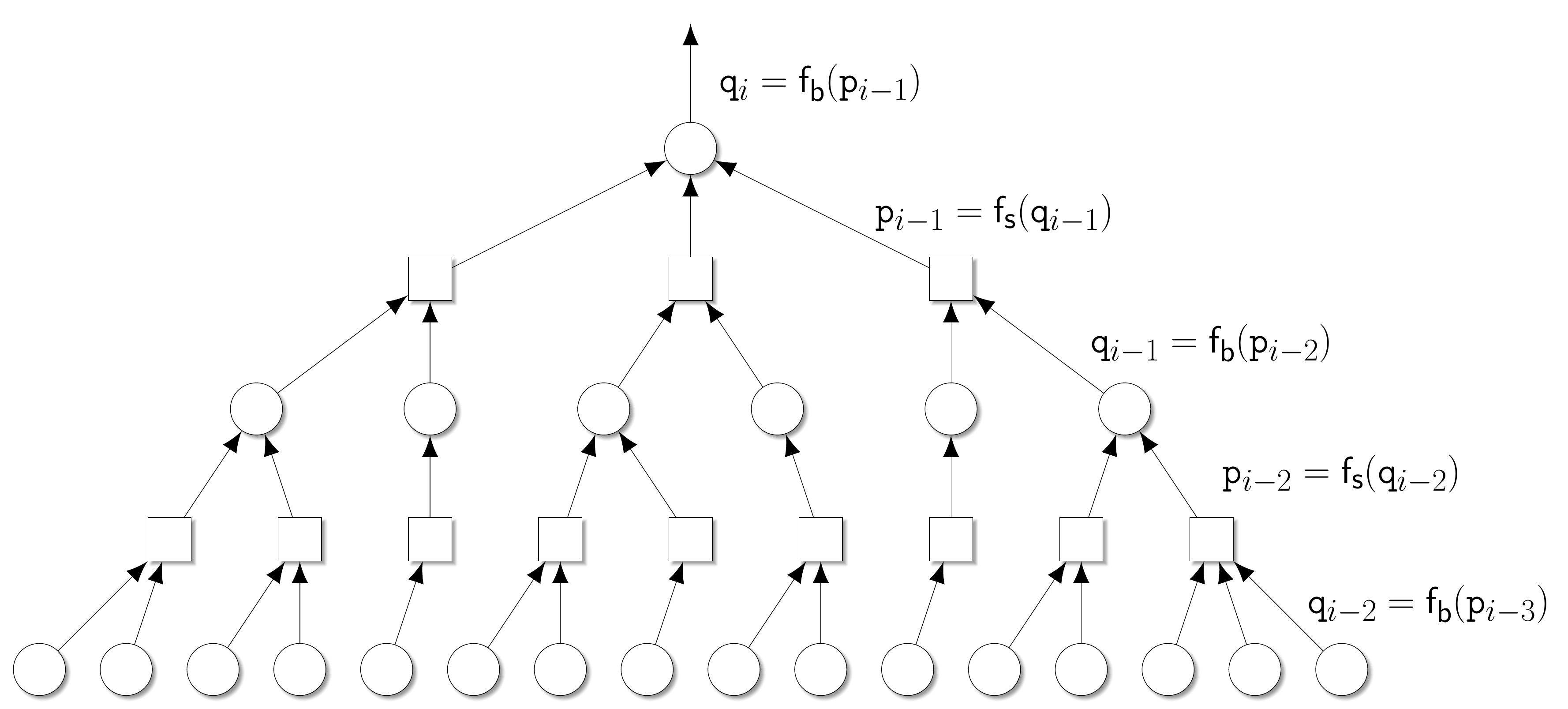}
\end{center}
\caption{Tree representation of the MAC frame.}\label{fig:and_or_tree}
\end{figure}

In this Section, we apply the technique of \ac{DE} in order to evaluate asymptotic performance of the proposed technique, \ie when $\numUs \rightarrow \infty$ and $ \numSlot \propto \numUs$.
For this purpose, we unfold the graph representation of the MAC frame (Figure~\ref{fig:graph}) into a tree, choosing a random burst node as its root, as depicted in Figure~\ref{fig:and_or_tree}.
The evaluation is performed in terms of probabilities that erasure messages are exchanged over the edges of the graph, where the erasure message denotes that the associated burst is not decoded.\footnote{For a more detailed introduction to the \ac{DE}, we refer the interested reader to \cite{RU2007}.}
The message exchanges are modeled as successive (\ie iterative) process, corresponding to the decoding algorithm described in Section~\ref{subsec:receiver} in the asymptotic case, when the lengths of the loops in the graph tends to infinity.
Specifically, the $i$-th iteration consists of the update of the probability $\q_i$ that an edge carries an erasure message from a burst node to a slot node, followed by the update of the probability $\p_i$ that an edge carries an erasure message from a slot node to a burst node. These probabilities are averaged over all edges in the graph. We proceed by outlining the details.

The probability that an edge carries an erasure message from burst nodes to slot nodes in the $i$-th iteration is
\begin{align}
\label{eq:final_q}
\q_i = \sum_{\dg=1}^{\dmax} \ld_\dg \, \q_i^{(\dg)}  = \sum_{\dg=1}^{\dmax} \ld_\dg \, \p_{i-1}^{\dg-1} =: \exitb(\p_{i-1})
\end{align}
where $\ld_\dg$ is the probability that an edge is connected to a burst node of degree $\dg$ (see Section~\ref{subsec:graph}) and $\q_i^{(\dg)}$ is the probability that an edge carries an erasure message given that it is connected to a burst node of degree $\dg$. In the second equality we used the fact that the an outgoing message from a burst node carries an erasure only if all incoming edges carry an erasure, \ie $ \q_i^{(\dg)} = \p_{i-1}^{\dg-1}$.

Similarly, the probability that an edge carries an erasure message from SNs to BNs in $i$-th iteration is
\begin{align}
\label{p_i_avg}
\p_i & = \sum_{\slotind=1}^{+\infty} \rd_\slotind \, \p_i^{(\slotind)}
\end{align}
where $\rd_\slotind$ is the probability that an edge is connected to a slot node of degree $\slotind$, and where $\p_i^{(\slotind)}$ is the probability that an edge carries an erasure message given that it is connected to a slot node of degree $\slotind$.
This probability may be expressed as
\begin{align}
\label{p_i_d}
\p_i^{(\slotind)} = 1 - \sum_{\sDegJ=1}^{\slotind} D( \sDegJ ) { \slotind - 1 \choose \sDegJ - 1} \q_{i}^{\sDegJ-1} ( 1 - \q_{i} )^{\slotind-\sDegJ},
\end{align}
where summation is done over all possible values of the reduced degree $\sDegJ$, \ie $ 1 \leq \sDegJ \leq \slotind$, the term ${ \slotind - 1 \choose \sDegJ - 1} \q_{i}^{\sDegJ-1} ( 1 - \q_{i} )^{\slotind - \sDegJ}$ corresponds to the probability that the degree of the slot node is reduced to $\sDegJ$ and $D (\sDegJ )$ is the probability that the burst corresponding to the outgoing edge is decoded when the (reduced) degree of the slot node is $\sDegJ$.\footnote{As it happens in the asymptotic case, the loops in the graph are assumed to be of infinite length, such that the tree representation in Figure~\ref{fig:and_or_tree} holds, the reduction of the slot degree happens only via inter-slot IC, which is implicitly assumed in the term ${ \slotind - 1 \choose \sDegJ - 1} \q_{i}^{\sDegJ-1} ( 1 - \q_{i} )^{\slotind - \sDegJ}$. On the other hand, $D (\sDegJ)$ expresses the probability that an outgoing edge from the slot node is decoded using intra-slot IC (see Section~\ref{sec:decoding_probs}). In other words, inter- and intra-slot IC are in the asymptotic evaluation separated over \ac{DE} iterations.}

Combining \eqref{p_i_d} and the expression for the edge-oriented slot-node degree distribution (see Section~\ref{subsec:graph}) into \eqref{p_i_avg} yields
\begin{equation}
\p_i = 1 - e^{ - \frac{\load }{ \rateIrsa } } \sum_{\slotind=1}^{\infty} \left( \frac{\load}{ \rateIrsa} \right)^{\slotind - 1} \sum_{\sDegJ=1}^{\slotind} \frac { D ( \sDegJ ) }{ ( \sDegJ - 1 ) ! } \q_i^{\sDegJ-1} ( 1 - \q_i )^{\slotind-\sDegJ}.\label{eq:p_i}
\end{equation}
It can be shown that in case of perfect IC, \eqref{eq:p_i} becomes
\begin{align}
\p_i & = 1 - e^{-\frac{\load}{\rateIrsa}} \sum_{\sDegJ = 1}^{+\infty} \frac{D ( \sDegJ )}{ ( \sDegJ - 1 )! } \left( \frac{\load}{\rateIrsa} \q_i \right)^{ \sDegJ - 1 } \sum_{\slotind = 0}^{+\infty} \frac{\left( \frac{\load}{\rateIrsa} ( 1 - \q_i ) \right)^{ \slotind } }{\slotind !}\\
& = 1 - e^{-\frac{\load}{\rateIrsa} \q_i } \sum_{\sDegJ=1}^{+\infty} \frac{D( \sDegJ )}{ ( \sDegJ - 1 )!} \left( \frac{\load}{\rateIrsa} \q_i \right)^{\sDegJ - 1 }, \label{eq:p_i_1}
\end{align}
where $ D ( \sDegJ ) = \sum_{\decStep =1}^{\sDegJ} D ( \sDegJ ,\decStep )$, see \eqref{eq:D(r|t)}.
Further, defining $\zt=(1+\snrthr)^\decStep$, \eqref{eq:p_i_1} becomes
\allowdisplaybreaks
\begingroup
\begin{align}
\label{eq:final_p}
\p_i & = 1 - e^{-\frac{\load}{\rateIrsa} \q_i } \sum_{\sDegJ=1}^{+\infty} \left( \frac{\load}{\rateIrsa} \q_i \right)^{\sDegJ - 1 } \sum_{\decStep = 1}^{\sDegJ} \frac{e^{ - \frac{1}{\snravg} ( \zt - 1 )}}{( \sDegJ - \decStep )! \zt^{ \left(\sDegJ - \frac{\decStep+1}{2}\right) }} \notag \\
& = 1 - e^{-\frac{\load}{\rateIrsa} \q_i } \sum_{\decStep=1}^{+\infty} \frac{ \left( \frac{\load}{\rateIrsa} \q_i \right)^{\decStep-1} }{ \zt^{\left(\frac{ \decStep - 1 }{ 2 }\right)}} e^{ - \frac{1}{\snravg} ( \zt - 1 )} \sum_{\sDegJ = 0 }^{+\infty} \frac{ \left( \frac{\load}{\rateIrsa} \q_i \right)^\sDegJ }{\sDegJ ! \zt^\sDegJ } \notag \\
& = 1-  \sum_{\decStep=1}^{+\infty} \frac{ \left( \frac{\load}{\rateIrsa} \q_i \right)^{\decStep-1} }{ \zt^{\left(\frac{ \decStep - 1 }{ 2 }\right)}} e^{ - ( \zt - 1 ) \left (\frac{1}{\snravg} + \frac{\frac{\load}{\rateIrsa} \q_i} {\zt }\right)  } =: \exits(\q_i) \, .
\end{align}
\endgroup

A \ac{DE} recursion is obtained combining \eqref{eq:final_q} with \eqref{eq:final_p}, consisting of one recursion for $\q_i$ and one for $\p_i$. In the former case, the recursion takes the form $\q_i = (\exitb \circ \exits) (\q_{i-1})$ for $i \geq 1$, with initial value $\q_0 = 0$. In the latter case, it assumes the form $\p_i = (\exits \circ \exitb) (\p_{i-1})$ for $i \geq 1$, with initial value $\p_0 = \exits(0)$. Note that the \ac{DE} recursion for $\p_i$ allows expressing the asymptotic \ac{PLR} of an \ac{IRSA} scheme in a very simple way. More specifically, let $\p_{\infty} (\load, \BNdegDist, \snravg,\snrthr) = \lim_{i \rightarrow \infty} \p_i$ be the limit of the \ac{DE} recursion, where we have explicitly indicated that the limit depends on the system load, on the burst node degree distribution, on the average \ac{SNR}, and on the threshold for successful intra-slot decoding. Since $\left[ \p_{\infty}(\load, \BNdegDist, \snravg, \snrthr) \right]^{\dg}$ represents the probability that a user packet associated with a burst node of degree $d$ is not successfully received at the end of the decoding process, the asymptotic \ac{PLR} is given by
\begin{align}
\plr(\load,\BNdegDist,\snravg, \snrthr) = \sum_{\dg=2}^{\dmax} \Ld_\dg \left[ \p_{\infty} (\load, \BNdegDist, \snravg, \snrthr) \right]^{\dg} \, .
\end{align}

Next, we introduce the concept of \emph{asymptotic decoding threshold} for an \ac{IRSA} scheme over the considered block fading channel model and under the decoding algorithm described in Section~\ref{subsec:receiver}. Let $\PLRtarget$ be a target \ac{PLR}. Then, the asymptotic decoding threshold, denoted by $\load^{\star} = \load^{\star}(\BNdegDist, \snravg, \snrthr, \PLRtarget)$, is defined as the supremum system load value for which the target \ac{PLR} is achieved in the asymptotic setting:
\begin{align}
\load^{\star} = \sup_{\load \geq 0} \{ \load : \plr(\load,\BNdegDist,\snravg, \snrthr) < \PLRtarget \}.
\end{align} 

\section{Numerical Results}\label{sec:Performance}

Table~\ref{tab:Distribution} shows some degree distributions designed combining the \ac{DE} analysis developed in Section~\ref{sec:DE} with the differential evolution optimisation algorithm proposed in \cite{diffEvol1997}. For each design we set $\PLRtarget=10^{-2}$, $\snravg=20$ dB, and $\snrthr=3$ dB, and we constrained the optimisation algorithm to find the distribution $\BNdegDist$ with the largest threshold $\load^{\star}$, subject to a given average degree\footnote{A constraint on the average degree $\davg$ can be turned into a constraint on the rate $\rateIrsa$ as there is a direct relation between the two; see also equation~\eqref{eq:rate}.} $\davg$ and maximum degree $\dmax=16$.

For all chosen average degrees, the $\load^{\star}$ threshold of the optimised distribution largely exceeds the value $1\,\mathrm{[pk/slot]}$, the theoretical limit under a collision channel model. In general, the higher is the average degree $\davg$, the larger is the load threshold $\load^{\star}$. However, as $\davg$ increases, more complex burst node distributions are obtained. For instance, under a $\davg=4$ constraint, the maximum degree is $\dmax=16$ (\ie a user may transmit up to $16$ copies of its packet); when reducing $\davg$, the optimisation converges to degree distributions with a lower maximum degree, and degree-$2$ nodes become increasingly dominant.

To assess the effectiveness of the proposed design approach, tailored to the block fading channel with capture, we optimised a distribution $\Ld_5(x)$ using the \ac{DE} recursion over the collision channel \cite{Liva2011} and again constraining the optimisation to $\davg=4$ and $\dmax=16$. As from Table~\ref{tab:Distribution}, due to the mismatched channel model, we observe a $7\%$ loss in terms of $\load^{\star}$ threshold w.r.t. the distribution $\Ld_1(x)$ that fulfills the same constraints but was obtained with the \ac{DE} developed in this Chapter.

\begin{table}[tb]

\caption{Optimised user node degree distribution and corresponding threshold  $\thr$ for $\PLRtarget=10^{-2}$.}\label{tab:Distribution}

\begin{center}
\begin{tabularx}{\columnwidth}{c|>{\centering}X|c}
\hline\hline
\rule{0pt}{3ex}
$\davg$ & Distribution $\Ld(x)$ & $\thr$ \\[1mm]
\hline
\rule{0pt}{2.5ex}
$4$ & $\Ld_1(x)=0.59 x^2 + 0.27 x^3 + 0.02 x^5 + 0.12 x^{16}$ & $1.863$ \\
\rule{0pt}{0ex}
$3$ & $\Ld_2(x)=0.61 x^2 + 0.25 x^3 + 0.03 x^6 + 0.02 x^7 + 0.07 x^8 + 0.02 x^{10}$ & $1.820$ \\
$2.5$ & $\Ld_3(x)=0.66 x^2 + 0.16 x^3 + 0.18 x^4$ & $1.703$ \\
$2.25$ & $\Ld_4(x)=0.65 x^2 + 0.33 x^3 + 0.02 x^4$ & $1.644$ \\
$4$ & $\Ld_5(x)=0.49 x^2 + 0.25 x^3 + 0.01 x^4 + 0.03 x^5 + 0.13 x^6 + 0.01 x^{13} + 0.02 x^{14} + 0.06 x^{16}$ & $1.734$ \\
\hline\hline
\end{tabularx}
\end{center}
\end{table}

\begin{figure}[tb]
\begin{center}
\includegraphics[width=0.8\columnwidth,draft=false]{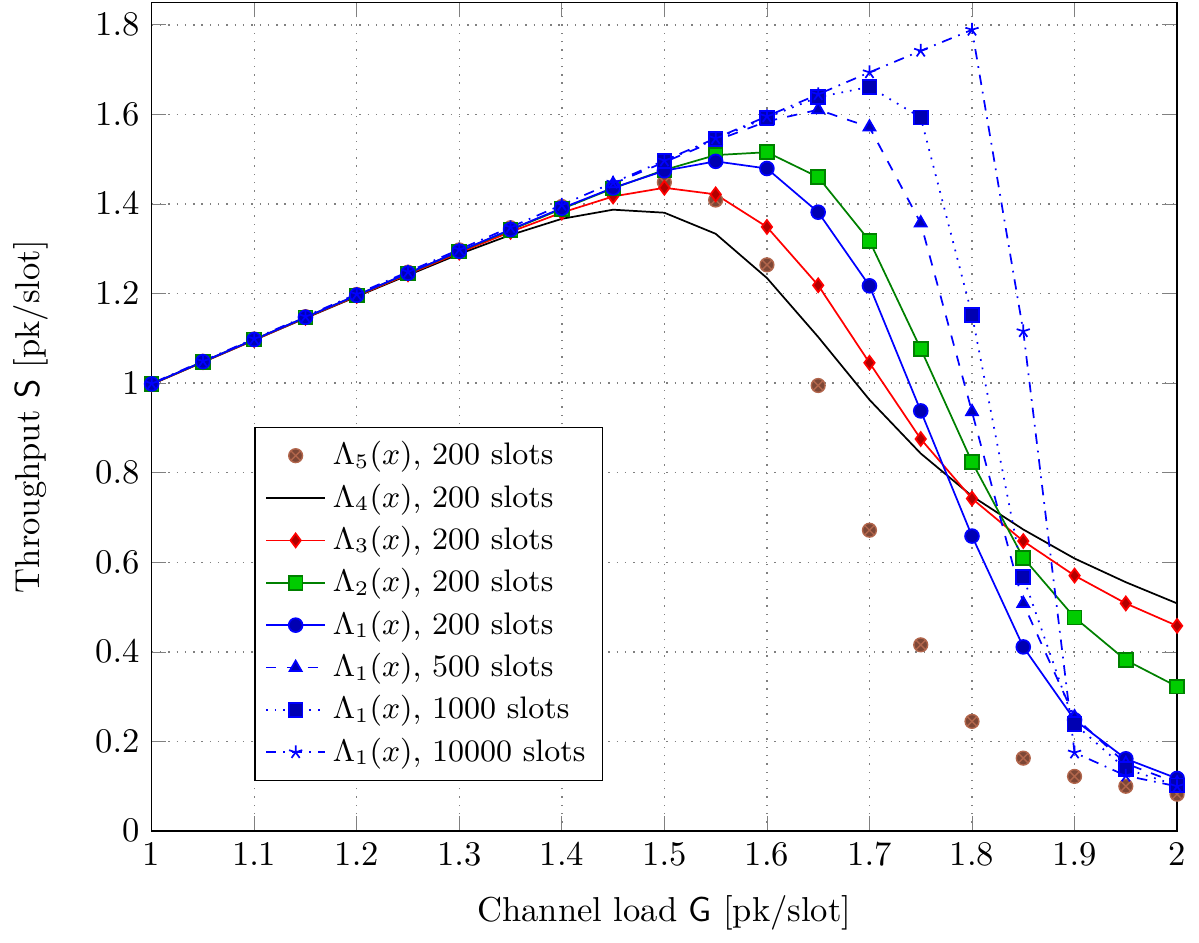}
\end{center}
\caption{Throughput values achieved by the burst node distributions in Table~\ref{tab:Distribution}, versus the channel load. Various frame sizes $\numSlot$, $\snravg=20$ dB, $\snrthr=3$ dB.}\label{fig:thr}
\end{figure}

\begin{sloppypar}
In order to investigate the performance of the optimised distributions in a finite frame length setting, we ran Monte Carlo simulations setting $\numSlot=200$ (unless otherwise stated), $\snravg=20$ dB, $\snrthr=3$ dB, and the maximum number of \ac{IC} iterations to $20$. Figure~\ref{fig:thr} illustrates the throughput $\tp$, defined as the average number of successfully decoded packets per slot, versus the channel load. Even in the relatively short frame regime $\numSlot=200$, all distributions exhibit peak throughput exceeding $1\,\mathrm{[pk/slot]}$. The distribution $\Ld_2(x)$ is the one achieving the highest throughput of $1.52\, \mathrm{[pk/slot]}$, although $\Ld_1(x)$ is the distribution with the highest threshold $\thr$. It is important to recall that the threshold is computed for a target \ac{PLR} of ${\PLRtarget=10^{-2}}$ and that when we move to finite frame lengths, the threshold effect on the \ac{PLR} tends to vanish and a more graceful degradation of the \ac{PLR} curve as the channel load increases is expected. Moreover, as there are user nodes transmitting as high as $16$ replicas, for short frames sizes the distribution $\Ld_1(x)$ is more penalised w.r.t. to $\Ld_2(x)$ for which there are at most $10$ replicas per user node.\footnote{Results for frame size of $500$ slots, not presented in the figures, show that the peak throughput for $\Ld_1(x)$ is $1.61$ while for $\Ld_2(x)$ is $1.60$, which is in line with this observation.} This effect, coupled with the fact that the threshold $\thr$ of $\Ld_1(x)$ is only slightly better than the one of $\Ld_2(x)$, explains the peak throughput behaviour. To investigate the benefit of larger frames, we selected $\Ld_1(x)$ and we increased the frame size up to $10000$ slots. As expected, the peak throughput is greatly improved from $1.49$ to $1.79\, \mathrm{[pk/slot]}$, \ie $20\%$ of gain.
\end{sloppypar}

\begin{figure}[tb]
\begin{center}
\includegraphics[width=0.8\columnwidth,draft=false]{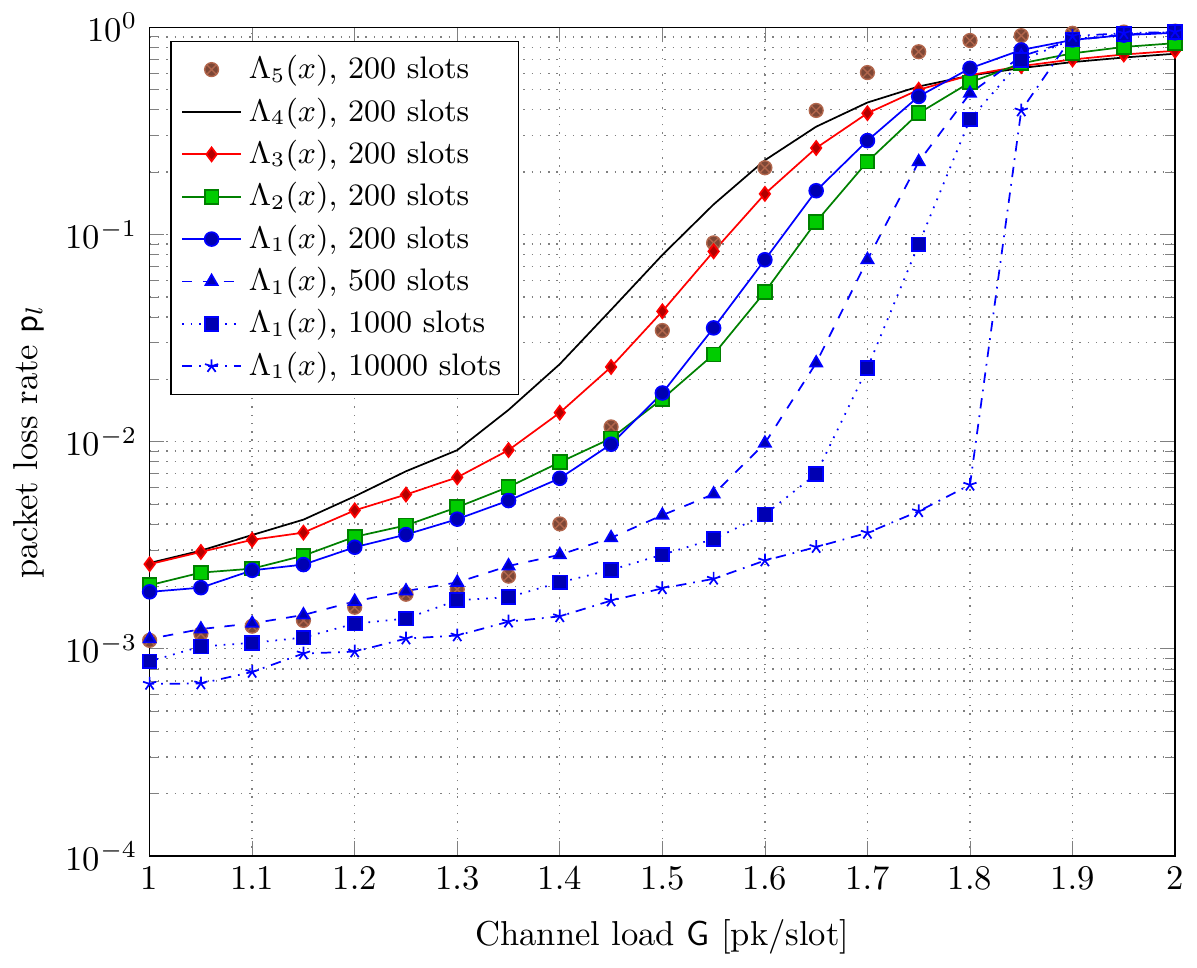}
\end{center}
\caption{\ac{PLR} values achieved by the burst node distributions in Table~\ref{tab:Distribution}, versus the channel load. Various frame sizes $\numSlot$, $\snravg=20$ dB, $\snrthr=3$ dB.}\label{fig:plr}
\end{figure}

The \ac{PLR} performance is illustrated in Figure~\ref{fig:plr}. Coherently with the optimisation results, the $\Ld_1(x)$ distribution achieves $\plr=10^{-2}$ for values of the channel load slightly larger than the ones required by $\Ld_2(x)$. Indeed, the steeper \ac{PLR} curve of $\Ld_1(x)$ is the reason for the slightly larger peak throughput of $\Ld_2(x)$ observed in Figure~\ref{fig:thr}.
Finally, as expected, an increment of the number of slots per frame yields to an increase in the channel load for which the target \ac{PLR} is achieved. 

\section{Conclusions}\label{sec:Conclusions}

The asymptotic analysis of \ac{IRSA} access schemes, assuming both a Rayleigh block fading channel and capture effect, was presented in this Chapter. The decoding probability of a burst replica due to intra-slot \ac{IC} was derived. The \ac{DE} analysis was modified considering the Rayleigh block fading channel and the user and slot nodes updates of the iterative procedure were explicitly derived. Due to the presence of fading, the optimisation procedure target was modified as well, to select distributions that were able to achieve higher channel load values without exceeding a properly defined target \ac{PLR}. Some degree distributions were designed for different values of average degree. Remarkably, all presented a load threshold that guaranteed \ac{PLR} below $10^{-2}$ for values well above $1\,\mathrm{[pk/slot]}$. In particular, the best distribution exceeded $1.8\,\mathrm{[pk/slot]}$. The derived distributions were shown to perform well also for finite frame durations. Even for relatively short frames with $200$ slots, the peak throughput exceeded $1.5\,\mathrm{[pk/slot]}$ and up to $1.45\,\mathrm{[pk/slot]}$ the \ac{PLR} remained below $10^{-2}$.  
\chapter{Random Access with Multiple Receivers}
\label{chapter7}
\thispagestyle{empty}
\ifpdf
    \graphicspath{{chapter7/figures/}}
\fi
\epigraph{It is the weight, not numbers of experiments that is to be regarded}{Isaac Newton}

\section{Introduction} \label{sec:introduction}

Until now we have focused on the simple a basic scenario in which a plurality of user terminals want to communicate in an uncoordinated fashion to a single receiver. Different \ac{RA} techniques have been investigated, both involving time synchronization (at slot or at frame level) or completely asynchronous. In this Chapter we follow a parallel research direction, in which the transmitter can rely on spatial diversity. In particular, we expand the simple topology typical of \ac{MAC} channel, in two directions. Firstly, we add a finite number of non-cooperative receivers (called also relays) and secondly, we complement the uplink of such a system with a downlink or wireless backhaul link,\footnote{In the following we will use the term downlink for referring to such link, in accordance to satellite communications nomenclature in which the return uplink refers to the users to satellite communication link and the downlink refers to the satellite to gateway link. In terrestrial networks the name downlink shall be substituted with the more appropriate wireless backhaul link.} where the receivers forward the successfully received packets to a single common \ac{gateway}. In order to take into account variability of user-relay links, a packet erasure channel is considered in the uplink phase. Specifically, each user-rely link is considered active with probability $1-\peras$ on a slot-by-slot basis. Furthermore, \ac{SA} is considered as channel access policy for each such link. In the downlink phase, the relays forward the correctly received packets or linear combinations of them.

In this Chapter we aim at characterising the uplink aggregate throughput seen at all receivers as well as the \ac{PLR}. We derive the upper bound on the downlink rate defined as number of total downlink transmission per uplink time slot and we show that with \ac{RLC} the bound can be achieved letting the observation of time slots in the uplink grow very large. We then propose alternative low complexity downlink policies that allow the relays to drop some of the correctly received packets eventually based on the observation of the uplink slots. Finally, we address the scenario in which the relays have a finite memory so that only a finite number of correctly decoded packets can be collected before starting the downlink phase. In this case, we first derive an analytical expression of the performance for \ac{RLC} and we compare this solution with the dropping policies. We are then able to show that under some dowlink rate conditions and for some buffer sizes, the simple dropping policies are able to largely outperform \ac{RLC}.

Note that part of the Chapter has been already published in \cite{Munari2013} and the author of the thesis is not co-author of that work. Nonetheless, a recall of the general framework is considered fundamental for giving a proper understanding of the contribution of the thesis' author which mainly focuses on Section~\ref{sec:buffer_downlink}.

\begin{sloppypar}
\ac{SA} with space (antenna) diversity was analysed in \cite{Zorzi1994_diversity} under the assumption of Rayleigh fading and shadowing, with emphasis on the two-antenna case. The authors in \cite{Jakovetic2015} address also a similar scenario, although in their case \ac{SIC} is introduced at the relays. Depending on the considered policy, both inter-relay and intra-relay \ac{SIC} are also applied.
\end{sloppypar} 
\section{System Model and Preliminaries} \label{sec:sysModel}

Throughout this Chapter, we focus on the topology depicted in Figure~\ref{fig:simple_topology}, where an infinite population of users wants to deliver information in the form of data packets to a collecting \ac{gateway}. The transmission process is divided in two phases, referred to as \emph{uplink} and \emph{downlink}, respectively. During the former, data are sent in an uncoordinated fashion over a shared wireless channel to a set of $\nrx$ receivers or relays, which, in turn, forward collected information to the \ac{gateway} in the downlink.

As to the uplink, time is divided in successive slots, and transmission parameters in terms of packet length, coding and modulation are fixed such that one packet can be sent within one time unit. Users are assumed to be slot-synchronized, and \ac{SA} \cite{Abramson1970} is employed as medium access policy. Furthermore, the number of users accessing the channel in a generic slot is modelled as a Poisson-distributed r.v. $\usersRV$ of intensity $\load$, with:
\begin{equation}
\textrm{Pr}\{ \usersRV = \user \} = \frac{\load^\user e^{-\load}}{\user!}\,.
\end{equation}

\begin{figure}
\centering
\includegraphics[width=0.8\columnwidth]{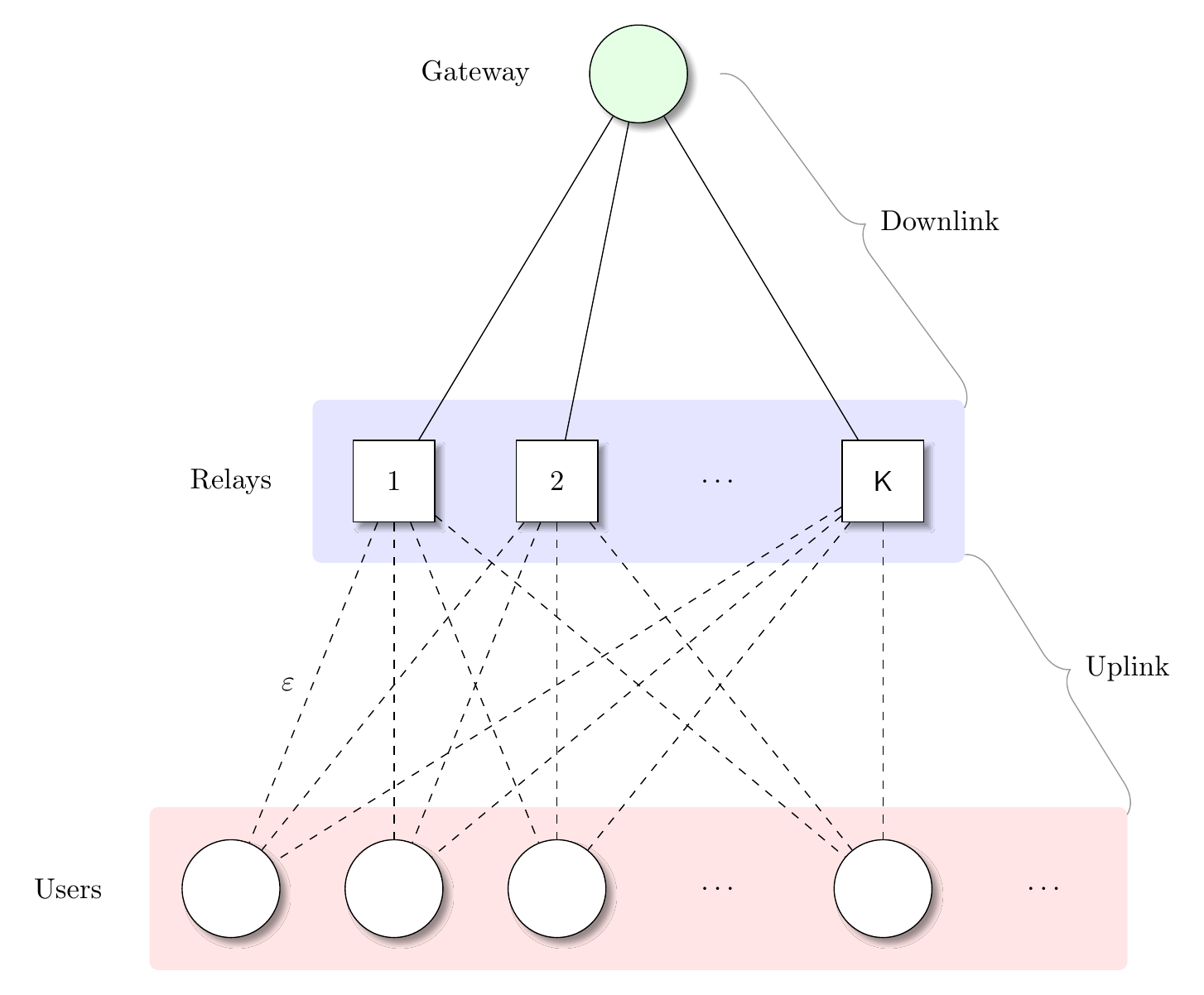}
\caption{Reference topology for the system under consideration.}
\label{fig:simple_topology}
\end{figure}

The uplink wireless link connecting user $i$ and receiver $j$ is described by a packet erasure channel with erasure probability $\peras_{i,j}$, where we assume independent realisations for any $(i,j)$ pair, as well as for a specific user-receiver couple across time slots. For the sake of mathematical tractability, we set $\peras_{i,j} = \peras, \: \forall \,\, i,\, j$. Following the on-off fading description \cite{OnOff2003}, we assume that a packet is either completely shadowed, not bringing any power or interference contribution at a receiver, or it arrives unfaded. While, such a model is especially useful to develop mathematically tractable approaches to the aim of highlighting the key tradeoffs of the considered scenario, it also effectively captures effects like fading and short-term receiver unavailability due, for instance, to the presence of obstacles.
Throughout our investigation, no multi-user detection capabilities are considered at the relays, so that collisions among non-erased data units are regarded as destructive and prevent decoding at a receiver.

Within this framework, the number of non-erased packets that arrive at a relay when $\user$ concurrent transmissions take place follows a binomial distribution of parameters $(\user,1-\peras)$ over one slot. Therefore, a successful reception occurs with probability $\user(1-\peras)\peras^{\user-1}$, and the average throughput experienced at each of the $\nrx$ receivers, in terms of decoded packets per slot, can be computed as:
\begin{equation}
\tpSA = \sum\limits_{\user=0}^{\infty} \frac{\load^\user e^{-\load}}{\user!} \, \user(1-\peras)\peras^{\user-1} = \load (1-\peras) e^{-\load(1-\peras)}\,,
\label{eq:Tsa}
\end{equation}
corresponding to the performance of a SA system with erasures.
On the other hand, a spatial diversity gain can be triggered when the  relays are considered jointly, since independent channel realisations may lead them to retrieve different information units over the same time slot. In order to quantify this beneficial effect, we label a packet as \emph{collected} when it has been received by at least one of the relays, and we introduce the \emph{uplink throughput} $\tpULk{\nrx}$ as the average number of collected packets per slot. Despite its simplicity, such a definition offers an effective characterisation of the beneficial effects of diversity, by properly accounting for both the possibility of retrieving up to $\min\{\user,\nrx\}$ distinct data units or multiple times the same data unit over a slot, as will be discussed in details in Section~\ref{sec:uplinkThroughput}. On the other hand, $\tpULk{\nrx}$ also quantifies the actual amount of information that can be retrieved by the set of receivers, providing an upper bound for the overall achievable end-to-end performance, and setting the target for the design of any relay-to-\ac{gateway} delivery strategy.

For the downlink phase, we focus on a \ac{DF} approach, so that each receiver re-encodes and transmits only packets it has correctly retrieved during the uplink phase, or possibly linear combinations thereof. A finite downlink capacity is assumed, and relays have to share a common bandwidth to communicate to the \ac{gateway} by means of a \ac{TDMA} scheme. In order to get an insightful characterisation of the optimum achievable system performance, we assume relay-to-\ac{gateway} links to be error free, and let resource allocation for the D\&F phase be performed ideally and without additional cost by the central collecting unit.

\subsection{Notation}
Prior to delving into the details of our mathematical framework, we introduce in the following some useful notation. All the variables will be properly introduced when needed in the discussion, and the present Section is simply meant to offer a quick reference point throughout the reading.

$\nrx$ relays are available, and, within time slot $\slot$, the countably infinite set of possible outcomes at each of them is labeled as $\evSlotSet_\slot:=\{\evSlot_0^\slot,\evSlot_1^\slot,\evSlot_2^\slot,\ldots,\evSlot_\infty^\slot\}$ for each $\slot=1,2,\ldots,\numulslot$. Here, $\evSlot_0^\slot$ denotes the erasure event (given either by a collision or by an idle slot), while $\evSlot_j^\slot$ indicates the event that the packet of the $j$-th user arriving in slot $\slot$ was received. According to this notation, we define as
$\obsSlot_\rxEl^\slot$ the random variables with alphabet $\mathbb N$, where $\obsSlot_\rxEl^\slot = j$ if $\evSlot_j^\slot$ was the observation at relay $\rxEl$.
When needed for mathematical discussion, we let the uplink operate for $\numulslot$ time slots. In this case, let $\setColPack_\rxEl^{\numulslot}$ be the set of collected packets after $\numulslot$ time slots at receiver $\rxEl$, where $\setColPack_\rxEl^{\numulslot} \subsetneq \bigcup_{\slot=1}^{\numulslot} \{ \evSlotSet_\slot \backslash \evSlot_0^\slot\}$. That is, we do not add the erasure events to $\setColPack_\rxEl^{\numulslot}$. The number of received packets at relay $\rxEl$ after $\numulslot$ time slots is thus $|\setColPack_\rxEl^{\numulslot}|$.

In general, the complement of a set $\setColPack$ is indicated as $\setColPack^{c}$. We write vectors as bold variables, \eg  $\obsPk$, while matrices and their transposes are labeled by uppercase bold letters, \eg  $\bm{B}$ and $\bm{B}^T$.

\section{Uplink Performance}  \label{sec:uplink}

With reference to the topology of Figure~\ref{fig:simple_topology}, we first consider the uplink phase. In order to gather a comprehensive description of the improvements enabled by receiver diversity, we characterise the system by means of two somewhat complementary metrics: uplink throughput (Section~\ref{sec:uplinkThroughput}) and packet loss rate (Section~\ref{sec:PLR}).

\subsection{Uplink Throughput} \label{sec:uplinkThroughput}
Let us focus on the random access channel, and, following the definition introduced in Section~\ref{sec:sysModel}, let $\collRV$ be the number of packets collected by the relays over one slot. $\collRV$ is a r.v. with outcomes in the set $\{0,1,2, \ldots, \nrx\}$, where the maximum value occurs when the $\nrx$ receivers decode distinct packets due to different erasure patterns. The average uplink throughput can thus be expressed by conditioning on the number of concurrent transmissions as:
\begin{equation}
\label{eq:truUplink_general}
\tpULk{\nrx} =\!\mathbb E_U [ \, \mathbb E[ \,\collRV \,|\, \usersRV \,]\, ] \!=\!\! \sum\limits_{\user=0}^\infty \frac{\load^\user e^{-\load}}{\user!}  \sum\limits_{\coll=0}^\nrx \coll\,\textrm{Pr}\{\collRV=\coll \,|\,\usersRV=\user\}.
\end{equation}
While equation~\eqref{eq:truUplink_general} formula holds for any $\nrx$, the computation of the collection probabilities intrinsically depends on the number of available relays. In this perspective, we articulate our analysis by first considering the two-receiver case, to then extend the results for an arbitrary topology.

\subsubsection{The Two-Receiver Case}

Let us first then focus on the case in which only two relays are available. Such a scenario allows a compact mathematical derivation of the uplink throughput, as the events leading to packet collection at the relays set can easily be expressed. On the other hand, it also represents a case of practical relevance, as it can be instantiated by simply adding a receiver to an existing \ac{SA}-based system.
\begin{figure}
\centering
\includegraphics[width=0.8\columnwidth]{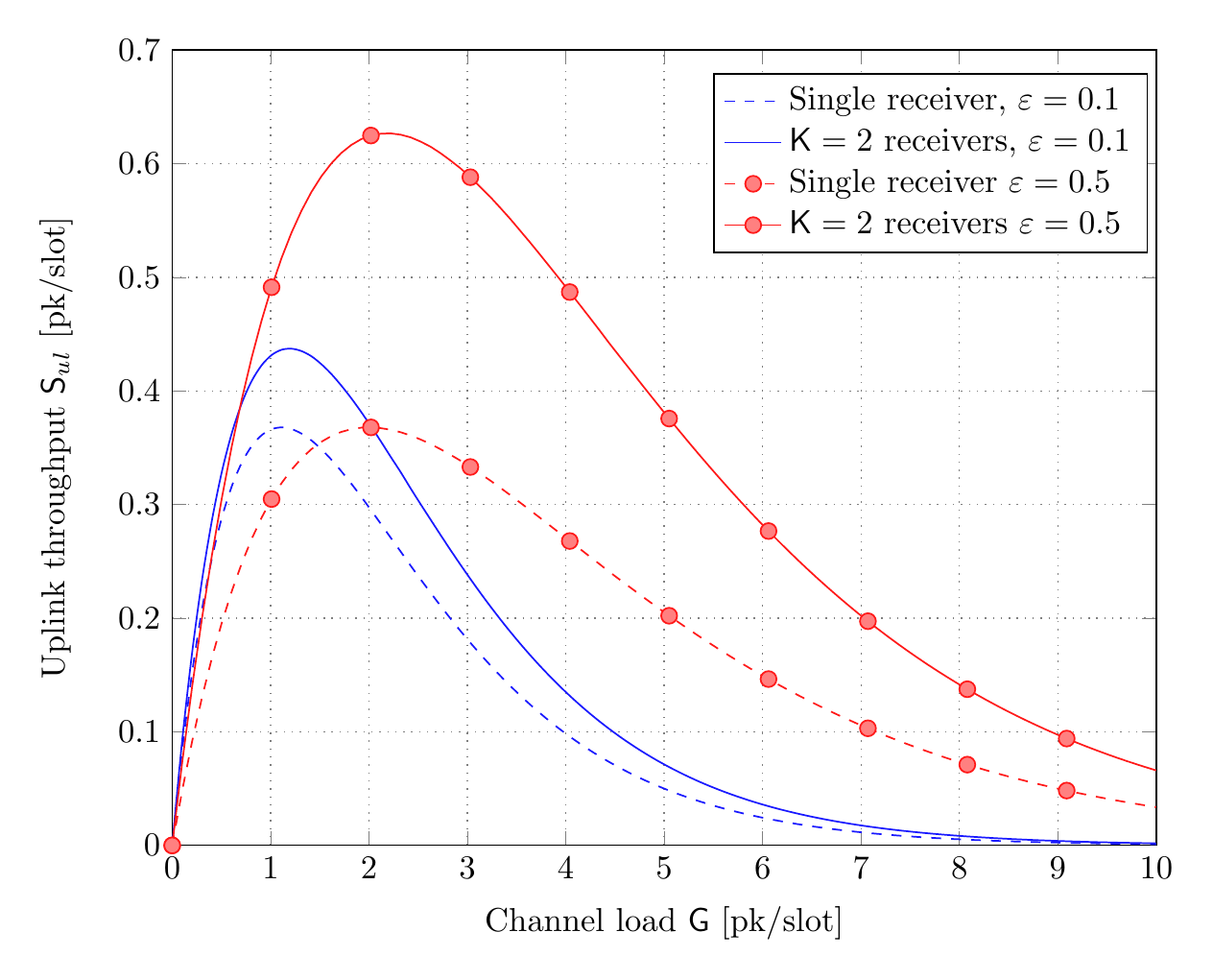}
\caption{Average uplink throughput vs channel load under different erasure probabilities. Continuous lines indicate the performance in the presence of two receivers, whereas dashed lines report the behaviour of pure \ac{SA}.}
\label{fig:truputUplink}
\end{figure}
When $\nrx=2$, the situation for $\collRV=1$ can be easily accounted for, since a single packet can be collected as soon as at least one of the relays does not undergo an erasure, \ie with overall probability $1-\peras^2$. On the other hand, by virtue of the binomial distribution of $\usersRV$, the event of collecting a single information unit over one slot occurs with probability
\begin{align}
\textrm{Pr}\{\collRV =1 \,|\,\usersRV=\user\} =& \,2u(1-\peras) \peras^{\user-1}\left[ 1 - \user(1-\peras) \peras^{\user-1} \right] + \user(1-\peras)^2 \peras^{2(\user-1)}
\end{align}
where the former addend accounts for the case in which one relay decodes a packet while the other does not (either due to erasures or to a collision), whereas the latter tracks the case of having the two relays decoding the same information unit. Conversely, a reward of two packets is obtained only when the receivers successfully retrieve distinct units, with probability
 \begin{equation}
\textrm{Pr}\{\collRV=2 \,|\,\usersRV=\user\} = \user(\user-1)(1-\peras)^2 \peras^{2(\user-1)}.
\end{equation}
Plugging these results into \eqref{eq:truUplink_general} we get, after some calculations, a closed-form expression for the throughput in the uplink, and thus, as discussed, also for the end-to-end \ac{DF} case with infinite downlink capacity:
\begin{equation}
\tpULk{2} = 2\load (1-\peras)\, e^{-\load (1-\peras)} - \load (1-\peras)^2 \, e^{-\load (1-\peras^2)}.
\label{eq:truUplink_closedForm}
\end{equation}
The trend of $\tpULk{2}$ is reported in Figure~\ref{fig:truputUplink} against the channel load $\load$ for different values of the erasure probability, and compared to the performance in the presence of a single receiver, \ie $\tpSA$.
equation~\eqref{eq:truUplink_closedForm} conveniently expresses $\tpULk{2}$ as twice the throughput of \ac{SA} in the presence of erasures, reduced by a loss factor which accounts for the possibility of having both relays decode the same information unit. In this perspective, it is interesting to evaluate the maximum throughput $\tpULkMax{2}(\peras)$ as well as the optimal working point $\loadMax(\peras)$ achieving it for the system uplink. The transcendental nature of \eqref{eq:truUplink_closedForm} does not allow to obtain a closed formulation of these quantities, which, on the other hand, can easily be estimated by means of numerical optimisation techniques.
\begin{figure}
\centering
\includegraphics[width=0.8\columnwidth]{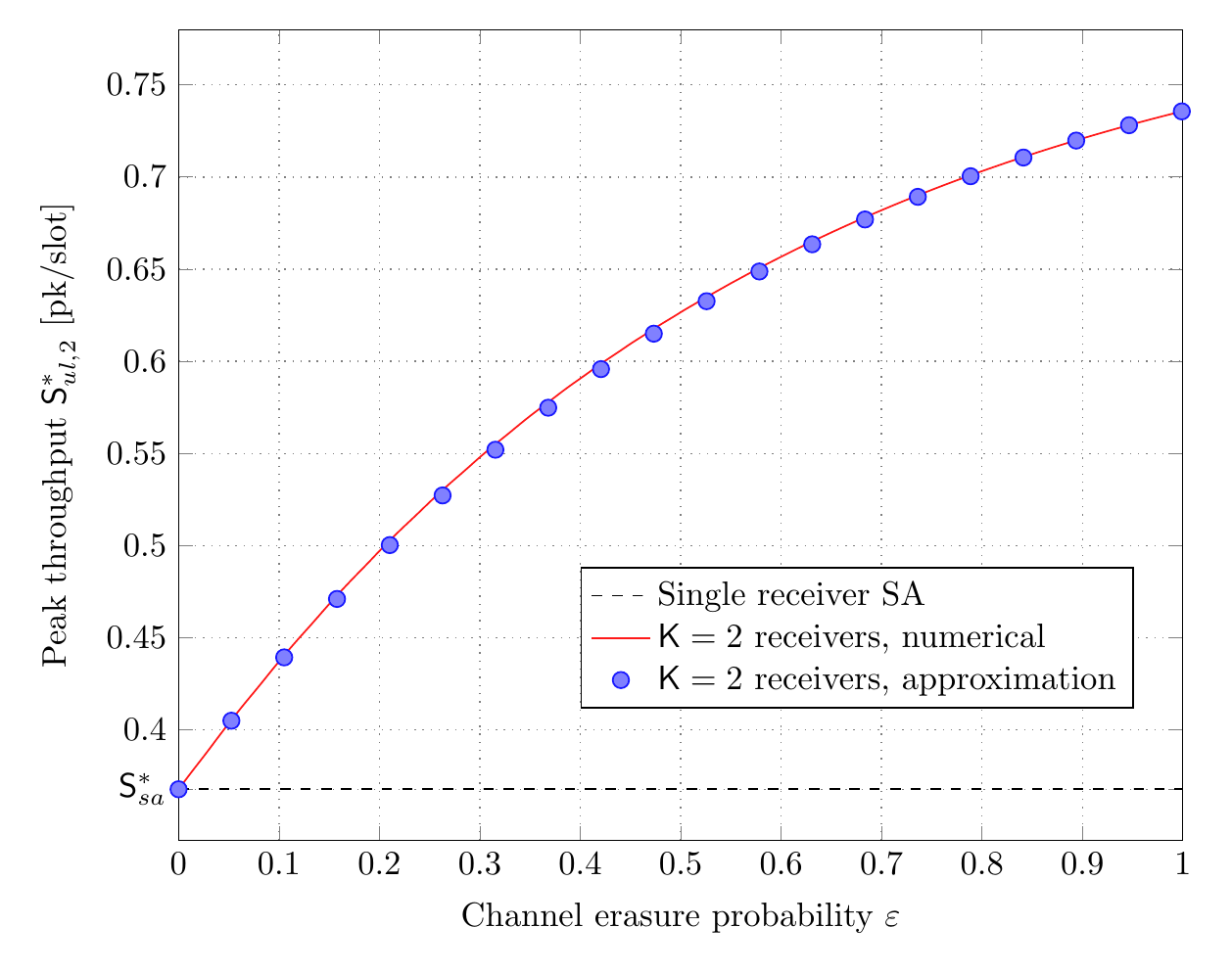}
\caption{Maximum uplink throughput vs erasure rate. The red continuous line reports the performance $\tpULkMax{2}$ of a two-receiver scheme, while blue circled markers indicate $\tpULk{2}(1/(1-\peras))$, and the dashed black line shows the behaviour of pure \ac{SA}.}
\label{fig:maxTruVsEpsilon}
\end{figure}
The results of this analysis are reported in Figure~\ref{fig:maxTruVsEpsilon}, where the peak throughput $\tpULMax$ is depicted by the red curve as a function of $\peras$ and compared to the performance of \ac{SA}, which clearly collects on average at most 0.36~ $\mathrm{[pk/slot]}$ regardless of the erasure rate. In ideal channel conditions, \ie $\peras=0$, no benefits can be obtained by resorting to multiple relays, as all of them would see the same reception set across slots. Conversely, higher values of $\peras$ favour a decorrelation of the pattern of packets that can be correctly retrieved, and consequently improve the achievable throughput at the expense of higher loss rates. The result is a monotonically increasing behaviour for $\tpULkMax{2}(\peras)$, prior to plummeting with a singularity to a null throughput for the degenerate case $\peras=1$. Figure~\ref{fig:maxTruVsEpsilon} also reports (circled-blue markers) the average throughput obtained for $\load=1/(1-\peras)$, \ie when the uplink of the system under consideration operates at the optimal working point for a single-receiver \ac{SA}, showing a tight match. In fact, even though the abscissa of the maximum $\loadMax(\peras)$ may differ from this value (they coincide only for the ideal case $\peras=0$), the error which is committed when approximating $\tpULkMax{2}$ with $\tpULk{2}(1/(1-\peras))$ can easily be shown numerically to never exceed $0.6$\%, due to the very small slope of the function in the neighborhood of $\loadMax(\peras)$. We can thus provide a very precise estimate of the peak uplink performance for a specific erasure rate as:
\begin{equation}
\tpULkMax{2}(\peras) \simeq \frac{2}{e} - (1-\peras)\, e^{-1-\peras}, \hspace{.5cm} 0\leq \peras < 1.
\label{eq:truUplink_peakApprox}
\end{equation}
which once again compactly captures the behaviour of a two-receiver scenario by quantifying the loss with respect to twice the performance of \ac{SA}.
In this perspective, two remarks shall be made. Firstly, in order to approach the upper bound, the system has to be operated at very high load, as $\loadMax \simeq 1/(1-\peras)$). These working points are typically not of interest, since very low levels of reliability can be provided by a congested channel with high erasure rates. Nevertheless, the presence of a second receiver triggers remarkable improvements already for loss probabilities that are of practical relevance, \eg under harsh fading conditions or for satellite networks. Indeed, with $\peras=0.1$ a $\sim 15$\% raise can be spotted, whereas a loss rate of $20$\% already leads to a $50$\% throughput gain.
Secondly, the proposed framework highlights how no modifications in terms of system load are needed with respect to plain \ac{SA} for a two-receiver system to be very efficiently operated. Such a result suggests that a relay node can be seamlessly and efficiently added to an already operating \ac{SA} uplink when available, triggering the maximum achievable benefit without the need to undergo a re-tuning of the system which might be particularly expensive in terms of resources.

\subsubsection{The General Case, $\nrx>2$}
\label{sec:general_throughput}

Let us now focus on the general topology reported in Figure~\ref{fig:simple_topology}, where $\nrx$ relays are available. While conceptually applicable, the approach presented to compute the uplink throughput in the two-receiver case becomes cumbersome as $\nrx$ grows, due to the rapidly increasing number of events that have to be accounted for.
In order to characterise $\tpULk{\nrx}$, then, we follow a different strategy.
With reference to a single slot $\slot$, let $\evSlotSet_\slot:=\{\evSlot_0^\slot,\evSlot_1^\slot,\evSlot_2^\slot,\ldots,\evSlot_\infty^\slot\}$ for each $\slot=1,2,\ldots,\numulslot$ be the countably infinite set of possible outcomes at each relay, where $\evSlot_0^\slot$ denotes the erasure event while $\evSlot_j^\slot$ indicates the event that the packet of the $j$-th user arriving in slot $\slot$ was received. Let us furthermore define as
$\obsSlot_\rxEl^\slot$ the random variables with alphabet $\mathcal X=\{0,1,2,\ldots,\infty\}$, where $\obsSlot_\rxEl^\slot = j$ if $\evSlot_j^\slot$ was the observation at relay $\rxEl$, so that $\obsSlot_\rxEl^1,\obsSlot_\rxEl^2,\ldots, \obsSlot_\rxEl^{\numulslot}$ is an i.i.d. sequence for each relay $\rxEl$. We let the uplink operate for $\numulslot$ time slots, and indicate as $\setColPack_\rxEl^{\numulslot}$ the set of packets collected at receiver $\rxEl$ over this time-span, where $\setColPack_\rxEl^{\numulslot} \subsetneq \bigcup_{\slot=1}^{\numulslot} \{ \evSlotSet_\slot \backslash \evSlot_0^\slot\}$ (\ie we do not add the erasure events to $\setColPack_\rxEl^{\numulslot}$). The number of received packets at relay $\rxEl$ after $\numulslot$ time slots is thus $|\setColPack_\rxEl^{\numulslot}|$ and, with reference to this notation, we prove the following result:
\begin{prop}
 For an arbitrary number of $\nrx$ relays, the throughput $\tpULk{\nrx}$ is given by
 \begin{align}
 \tpULk{\nrx} = \sum_{\rxEl=1}^\nrx (-1)^{\rxEl-1} {\nrx \choose \rxEl} \load (1-\peras)^\rxEl e^{-\load(1-\peras^\rxEl)}
\end{align}
\end{prop}
\vspace{2mm}
\begin{proof}
We have $|\setColPack_\rxEl^{\numulslot}| = \sum_{\slot=1}^{\numulslot} \mathbb{I}_{\{\obsSlot_\rxEl^\slot\not = 0\}}$, where $\mathbb{I}_{\{E\}}$ denotes the indicator random variable that takes on the value $1$ if the event $E$ is true and $0$ otherwise.
The throughput seen by a single relay can then be written as $ \tpULk{1}=\mathbb E[\mathbb{I}_{\{\obsSlot_\rxEl^\slot\not = 0\}}] = \Pr\{\obsSlot_\rxEl^\slot\not = 0\}$, and does not depend on the specific receiver being considered. By the weak law of large numbers,
\begin{align}
 \tpULk{1} = \lim_{\numulslot\rightarrow \infty} \frac{|\setColPack_\rxEl^{\numulslot}| }{\numulslot}
\end{align}
or, more formally,
\begin{align}
 \lim_{\numulslot\rightarrow\infty} \Pr\left\{ \left| \frac{|\setColPack_\rxEl^{\numulslot}| }{\numulslot} - \tpULk{1} \right| > \epsilon \right\} = 0 \quad ~\text{for some }\epsilon > 0.
\end{align}
Similarly, for $\nrx$ relays we have
\begin{align}
\tpULk{\nrx} = \lim_{\numulslot\rightarrow\infty} \frac{|\bigcup_{\rxEl=1}^\nrx \setColPack_\rxEl^{\numulslot}|}{\numulslot}
\end{align}
By the inclusion-exclusion principle (see, \eg \cite{slomson1991introduction}), we have
\begin{align}
 \left|\bigcup_{\rxEl=1}^\nrx \setColPack_\rxEl^{\numulslot}\right| = \sum_{\setRel\subseteq \{1,\ldots,\nrx\}, \setRel \not = \emptyset} (-1)^{|\setRel|-1} \left| \setCommPack_{\setRel}^{\numulslot} \right| \quad
 \text{with } \setCommPack_{\setRel}^{\numulslot} = \bigcap_{\rxEl \in \setRel} \setColPack_\rxEl^{\numulslot} \label{eq:defI_S}
\end{align}
Here, $\setCommPack_{\setRel}^{\numulslot}$ denotes the set of packets that all the relay nodes specified by $\setRel = \{\rxEl_1,\rxEl_2,\ldots,\rxEl_{|\setRel|}\}$ have in common:
\begin{align}
 \left| \setCommPack_{\setRel}^{\numulslot}\right|=\left|\bigcap_{\rxEl \in \setRel} \setColPack_\rxEl^{\numulslot}\right| = \sum_{\slot=1}^ {\numulslot} \mathbb{I}_{\{0 \not = \obsSlot_{\rxEl_1}^\slot = \obsSlot_{\rxEl_2}^\slot = \ldots = \obsSlot_{\rxEl_{|\setRel|}}^\slot \} }
\end{align}
Due to symmetry in the setup, the value of $\left| \setCommPack_{\setRel}^{\numulslot}\right|$ only depends on the cardinality of $\setRel$ but not the explicit choice, so that $\left|\setCommPack_{\setRel}^{\numulslot}\right| = \carI_\rxEl^{\numulslot}$ for $\rxEl = |\setRel|$, and,
$$\left|\bigcup_{\rxEl=1}^\nrx \setColPack_\rxEl^{\numulslot}\right| = \sum_{\rxEl=1}^\nrx (-1)^{\rxEl-1} {\nrx \choose \rxEl} \carI_\rxEl^{\numulslot}.$$
As $\obsSlot_\rxEl^1,\obsSlot_\rxEl^2,\ldots, \obsSlot_\rxEl^{\numulslot}$ are i.i.d., by the weak law of large numbers we have:
\begin{align}
 \lim_{\numulslot\rightarrow\infty}\frac{\left|\setCommPack_{\setRel}^{\numulslot} \right|}{\numulslot} =\Pr[\{0 \not = \obsSlot_{\rxEl_1}^\slot = \obsSlot_{\rxEl_2}^\slot = \ldots = \obsSlot_{\rxEl_{|\setRel|}}^\slot \}].
\end{align}
We can compute the latter probability as
\begin{align}
\Pr\{0 &\not = \obsSlot_{\rxEl_1}^\slot = \ldots = \obsSlot_{\rxEl_{|\setRel|}}^\slot \}
 = \sum_{\user} \Pr\{0 \not = \obsSlot_{\rxEl_1}^\slot = \ldots = \obsSlot_{\rxEl_{|\setRel|}}^\slot |\usersRV = \user\}\Pr\{\usersRV=\user\} \nonumber\\
 = & \sum_{\user=1}^{\infty} \frac{e^{-\load}\load^\user}{\user!} {\user \choose 1} \left((1-\peras) \peras^{\user-1}\right)^{|\setRel|}
 = (1-\peras)^{|\setRel|}\load e^{-\load(1-\peras^{|\setRel|})}
\end{align}
As $\lim_{\numulslot\rightarrow\infty} \frac{\carI_\rxEl^{\numulslot}}{\numulslot} = (1-\peras)^{\rxEl}\load e^{-\load(1-\peras^{\rxEl})}$, the proposition follows.
\end{proof}

\begin{figure}
\centering
\includegraphics[width=0.8\columnwidth]{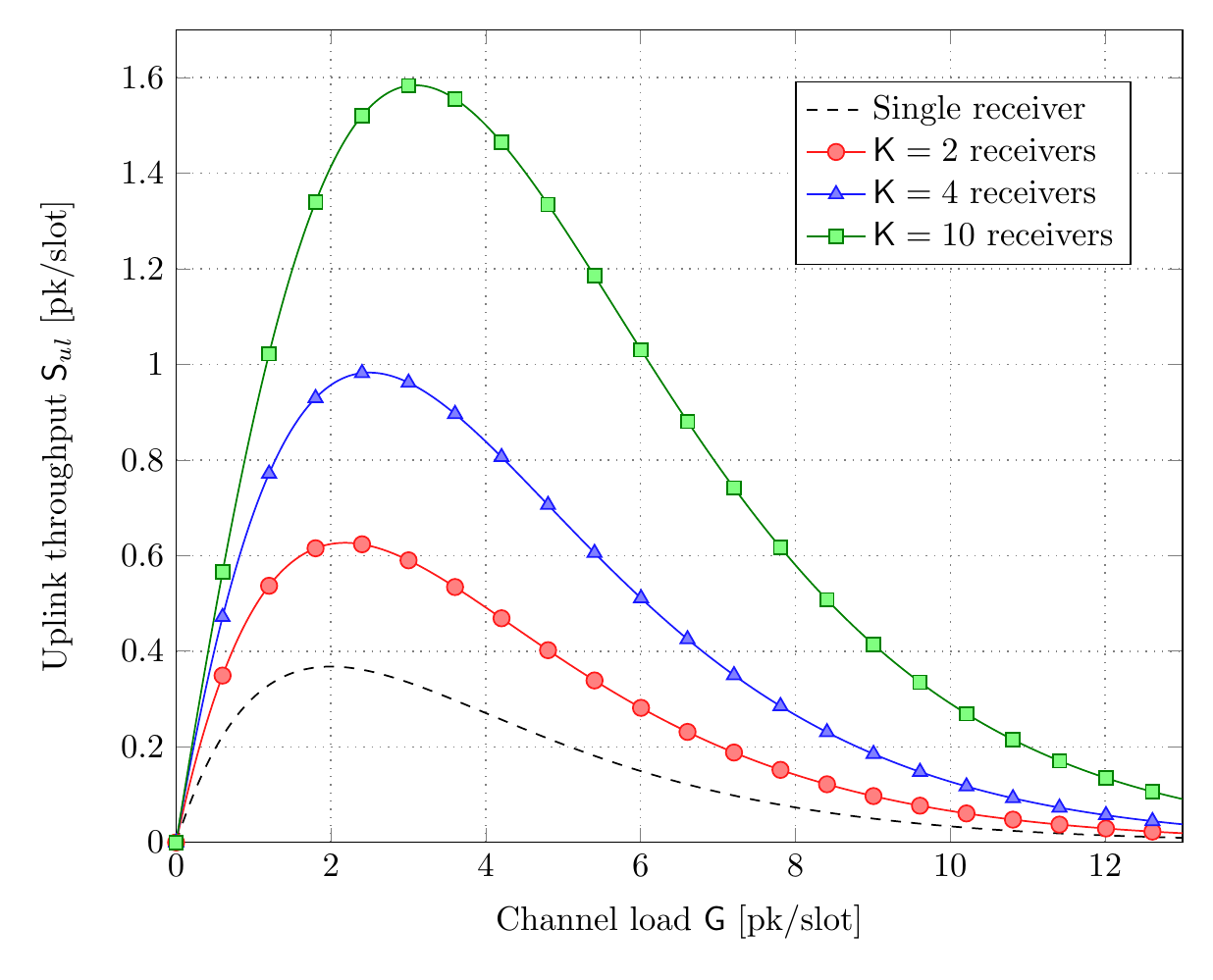}
\caption{Average uplink throughput vs channel load for different number of relays $\nrx$. The erasure probability has been set to $\peras=0.5$.}
\label{fig:truputUplink5rec}
\end{figure}

The performance achievable by increasing the number of relays is reported against the channel load in Figure~\ref{fig:truputUplink5rec} for a reference erasure rate $\peras=0.5$. As expected, $\tpULk{\nrx}$ benefits from a higher degree of spatial diversity, showing how the system can collect more than one packet per uplink slot as soon as more than four receivers are available, for the parameters under consideration. Such a result stems from two main factors. On the one hand, increasing $\nrx$ enables larger peak throughput over a single slot, as up to $\nrx$ different data units can be simultaneously retrieved. On the other hand, broader receiver sets improve the probability of decoding packets in the presence of collisions even when less than $\nrx$ users accessed the channel, by virtue of the different erasure patterns they experience.
The uplink throughput characterisation is complemented by Figure~\ref{fig:asymptoticULTru}, which reports the peak value for $\tpULkMax{\nrx}$ (solid red curve), obtained by properly setting the channel load to $\loadMax$ (whose values are shown by the solid blue curve), for an increasing relay population.\footnote{As discussed for the $\nrx=2$ case, a mathematical derivation of the optimal working point load $\loadMax$ is not straightforward, and simple numerical maximisation techniques were employed to obtain the results of Figure~\ref{fig:asymptoticULTru}.} The plot clearly highlights how the benefit brought by introducing an additional receiver to the scheme, quantified by equation~\eqref{eq:incrementalGain}, progressively reduces, leading to a growth rate for the achievable throughput that is less than linear and that exhibits a logarithmic-like trend in $\nrx$.
\begin{align}
\tpULk{\nrx} - \tpULk{\nrx-1} = \sum_{\rxEl=1}^\nrx (-1)^{\rxEl-1} {\nrx-1 \choose \rxEl-1} \load (1-\peras)^\rxEl e^{\load (1-\peras^\rxEl)}.
\label{eq:incrementalGain}
\end{align}
\begin{figure}
\centering
\includegraphics[width=0.8\columnwidth]{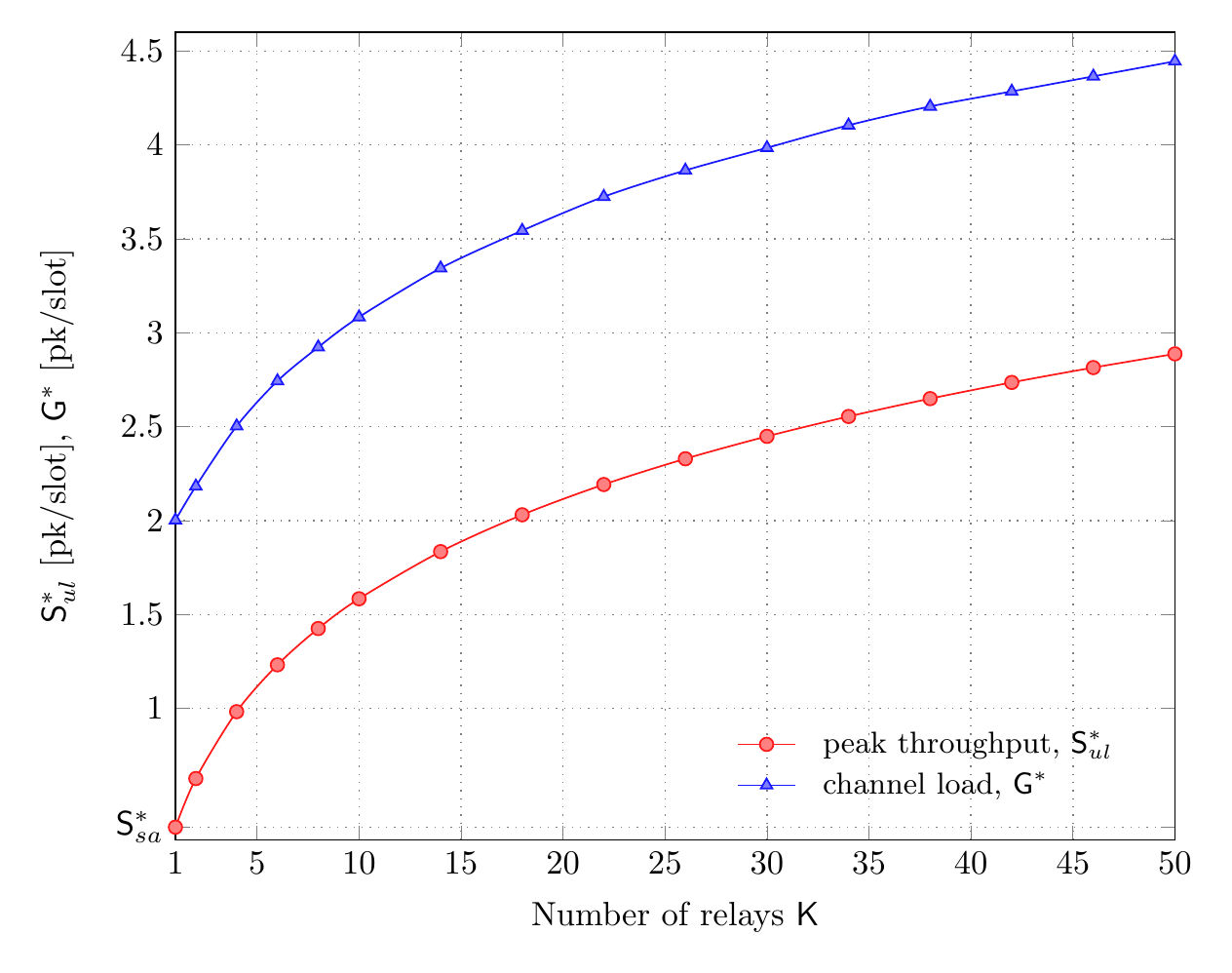}
\caption{Maximum achievable throughput $\tpULkMax{\nrx}$ as a function of the number of relays $\nrx$ for an erasure rate $\peras=0.5$. The gray curve reports the load on the channel $\loadMax_{\nrx}$ needed to reach $\tpULkMax{\nrx}$.}
\label{fig:asymptoticULTru}
\end{figure}

\subsection{Packet Loss Probability} \label{sec:PLR}
The aggregate throughput derived in Section~\ref{sec:uplinkThroughput} represents a metric of interest towards understanding the potential of \ac{SA} with diversity when aiming at reaping the most out of uplink bandwidth. On the other hand, operating an ALOHA-based system at the optimal load $\loadMax$ exposes each transmitted packet to a loss probability that may not be negligible. In the classical single-receiver case without fading, for instance, the probability for a data unit not to be collected evaluates to $1-e^{-1}\simeq 0.63$. From this standpoint, in fact, several applications may resort to a lightly loaded random access uplink, aiming at a higher level of delivery reliability rather than at a high throughput. This is the case, for example, of channels used for logon and control signalling in many practical wireless networks. In order to investigate how diversity can improve performance in this direction, we extend our framework by computing the probability $\ulrk{\nrx}$ that a user accessing the channel experiences a data loss, \ie that the information unit he sends is not collected, either due to fading or to collisions, by any of the $\nrx$ relays.

To this aim, let $\evNotRx$ describe the event that the packet of the observed user sent over time slot $\slot$ is not received by any of the receivers. Conditioning on the number of interferers $\intNum$, \ie data units that were concurrently present on the uplink channel at $\slot$, the sought probability can be written as:
\begin{align}
 \ulrk{\nrx} = \sum_{\intNum=0}^\infty \Pr[\evNotRx|\intRV=\intNum] \Pr[\intRV=\intNum].
\end{align}
Here, the conditional probability can easily be determined recalling that each of the $\nrx$ relays experiences an independent erasure pattern, obtaining $\Pr[\evNotRx|\intRV=\intNum] = (1-(1-\peras)\peras^{\intNum})^\nrx$ for an individual packet and for $\nrx$ relays with independent erasures on all individual links.
By resorting to the binomial theorem, such an expression can be conveniently reformulated as:
\begin{align}
 \Pr[\evNotRx|\intRV=\intNum] = \sum_{\rxEl=0}^\nrx (-1)^\rxEl {\nrx \choose \rxEl} \left((1-\peras)\peras^{\intNum}\right)^\rxEl.
\end{align}
On the other hand, the number of interferers seen by a user that accesses the channel at time $\slot$ still follows a Poisson distribution of intensity $\load$, so that, after simple calculations we finally get:
\begin{align}
 \ulrk{\nrx} = \sum_{\rxEl=0}^\nrx (-1)^\rxEl {\nrx \choose \rxEl} (1-\peras)^\rxEl e^{-\load(1-\peras^\rxEl)}.
 \label{eq:zeta_K}
\end{align}
\begin{figure}
\centering
\includegraphics[width=0.8\columnwidth]{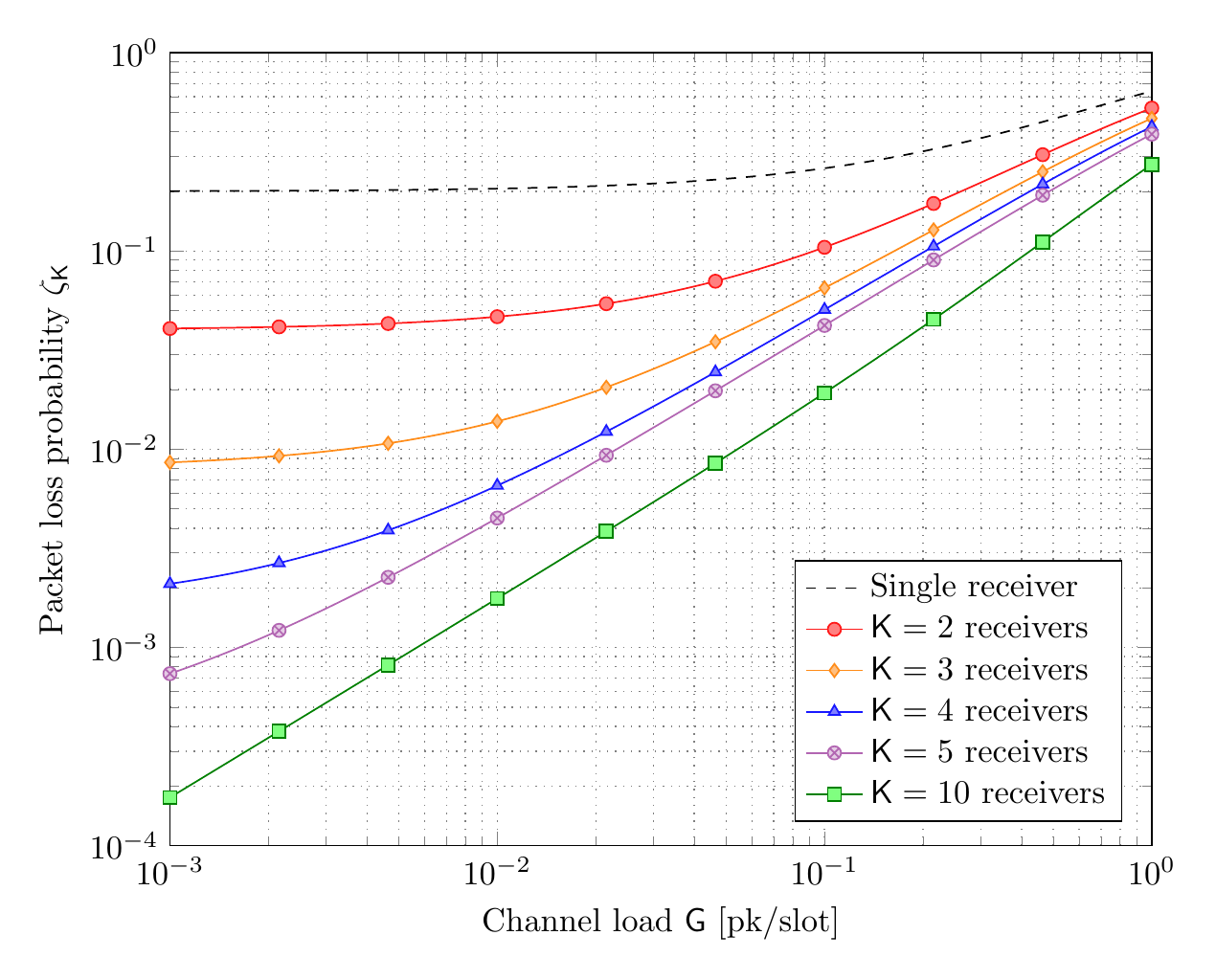}
\caption{Probability $\ulrk{\nrx}$ that a packet sent by a user is not received by any of the relays. Different curves indicate different values of $\nrx$, while the erasure probability has been set to $\peras=0.2$.}
\label{fig:erasure_eff}
\end{figure}

Figure~\ref{fig:erasure_eff} reports the behaviour of $\ulrk{\nrx}$ as a function of $\load$ when the erasure rate over a single link is set to $\peras = 0.2$. Different lines indicate the trend when increasing the number of receivers from $1$ to $10$. As expected, when $\load\rightarrow 0$, a user accessing the channel is not likely to experience any interference, so that failures can only be induced by erasures, leading to an overall loss probability of $\peras^\nrx$. In this perspective, the availability of multiple receivers triggers a dramatic improvement, enabling levels of reliability that would otherwise not be possible irrespective of the channel configuration. On the other hand, equation~\eqref{eq:zeta_K} turns out to be useful for system design, as it allows to determine the load that can be supported on the uplink channel while guaranteeing a target loss rate. Also in this case diversity can significantly ameliorate the performance. As shown in Figure~\ref{fig:erasure_eff}, for example, a target loss rate $\ulrk{\nrx}=5\cdot 10^{-2}$ is achieved by a three- and four-receiver scheme under $6-$ and $10-$fold larger loads compared to the $\nrx=2$ case, respectively.

\section{An Achievable Downlink Upper Bound} \label{sec:downlink_bound}

The analysis carried out in Section~\ref{sec:uplink} has characterised the average number of packets that can be decoded at the relay set when \ac{SA} is used in the uplink. We now consider the complementary task of delivering what has been collected to a central \ac{gateway}. In doing so, we aim at employing the minimum number of resources in terms of transmissions that have to be performed by the relays, while not allowing any information exchange among them. In particular, we consider a finite-capacity downlink, where the $\nrx$ receivers share a common bandwidth to communicate with the \ac{gateway} by means of a \ac{TDMA} scheme, and we assume that each of them can reliably deliver exactly one packet, possibly composed of a linear combination of what has been collected, over one time unit. It is clear that $\nrx\,\tpSA$ slots are sufficient to deliver all the collected data on average, but it can be inefficient since the same packet can be collected by more than one relay. We focus on a horizon of $\numulslot$ slots to operate the uplink, after which the downlink phase starts.

We structure our analysis in two parts. First, in Section~\ref{sec:bounds}, we derive lower bounds for the rates (in terms of downlink slots allocated per uplink slot) that have to be assigned to receivers in order to deliver the whole set of data units collected in the uplink over the $\numulslot$ slots. Then, Section~\ref{sec:nc} shows how a simple forwarding strategy based on random linear network coding suffices to achieve optimality, completing the downlink phase in $\tpULk{\nrx}$ slots for asymptotically large values of $\numulslot$.

\begin{sloppypar}
Prior to delving into the details, let us introduce some useful notation. We denote the $L$-bit data part of packets of the $j$-th user arriving in time slot $\slot$ as $\dataPart_\rxEl^{\slot} \in \dataSet$, with ${\dataSet = \mathbb{F}_{2^L} \cup e}$, where $e$ is added as the erasure symbol. We furthermore assume that the receiver can determine the corresponding user through a packet header, \ie the receiver knows both $j$ and $\slot$ after successful reception.
As the uplink operates over $\numulslot$ time slots, relay $\rxEl$ observes the vector $\obsPk_\rxEl = [\dataPart_\rxEl^1, \ldots, \dataPart_\rxEl^{\numulslot}]$. In each time slot, the tuple $(\dataPart_1^{\slot}, \dataPart_2^{\slot}, \ldots, \dataPart_\nrx^{\slot})$ is drawn from a joint probability distribution $\Prob_{\dataPart_1 \ldots \dataPart_\nrx}$ which is governed by the uplink, and different relays might receive the same packet.
\end{sloppypar}

\subsection{Bounds for Downlink Rates} \label{sec:bounds}
Each relay $\rxEl$ transmits a packet in each of its $\numulslot \rateDw_\rxEl$ downlink slots.
We are interested in the set of rates $(\rateDw_\rxEl)_{\rxEl=1}^\nrx$ such that the gateway can recover all packets (with high probability).

This is essentially the problem of distributed source coding (SW-Coding \cite{slepian1973noiseless}), with the following modification: SW-coding ensures that the gateway can recover all $\nrx$ observed strings $\obsPk_\rxEl$, $\rxEl=1,2,\ldots,\nrx$ perfectly.
In this setup, the gateway should be able to recover every packet that was received at any relay. However, is the gateway neither interested in erasures symbols at the relays, \ie whenever $\dataPart_\rxEl^{\slot}=e$ for any $\rxEl$, $\slot$, nor in reconstructing each relay sequence perfectly. The authors in \cite{dana2006capacity} overcame this problem by assuming that the decoder knows all the erasure positions of the whole network. This assumption applies in our case as packet numbers are supposed to be known via a packet header.
Let all erasure positions be represented by $\erPos$.

The rates $(\rateDw_1, \ldots, \rateDw_\nrx)$ are achievable \cite{slepian1973noiseless} if
\begin{align}
 \sum_{\rxEl \in \setRel} \rateDw_\rxEl \geq \entropy(\dataPart_{\setRel}|\dataPart_{\overline{\setRel}},\erPos), \quad \forall~ \setRel\subseteq [1,2,\ldots, \nrx]
\end{align}
where $ \dataPart_{\setRel} = (\dataPart_{\rxEl_1}^{\slot}, \dataPart_{\rxEl_2}^{\slot}, \ldots, \dataPart_{\rxEl_{|\setRel|}}^{\slot})$,  denotes the observations at some time $\slot$ at the subset of receivers specified by $\setRel = \{\rxEl_1, \rxEl_2, \ldots, \rxEl_{|\setRel|}\}$.
The quantity $\erPos$ has the effect of removing the influence of the erasure symbols on the conditional entropies. Computing the entropies however requires the full probability distribution $\Prob_{\dataPart_1 \ldots \dataPart_\nrx}$ which is a difficult task in general. By different means, we can obtain the equivalent conditions:
\vspace{1mm}
\begin{prop}
 The rates $(\rateDw_\rxEl)_{\rxEl=1}^\nrx$ have to satisfy
 \begin{align}
  \sum_{\rxEl \in \setRel} \rateDw_\rxEl \!\geq \tpULk{\nrx} \!+\!\! \sum_{\rxEl=1}^{\nrx-|\setRel|} (-1)^\rxEl {\nrx-|\setRel| \choose \rxEl} \load (1-\peras)^\rxEl e^{-\load (1-\peras^\rxEl)}, \quad \forall \setRel \subseteq \{1,\ldots,\nrx\}
 \end{align}
\end{prop}
\vspace{1mm}
\begin{proof}
Consider a subset of relays $\setRel \subseteq \{1,\ldots,\nrx\}$ and their buffer contents $\bigcup_{\rxEl\in \setRel}\setColPack_\rxEl^{\numulslot}$ after $\numulslot$ time slots. In order to satisfy successful recovery at the gateway, at least all packets that have been collected only by nodes in the set $\setRel$ and not by anyone else have to be communicated to the gateway. That is,
 \begin{align}
  \sum_{\rxEl \in \setRel} \numulslot \, \rateDw_\rxEl \geq \left| \bigcup_{\rxEl\in \setRel} \setColPack_\rxEl^{\numulslot} \backslash \bigcup_{\rxEl \in \overline{\setRel}} \setColPack_\rxEl^{\numulslot} \right|,
  \label{eq:ratebound1}
 \end{align}
 with $\overline{\setRel} = \{1,\ldots,\nrx\} \backslash \setRel$.
 Note that $$\bigcup_{\rxEl\in \setRel} \setColPack_\rxEl^{\numulslot} \backslash \bigcup_{\rxEl \in \overline{\setRel}} \setColPack_\rxEl^{\numulslot} = \setColPack^{\numulslot}\backslash \bigcup_{\rxEl \in \overline{\setRel}} \setColPack_\rxEl^{\numulslot},$$ so by the inclusion-exclusion principle and due to $|\overline{\setRel}| = \nrx - |\setRel|$
 \begin{align}
  \left|\bigcup_{\rxEl\in \setRel} \setColPack_\rxEl^{\numulslot} \backslash \bigcup_{\rxEl \in \overline{\setRel}} \setColPack_\rxEl^{\numulslot} \right| = |\setColPack_\nrx^{\numulslot}| + \sum_{\rxEl=1}^{\nrx-|\setRel|} (-1)^\rxEl {\nrx-|\setRel| \choose \rxEl} a_\rxEl^{\numulslot}
 \end{align}
with $\left|\setCommPack_{\setRel}^{\numulslot}\right| = a_\rxEl^{\numulslot}$ for $\rxEl = |\setRel|$ as before.
By plugging in the value for $\lim_{\numulslot\rightarrow\infty} \frac{a_\rxEl^{\numulslot}}{\numulslot}$, the proposition follows.
\end{proof}

\vspace{-1mm}
\subsection{Random Linear Coding} \label{sec:nc}
By means of Proposition 2, we have derived a characterisation of the rates that have to be assigned to relays in order to deliver the whole set of collected packets to the \ac{gateway}. In this Section, we complete the discussion by proposing a strategy that is capable of matching such conditions, thus achieving optimality. The solution that we employ is based on a straight-forward application of the well-known random linear coding scheme in \cite{ho2006random}, and will be therefore only briefly sketched in the following.

Each relay $\rxEl$ generates a matrix $\MxComb_\rxEl \in \mathbb F_{2^L}^{\numulslot \rateDw_\rxEl \times \numulslot}$ and obtains the data part of its $\numulslot \rateDw_\rxEl$ transmit packets by $\colVec_\rxEl^T = \MxComb_\rxEl \obsPk_\rxEl^T$. Whenever an element of $\obsPk_\rxEl$ is an erasure symbol, the corresponding column of $\MxComb_\rxEl$ is an all-zero column. Erasure symbols thus have no contributions to the transmit packets $\colVec_\rxEl$. All other elements of $\MxComb_\rxEl$  are drawn uniformly at random from $\mathbb F_{2^L}^*$, where $\mathbb F_{2^L}^*$ denotes the multiplicative group of $\mathbb F_{2^L}$.

The gateway collects all incoming packets and obtains the system of linear equations
\begin{align}
 \underbrace{\left( \begin{array}{c} \colVec_1^T\\ \colVec_2^T \\ \vdots \\ \colVec_\nrx^T \end{array} \right)}_{\colVec^T} =
 \underbrace{\left(\begin{array}{cccc} \MxComb_1 & 0 & \ldots & 0\\0 & \MxComb_2 & \ldots & 0 \\ 0 & 0 & \ddots & 0 \\ 0 & 0 & \ldots & \MxComb_\nrx \end{array}\right)}_{\MxComb}
 \underbrace{\left( \begin{array}{c} \obsPk_1^T\\ \obsPk_2^T \\ \vdots \\ \obsPk_\nrx^T \end{array} \right)}_{\obsPk^T}
\end{align}
where $\MxComb\in \mathbb F_{2^L}^{\numulslot\sum_\rxEl \rateDw_\rxEl \times \numulslot \nrx}$. Note that some elements of $\obsPk$ can be identical because they were received by more than one relay and thus are elements of some $\obsPk_{\rxEl_1}$, $\obsPk_{\rxEl_2}, \ldots$. One can merge these entries in $\obsPk$ that appear more than once. Additionally, we drop all erasure-symbols in $\obsPk$ and delete the corresponding columns in $\MxComb$ to obtain the reduced system of equations
\begin{align}
 \colVec^T = \MxCombTwo \obsPkTwo^T
\end{align}
where $\obsPkTwo \in \mathbb F_{2^L}^{|\setColPack^{\numulslot}|}$ contains only distinct received packets and no erasure symbols. Clearly, there are $\left| \bigcup_{\rxEl\in \setRel} \setColPack_\rxEl^{\numulslot} \right| $ elements in $\obsPkTwo$.

\begin{sloppypar}
We partition the entries in $\obsPkTwo$ into $2^\nrx-1$ vectors $\obsPkTwo_{\setRel}$ for each nonempty subset ${\setRel\subseteq \{1,2,\ldots,\nrx\}}$: Each vector $\obsPkTwo_{\setRel}$ contains all packets that have been received only by all relays specified by $\setRel$ and not by anyone else. That is, $\obsPkTwo_{\setRel}$ corresponds to the set ${\setUnique_{\setRel}^{\numulslot} = \bigcap_{\rxEl\in \setRel} \setColPack_\rxEl^{\numulslot} \backslash \bigcup_{\rxEl \in \overline{\setRel}} \setColPack_\rxEl^{\numulslot}}$, and its length is $|\setUnique_{\setRel}^{\numulslot}|$.
\end{sloppypar}

The columns in $\MxCombTwo$ and rows in $\obsPkTwo^T$ can be permuted such that one can write
\begin{align}
 \colVec_\rxEl^T = \MxCombTwo_\rxEl \obsPkTwo:= \sum_{\setRel\subseteq \{1,2,\ldots,\nrx\}} \MxCombTwo_{\rxEl,\setRel} \cdot \obsPkTwo_{\setRel}, \quad \forall~\rxEl=1,\ldots,\nrx.
\end{align}
Each of the matrices $\MxCombTwo_{\rxEl,\setRel} \in \mathbb F_{2^L}^{\numulslot \rateDw_\rxEl \times |\setUnique_{\setRel}^{\numulslot}|}$ contains only elements from $\mathbb F_{2^L}^*$ if $\rxEl\in \setRel$ and is an all-zero matrix otherwise. A compact representation for $\nrx=3$ is shown in \eqref{eq:exampleG_K=3} at the bottom of the page.

\begin{sloppypar}
The variables that are involved only in $\numulslot \sum_{\rxEl \in \setRel}\rateDw_\rxEl$ equations are those in ${\obsPkTwo_{\setVar},~ \setVar \subseteq \setRel}$, for each subset $\setRel\subseteq \{1,2,\ldots, \nrx\}$.
For decoding, the number of equations has to be larger or equal to the number of variables, so a necessary condition for decoding is that $\numulslot \sum_{\rxEl \in \setRel}\rateDw_\rxEl \geq \sum_{\setVar \subseteq \setRel}  \left|\setUnique_{\setVar}^{\numulslot}\right|$. This is satisfied by (\ref{eq:ratebound1}), since ${\sum_{\setVar \subseteq \setRel}  \left|\setUnique_{\setVar}^{\numulslot}\right| = \left|\bigcup_{\rxEl\in \setRel} \setColPack_\rxEl^{\numulslot} \backslash \bigcup_{\rxEl \in \overline{\setRel}} \setColPack_\rxEl^{\numulslot}\right|}$.
\end{sloppypar}

A sufficient condition is that the matrix $\MxCombTwo_\rxEl,~ \rxEl\in \setRel$,
representing $\numulslot \sum_{\rxEl \in \setRel}\rateDw_\rxEl$ equations, has rank $\sum_{\setVar \subseteq \setRel}  \left|\setUnique_{\setVar}^{\numulslot}\right|$
for each subset $\setRel\subseteq \{1,2,\ldots, \nrx\}$.
Denote the set of indices of nonzero columns of matrix $\MxCombTwo_\rxEl$ as the support of $\MxCombTwo_\rxEl$.
Note that a row of matrix $\MxCombTwo_\rxEl$ has a different support than a row of matrix $\MxCombTwo_l$, for $\rxEl\not = l$. These rows are thus linearly independent. It thus suffices to check that all rows of matrix $\MxCombTwo_\rxEl$ are linearly independent. As all nonzero elements are randomly drawn from $\mathbb F_{2^L}^*$, the probability of linear dependence goes to zero as $L$ grows large, completing our proof, and showing that the presented forwarding scheme achieves the bounds of Proposition 2.


\begin{figure*}[!bp]
\normalsize
\hrulefill
\begin{gather}
 \left( \begin{array}{c} \colVec_1^T\\ \colVec_2^T \\  \colVec_3^T \end{array} \right) =
  \left(\begin{array}{c}\MxCombTwo_1\\ \MxCombTwo_2 \\ \MxCombTwo_3
                            \end{array}\right)
                            \obsPkTwo^T
 =\\
 =
 \left(\begin{array}{ccc} \MxCombTwo_{1,\{1\}} & 0 & 0\\
                        0 & \MxCombTwo_{2,\{2\}} & 0\\
                        0 & 0 & \MxCombTwo_{3,\{3\}}\\
                        \MxCombTwo_{1,\{1,2\}} & \MxCombTwo_{2,\{1,2\}} & 0\\
                        \MxCombTwo_{1,\{1,3\}} & 0 & \MxCombTwo_{3,\{1,3\}}\\
                        0 & \MxCombTwo_{2,\{2,3\}} & \MxCombTwo_{3,\{2,3\}}\\
                        \MxCombTwo_{1,\{1,2,3\}} & \MxCombTwo_{2,\{1,2,3\}} & \MxCombTwo_{3,\{1,2,3\}}\\
                            \end{array}\right)^T
 \left( \begin{array}{l} \obsPkTwo_{\{1\}}^T\\ \obsPkTwo_{\{2\}}^T \\ \obsPkTwo_{\{3\}}^T \\ \obsPkTwo_{\{1,2\}}^T \\ \obsPkTwo_{\{1,3\}}^T \\ \obsPkTwo_{\{2,3\}}^T \\ \obsPkTwo_{\{1,2,3\}}^T
 \end{array} \right)
 \label{eq:exampleG_K=3}
\end{gather}
\end{figure*} 
\section{Simplified Downlink Strategies} \label{sec:capacityRegion}

The analysis developed in Section~\ref{sec:downlink_bound} showed how a strategy based on RLNC is capable of delivering to the gateway the whole information collected over the uplink resorting to the minimum amount of resources. When brought to implementation, however, network-coded schemes incur in drawbacks that partly counter-balance such benefits. On the one hand, an increased complexity is triggered both at the relays and at the collecting node to process linear combinations of data units. Moreover, an efficiency cost in terms of bandwidth is undergone to notify the gateway of the coefficients employed to encode the transmitted  packets, all the more so when the uplink is observed over a long time interval prior to triggering the network-coded phase. From this standpoint, the definition of alternative and simpler downlink strategies becomes relevant to unleash the potential of receiver diversity in practical settings.

Let us focus in particular on a scenario where $\nrx=2$ relays communicate with the collecting unit as described in Section~\ref{sec:sysModel}, \ie being coordinated to share a finite bandwidth via \ac{TDMA} over an error-free channel with the gateway, while no information exchange among them is possible. Recalling that the uplink channel is fed via \ac{SA}, each relay can deliver on average to the gateway all its incoming traffic as soon as at least  $\tpSA$ transmission opportunities per uplink slot are allocated to it in the downlink. Conversely, if the available resources are not sufficient to forward the whole set of incoming data and no coding across packets is permitted, the relay has to selectively decide which units to place on the downlink channel.\footnote{While such a configuration is not particularly meaningful when a single relay is available, it becomes interesting when multiple receivers are considered. In fact, should each of them be allocated on average $\tpSA$ or more downlink transmission opportunities per uplink slot, the gateway would receive redundant information in the form of duplicate packets. A tradeoff between the amount of downlink resources and collected packets at the gateway then arises, leaving space for non-trivial optimisations.} This condition is epitomised by a policy in which, upon retrieval of a packet, the relay may either drop it or enqueue it for later transmission. In general, the decision can be made considering side-information, \eg the state of the uplink channel, leading to the following definition which will be used as reference in the remainder of our discussion:

\begin{defin}[Dropping policy]
\label{def:dropping_policy}
Let $\setInterest$ be the event set of a probability space of interest, and let $\eventInterest_j$, $j=1,\dots,\cardSpace+1$ be events in $\setInterest$ such that $\bigcup_{j=1}^{\cardSpace+1} \eventInterest_j = \setInterest$. Furthermore, let each relay $\rxEl$, $\rxEl\in\{1,2\}$ be associated with a vector $\dropVec_{\rxEl} = [\drop^{(1)}_\rxEl, \dots, \drop^{(\cardSpace+1)}_{\rxEl}]$ of cardinality $\cardSpace+1$, and such that $\drop^{(j)}_\rxEl\in[0,1] \,\, \forall j$. Following a \emph{dropping policy}, whenever a packet is decoded over the uplink, relay $\rxEl$ discards it with probability $1-\drop^{(j)}_\rxEl$ or enqueues it in a FIFO buffer with probability $\drop^{(j)}_\rxEl$, where the superscript $j$ indicates that event $\eventInterest_j$ was observed at the time of reception.
\end{defin}

This family of strategies do not entail any complexity in terms of packet-level coding, and represent a simple and viable alternative for the downlink. On the other hand, in contrast to NC-based solutions, dropping policies are inherently not able to ensure delivery of the whole retrieved information set as soon as non-null dropping probabilities are considered, since a data unit may be discarded by all the relays that successfully decoded it. In order to further investigate this tradeoff, we will present and discuss some variations of such an approach, optimising the probability of discarding a packet to maximise information delivered to the gateway for a given downlink rate. From this standpoint, while the provided definition is fairly general, out of practical considerations we will restrict ourselves to strategies that base the buffering decision on observations of the uplink channel, \eg on the presence or absence of interference affecting the retrieved packet.

Let us introduce some notation and specify the framework that will be later used throughout this Section. We model each relay as an infinite buffer which, under any dropping policy, experiences an arrival rate that is a fraction of $\tpSA$. Furthermore, we assume the downlink to be dimensioned for the system to be stable, \ie such that all enqueued packets can be eventually delivered to the gateway and, consistently with Section~\ref{sec:downlink_bound}, indicate with $\rateDw_\rxEl$ the average number of transmissions performed by relay $\rxEl$ in the downlink per uplink slot. Under these hypotheses, the aggregate rate $\rateDw = \sum_\rxEl \rateDw_\rxEl$ can be regarded as a key design parameter, as it characterises the minimum amount of downlink resources that have to be allocated to properly implement a dropping policy under consideration.\footnote{From a queueing theory viewpoint, the downlink is stable as soon as the actual rate $\rateDw_\rxEl^*$ offered to relay $\rxEl$ satisfies $\rateDw^*_\rxEl > \rateDw_\rxEl$, $\forall \rxEl$.}
In turn, we evaluate the performance of such strategies by means of the downlink throughput $\tpDL$, computed as the average number of packets collected at the gateway per each uplink slot. Clearly, this quantity is bounded from above by $\tpUL$, \ie the amount of information successfully gathered by the set of relays from the user population. Finally, we introduce the downlink capacity of the system as $\capDL(\rateDw) = \sup\{ \tpDL \,|\, \rateDw, \plr \}$, where $\plr$ indicates the probability that a packet collected by the set of relays is not retrieved at the gateway.\footnote{In the case of dropping policies, this may happen when all the relays that received the data unit decide not to enqueue it in their buffers for subsequent transmission.} Within this framework, the insights of Section~\ref{sec:bounds} can be revisited and summarised by the following result:
\begin{coroll}
The downlink capacity of the system satisfies:
\begin{equation*}
\capDL(\rateDw) = \left\{
\begin{aligned}
\rateDw& \quad \mbox{for} \quad \rateDw < \tpUL \\
\tpUL & \quad \mbox{for} \quad \rateDw \geq \tpUL \\
\end{aligned}
\right.
\end{equation*}
\label{th:capRegion}
\end{coroll}
\begin{proof}
For $\rateDw\geq \tpUL$, the result is simply a reformulation of the propositions of Section~\ref{sec:bounds} when $\nrx=2$. Conversely, let $\beta=\rateDw/\tpUL<1$, and assume that each relay drops a packet received over the uplink with probability $1-\beta$. It immediately follows that the average number of collected data units evaluates to $\beta\, \tpUL$, so that the downlink phase is equivalent to the one of a system serving an uplink throughput of $\rateDw$ packets per slot. The propositions of Section~\ref{sec:bounds} apply again to the scaled downlink, proving the result.
\end{proof}

Theorem~{\ref{th:capRegion}} characterises the capacity region of the complete topology under consideration, specifying for any uplink configuration, \ie for any $(\load,\peras)$ pair, the achievable throughput $\tpDL$ when a certain amount of resources is allocated to the downlink. In this perspective, thus, it offers an upper bound to the performance of any forwarding strategy, and allows us to evaluate the efficiency of the different dropping policies that will be introduced in the following.

\subsection{Common Transmission Probability}
As a starting point, consider the basic case in which both relays employ a common and single probability $1-\drop$ to drop packets irrespectively of any external observation (\ie $\cardSpace=0$). In this condition, the average traffic forwarded by each receiver to the gateway scales by a factor $\drop$ with respect to the incoming uplink flow, and the overall average amount of employed resources in the downlink evaluates to $\rateDw = 2\drop\,\tpSA$.
Due to the lack of coordination among receivers, not all the performed transmissions bring innovative information to the collecting unit, so that the downlink throughput exhibits a loss with respect to $\rateDw$. More precisely, by the combinatorial technique discussed in Section~\ref{sec:uplink}, the probability of having the same packet forwarded twice can be expressed as $\drop^2\,\user(1-\peras)^2 \peras^{2(\user-1)}$, under the assumption that $\user$ users concurrently accessed the uplink channel in the slot observed by the receivers.  Taking the expectation over the Poisson traffic distribution, we then get
\begin{equation}
\tpDL= \rateDw-\drop^2 \, \load(1-\peras)^2 e^{-\load(1-\peras^2)} = \rateDw - \rateDw^2\, \left(e^{\load(1-\peras)^2}/4\load \right),
\label{eq:tru_common_tx_prob}
\end{equation}
where the second equality directly follows from the expression of $\tpSA$ in (\refeq{eq:Tsa}).
The achievable performance is summarized in Figure~\ref{fig:capacityRegionTheoretical}, where the red slid line (labelled \emph{Common Prob.}) reports the dependency of $\tpDL$ on the overall downlink rate as per equation (\refeq{eq:tru_common_tx_prob}) in the exemplary case $\peras = 0.3$ and $\load=1/(1-\peras)$. The plot also reports the system capacity curve $\capDL(\rateDw)$, which divides the plane in two regions and highlights downlink throughput values that can be aimed for. Moreover, we show results for $\tpSA \leq \rateDw \leq 2\,\tpSA$. In fact, while for $\rateDw > 2\, \tpSA$ no dropping is required, and on average all the traffic can be delivered to the gateway, the investigated rates correspond to configurations in which a non-trivial optimisation can be performed to leverage receiver diversity.
It is apparent that the simple dropping scheme under analysis exhibits quite a large performance gap with respect to the upper bound.  This inefficiency is further stressed by the fact that there exists a threshold rate $\rateDw^* > \tpSA$ such that for downlink rates lower than $\rateDw^*$, it holds $\tpDL(\rateDw) < \tpSA$.\footnote{By simple manipulations of (\refeq{eq:tru_common_tx_prob}), $\rateDw^*=2\load\,e^{-\load (1-\peras)^2}\, \left( 1 - \sqrt{1 - (1-\peras) \,e^{-\load \peras (1-\peras)}} \right)$.} In other words, in the operating region  $\tpSA<\rateDw< \rateDw^*$, a system enjoying receiver diversity in the uplink is seen at the gateway as performing worse than a simpler configuration with a single relay with allocated fewer downlink resources.\footnote{Notice in fact that a single relay without any dropping policy would employ on average $\tpSA$ downlink transmissions to deliver to the gateway the whole information set it collects in the uplink.}
This comes as no surprise, since not only does the considered policy resort to no form of coding, but also, in the absence of coordination among relays, it makes no attempt to reduce the likelihood of forwarding duplicate packets.
In order to overcome these limitations, we propose in the following the use of distinct dropping probabilities for the relays and a careful tuning of their value exploiting available side information.

\begin{figure}
\centering
\includegraphics[width=0.8\columnwidth]{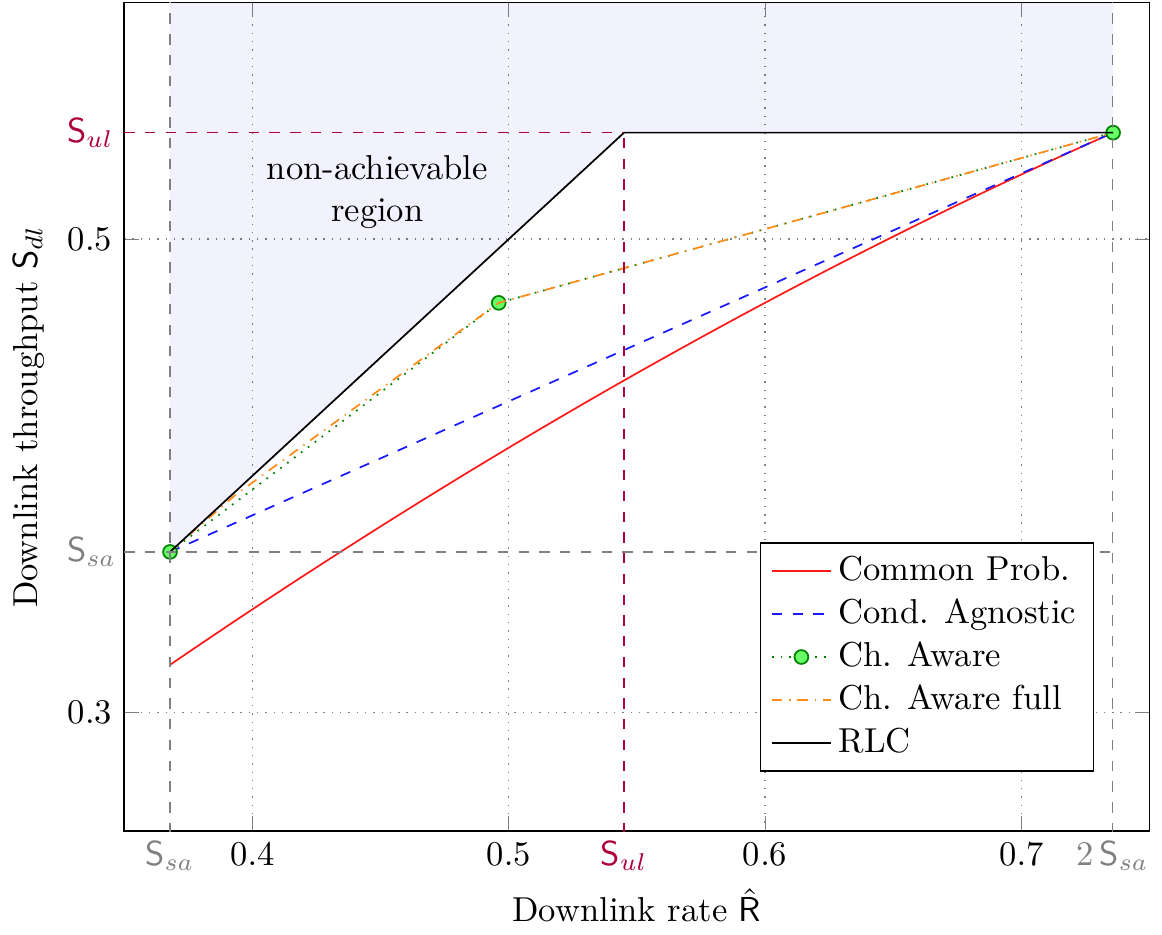}
\caption{Downlink throughput vs rate for different implementation of a dropping policy. $\peras=0.3$, $\load=1/(1-\peras)\sim1.43$.}
\label{fig:capacityRegionTheoretical}
\end{figure}

\subsection{Distinct Transmission Probabilities}

\subsubsection{Uplink-Condition Agnostic Policy}
A straightforward modification of the \emph{balanced} policy consists in letting receivers drop incoming packets with different probabilities. We still assume that buffering decisions are made at a receiver blindly, \ie they are not based on any observation of the uplink channel ($\cardSpace=0$). For this \emph{uplink-agnostic} strategy, the downlink phase is then completely specified by the pair $(\drop_1,\,\drop_2)$.\footnote{We omit the superscripts in $\drop_1^{(1)}$ and $\drop_2^{(1)}$ for the sake of notation readability.}  In this case,  the average number of transmissions in the downlink evaluates to $\rateDw = (\drop_1 + \drop_2)\,\tpSA$, and the throughput can  readily be expressed following combinatorial arguments as:
\begin{equation}
\tpDL = \rateDw - \drop_1 \drop_2 \, \load(1-\peras)^2 e^{-\load(1-\peras^2)}.
\label{eq:tru_channelAgnostic}
\end{equation}
The introduction of distinct dropping probabilities, thus, paves the road for a system optimisation, tuning $\drop_1$ and $\drop_2$ so as to maximise the downlink throughput while satisfying the constraint $\rateDw = (\drop_1 + \drop_2)\,\tpSA$ on the downlink rate, \ie on the minimum dimensioning in terms of resources allocated for transmission to the gateway.
From this standpoint, it is easy to observe that, for any $\rateDw$, the maximum throughput is achieved by minimising the loss factor expressed by the second addend in (\refeq{eq:tru_channelAgnostic}). In other words, given an uplink configuration in terms of $(\load,\peras)$, we are interested in finding the $(\drop_1,\,\drop_2)$ pair that minimises the product $\drop_1 \drop_2$ under the constraint $\drop_1+\drop_2 \leq \rateDw/\tpSA$. The general solution to this problem is provided by the following lemma.
\lemma \label{lemma:sim_pol} For any $\alpha \in [0,2]$, let $\hat{x}$ and $\hat{y}$ be real numbers such that $\hat{x},\hat{y} \in [0,1]$ and $\hat{x}+\hat{y}= \alpha$.
Then, the $\{\hat{x},\hat{y}\}$ pairs that minimise the product $\hat{z}= \hat{x}\hat{y}$ are given by
\begin{align}
\textrm{for} \quad \alpha\in[0,1]:& \quad \{\hat{x}=0, \hat{y}= \alpha \} \quad \textrm{or,} \quad \{\hat{x}=\alpha, \hat{y}=0\} \\
\textrm{for} \quad \alpha\in(1,2]:& \quad \{\hat{x}=1, \hat{y}= \alpha-1\} \quad \textrm{or,} \quad \{\hat{x}=\alpha-1, \hat{y}=1\} \\
\end{align}
\begin{proof}
\begin{sloppypar}
Writing $\hat{z}$ as a function of $\hat{x}$, we get $\hat{z}=-\hat{x}^2+\alpha \hat{x}$, which represents a parabola with downward concavity and zeroes for $\hat{x}=0$ and $\hat{x}=\alpha$. Furthermore, imposing the conditions $\hat{x}\in[0,1]$ and $\alpha-\hat{x}=\hat{y}\in[0,1]$, the region of interest restricts to ${\max\{0,\alpha-1\}\leq \hat{x}\leq \min \{ \alpha, 1 \}}$. When $\alpha\leq 1$, the minimum value of $\hat{z}$ in the studied domain is thus reached either when $\hat{x}=0$ or $\hat{x}=\alpha$, while, for $\alpha\in(1,2]$ the minimum value is obtained for $\hat{x}=\alpha-1$ and $\hat{x}=1$. The values of $\hat{y}$ follow immediately.
\end{sloppypar}
\end{proof}

In this way, the optimal setup for an uplink-agnostic policy is given by picking $\drop_1=1$ and $\drop_2=\rateDw/\tpSA-1$.\footnote{Clearly, the symmetric configuration $\drop_1=\rateDw/\tpSA-1$, $\drop_2=1$  yields the same result. This remark holds for all the considered policies, and will not be further stressed.} Under these conditions, the maximum throughput evaluates to
\begin{equation}
\tpDL = \rateDw \left ( \,1-(1-\peras)\,e^{-\load \peras (1-\peras)}\, \right ) + \load(1-\peras)^2 e^{-\load(1-\peras^2)},
\label{eq:maxTru_channelAgnostic}
\end{equation}
reported as a function of $\rateDw$ by the dashed blue line in Figure~\ref{fig:capacityRegionTheoretical}. It is clear that the linear trend exhibited by the presented solution triggers remarkable improvement over the \emph{balanced} approach, especially when little resources are available in the downlink. Moreover, straightforward manipulations of (\refeq{eq:maxTru_channelAgnostic}) show how allowing the relays to employ different dropping probabilities is sufficient for the system to always outperform a single-receiver configuration in the downlink rate region of interest (\ie $\tpDL \geq \tpSA$), thus overcoming the inefficiency that beset the baseline scheme. Such a result is achieved by simply letting one of the two receivers always forward all its incoming packets, allocating to it (on average) $\tpSA$ downlink transmissions per uplink slot, while resorting to the contributions of its fellow relay only when additional resources are available. In this perspective, the policy under discussion also represents a smart way to seamlessly take advantage of diversity in already deployed systems. In fact, a scenario operated via a single relay can be upgraded by plugging in an additional receiver, and by incrementally allocating to it downlink bandwidth, without any change to the forwarding policy of the original relay.

\subsubsection{Channel-Aware Policies}

A further extension of the considered policies consists in having relays make an educated choice on whether to drop or enqueue a data unit based on the observation of what happened on the uplink channel. The intuition in fact suggests that it is more likely for a packet to be retrieved by both receivers if it was the only one sent over the slot of interest, whereas the presence of several information units in an uplink slot reduces the chance for one of them to be decoded twice. A reasonable approach, even in the absence of any communication among relays, would then then be to increase the dropping probability in the former case, and reduce it in the latter to prevent duplicate transmissions in the downlink.

More formally, let us first consider the simple case in which each collector distinguishes between two situations, basing its decision on whether the retrieved packet was the only one on the channel (\ie no interference) or whether more than one users accessed the uplink channel over that slot  (\ie presence of interference). In the former condition, receiver $i$ drops with probability $\drop_i^{(1)}$, while in the latter case the packet is discarded with probability $\drop_i^{(2)}$. While this decision rule may appear elusive under the abstraction of on-off fading, as no incoming power is detected when a packet is erased, it turns out to be of high relevance in practical systems where a data unit can often be decoded even in the presence of a certain level of interference \eg by leveraging the capture effect \cite{Munari15_Massap}.\footnote{From this viewpoint, we also remark that detecting the presence or absence of interference can be rather easily accomplished by evaluating the noise level that affects the reception of a decoded data packet.}

\begin{sloppypar}
In accordance with the introduced framework, we can express the average number of downlink transmissions per uplink slot as
\begin{equation}
\rateDw = \left(\drop_1^{(1)} + \drop_2^{(1)}\right) (1-\peras) \load e^{-\load} + \sum\limits_{\user=2}^{\infty} \frac{\load^\user e^{-\load} }{\user!} \,\user\, (1-\peras) \peras^{\user-1} \left(\drop_1^{(2)} + \drop_2^{(2)}\right)\,,
\end{equation}
where the first member accounts for the case in which no interference was detected, whereas the summation considers all the situations in which at least one interfering packet affected the reception of the data unit of interest. Similarly, we can compute the probability that a data unit is forwarded twice towards the gateway, so that the overall throughput evaluates to
\begin{equation}
\tpDL = \rateDw - \drop_1^{(1)} \drop_2^{(1)} \load (1-\peras)^2  e^{-\load} - \sum\limits_{\user=2}^{\infty} \frac{\load^\user e^{-\load} }{\user!} \,\user\, (1-\peras)^2 \peras^{2(\user-1)} \drop_1^{(2)} \drop_2^{(2)}\,.
\end{equation}
Simple  manipulations  lead to the tackled optimisation problem:
\begin{align}
\begin{aligned}
\textrm{maximise} \quad &\tpDL = \rateDw - \coeffc\, \drop_1^{(1)} \drop_2^{(1)}- \coeffd \, \drop_1^{(2)} \drop_2^{(2)} \\
\textrm{s.t.}  \quad &\rateDw = \coeffa \left( \drop_1^{(1)} + \drop_2^{(1)} \right) + \coeffb \left(\drop_1^{(2)}+\drop_2^{(2) }\right)
\end{aligned}
\label{eq:opt_problem_interfAware}
\end{align}
with coefficients $\coeffa = \load (1-\peras ) e^{-\load}$, $\coeffb = \load (1-\peras ) e^{-\load} \left( e^{\load \peras} - 1 \right)$, $\coeffc = \load (1-\peras )^2 e^{-\load}$, and ${\coeffd = \load (1-\peras )^2 e^{-\load} \left( e^{\load \peras^2} - 1 \right)}$.
The solution is offered by the following result:
\prop \label{prop:sim_pol} Let $\rateDwSet_{1} = \left\{ \rateDw \,|\, \rateDw \in [ \tpSA, \coeffa+2\coeffb) \right\}$ and $\rateDwSet_{2} = \left\{ \rateDw \,|\, \rateDw \in [ \coeffa+2\coeffb, 2 \tpSA] \right\}$. Then, the maximum downlink throughput achievable with the proposed interference-aware policy is given by
\begin{equation}
\tpDL =
\begin{cases}
\rateDw \, (1 - \coeffd/\coeffb) + \coeffd \tpSA/\coeffb &\text{for } \rateDw \in \rateDwSet_1\\
\rateDw \, (1 - \coeffc/\coeffa)  + \coeffc (\tpSA + \coeffb )/\coeffa - \coeffd  &\text{for } \rateDw \in \rateDwSet_2
\end{cases}
\label{eq:opt_interfAware}
\end{equation}
As a preliminary remark, note that all the coefficients $\coeffa,\coeffb,\coeffc,\coeffd$ in (\refeq{eq:opt_problem_interfAware}) are strictly positive, and that, for any admissible values of $\load$ and $\peras$, the inequality $\coeffc/\coeffa > \coeffd/\coeffb$ holds. For the sake of a simplified notation, let $\drop_1^{(1)} = \hat{x}_1$, $\drop_2^{(1)} = \hat{x}_2$, $\drop_1^{(2)} = \hat{y}_1$, and $\drop_2^{(2)} = \hat{y}_2$. Furthermore, let us introduce $\alpha = \hat{x}_1+\hat{y}_1$ and $\beta = \hat{x}_2 + \hat{y}_2$, with $\alpha,\, \beta \in [0,2]$, as well as the auxiliary function $f(\bm{\hat{x}}) = \coeffc\, (\hat{x}_1 \hat{y}_1) + \coeffd\,(\hat{x}_2 \hat{y}_2)$, where $\bm{\hat{x}} =[\hat{x}_1,\hat{x}_2,\hat{y}_1,\hat{y}_2]$ and  $f(\bm{\hat{x}}) \geq 0$. We are then interested in maximising $\tp = \rateDw - f(\bm{\hat{x}})$ subject to $\rateDw = \coeffa\alpha + \coeffb \beta$ or, equivalently in minimising $f(\bm{\hat{x}})$ under the same constraint. Let us notice that the first addend of $f(\bm{\hat{x}})$ only contains the variables that determine $\alpha$, while the second addend of $f(\bm{\hat{x}})$ solely defines the value of $\beta$. It is then possible to solve the optimisation problem by considering four non-overlapping regions: $\load_1 = \{\bm{\hat{x}} \, | \, \alpha \in [0,1], \beta \in [0,1] \}$, $\load_2 = \{\bm{\hat{x}} \, | \, \alpha \in (1,2], \beta \in [0,1] \}$, $\load_3 = \{\bm{\hat{x}} \, | \, \alpha \in [0,1], \beta \in (1,2] \}$, $\load_4 = \{\bm{\hat{x}} \, | \, \alpha \in (1,2], \beta \in (1,2] \}$. $\load_1$ can immediately be discarded as $\alpha,\beta < 1$ imply ${\rateDw < \coeffa+\coeffb = \tpSA}$, identifying a condition which is not of interest. In the remaining regions, for any $\alpha$ and $\beta$ the  values of the optimisation variables that maximise the throughput can be found resorting to Lemma~\ref{lemma:sim_pol}. In particular:
\begin{itemize}
\item for $\bm{\hat{x}} \in \load_2$: $\alpha \in [0,1]$ implies $\hat{x}_1=1$ and $\hat{y}_1=\alpha-1$, while $\beta \in (1,2]$ implies $\hat{x}_2=\beta$ and $\hat{y}_2=0$. By the last condition we can write $f(\bm{\hat{x}}) = \coeffc(\alpha-1)$ so that the optimum lies in the $(\alpha,\beta$) pair that satisfies the constraint on $\rateDw$ with minimum $\alpha$. The solution follows as $\beta=1$, $\alpha= (\rateDw-\coeffb)/\coeffa$ with a  corresponding throughput $\tp=\rateDw-(\coeffc/\coeffa) (\rateDw-\tpSA)$.
\item by a symmetrical reasoning, for $\bm{\hat{x}} \in \load_3$ the optimal solution is given by $\alpha=1$ and $\beta= (\rateDw-\coeffa)/\coeffb$, with a throughput $\tp=\rateDw-(\coeffd/\coeffb) (\rateDw-\tpSA)$ achieved for $\hat{x}_1=1$, $\hat{y}_1=0$, $\hat{x}_2=1$ and $\hat{y}_2=(\rateDw-\tpSA)/\coeffb$.
\item for $\bm{\hat{x}} \in \load_4$: by Lemma~\ref{lemma:sim_pol}, $\hat{x}_1=1$, $\hat{x}_2=1$, so that $f(\bm{\hat{x}})=\coeffc(\alpha-1)+\coeffd(\beta-1)$. Recalling that $\beta=(\rateDw-\coeffa\alpha)/\coeffb$, we can then write $f(\bm{\hat{x}})=\alpha(\coeffc-\coeffa \coeffd/\coeffb) - \coeffc - \coeffd - \rateDw \coeffd/\coeffb$, which represents a straight line with positive slope and minimum in the left extremal point of the $\alpha$ domain. Imposing $\beta \in (1,2]$, the support of interest follows as: ${\max \{ 1, (\rateDw-2 \coeffb)/\coeffa\} \leq \alpha \leq \min \{ 2, (\rateDw-\coeffb)/\coeffa\}}$. Two cases have then to be distinguished. When $(\rateDw-2\coeffb)/\coeffa < 1$, $\alpha=1$ and the problem collapses to the solution found for region $\load_3$. Conversely, when $\rateDw \geq \coeffa+2\coeffb$, the optimum is achieved for $\alpha=(\rateDw-2\coeffb)/\coeffa$ and $\beta=2$, for a throughput $\tp=\rateDw-(\coeffc/\coeffa)(\rateDw-\tpSA-\coeffb)-\coeffd$ with $\hat{x}_1=1$, $\hat{y}_1=(\rateDw-\tpSA-\coeffb)/\coeffa$, $\hat{x}_2=1$, $\hat{y}_2=1$.
\end{itemize}
Comparing the  throughputs of the different configurations and taking advantage of the inequality $\coeffc/\coeffa > \coeffd/\coeffb$ it is immediate to verify that the optimal solution is to pick $\bm{\hat{x}} \in \load_3$ for $\rateDw\in [ \tpSA, \coeffa+2\coeffb )$ and $\bm{\hat{x}} \in \load_4$ for $\rateDw\in [ \coeffa+2\coeffb, 2\tpSA ]$, stating the result of the proposition.  \hfill $\blacksquare$
\end{sloppypar}

\begin{sloppypar}
The solution is obtained by setting the dropping probabilities for relay $\rxEl=1$ as ${[\,\drop_1^{(1)}=1, \,\drop_1^{(2)}=1\,]}$ and for relay $\rxEl=2$ as $[\,\drop_2^{(1)}=0, \,\drop_2^{(2)}=(\rateDw -\tpSA)/\coeffb\,]$ when $\rateDw \in \rateDwSet_1$.  Conversely, for $\rateDw \in \rateDwSet_2$, the  optimal working point is achieved for $[\,\drop_1^{(1)}=1, \,\drop_1^{(2)}=1\,]$, and ${[\,\drop_2^{(1)}=(\rateDw-\tpSA-\coeffb)/\coeffa, \,\drop_2^{(2)}=1\,]}$.
\end{sloppypar}

The performance offered by the considered scheme is once again reported in Figure~\ref{fig:capacityRegionTheoretical} (green dotted line). As per (\refeq{eq:opt_interfAware}), the plot highlights two regions both exhibiting a linear dependence of the achievable downlink throughput on $\rateDw$. In particular, when the overall downlink rate is lower that $\coeffa+2\coeffb$, the optimal allocation consists in having one of the relay, say $\rxEl=1$, forwarding all the incoming traffic, while the other, say $\rxEl=2$, only forwarding packets that were decoded in the uplink in the presence of interference. This confirms the intuition that data units collected when the \ac{SA} channel was accessed by more than one user bring a higher reward in terms of throughput when forwarded to the gateway. On the other hand, when enough resources in the downlink are available for $\rxEl=2$ to deliver all such packets, the policy naturally enables it to gradually enqueue and transmit also information units collected in the absence of interference.  The higher probability for them to be duplicates of what is forwarded by $\rxEl=1$ is reflected in the lower slope of the throughput  curve in the rightmost region. The figure clearly stresses the remarkable improvement unleashed by taking into account even partial information on the state of the uplink channel, proving how simple strategies can indeed provide performance that are not too far from the bound represented by \ac{RLC}. From this standpoint, two remarks are in order. In the first place, we notice that the switching point between $\rateDwSet_1$ and $\rateDwSet_2$ can be expressed as $\coeffa+2\coeffb=2\tpSA - \load(1-\peras)e^{-\load}$. For a given congestion level $\load$ in the uplink, then, higher erasure rates result in an extension of the region with higher throughput slope, further reducing the gap of the interference-aware dropping policy with respect to the upper bound. This trend is consistent with the stronger uncorrelation induced by larger values of $\peras$ over the sets of collected packets at the two relays and hints at how the additional diversity likely to characterise uplink channels in real systems (see, \eg \cite{Munari15_Massap}) may further benefits the class of proposed downlink strategies.
Moreover, while interference detection represents a simple and practically viable basis to tune the dropping probabilities, the question on how more detailed side-information would impact the performance naturally arises.

In order to investigate this aspect, let us focus on the ideal case in which, at every given time slot, both relays seamlessly and perfectly know how many packets were concurrently transmitted over the uplink. Under this hypothesis, we consider a dropping policy in which each relay has $\cardSpace+1$ transmission probabilities. More specifically, collector $\rxEl$ enqueues a decoded packet with probability $\drop_\rxEl=1^{(j)}$, if $j = 1,\dots,\,\cardSpace$ users accessed the uplink in the slot under consideration (\ie $j-1$ packets were erased, given  the model under analysis). Conversely, if more than $\cardSpace$ data units populate the slot, the retrieved packet is buffered for later forwarding with probability $\drop_\rxEl=1^{(\cardSpace+1)}$. The combinatorial approach followed so far can be employed in this case as well to evaluate the downlink throughput and the average number of transmissions towards the gateway. After some simple yet tedious calculations due to the larger number of cases that need to be considered, we can express the sought quantities as:
\begin{equation}
\begin{cases}
\rateDw  &=  \sum\limits_{j=1}^{\cardSpace+1} \coeffa_j \left( \drop_1^{(j)} + \drop_2^{(j)} \right)\\
\tpDL &= \rateDw-  \sum\limits_{j=1}^{\cardSpace+1}  \coeffb_j \, \drop_1^{(j)} \drop_2^{(j)} \\
\end{cases}
\label{eq:tru_channelAware}
\end{equation}
where the coefficients $\coeffa_j$ and $\coeffb_j$, $j=1,\dots, \cardSpace+1$ are reported in equation (\refeq{eq:coeffs_channelAware}) and  $\Gamma(s,x)$ is the incomplete superior gamma function, defined as $\Gamma(s,x)=\int_x^{+\infty}t^{s-1}\, e^{-t} \, dt$.

\begin{EqAtPageBottom}
\hrulefill
\begin{align}
\begin{split}
\coeffa_i &= \load(1-\peras) e^{-\load} \, \frac{(\load \peras)^{i-1}}{(i-1)!}, \quad 1\leq i \leq \cardSpace  \\
\coeffa_{\cardSpace+1} &= \load(1-\peras) e^{-\load} \, \sum_{i=\cardSpace+1}^{\infty}
\frac{(\load \peras)^{i-1}}{(i-1)!} = \tpSA \left( 1-\frac{\Gamma(\cardSpace,\load \peras)}{\Gamma(\cardSpace)}\right) \\
\coeffb_i &= \load(1-\peras)^2 e^{-\load} \,  \frac{(\load \peras^2)^{i-1}}{(i-1)!} , \quad 1\leq i \leq \cardSpace\\
\coeffb_{\cardSpace+1} &= \load(1-\peras)^2 e^{-\load} \, \sum\limits_{i=\cardSpace+1}^{\infty}
\frac{(\load \peras^2)^{i-1}}{(i-1)!} = \load (1-\peras)^2 e^{-\load(1-\peras^2)}\left( 1-\frac{\Gamma(\cardSpace,\load \peras^2)}{\Gamma(\cardSpace)}\right) \\
\end{split}
\label{eq:coeffs_channelAware}
\end{align}
\end{EqAtPageBottom}

Starting from (\refeq{eq:tru_channelAware}), an optimisation problem analogous to the one in (\refeq{eq:opt_problem_interfAware}) can be stated, aiming at the dropping probabilities that maximise the downlink throughput for a given rate $\rateDw$. For arbitrary and potentially large values of $\cardSpace$, however, an analytical solution of the problem can be elusive. We thus follow a different approach, and conjecture that the idea underpinning the optimal working point for $\cardSpace=1$ derived in Proposition \ref{prop:sim_pol} extends to any value of $\cardSpace$. More specifically, for the rate-region of interest, we let one of the relays, say $\rxEl=1$, always enqueue and forward all the received packets, \ie $\drop_1^{(j)}=1$, $\forall j \in[1,\cardSpace+1]$. On the other hand, when $\rateDw > \tpSA$, the second relay starts by only buffering the data units it receives in the uplink that are less likely to have also been decoded at $\rxEl=1$. This translates to a choice of $\drop_2^{(j)=0}$, $j=1,\dots,\cardSpace$, while $\drop_2^{(\cardSpace+1)}$ is linearly increased with the amount of resources available in the downlink. Further increasing $\rateDw$, $\drop_2^{(\cardSpace+1)}$ eventually saturates at one. The corresponding value of the downlink rate can easily be computed from (\refeq{eq:tru_channelAware})-(\refeq{eq:coeffs_channelAware}), by setting $\rateDw = \tpSA + \coeffa_{\cardSpace+1} = \tpSA \left( 2 - \frac{\Gamma(\cardSpace,\load \peras)}{\Gamma(\cardSpace)} \right)$. Here, the first addend accounts for the resources allocated to the first relay $\rxEl=1$ to forward to the gateway all the incoming uplink packets, whereas the second term is simply the rate given to $\rxEl=2$ when $\drop_2^{(\cardSpace+1)}=1$ and $\drop_2^{(j)=0}$, $j=1,\dots,\cardSpace$. After this point, additional resources allocated to the second relay will be used to store and forward packets received over slots accessed by $\cardSpace$ users in the uplink, \ie $\drop_2^{\cardSpace}>0$.  Along the same line of reasoning, $\drop_2^{\cardSpace}$ is linearly increased until it reaches one, \ie for $\rateDw = \tpSA + \coeffa_{\cardSpace+1} + \coeffa_\cardSpace$. Iterating this approach, $\cardSpace+1$ rate regions can be identified, where relay $\rxEl=2$ is progressively allowed to deliver information units which are more likely to be duplicates of what forwarded by its fellow relay. More formally, we provide the following result.

For the channel-aware dropping policy defined by (\refeq{eq:tru_channelAware}), let $\rateDwSet_i$, $i=1,\dots,\cardSpace+1$ be $\cardSpace+1$ downlink rate regions defined as:
\begin{align}
\begin{split}
&\rateDwSet_{\cardSpace+1} = \left\{ \rateDw \,|\, \rateDw \in [ \tpSA, \rateDw_{\cardSpace+1}) \right\}, \,\,  \rateDw^*_{\cardSpace+1} = \tpSA \left( 2 - \frac{\Gamma(\cardSpace,\load \peras)}{\Gamma(\cardSpace)} \right)  \\
&\rateDwSet_{i} = \left\{ \rateDw \,|\, \rateDw \in [ \rateDw^*_{i+1}, \rateDw^*_{i}) \right\}, \,\, \rateDw^*_i = \tpSA \left( 2  - \frac{\Gamma(i,\load \peras) - e^{-\load \peras} (\load \peras)^{i-1}}{\Gamma(i)} \right), \, i = \cardSpace,\dots,1
\end{split}
\label{eq:regions_channelAware}
\end{align}
We conjecture that the probability vectors $\dropVec_1$, $\dropVec_2$ maximising the achievable downlink throughput for any value of $\rateDw \in [\tpSA, 2 \tpSA]$ are given by:
\begin{equation}
\drop_1^{(i)} = 1, \, i=1,\dots,\cardSpace+1,  \quad
\drop_2^{(i)}=
\begin{cases}
1 \quad & i=\hat{k}+1,\dots,\cardSpace+1 \\
\frac{\rateDw-\tpSA - \sum_{j=\bar{i}+1}^{\cardSpace+1}\coeffa_j}{\coeffa_k} \quad & i=\hat{k}\\
 0 \quad & i=1,\dots,\hat{k}-1 \\
\end{cases}
\label{eq:optQ_channelAware}
\end{equation}
Under this choice, introducing the auxiliary variable $\coeffe_i = \coeffb_i/\coeffa_i$, the downlink throughput evaluates to:
\begin{equation}
\tpDL =
\begin{cases}
\left(1-\coeffe_{\cardSpace+1} \right) \rateDw + \tpSA \coeffe_{\cardSpace+1}  & \,  \rateDw \in \rateDwSet_{\cardSpace+1} \\[2mm]
\left( 1 - \coeffe_i \right) \rateDw + \left( 2 - \frac{\Gamma(i,\load \peras)}{\Gamma(i)} \right) \tpSA \coeffe_i
- \load (1-\peras)^2 e^{-\load(1-\peras^2)}\,\left( 1 - \frac{\Gamma(i,\load \peras^2)}{\Gamma(i)} \right)  & \, \rateDw \in \rateDwSet_{i\leq \cardSpace}
\end{cases}
\label{eq:opt_channelAware}
\end{equation}

The derivation of (\refeq{eq:opt_channelAware}) follows directly by some manipulation after plugging the probability values of (\refeq{eq:optQ_channelAware}) into (\refeq{eq:tru_channelAware}). Secondly, the accuracy of the conjecture has been verified by means of constrained numerical optimisation techniques applied to (\refeq{eq:tru_channelAware}) for a variety of uplink configurations $(\load,\peras)$, always obtaining values in excellent agreement with the presented analytical expressions.
Leaning on this result, we report in Figure~\ref{fig:capacityRegionTheoretical} the achievable throughput against the downlink rate when a very accurate knowledge of the uplink channel conditions in terms of size of the collision sets is available at relays, \ie $\cardSpace = 20$ (dash-dotted orange line). An accurate inspection of (\refeq{eq:regions_channelAware}) reveals how the starting point of the rightmost region $\rateDwSet_1$ (and the downlink throughput achieved therein) does not vary with $\cardSpace$. Increasing the size of the dropping probability vectors, thus, leads to a larger number of smaller partitions of the region $\tpSA \leq \rateDw \leq \rateDw_1^*$. On the other hand, such leftmost regions are precisely the ones characterised by a stronger slope of the throughput curve, earned leveraging additional side information. The combination of the two effects significantly curbs the benefits brought by a more accurate knowledge of the number of users accessing the uplink channel. This is clearly highlighted in the plot, where the $\cardSpace=20$ curve exhibits a trend which is very close to the one of its $\cardSpace=1$ counterpart, with a limited gain only in the downlink rate region which is in fact of less interest for multi-receiver systems (\ie when the total available rate is slightly larger than the one necessary to collect the traffic of a single relay). Such a result is remarkable, and suggests how a simple and practically viable strategy which makes forwarding decisions only based on interference detection can indeed reap a noticeable fraction of the downlink throughput  achievable by means of a large family dropping policies, offering performance not too far away from the ones of the \ac{RLC} upper bound.

\section{On the Impact of Finite-Buffer Size on Downlink Strategies}
\label{sec:buffer_downlink}

We observed in Section~\ref{sec:downlink_bound}, that \ac{RLC} is able to achieve the downlink rate upper bound. Nevertheless, \ac{RLC} requires an infinite observation window of the uplink to achieve such upper bound. In this Section, we derive the analytical model for the \ac{RLC} under a finite observation window.

We focus on the $\nrx=2$ relay case. In this scenario, the gateway collects the incoming packets forwarded by the relays building the system of linear equations
\begin{align}
\left( \begin{array}{c} \linCombK{1}^T\\ \linCombK{2}^T \end{array} \right) =
\MxComb \colVec^T =
\left(\begin{array}{ccc} \MxCombK{1,1} & \bm 0 &  \MxCombK{1,1\wedge2} \\
\bm 0 & \MxCombK{2,2} & \MxCombK{2,1\wedge2}
\end{array}\right)
\left( \begin{array}{l} \colVecK{1}^T\\
\colVecK{2}^T \\
\colVecK{1\wedge2}^T
\end{array} \right),
\label{eq:exampleG_K=2}
\end{align}
where $\MxCombK{1,1} \in \mathbb{F}_{\fieldOrd}^{\linllK{1} \times \collK{1}}$ is the matrix representing the coefficients applied to linear combinations of uplink decoded packets by the first relay only. The number of columns in $\MxCombK{1,1}$, \ie $\collK{1}$, corresponds to the number of uplink decoded packets by the first relay only while the number of rows, \ie $\linllK{1}$, corresponds to the number of linear combination of these packets generated and sent to the gateway by the first relay only. Similarly, $\MxCombK{2,2} \in \mathbb{F}_{\fieldOrd}^{\linllK{2} \times \collK{2}}$ is the matrix representing the coefficients applied to linear combinations of uplink decoded packets by the second relay only. Finally, $\MxCombK{1,1\wedge2} \in \mathbb{F}_{\fieldOrd}^{\linllK{1} \times \collK{1\wedge2}}$ and $\MxCombK{2,1\wedge2} \in \mathbb{F}_{\fieldOrd}^{\linllK{2} \times \collK{1\wedge2}}$ are the coefficients of the packets linear combinations generated respectively by the first and second relay involving uplink decoded packets by both relays. The vector $\colVec$ corresponds to the uplink decoded packets, where $\colVecK{1}$, $\colVecK{2}$ are the uplink decoded packets by the first relay only and second relay only respectively, and $\colVecK{1\wedge2}$ are the uplink decoded packets by both of them.

\begin{sloppypar}
The gateway resolves the system of equations applying Gauss-Jordan elimination exploiting the knowledge of the linear combinations coefficients. Two Markov chains are needed to track the rank of the sub-matrices and the number of correctly decoded uplink packets at the relays. We start applying Gauss-Jordan elimination to $\MxComb$ and we obtain the matrix $\MxComb'$ in the following form
\begin{align}
\MxComb' =
\left( \begin{array}{ccccc}
\bm{I}_{\rankK{1}} & \bm{A} & \bm{0} & \bm{0} & \bm{U}_1 \\
\bm{0} & \bm{0} & \bm{0} & \bm{0} & \bm{L}_1 \\
\hline
\bm{0} & \bm{0} & \bm{I}_{\rankK{2}} & \bm{B} & \bm{U}_2 \\
\bm{0} & \bm{0} & \bm{0} & \bm{0} & \bm{L}_2
\end{array}
\right).
\label{eq:Mx_prime}
\end{align}
The size of the two identity matrices are $\rankK{1}$ and $\rankK{2}$ and represent the rank of $\MxCombK{1,1}$ and $\MxCombK{2,2}$ respectively, while $\bm{A} \in \mathbb{F}_{\fieldOrd}^{\rankK{1}\times \mxColK{1}}$ and $\bm B \in \mathbb{F}_{\fieldOrd}^{\rankK{2} \times \mxColK{2}}$. The following relations hold
\begin{align}
\mxColK{1}=\collK{1}-\rankK{1}\\
\mxColK{2}=\collK{2}-\rankK{2}.
\end{align}
We reorder the rows inverting the second and third block in matrix $\MxComb'$. Applying Gauss-Jordan elimination on the sub-matrices $\left(\begin{smallmatrix}
\bm{L}_1\\ \bm{L}_2
\end{smallmatrix}\right)$ and we obtain $\MxComb''$ in the form
\begin{align}
\MxComb'' =
\left( \begin{array}{cccccc}
\bm{I}_{\rankK{1}} & \bm{A} & \bm{0} & \bm{0} & \bm{0} & \bm{U'}_1 \\
\bm{0} & \bm{0} & \bm{I}_{\rankK{2}} & \bm{B} & \bm{0} & \bm{U'}_2 \\
\bm{0} & \bm{0} & \bm{0} & \bm{0} &  \bm{I}_{\rankK{3}} & \bm{L'} \\
\bm{0} & \bm{0} & \bm{0} & \bm{0} & \bm{0} & \bm{0}
\end{array}
\right).
\label{eq:Mx_second_two}
\end{align}
In $\MxComb''$, $\bm{U'}_1 \in \mathbb{F}_{\fieldOrd}^{\rankK{1} \times \mxColK{3}}$, $\bm{U'}_2 \in \mathbb{F}_{\fieldOrd}^{\rankK{2} \times \mxColK{3}}$ and $\bm{L'} \in \mathbb{F}_{\fieldOrd}^{\rankK{3} \times \mxColK{3}}$, where $\rankK{3}\leq \left[(\linllK{1} + \linllK{2}) - (\rankK{1} + \rankK{2})\right]$. The following relation holds
\[
\mxColK{3}=\collK{1\wedge2}-\rankK{3}.
\]
We now rearrange the columns resulting in
\begin{align}
\MxComb'' =
\left( \begin{array}{ccc|ccc}
\bm{I}_{\rankK{1}} & \bm{0} & \bm{0} & \bm{A} & \bm{0} & \bm{U'}_1 \\
\bm{0} & \bm{I}_{\rankK{2}} & \bm{0} & \bm{0} & \bm{B} & \bm{U'}_2 \\
\bm{0} & \bm{0} & \bm{I}_{\rankK{3}} & \bm{0} & \bm{0} & \bm{L'} \\
\hline
\bm{0} & \bm{0} & \bm{0} & \bm{0} & \bm{0} & \bm{0}.
\end{array}
\right)
\label{eq:Mx_second_two}
\end{align}
The metric of interest is the donwlink throughput $\tpDL$, that is
\[
\tpDL = \frac{\mathbb{E}[\widehat N_1] + \mathbb{E}[\widehat N_2] + \mathbb{E}[\widehat N_{12}]}{\numulslot},
\]
being $\mathbb{E}[\widehat N_1]$ the expectation of the number of decoded packets forwarded only by the first relay to the gateway. Similarly, $\mathbb{E}[\widehat N_2]$ is the expected number of decoded packets forwarded to the gateway only by the second relay. Finally, $\mathbb{E}[\widehat N_{12}]$ is the expected number of decoded packets forwarded by both relays to the gateway. We define $\colNumVec\triangleq\left[\collK{1}, \collK{2}, \collK{1\wedge2}\right]$ and ${\rankNumVec\triangleq\left[\rankK{1}, \rankK{2}, \rankK{3}\right]}$. We first elaborate on $\mathbb{E}[\widehat N_1]$
\[
\mathbb{E}[\widehat N_1] = \sum_{\colNumVec} \sum_{\rankNumVec} \mathbb{E}[\widehat N_1|\colNumVec ,\rankNumVec] \Pr(\rankNumVec |\colNumVec) \Pr(\colNumVec).
\]
One packet can be successfully decoded at the gateway iff the row corresponding to the specific packet in the matrix $\MxComb''$ resulting after Gauss-Jordan elimination, has only one non-zero element. The probability for this to happen is $\frac{1}{{\fieldOrd}^{\mxColK{1}+\mxColK{3}}}$, for uplink packets successfully decoded only by the first relay. It corresponds to the probability that all entries in the corresponding row of the matrices $\bm{A}$ and $\bm{U'}_1$ are zeros. The $\mathbb{E}[\widehat N_1|\colNumVec ,\rankNumVec]$ is the expectation of a binomial distribution of parameters $\mathcal{B}\left(\collK{1}, \frac{1}{{\fieldOrd}^{\mxColK{1}+\mxColK{3}}}\right)$, so it holds
\begin{equation}
\mathbb{E}[\widehat N_1|\colNumVec ,\rankNumVec] = \collK{1}\left( \frac{1}{\fieldOrd}\right)^{\mxColK{1}+\mxColK{3}}.
\end{equation}
Now, we can write
\begin{align}
\mathbb{E}[\widehat N_1] &= \sum_{\colNumVec} \sum_{\rankNumVec} \collK{1}\left( \frac{1}{\fieldOrd}\right)^{\mxColK{1}+ \mxColK{3}} \Pr(\rankNumVec| \colNumVec) \Pr(\colNumVec) =\\
&= \sum_{\collK{1}=0}^{\numulslot} \sum_{\collK{2}=0}^{\numulslot} \sum_{\collK{1\wedge2}=0}^{\numulslot} \sum_{\rankK{1}=0}^{\rankKMax{1}} \sum_{\rankK{2}=0}^{\rankKMax{2}} \sum_{\rankK{3}=0}^{\rankKMax{3}} \collK{1}\left( \frac{1}{\fieldOrd}\right)^{\mxColK{1}+\mxColK{3}} \Pr(\rankK{1}|\collK{1}) \Pr(\rankK{2}|\collK{2})  \Pr(\rankK{3}|\collK{1\wedge2}) \\
&\cdot \Pr(\collK{1},\collK{2},\collK{1\wedge2}),
\label{eq:n1_avg}
\end{align}
where $\rankKMax{1}=\min \{\linllK{1},\collK{1}\}$, $\rankKMax{2}=\min \{\linllK{2},\collK{2}\}$ and $\rankKMax{3}=\min \{(\linllK{1} + \linllK{2}) - (\rankK{1} + \rankK{2}), \collK{1\wedge2}\}$. The last equality holds because the rank of the sub-matrices $\rankK{1}$, $\rankK{2}$ and $\rankK{3}$ are independent with each other. The Markov chains that show how the rank of the matrix evolves as row vectors $\mathbb{F}_{\fieldOrd}^{1 \times \coll}$ are added one by one as shown in Figure~\ref{fig:Mk_chain_rank}. The probability that a matrix belonging to $\mathbb{F}_{\fieldOrd}^{\rank \times \coll}$ has rank $x_{\rank}= \min\{\rank,\coll\}$ can be written in a recursive way \cite{Landsberg_1893, Kolchin_RG_1999} as
\begin{equation}
\Pr(\rank| \coll) =
\begin{cases}
\begin{aligned}
\Prob_{1,\coll}(0) &= \frac{1}{\fieldOrd^{\coll}}\\
\Prob_{1,\coll}(1) &= 1-\frac{1}{\fieldOrd^{\coll}}\\
\Prob_{\rank,\coll}(x_{\rank}) &= \left(\frac{\fieldOrd^{x_{\rank}}}{\fieldOrd^{\coll}}\right) \Prob_{\rank-1,\coll}(x_{\rank})+ \left(1-\frac{\fieldOrd^{x_{\rank}-1}}{\fieldOrd^{\coll}}\right) \Prob_{\rank-1,\coll}(x_{\rank}-1),\quad x_{\rank}=\min\{\rank,\coll\}.
\end{aligned}
\end{cases}
\label{eq:Mk_rank_m1_n1}
\end{equation}
\end{sloppypar}

\begin{figure}
\centering
\includegraphics[width=0.8\columnwidth]{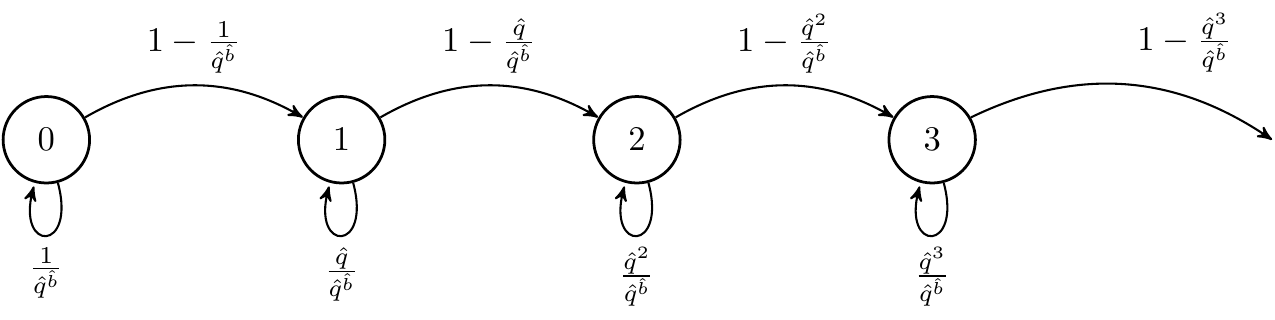}
\caption{Markov chain of the evolution of the matrix rank as row vectors are added. The state number represents the rank value.}
\label{fig:Mk_chain_rank}
\end{figure}
Where $\rank$ and $\coll$ can be substituted with any couple $(\rankK{1}, \collK{1})$, or $(\rankK{2}, \collK{2})$ or also $(\rankK{3}, \collK{1\wedge2})$. The quantity $\mathbb{E}[\widehat N_2|\colNumVec ,\rankNumVec]$ is the expectation of the binomial distribution $\mathcal{B}\left(\collK{2}, \frac{1}{\fieldOrd^{\mxColK{2}+\mxColK{3}}}\right)$, while $\mathbb{E}[\widehat N_{12}|\colNumVec ,\rankNumVec]$ is the expectation of the binomial distribution $\mathcal{B}\left(\collK{1\wedge2}, \frac{1}{\fieldOrd^{\mxColK{3}}}\right)$, so we can write,
\begin{align}
\mathbb{E}[\widehat N_2] &= \sum_{\colNumVec} \sum_{\rankNumVec} \collK{2}\left( \frac{1}{\fieldOrd}\right)^{\mxColK{2}+ \mxColK{3}} \Pr(\rankNumVec| \colNumVec) \Pr(\colNumVec)\\
\mathbb{E}[\widehat N_{12}] &= \sum_{\colNumVec} \sum_{\rankNumVec} \collK{1\wedge 2}\left( \frac{1}{\fieldOrd}\right)^{ \mxColK{3}} \Pr(\rankNumVec| \colNumVec) \Pr(\colNumVec).
\label{eq:n2_n12_avg}
\end{align}
The joint probability mass function $\Pr(\colNumVec)$ for the random vector $\{\collK{1}, \collK{2}, \collK{1\wedge2}\}$ can be tracked effectively by means of a homogeneous Markov chain, leaning on the assumed independence of channel realizations across uplink slots. To this aim, let $\obsSlot_\slot = \{ \collK{1}^{(\slot)}, \collK{2}^{(\slot)}, \collK{1\wedge2}^{(\slot)} \}$ be the state at the start of slot $\slot$, indicating the number of packets received so far by the first relay solely, by the second relay and by both of them, respectively. For the sake of compactness, let us furthermore denote $\textrm{Pr}\left\{ \obsSlot_{\slot+1} = \{ i',j',k'\}  \, \big{|} \, \obsSlot_\slot = \{ i,j,k  \} \right\}$ as $\Prob_{(i,j,k)\rightarrow(i',j',k')}$. Following this notation, each time unit can see five possible transitions for the chain, whose probabilities follow by simple combinatorial arguments similar to the ones discussed in Section~\ref{sec:uplink}:
\begin{equation}
\begin{cases}
\begin{aligned}
\Prob_{(i,j,k)\rightarrow(i,j,k)}
&= 1 - 2\load (1-\peras) e^{-\load (1-\peras)} + \load ( 1-\peras )^2\, e^{-\load (1-\peras^2)} \left( 1+\load \peras^2\right) \\
\Prob_{(i,j,k)\rightarrow(i+1,j,k)}
&= \load (1-\peras) e^{-\load (1-\peras)} - \load ( 1-\peras )^2\, e^{-\load (1-\peras^2)} \left( 1+\load \peras^2\right) \\
\Prob_{(i,j,k)\rightarrow(i,j+1,k)}
&= \load (1-\peras) e^{-\load (1-\peras)} - \load ( 1-\peras )^2\, e^{-\load (1-\peras^2)} \left( 1+\load \peras^2\right) \\
\Prob_{(i,j,k)\rightarrow(i,j,k+1)}
&= \load ( 1-\peras )^2\, e^{-\load (1-\peras^2)} \\
\Prob_{(i,j,k)\rightarrow(i+1,j+1,k)}
&=  (\load \peras)^2\,( 1-\peras )^2\, e^{-\load (1-\peras^2)} \\
\end{aligned}
\end{cases}
\label{eq:markov_transitions}
\end{equation}

The defined probabilities uniquely identify the transition matrix for the Markov chain under consideration, so that the sought probability mass function follows as its $\numulslot$-th step evolution when forcing the initial state as $\obsSlot_0 = \{0,0,0\}$.

We present first the results for \ac{RLC} for finite buffer size in Figure~\ref{fig:NC_finite}. If the buffer size is small, \ie $25$ slots, the performance of \ac{RLC} is particularly degraded with respect to the results for infinite buffer length, with a loss that can exceed $50\%$ for small downlink rates. A very tight match of the Monte Carlo simulations with the Markov model presented in this Section can also be observed, especially for buffer sizes of $80$ slots. For this reason, we depicted only the Monte Carlo simulations for higher buffer sizes. The results show a threshold behaviour: below the downlink rate that allows the two relays to send to the gateway all the received packets, which corresponds to $\rateDw=\tpUL=0.545$ in this scenario, the downlink throughput is particularly limited. The reason relies on the fact that in most of the cases the matrix is rank-deficient and very few packets can be decoded. This is true in general, regardless of the buffer size. Instead, for $\rateDw>0.545$, \ac{RLC} allows the correct decoding of most of the packets and, for high buffer sizes, \ac{RLC} is able to reach the infinite buffer bound. This is the ultimate limit of the scheme and determines the non-achievable region of downlink throughput. The higher the buffer size, the lower is the downlink rate necessary for reaching the infinite buffer throughput bound. The high sensitivity of \ac{RLC} on the downlink rate is detrimental, especially for scenarios where the uplink throughput is subject to quick and large variations. In these scenarios, in fact, the downlink rate shall be adapted accordingly, in order to avoid the region of low throughput.

\begin{figure}
\centering
\includegraphics[width=\figw]{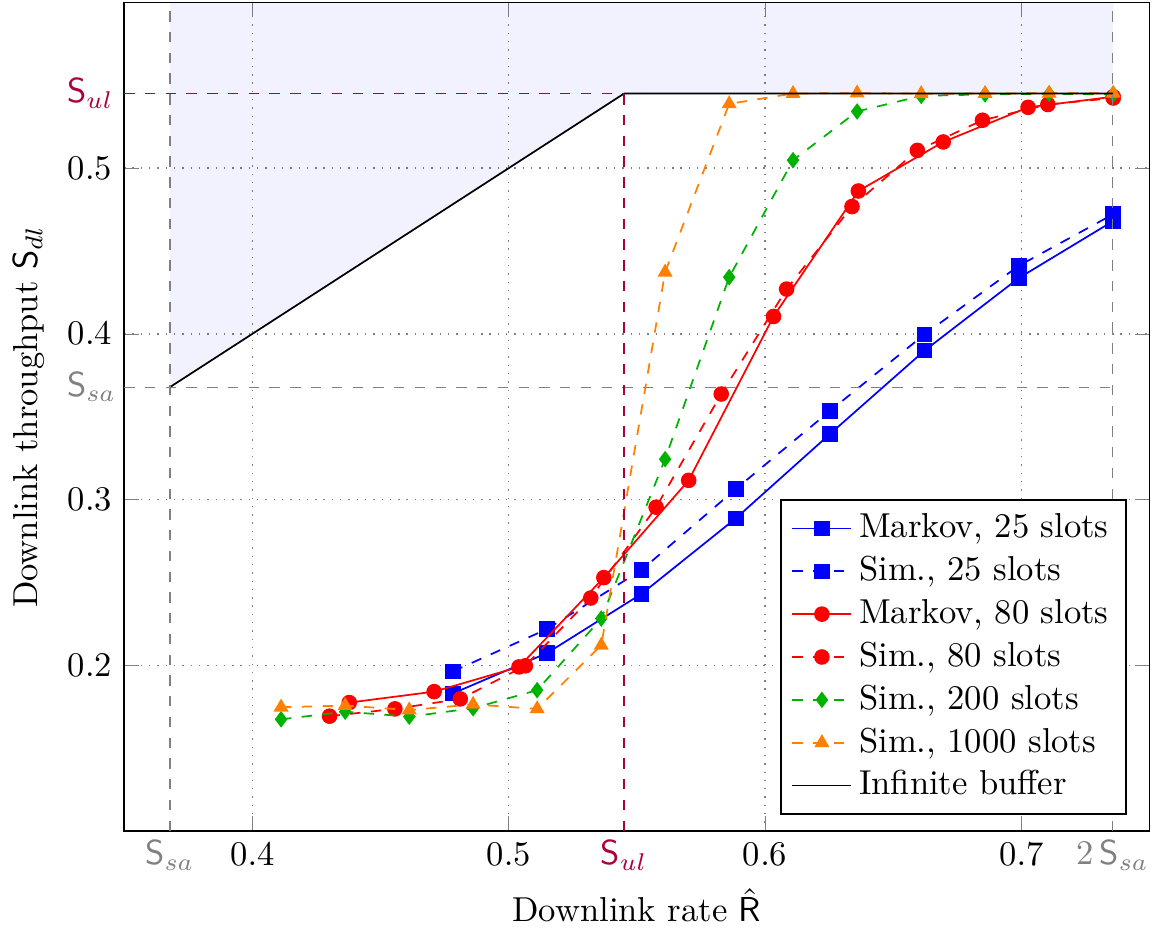}
\caption{\Ac{RLC} under finite buffer size. We compare the Monte Carlo simulations with the developed analytical tool using Markov chains, up to 80 slots as buffer size. Higher buffer sizes are obtained via Monte Carlo simulations.}
\label{fig:NC_finite}
\end{figure}

In Figure~\ref{fig:Comparison2} we show the comparison between \ac{RLC} and two dropping policies, the common transmission probability and the channel-aware policy. The choice of showing only two dropping policies is driven by the sake of clarity. The two chosen policies can be regarded as worst and best case respectively. Three buffer sizes are selected, \ie $25$, $100$ and $500$ slots for comparing three different scenarios: small buffer size, which is the case when packets have stringent delay constraint, medium buffer size and large buffer size, where packets have very relaxed delay constraints. We can observe that both dropping policies are particularly robust against reduction in the buffer size, showing limited performance loss w.r.t. the highest buffer size. Already for $100$ slots buffer size, both dropping policies reach almost their infinite buffer bound. Remarkably, for $25$ slots buffer size, both dropping policy outperform \ac{RLC} for any dowlink rate $\rateDw$. Increasing the buffer size allow \ac{RLC} to outperform both dropping policies but only for high downlink rates, \eg or $100$ slots buffer size $\rateDw>0.65$ against the channel-aware policy, and $\rateDw>0.63$ against the common transmission policy. Regardless the buffer size, for low donwlink rates both dropping policies largerly outperform \ac{RLC}, with gains exceeding $100\%$.

\begin{figure}
\centering
\includegraphics[width=\figw]{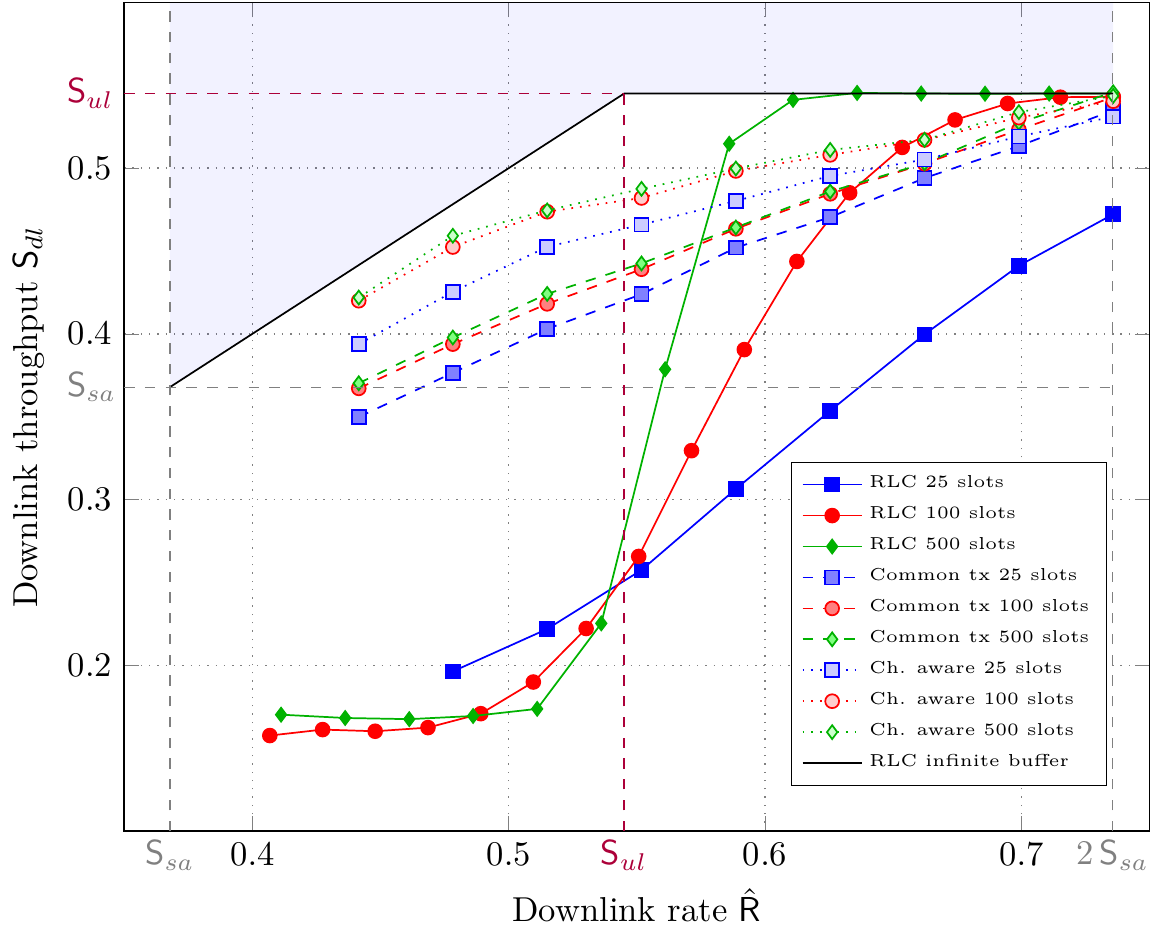}
\caption{Comparison of dropping policies and \ac{RLC} for finite-buffer scenarios. Various buffer dimensions are selected.}
\label{fig:Comparison2}
\end{figure} 
\section{Conclusions} \label{sec:conclusions}

In this Chapter we presented the beneficial effect brought by the addition of multiple receivers in a \ac{RA} scenario. In the uplink, the user population adopted \ac{SA} under a packet erasure channel and the relays collected only packets received collision free. In the downlink, the relays did not coordinate and forwarded the collected packets to a single \ac{gateway}. An analytical derivation of the uplink throughput for the entire network was derived together with the \ac{PLR}. Improvements w.r.t. the single receiver scenario were highlighted in the multi-receiver case and a logarithmic trend in the throughput increase as the number of relays was increased was shown. For the downlink, first we derived an upper bound on the downlink rate, measured in downlink transmission per uplink slot. Through the use of \ac{RLC}, a solution that reached the upper bound was presented that nonetheless relied on the assumption of infinite observation of the uplink. It also required the linear combination of collected packets at the relays and solving a system of equations at the gateway. In order to reduce complexity at the gateway alternative solutions based on dropping policies were also introduced and analysed. The key observation was that part of the collected packets at one relay was also correctly collected at a second one, with a given probability. In this way, properly sizing the queueing probability of collected packets leaded to a more efficient use of the downlink. Although viable, these solutions had a non-negligible performance loss compared to \ac{RLC} when an infinite uplink observation was guaranteed. We relaxed this assumption moving towards relay having finite buffers. In this scenario we derived a Markov chain model capturing the performance of the \ac{RLC} solution. Comparing it with the dropping policies showed that for limited downlink resources and stringent buffer constraints, the dropping policies largely outperformed \ac{RLC}. Nonetheless, \ac{RLC} could be improved following two paths: on one hand, one could adopt a low-density approach, in which the number of packets involved in the linear combinations sent to the \ac{gateway} becomes smaller, which corresponds to a sparse matrix. A second possibility would be to back-off in the number of linear combinations forwarded. Reducing the number of generated linear combinations would have a two-folded effect: could help when the downlink resources are scarce, but does not guarantee vanishing small error probability.

\renewcommand{\chaptermark}[1]{\markboth{#1}{}}
\setcounter{secnumdepth}{-1}
\chapter{Conclusions}
\thispagestyle{empty}
\epigraph{The good ideas will survive}{Quentin Tarantino}

\begin{sloppypar}
The thesis focused on medium access \ac{RA} protocols dedicated to variable and bursty channel traffic. Several aspects were addressed, from the performance of asynchronous \ac{RA} protocols employing \ac{SIC} and combining techniques, to optimisation of the repetition degree for time slotted schemes under fading channels. The key findings of the thesis and future research directions are summarised as follows.
\begin{itemize}
\item \textbf{Asynchronous \ac{RA}:} Bundling the use of \ac{SIC} together with combining techniques as \ac{SC} or \ac{MRC} in \ac{RA} with time diversity had a strong beneficial impact, leading such schemes to compete with comparable slot synchronous ones. An analytical approximation of the \ac{PLR} tight for low to moderate channel loads was derived and was shown to be useful for the design of such schemes. Practical aspects drove the design of a blind physical layer scheme for detection and combining that did not require any prior knowledge on the association between received packets and the user to which they belong. Numerical results demonstrated its applicability with very limited losses compared to the ideal scenario.

    Comparison at higher layer between asynchronous and slot synchronous schemes showed that higher benefits could be expected for the former ones especially when packets exceeded the slot duration and fragmentation was required in the latter ones. Analytical derivation of higher layer performance for the latter schemes was provided, so that no simulations of higher layers were required. When no knowledge of the probability measure defining the packet duration was available, an approximation based on the solely mean was given.

    One of the key research directions in this area shall focus on deriving a proper analytical model capturing the behaviour of \ac{SIC}. This would open an entire new field where for example, the users can use irregular repetition as in the slot synchronous case. With such a model, one could optimise the \ac{PMF} adopted by the user in order to maximise the throughput.

    Another interesting aspect is the delay that a user experience from transmission of the first replica until its correct decoding. This aspect is particularly relevant for the layer~$3$ comparison and shall be addressed in the future. A possible improvement of the slot  synchronous schemes include the exploitation of the outcome for the first fragment sent in the first frame in order to choose the number of replicas to be sent in the second frame.

\item \textbf{Time slotted irregular repetition \ac{RA}:} Irregular repetition of the physical layer packets brought undoubtable gains in the performance compared to the regular case. Extension of the optimisation procedure in order to include the Rayleigh block fading channel and the capture effect was provided. A non-negligible improvement in terms of achievable throughput and \ac{PLR} was shown compared to the optimisation over the collision channel.

    A straightforward extension of the framework shall consider other types of fading, e.g. Rice. In order to identify the ultimate performance of such scheme under fading, we shall derive an upper bound on the achievable rate, as is done for the collision channel. Although very appealing, its practical feasibility shall be considered.

\item \textbf{\ac{RA} with multiple receivers:} Expanding the topology to consider multiple non-cooperative receivers for a \ac{RA} uplink under the packet erasure channel, called for development of efficient solutions to forward the collected messages to a common gateway. \ac{RLC} was proven to achieve the best possible performance, but required an infinite amount of time for observing and collecting the uplink packets. Dropping policies were less efficient, but did not require infinite observation of the uplink. A buffer limited analytical model for the \ac{RLC} was here derived and it was shown via numerical simulations that dropping policies were able to largely outperform \ac{RLC} for many practical buffer sizes especially when the downlink resource was scarce.

    A possible extension of this framework would be to consider more advanced random access strategies for the uplink, as for example \ac{CSA}. As for the downlink, we could consider network coding strategies with probability of dropping packets correctly received from the uplink, in order to cope with finite length buffers so to compete with the simplified dropping policies.

\end{itemize}
\end{sloppypar}

Finally, we present a quick recap of the thesis contributions.
\begin{itemize}
\item A novel asynchronous \ac{RA} scheme named \ac{ECRA} was presented in Chapter~\ref{chapter4};
\item A novel detection and combining methodology for asynchronous \ac{RA} protocols was developed in Chapter~\ref{chapter4};
\item A new methodology for assessing the layer~$3$ performance of \ac{RA} protocols was developed in Chapter~\ref{chapter5};
\item A new methodology for the analysis of irregular slot synchronous \ac{RA} schemes was presented in Chapter~\ref{chapter6};
\item A new methodology for the analysis of finite buffer \ac{RLC} was presented in Chapter~\ref{chapter7}.
\end{itemize} 
\chapter{Selected List of Publications}
\thispagestyle{empty}

Part of the work presented in this thesis has also appeared in the articles reported below.
\vspace{0.8cm}

\textbf{Journal Papers}
\begin{itemize}
\item [{[J2]}] F. Clazzer, C. Kissling, M. Marchese, "Enhancing Contention Resolution ALOHA
    using Combining Techniques", submitted to IEEE Transactions on Communications.
\item [{[J1]}] A. Munari, M. Heindlmaier, F. Clazzer, G. Liva, "On the Capacity of Slotted ALOHA
    with Multiple Receivers", in preparation for the IEEE Transaction on Information Theory.
\end{itemize}

\vspace{0.8cm}
\textbf{Conference Papers}
\begin{itemize}
\item [{[C6]}] F. Clazzer, E. Paolini, I. Mambelli and C. Stefanovi{\'c},
 "Irregular Repetition Slotted ALOHA over the Rayleigh Block Fading Channel with Capture",
 accepted for publication in the IEEE International Conference on Communications (ICC), 2017.
\item [{[C5]}] F. Clazzer, F. L\'azaro, G. Liva, M. Marchese, "Detection and Combining Techniques for
Asynchronous Random Access with Time Diversity", in Proceedings of the 11th International ITG
Conference on Systems, Communications and Coding (SCC), Hamburg, Germany, February 2017.
\item [{[C4]}] A. Munari, F. Clazzer, "Multi-Receiver Aloha Systems - a Survey and New Results", in
Proceedings of the IEEE International Conference on Communications (ICC), Workshop on Massive
Uncoordinated Access Protocols (MASSAP), London, UK, June 2015.
\item [{[C3]}] F. Clazzer, M. Marchese, "Layer $3$ Throughput Analysis for Advanced ALOHA Protocols", in
Proceedings of the IEEE International Conference on Communications (ICC), Workshop on Massive
Uncoordinated Access Protocols (MASSAP), pp.533-538, Sydney, Australia, June 2014.
\item [{[C2]}] F. Clazzer, C. Kissling, "Optimum Header Positioning in Successive Interference
Cancellation (SIC) based ALOHA", in Proceedings of the IEEE International Conference on
Communications (ICC), pp.2869-2874, Budapest, Hungary, June 2013.
\item [{[C1]}] F. Clazzer, C. Kissling, "Enhanced Contention Resolution ALOHA - ECRA", in Proceedings
of the 9th International ITG Conference on Systems, Communications and Coding (SCC), pp.1-6,
Munich, Germany, January 2013.
\end{itemize}

\chapter{Other Publications}
\thispagestyle{empty}

Below are reported other articles submitted or published during the Ph.D. that are not appearing in the thesis.
\vspace{0.8cm}

\textbf{Journal Papers}
\begin{itemize}
\item [{[J3]}] F. Clazzer, F. L\'azaro, S. Plass, "Enhanced AIS Receiver Design for Satellite
    Reception", CEAS Space Journal, Vol. 8, No. 4, pp. 257-268, December 2016.
\item [{[J2]}] Z. Katona, F. Clazzer, K. Shortt, S. Watts, Lexow, H.P., R. Winduratna,
    "Performance, Cost Analysis, and Ground Segment Design of Ultra High Throughput Multi-Spot
    Beam Satellite Networks applying different Capactiy Enhancing Techniques",
    International Journal of Satellite Communications and Networking, Vol. 34, No. 4, pp.
    547-573, July/August 2016.
\item [{[J1]}] P.-D. Arapoglou, A. Ginesi, S. Cioni, S. Erl, F. Clazzer, S. Andrenacci, A.
    Vanelli-Coralli, "DVB-S2X-Enabled Precoding for High Throughput Satellite Systems",
    International Journal of Satellite Communications and Networking, Vol. 34, No. 3, pp.
    439-455, May/June 2016.
\end{itemize}

\vspace{0.8cm}
\textbf{Conference Papers}
\begin{itemize}
\item [{[C10]}] F. Clazzer, A. Munari, F. Giorgi, "Asynchronous Random Access Schemes
    for the VDES Satellite Uplink", submitted to the IEEE/MTS Oceans'17, 2017.
\item [{[C9]}] F. Giorgi, F. Clazzer, A. Munari, "Asynchronous Random Access Protocols for the
    Maritime Satellite Channel", submitted to the IEEE Global Communications Conference
    (GLOBECOM), 2017.
\item [{[C8]}] A. Stajkic, F. Clazzer, G. Liva, "Neighbor Discovery in Wireless Networks: A
    Graph-based Analysis and Optimization", in Proceedings of the IEEE International Conference
    on Communications (ICC), Workshop on Massive Uncoordinated Access Protocols (MASSAP),
    pp. 511-516, Kuala Lumpur, Malaysia, May 2016.
\item [{[C7]}] S. Plass, F. Clazzer, F. Bekkadal, "Current Situation and Future Innovations in
    Arctic Communications", in Proceedings of the IEEE 82nd Vehicular Technology Conference
    (VTC Fall), pp. 1-7, Boston, MA, USA, September 2015.
\item [{[C6]}] F. Clazzer, F. L\'azaro, S. Plass, "New Receiver Algorithms for Satellite AIS", in
    Proceedings of the Deutscher Luft- und Raumfahrtkongress, pp. 1-7, Rostock, Germany,
    September 2015.
\item [{[C5]}] F. Clazzer, A. Munari, "Analysis of Capture and Multi-Packet Reception on the
    AIS Satellite System", Proceedings of the IEEE/MTS Oceans'15, pp.1-9, Genova, Italy, May
    2015.
\item [{[C4]}] S. Plass, F. Clazzer, F. Bekkadal, Y. Ibnyahya, M. Manzo, "Maritime Communications
    - Identifying Current and Future Satellite Requirements \& Technologies", in Proceedings of
    the 20th Ka and Broadband Communications, Navigation and Earth Observation Conference,
    pp. 1-8, Vietri, Italy, October 2014.
\item [{[C3]}] F. Clazzer, A. Munari, S. Plass, B. Suhr, "On the Impact of Coverage Range on AIS
    Message Reception at Flying Platforms", in Proceedings of the 7th Advanced Satellite
    Multimedia Conference and the 13th Signal Processing for Space Communications Workshop
    (ASMS/SPSC), pp.128-135, Livorno, Italy, September 2014.
\item [{[C2]}] F. Clazzer, A. Munari, M. Berioli, F. L\'azaro, "On the Characterization of AIS
    Traffic at the Satellite", in Proceedings of the IEEE/MTS Oceans'14, pp.1-9, Taipei, Taiwan,
    April 2014.
\item [{[C1]}] C. Kissling, F. Clazzer, "LDPC Code Performance and Optimum Code Rate for
    Contention Resolution Diversity ALOHA", in Proceedings of the IEEE Global Communication
    Conference (GLOBECOM), pp.2954-2960, Atlanta, GA, USA, December 2013.
\end{itemize}

\appendix

\renewcommand{\chaptermark}[1]{\markboth{\appendixname\ \thechapter.\ #1}{}}

\backmatter

\renewcommand{\chaptermark}[1]{\markboth{#1}{}}

\bibliographystyle{IEEEtran}
\bibliography{IEEEabrv,References}

\begin{thebibliography}{100}
\providecommand{\url}[1]{#1}
\csname url@samestyle\endcsname
\providecommand{\newblock}{\relax}
\providecommand{\bibinfo}[2]{#2}
\providecommand{\BIBentrySTDinterwordspacing}{\spaceskip=0pt\relax}
\providecommand{\BIBentryALTinterwordstretchfactor}{4}
\providecommand{\BIBentryALTinterwordspacing}{\spaceskip=\fontdimen2\font plus
\BIBentryALTinterwordstretchfactor\fontdimen3\font minus
  \fontdimen4\font\relax}
\providecommand{\BIBforeignlanguage}[2]{{%
\expandafter\ifx\csname l@#1\endcsname\relax
\typeout{** WARNING: IEEEtran.bst: No hyphenation pattern has been}%
\typeout{** loaded for the language `#1'. Using the pattern for}%
\typeout{** the default language instead.}%
\else
\language=\csname l@#1\endcsname
\fi
#2}}
\providecommand{\BIBdecl}{\relax}
\BIBdecl

\bibitem{Massey1981}
J.~L. Massey, ``{Collision-Resolution Algorithm and Random-Access
  Communication},'' \emph{Multi-User Communication Systems}, vol. Editor G.
  Longo, pp. 73--137, 1981.

\bibitem{MAC_1990}
R.~Rom and M.~Sidi, \emph{{Multiple Access Protocols: Performance and
  Analysis}}, M.~Gerla, A.~Lazar, and P.~Kuehn, Eds.\hskip 1em plus 0.5em minus
  0.4em\relax Springer-Verlag, 1990.

\bibitem{Abramson1970}
N.~Abramson, ``{The ALOHA system: Another alternative for computer
  communications},'' in \emph{Proceedings of the 1970 Fall Joint Computer
  Conference, AFIPS Conference}, vol.~37, Montvale, N.~J., 1970, pp. 281--285.

\bibitem{Roberts1975}
L.~G. Roberts, ``{ALOHA packet system with and without slots and capture},''
  \emph{Proceedings of the International Conference of Special Interest Group
  on Data Communications (SIGCOMM)}, vol.~5, pp. 28--42, April 1975.

\bibitem{Data_Networks_1987}
D.~Bertsekas and R.~Gallager, \emph{{Data Networks}}.\hskip 1em plus 0.5em
  minus 0.4em\relax Prentice Hall, 1987.

\bibitem{Abramson1994}
N.~Abramson, ``{Multiple Access in Wireless Digital Networks},''
  \emph{Proceedings of IEEE}, vol.~82, no.~9, pp. 1360--1370, September 1994.

\bibitem{Wu_2011}
G.~Wu, S.~Talwar, K.~Johnsson, N.~Himayat, and K.~Johnson, ``{M2M: From Mobile
  to Embedded Internet},'' \emph{IEEE Communications Magazine}, vol.~49, no.~4,
  pp. 36--43, April 2011.

\bibitem{MUD_1998}
S.~Verd{\'u}, \emph{{Multiuser Detection}}.\hskip 1em plus 0.5em minus
  0.4em\relax Cambridge University Press, 1998.

\bibitem{Menouar2006}
H.~Menouar, F.~Filali, and M.~Lenardi, ``{A Survey and Qualitative Analysis of
  MAC Protocols for Vehicular Ad Hoc Networks},'' \emph{IEEE Wireless
  Communications}, vol.~13, no.~5, pp. 30--35, October 2006.

\bibitem{He2013}
Y.~He and X.~Wang, ``{An ALOHA-Based Improved Anti-Collision Algorithm for RFID
  Systems},'' \emph{IEEE Wireless Communications}, vol.~20, no.~5, pp.
  152--158, October 2013.

\bibitem{Peyravi1999}
H.~Peyravi, ``{Medium Access Control Protocols Performance in Satellite
  Communications},'' \emph{IEEE Communications Magazine}, vol.~37, no.~3, pp.
  62--71, March 1999.

\bibitem{Pompili2009}
D.~Pompili and I.~Akyildiz, ``{Overview of Networking Protocols for Underwater
  Wireless Communications},'' \emph{IEEE Communications Magazine}, vol.~47,
  no.~1, pp. 97--102, January 2009.

\bibitem{Laya2014}
A.~Laya, L.~Alonso, and J.~Alonso-Zarate, ``{Is the Random Access Channel of
  LTE and LTE-A Suitable for M2M Communications? A Survey of Alternatives},''
  \emph{IEEE Communications Surveys and Tutorials}, vol.~16, no.~1, pp. 4--16,
  First Quarter 2014.

\bibitem{oneWeb}
``http://oneweb.world/.''

\bibitem{Kleinrock1976_book}
L.~Kleinrock, \emph{{Queueing Systems - Volume II: Computer
  Applications}}.\hskip 1em plus 0.5em minus 0.4em\relax John Wiley \& Sons,
  1976.

\bibitem{Abramson1985}
N.~Abramson, ``{Development of the ALOHANET},'' \emph{IEEE Transactions on
  Information Theory}, vol. IT-31, no.~2, pp. 119--123, March 1985.

\bibitem{Metcalfe1976}
B.~Metcalfe and D.~R. Boggs, ``{Ethernet: Distributed Packet Switching for
  Local Computer Networks},'' \emph{Communications of the ACM}, vol.~19, no.~7,
  pp. 395--404, July 1976.

\bibitem{Marisat1977}
D.~W. Lipke, D.~W. Swearingen, J.~F. Parker, E.~E. Steinbrecher, T.~O. Clavit,
  and H.~Dodel, ``{Marisat - A Maritime Satellite Communications System},''
  \emph{COMSAT Technical Review}, vol.~7, pp. 351--391, November 1977.

\bibitem{Cai1997}
J.~Cai and D.~J. Goodman, ``{General Packet Radio Service in GSM},'' \emph{IEEE
  Communications Magazine}, vol.~35, no.~10, pp. 122--131, October 1997.

\bibitem{Laya2016}
A.~Laya, C.~Kalalas, F.~Vazquez-Gallego, L.~Alonso, and J.~Alonso-Zarate,
  ``{Goodbye, ALOHA!}'' \emph{IEEE Access}, vol.~4, pp. 2029--2044, April 2016.

\bibitem{Kleinrock1975}
L.~Kleinrock and S.~Lam, ``{Packet Switching in a Multiaccess Broadcast
  Channel: Performance Evaluation},'' \emph{IEEE Transactions on
  Communications}, vol. COM-23, no.~4, pp. 410--423, April 1975.

\bibitem{Lam1975}
S.~S. Lam and L.~Kleinrock, ``{Packet Switching in a Multiaccess Broadcast
  Channel: Dynamic Control Procedures},'' \emph{IEEE Transactions on
  Communications}, vol. COM-23, no.~9, pp. 891--904, September 1975.

\bibitem{Carleial1975}
A.~B. Carleial and M.~E. Hellman, ``{Bistable Behavior of ALOHA-Type
  Systems},'' \emph{IEEE Transactions on Communications}, vol. COM-23, no.~4,
  pp. 401--410, April 1975.

\bibitem{Jenq1980}
Y.-C. Jenq, ``{On the Stability of Slotted ALOHA Systems},'' \emph{IEEE
  Transactions on Communications}, vol. COM-28, no.~11, pp. 1936--1939,
  November 1980.

\bibitem{Tsybakov1979_Stability}
B.~S. Tsybakov and V.~A. Mikhailov, ``{Ergodicity of sloted ALOHA system},''
  \emph{Problemy Peredachi Informatsii}, vol.~15, pp. 73--87, 1979.

\bibitem{Malyshev1972}
V.~A. Malyshev, ``{Classification of two-dimensional positive random walks and
  almost linear semimartingals},'' \emph{Dokl. Akad., Nauk SSSR}, vol.~22, pp.
  136--138, 1972.

\bibitem{Mensikov1974}
M.~V. Mensikov, ``{Ergodicity and transience conditions for random walks in the
  positive octant of space},'' \emph{Soviet. Math. Dokl.}, vol.~15, pp.
  1118--1121, 1974.

\bibitem{Malyshev1981}
V.~A. Malyshev and M.~V. Mensikov, ``{Ergodicity, continuity and analyticity of
  countable Markov chains},'' \emph{Trans. Moscow Math. Soc.}, vol.~39, pp.
  3--48, 1979.

\bibitem{Rao1988}
R.~R. Rao and A.~Ephremides, ``{On the Stability of Interacting Queues in A
  Multiple-Access System},'' \emph{IEEE Transactions on Information Theory},
  vol.~34, no.~5, pp. 918--930, September 1988.

\bibitem{Metcalfe1975}
B.~Metcalfe, ``{Steady-state analysis of a slotted and controlled Aloha system
  with blocking},'' \emph{Proceedings of the International Conference of
  Special Interest Group on Data Communications (SIGCOMM)}, vol.~5, no.~1, pp.
  24--31, January 1975.

\bibitem{Ferguson1975}
M.~J. Ferguson, ``{On the Control, Stability, and Waiting Time in a Slotted
  ALOHA Random-Access System},'' \emph{IEEE Transactions on Communications},
  vol.~23, no.~11, pp. 1306--1311, November 1975.

\bibitem{Jenq1981}
Y.-C. Jenq, ``{Optimal Retransmission Control of Slotted ALOHA Systems},''
  \emph{IEEE Transactions on Communications}, vol. COM-29, no.~6, pp. 891--895,
  June 1981.

\bibitem{Saadawi1981}
T.~N. Saadawi and A.~Ephremides, ``{Analysis, Stability, and Optimization of
  Slotted ALOHA with a Finite Number of Buffered Users},'' \emph{IEEE
  Transactions on Automatic Control}, vol. AC-26, no.~3, pp. 680--689, June
  1981.

\bibitem{Szpankowski1994}
W.~Szpankowski, ``{Stability Conditions for some Distributed Systems: Buffered
  Random Access Systems},'' \emph{Advances in Applied Probability}, vol.~26,
  no.~2, pp. 498--515, June 1994.

\bibitem{Capetanakis1979}
J.~I. Capetanakis, ``{Tree Algorithms for Packet Broadcast Channels},''
  \emph{IEEE Transactions on Information Theory}, vol.~25, no.~5, pp. 505--515,
  September 1979.

\bibitem{Capetanakis1977}
------, ``{The multiple access broadcast channel: Protocol and capacity
  considerations},'' Ph.D. dissertation, Department of Electrical Engineering,
  Massachusset Institute of Technology, Cambridge, MA, August 1977.

\bibitem{Tsybakov1978}
B.~S. Tsybakov and V.~A. Mikhailov, ``{Free synchronous packet access in a
  broadcast channel with feedback},'' \emph{Problemy Peredachi Informatsii},
  vol.~14, no.~4, pp. 32--59, October 1978.

\bibitem{Gallager1985}
R.~Gallager, ``{A perspective on Multiaccess Channels},'' \emph{IEEE
  Transactions on Information Theory}, vol.~31, no.~2, pp. 124--142, March
  1985.

\bibitem{Tsybakov1980}
B.~S. Tsybakov and N.~B. Likhanov, ``{Upper bound on the capacity of a random
  multiple-access system},'' \emph{Problemy Peredachi Informatsii}, vol.~23,
  no.~3, pp. 64--78, 1987.

\bibitem{Chan2012}
D.~S. Chan and T.~Berger, ``{Upper Bound for the Capacity of Multiple Access
  Protocols on Multipacket Reception Channels},'' in \emph{Proceedings of the
  IEEE International Symposium on Information Theory (ISIT)}, Cambridge, MA,
  USA, July 2012, pp. 1603--1607.

\bibitem{Tsybakov1985}
B.~S. Tsybakov, ``{Survey of USSR Contributions to Random Multiple-Access
  Communications},'' \emph{IEEE Transactions on Information Theory}, vol.
  IT-31, no.~2, pp. 143--165, March 1985.

\bibitem{Mathys_1985}
P.~Mathys and P.~Flajolet, ``{Q-ary collision resolution algorithms in
  random-access systems with free or blocked channel access},'' \emph{IEEE
  Transactions on Information Theory}, vol.~31, no.~2, pp. 217--243, March
  1985.

\bibitem{Massey1985_RAWithoutFeedback}
J.~L. Massey and P.~Mathys, ``{The Collision Channel Without Feedback},''
  \emph{IEEE Transactions on Information Theory}, vol. IT-31, no.~2, pp.
  192--204, March 1985.

\bibitem{Massey1982_ISIT}
J.~L. Massey, ``{The Capacity of the Collision Channel Without Feedback},'' in
  \emph{Proceedings of the IEEE International Symposium of Information Theory
  (ISIT)}, Les Arcs, France, 1982.

\bibitem{Tsybakov1983_RAWF}
B.~S. Tsybakov and N.~B. Likhanov, ``{Packet Switching in a Channel Without
  Feedback},'' \emph{Problemy Peredachi Informatsii}, vol.~19, no.~2, pp.
  69--84, Apr.-June 1983.

\bibitem{Hui1984}
J.~Y.~N. Hui, ``{Multiple Accessing for the Collsion Channel Without
  Feedback},'' \emph{IEEE Journal on Selected Areas in Communications}, vol.
  SAC-2, no.~4, pp. 575--582, July 1984.

\bibitem{Scholtz_1977}
R.~A. Scholtz, ``{The Spread Spectrum Concept},'' \emph{IEEE Transactions on
  Communications}, vol. COM-25, no.~8, pp. 748--755, August 1977.

\bibitem{Kahn_1978}
R.~E. Kahn, S.~A. Gronemeyer, J.~Burchfiel, and R.~C. Kunzelman, ``{Advances in
  Packet Radio Technology},'' \emph{Proceedings of the IEEE}, vol.~66, no.~11,
  pp. 1468--1496, November 1978.

\bibitem{Pickholtz_1982}
R.~L. Pickholtz, D.~L. Schilling, and L.~B. Milstein, ``{Theory of
  Spread-Spectrum Communications - A Tutorial},'' \emph{IEEE Transactions on
  Communications}, vol. COM-30, no.~5, pp. 855--884, May 1982.

\bibitem{Pursley_1987}
M.~B. Pursley, ``{The Role of Spread Spectrum in Packet Radio Networks},''
  \emph{Proceedings of the IEEE}, vol.~75, no.~1, pp. 116--134, January 1987.

\bibitem{Raychaudhuri_1980}
D.~Raychaudhuri, ``{Performance Analysis of Random Access Packet-Switched Code
  Division Multiple Access Systems},'' \emph{IEEE Transactions on
  Communications}, vol. COM-29, no.~6, pp. 895--901, June 1981.

\bibitem{Pursley_1986}
M.~B. Pursley, ``{Frequency-Hop Transmission for Satellite Packet Switching and
  Terrestrial Packet Radio Networks},'' \emph{IEEE Transactions on Information
  Theory}, vol. IT-32, no.~5, pp. 652--667, September 1986.

\bibitem{Polydoros_1987}
A.~Polydoros and J.~Silvester, ``{Slotted Random Access Spread-Spectrum
  Networks: An Analytical Framework},'' \emph{IEEE Journal on Selected Areas in
  Communications}, vol. SAC-5, no.~6, pp. 989--1002, July 1987.

\bibitem{Kleinrock1975a}
L.~Kleinrock and F.~A. Tobagi, ``{Packet Switching in Radiochannelsd: Part
  I-Carrier Sense Multiple-Access Modes and Their Throughput-Delay
  Characteristics},'' \emph{IEEE Transactions on Communications}, vol. COM-23,
  no.~12, pp. 1400--1416, December 1975.

\bibitem{Tobagi1975}
F.~A. Tobagi and L.~Kleinrock, ``{Packet Switching in Radio Channels: Part
  II-The Hidden Terminal Problem in Carrier Sense Multiple-Access and the
  Busy-Tone Solution},'' \emph{IEEE Transactions on Communications}, vol.
  COM-23, no.~12, pp. 1417--1433, December 1975.

\bibitem{Colvin1983}
A.~Colvin, ``{CSMA with collision avoidance},'' \emph{Computer Communications},
  vol.~6, no.~5, pp. 227--235, October 1983.

\bibitem{Metzener1976}
J.~J. Metzener, ``{On Improving Utilization in ALOHA Networks},'' \emph{IEEE
  Transactions on Communications}, vol. COM-24, no.~4, pp. 447--448, April
  1976.

\bibitem{Abramson1977}
N.~Abramson, ``{The Throughput of Packet Broadcasting Channels},'' \emph{IEEE
  Transactions on Communications}, vol.~25, no.~1, pp. 117--128, January 1977.

\bibitem{Raychaudhuri1981}
D.~Raychaudhuri, ``{Performance Analysis of Random Access Packet-Switched Code
  Division Multiple Access Systems},'' \emph{IEEE Transactions on
  Communications}, vol. COM-29, no.~6, pp. 895--901, June 1981.

\bibitem{Ghez1988}
S.~Ghez, S.~Verd{\'u}, and S.~C. Schwartz, ``{Stability Properties of Slotted
  Aloha with Multipacket Reception Capability},'' \emph{IEEE Transactions on
  Automatic Control}, vol.~33, no.~7, pp. 640--649, July 1988.

\bibitem{Zorzi1994}
M.~Zorzi and R.~R. Rao, ``{Capture and Retransmission Control in Mobile
  Radio},'' \emph{IEEE Journal on Selected Areas in Communications}, vol.~12,
  no.~8, pp. 1289--1298, October 1994.

\bibitem{Tong2004}
L.~Tong and P.~Venkitasubramaniam, ``{Signal Processing in Random Access},''
  \emph{IEEE Signal Processing Magazine}, vol.~21, no.~5, pp. 29--39, September
  2004.

\bibitem{TseBook_2005}
D.~Tse and P.~Viswanath, \emph{Fundamentals of Wireless Communications}.\hskip
  1em plus 0.5em minus 0.4em\relax Cambridge University Press, 2005.

\bibitem{Biglieri2005}
E.~Biglieri, \emph{{Coding for Wireless Channels}}, R.~Gallager and J.~K. Wolf,
  Eds.\hskip 1em plus 0.5em minus 0.4em\relax Springer, 2005.

\bibitem{Warrier_1998}
D.~Warrier and U.~Madhow, ``{On the Capacity of Cellular CDMA with Successive
  Decoding and Controlled Power Disparities},'' in \emph{Proceedings of the
  48th IEEE Vehicular Technology Conference (VTC)}, Ottawa, Canada, May 1998,
  pp. 1873--1877.

\bibitem{Andrews_2003}
J.~G. Andrews and T.~H. Meng, ``{Optimum Power Control for Successive
  Interference Cancellation With Imperfect Channel Estimation},'' \emph{IEEE
  Transactions on Wireless Communications}, vol.~2, no.~2, pp. 375--383, March
  2003.

\bibitem{Agrawal_2005}
A.~Agrawal, J.~G. Andrews, J.~M. Cioffi, and T.~H. Meng, ``{Iterative Power
  Control for Imperfect Successive Interference Cancellation},'' \emph{IEEE
  Transactions on Wireless Communications}, vol.~4, no.~3, pp. 878--884, May
  2005.

\bibitem{Shannon_1948}
C.~E. Shannon, ``{A Mathematical Theory of Communication},'' \emph{The Bell
  System Technical Journal}, vol.~27, pp. 379--423, 623--656, July, October
  1948.

\bibitem{Hartley_1928}
R.~V.~L. Hartley, ``{Transmission of Information},'' \emph{The Bell Labs
  Technical Journal}, vol.~7, no.~3, pp. 535--563, July 1928.

\bibitem{McEliece1984}
R.~J. McEliece and W.~E. Stark, ``{Channels with Block Interference},''
  \emph{IEEE Transactions on Information Theory}, vol. IT-30, no.~1, pp.
  44--53, January 1984.

\bibitem{cover2006}
T.~M. Cover and J.~A. Thomas, \emph{{Elements of Information Theory, 2nd
  edition}}.\hskip 1em plus 0.5em minus 0.4em\relax John Wiley \& Sons, 2006.

\bibitem{Choudhury1983}
G.~Choudhury and S.~Rappaport, ``{Diversity ALOHA - A Random Access Scheme for
  Satellite Communications},'' \emph{IEEE Transactions on Communications},
  vol.~31, no.~3, pp. 450--457, March 1983.

\bibitem{Casini2007}
E.~Casini, R.~De~Gaudenzi, and O.~del Rio~Herrero, ``{Contention Resolution
  Diversity Slotted {ALOHA} ({CRDSA}): An Enhanced Random Access Scheme for
  Satellite Access Packet Networks},'' \emph{IEEE Transactions on Wireless
  Communications}, vol.~6, no.~4, pp. 1408--1419, April 2007.

\bibitem{Herrero2014}
O.~del Rio~Herrero and R.~de~Gaudenzi, ``{Generalized Analytical Framework for
  the Performance Assessment of Slotted Random Access Protocols},'' \emph{IEEE
  Transactions on Wireless Communications}, vol.~13, no.~2, pp. 809--821,
  February 2014.

\bibitem{Kissling2011}
C.~Kissling, ``{On the stability of Contention Resolution Diversity Slotted
  ALOHA ({CRDSA})},'' in \emph{Proceedings of the IEEE Global Communication
  Conference (GLOBECOM)}, Houston, TX, USA, December 2011, pp. 1--6.

\bibitem{Ivanov2015}
M.~Ivanov, F.~Br\"{a}nnstr\"{o}m, A.~Graell~i Amat, and P.~Popovski, ``{Error
  Floor Analysis of Coded Slotted ALOHA Over Packet Erasure Channel},''
  \emph{IEEE Communications Letters}, vol.~19, no.~3, pp. 419--422, March 2015.

\bibitem{Liva2011}
G.~Liva, ``{Graph-Based Analysis and Optimization of Contention Resolution
  Diversity Slotted ALOHA},'' \emph{IEEE Transactions on Communications},
  vol.~59, no.~2, pp. 477--487, February 2011.

\bibitem{Paolini2011}
E.~Paolini, G.~Liva, and M.~Chiani, ``{High Throughput Random Access via Codes
  on Graphs: Coded Slotted ALOHA},'' in \emph{Proceedings of the IEEE
  International Conference on Communications (ICC)}, Kyoto, Japan, June 2011,
  pp. 1--6.

\bibitem{Paolini2015}
------, ``{Coded Slotted ALOHA: A Graph-Based Method for Uncoordinated Multiple
  Access},'' \emph{IEEE Transactions on Information Theory}, vol.~61, no.~12,
  pp. 6815--6832, December 2015.

\bibitem{Stefanovic2013_J}
C.~Stefanovic and P.~Popovski, ``{ALOHA Random Access that Operates as a
  Rateless Code},'' \emph{IEEE Transactions on Communications}, vol.~61,
  no.~11, pp. 4653--4662, November 2013.

\bibitem{Narayanan2012}
K.~R. Narayanan and H.~D. Pfister, ``{Iterative Collision Resolution for
  Slotted ALOHA: An Optimal Uncoordinated Transmission Policy},'' in
  \emph{Proceedings of the 7th International Symposium on Turbo Codes and
  Iterative Information Processing (ISTC)}, Gothenburg, Sweden, August 2012,
  pp. 136--139.

\bibitem{Luby1998}
J.~Byers, M.~Luby, M.~Mitzenmacher, and A.~Rege, ``{A Digital Fontain approach
  to Reliable Distribution of Bulk Data},'' in \emph{Proceedings of the
  International Conference of Special Interest Group on Data Communications
  (SIGCOMM)}, vol.~28, no.~4, October 1998, pp. 56--67.

\bibitem{Meloni2012}
A.~Meloni, M.~Murroni, C.~Kissling, and M.~Berioli, ``{Sliding Window-Based
  Contention Resolution Diversity Slotted ALOHA},'' in \emph{Proceedings of the
  IEEE Global Communications Conference (GLOBECOM)}, Anaheim, CA, USA, December
  2012, pp. 3305--3310.

\bibitem{Sandgren2016}
E.~Sandgren, A.~Graell~i Amat, and F.~Br\"{a}nnstr\"{o}m, ``{On Frame
  Asynchronous Coded Slotted ALOHA: Asymptotic, Finite Length, and Delay
  Analysis},'' \emph{Accepted for the IEEE Transactions on Communications},
  2016.

\bibitem{Andrews2005}
J.~G. Andrews, ``{Interference Cancellation for Cellular Systems: A
  Contemporary Overview},'' \emph{IEEE Wireless Communications}, vol.~12,
  no.~2, pp. 19--29, April 2005.

\bibitem{Gamb_Schlegel2013}
M.~Ghanbarinejad and C.~Schlegel, ``{Irregular Repetition Slotted ALOHA with
  Multiuser Detection},'' in \emph{Proceedings of the 10th Annual Conference on
  Wireless On-demand Network Systems and Services (WONS)}, Banff, AB, March
  2013, pp. 201--205.

\bibitem{Munari2013}
A.~Munari, M.~Heindlmaier, G.~Liva, and M.~Berioli, ``{The Throughput of
  Slotted ALOHA with Diversity},'' in \emph{Proceedings of the 51st Annual
  Allerton Conference on Communications, Control, and Computing}, Monticello,
  IL, USA, October 2013, pp. 698--706.

\bibitem{Jakovetic2015}
D.~Jakovetic, D.~Bajovic, D.~Vukobratovic, and V.~Crnojevic, ``{Cooperative
  Slotted Aloha for Multi-Base Station Systems},'' \emph{IEEE Transactions on
  Communications}, vol.~63, no.~4, pp. 1443--1456, April 2015.

\bibitem{Ivanonv2015}
M.~Ivanov, F.~Br\"annstr\"om, G.~Graell~i Amat, and P.~Popovski, ``{All-to-all
  Broadcast for Vehicular Networks Based on Coded Slotted ALOHA},'' in
  \emph{Proceedings of the IEEE International Conference on Communications -
  Workshop on Massive Uncoordinated Access Protocols (MASSAP)}, London, UK,
  June 2015, pp. 2046--2050.

\bibitem{Bui2015}
H.-C. Bui, K.~Zidane, J.~Lacan, and M.-L. Boucheret, ``{A Multi-Replica
  Decoding Technique for Contention Resolution Diversity Slotted Aloha},'' in
  \emph{Proceedings of the 82nd IEEE Vehicular Technology Conference
  (VTC-Fall)}, Boston, MA, September 2015, pp. 1--6.

\bibitem{Kissling2011a}
C.~Kissling, ``{Performance Enhancements for Asynchronous Random Access
  Protocols over Satellite},'' in \emph{Proceedings of the IEEE International
  Conference on Communications (ICC)}, Kyoto, Japan, June 2011, pp. 1--6.

\bibitem{deGaudenzi2014_ACRDA}
R.~De~Gaudenzi, O.~del Rio~Herrero, G.~Acar, and E.~Garrido~Barrabes,
  ``{Asynchronous Contention Resolution Diversity ALOHA: Making CRDSA Truly
  Asynchronous},'' \emph{IEEE Transactions on Wireless Communications},
  vol.~13, no.~11, pp. 6193--6206, November 2014.

\bibitem{delRioHerrero_2012}
O.~del Rio~Herrero and R.~De~Gaudenzi, ``{High Efficiency Satellite Multiple
  Access Scheme for Machine-to-Machine Communications},'' \emph{IEEE
  Transactions on Aerospace and Electronic Systems}, vol.~48, no.~4, pp.
  2961--2989, October 2012.

\bibitem{Gallinaro_2014}
G.~Gallinaro, F.~Di~Cecca, M.-A. Marchitti, R.~De~Gaudenzi, and O.~del
  Rio~Herrero, ``{Enhanced spread spectrum ALOHA system level performance
  assessment},'' \emph{International Journal on Satellite Communications and
  Networking}, vol.~32, no.~6, pp. 485--503, November-December 2014.

\bibitem{Scalise_2013}
S.~Scalise, C.~P\'{a}rraga~Niebla, R.~De~Gaudenzi, O.~del Rio~Herrero,
  D.~Finocchiaro, and A.~Arcidiacono, ``{S-MIM: A Novel Radio Interface for
  Efficient Messaging Services over Satellite},'' \emph{IEEE Communications
  Magazine}, vol.~51, no.~3, pp. 119--125, March 2013.

\bibitem{ZigZag}
S.~Gollakota and D.~Katabi, ``{ZigZag Decoding: Combating Hidden Terminals in
  Wireless Networks},'' in \emph{Proceedings of the International Conference of
  Special Interest Group on Data Communications (SIGCOMM)}, Seattle, WA, USA,
  August 2008, pp. 159--170.

\bibitem{SigSag}
A.~S. Tehrani, A.~G. Dimakis, and M.~J. Neely, ``{SigSag: Iterative Detection
  through Soft Message-Passing},'' \emph{IEEE Journal of Selected Topics in
  Signal Processing}, vol.~5, no.~8, pp. 1512--1523, December 2011.

\bibitem{Kschischang_2001}
F.~R. Kschischang, B.~J. Frey, and H.~A. Loeliger, ``{Factor Graphs and the
  Sum-Product Algorithm},'' \emph{IEEE Transactions on Information Theory},
  vol.~47, no.~2, pp. 498--519, Febraury 2001.

\bibitem{Richardson_2001}
T.~J. Richardson and R.~L. Urbanke, ``{The Capacity of Low-Density Parity-Check
  Codes Under Message-Passing Decoding},'' \emph{IEEE Transactions on
  Information Theory}, vol.~47, no.~2, pp. 599--618, February 2001.

\bibitem{Giannakis_2007}
Y.~Yu and G.~B. Giannakis, ``{High-Throughput Random Access Using Successive
  Interference Cancellation in a Tree Algorithm},'' \emph{IEEE Transactions on
  Information Theory}, vol.~53, no.~12, pp. 4628--4639, December 2007.

\bibitem{Massey1988}
J.~L. Massey, ``{Channel Models for Random Access Systems},'' \emph{Performance
  Limits in Communication Theory and Practice}, vol. 142, Part 3, pp. 391--402,
  1988, {NATO ASI Series}.

\bibitem{Pippinger1981}
N.~Pippinger, ``{Bound on the Performance of Protocols for a Multiple-Access
  Broadcast Channel},'' \emph{IEEE Transactions on Information Theory}, vol.
  IT-27, no.~2, pp. 145--151, March 1981.

\bibitem{Peeters2015}
G.~T. Peeters and B.~Van~Houdt, ``{On the Capacity of a Random Access Channel
  with Successive Interference Cancellation},'' in \emph{Proceedings of the
  IEEE International Conference on Communications (ICC), Workshop on Massive
  Uncoordinated Access Protocols (MASSAP)}, London, U.K., June 2015, pp.
  2051--2056.

\bibitem{Zhang2016}
Y.~Zhang, Y.-H. Lo, W.~S. Wong, and F.~Shu, ``{Protocol Sequences for the
  Multiple-Packet Reception Channel Without Feedback},'' \emph{IEEE
  Transactions on Communications}, vol.~64, no.~4, pp. 1687--1698, April 2016.

\bibitem{Zhang2016a}
J.~Zhang, Y.~Chen, Y.-H. Lo, and W.~S. Wong, ``{Achieving Zero-Error Capacity 1
  for a Collision Channel Without Feedback},'' \emph{Available at:
  http://arxiv.org/pdf/1609.05448v1.pdf}, 2016.

\bibitem{Bae2014}
Y.~H. Bae, B.~D. Choi, and A.~S. Alfa, ``{Achieving Maximum Throughput in
  Random Access Protocols with Multipacket Reception},'' \emph{IEEE
  Transactions on Mobile Computing}, vol.~13, no.~3, pp. 497--511, March 2014.

\bibitem{Shum2009}
K.~W. Shum, C.~S. Chen, C.~W. Sung, and W.~S. Wong, ``{Shift-Invariant Protocol
  Sequences for the Collision Channel Without Feedback},'' \emph{IEEE
  Transactions on Information Theory}, vol.~55, no.~7, pp. 3312--3322, July
  2009.

\bibitem{dvb_rcs2}
\emph{{Digital Video Broadcasting (DVB); Second Generation DVB Interactive
  Satellite System (DVB-RCS2); Part 2: Lower Layers for Satellite standard}},
  Digital Video Boradcasting (DVB) Std. EN 301 545-2, Rev. 1.2.1, April 2014.

\bibitem{Ais_std}
\emph{{Technical characteristics for an automatic identification system using
  time-division multiple access in the VHF maritime mobile band}},
  International Telecommunications Union Std. M.1371-5, Febraury 2014.

\bibitem{Vdes_std}
\emph{{Technical characteristics for a VHF data exchange system in the VHF
  maritime mobile band}}, International Telecommunications Union Std. M.2092-0,
  October 2015.

\bibitem{Clazzer2012}
F.~Clazzer and C.~Kissling, ``{Enhanced Contention Resolution ALOHA} -
  {ECRA},'' in \emph{Proceedings of the International ITG Conference on
  Systems, Communications and Coding (SCC)}, Munich, Germany, January 2013.

\bibitem{5GNow}
G.~Wunder, P.~Jung, M.~Kasparick, T.~Wild, F.~Schaich, Y.~Chen, S.~Ten~Brink,
  I.~Gaspar, N.~Michailow, A.~Festag, L.~Mendes, N.~Cassiau, D.~Ktenas,
  M.~Dryianski, S.~Pietrzyk, B.~Eged, P.~Vago, and F.~Wiedmann, ``{5GNOW:
  Non-Orthogonal, Asynchronous Waveforms for Future Mobile Applications},''
  \emph{IEEE Communicatons Magazine}, vol.~52, no.~2, pp. 97--105, February
  2014.

\bibitem{Chen2014}
K.~Chen, M.~Ma, E.~Cheng, F.~Yuan, and W.~Su, ``{A Survey on MAC Protocols for
  Underwater Wireless Sensor Networks},'' \emph{IEEE Communications Surveys and
  Tutorials}, vol.~16, no.~3, pp. 1433--1447, Third Quarter 2014.

\bibitem{Cui2006}
J.~Cui, J.~Kong, M.~Gerla, and S.~Zhou, ``{The Challenges of Building Scalable
  Mobile Underwater Wireless Sensor Networks for Aquatic Applications},''
  \emph{IEEE Networks}, vol.~20, no.~3, pp. 12--18, May-June 2006.

\bibitem{ESA_survey_2016}
R.~De~Gaudenzi, O.~Del Rio~Herrero, G.~Gallinaro, S.~Cioni, and P.-D.
  Arapoglou, ``{Random Access SScheme for Satellite Networks: from VSAT to M2M
  - A Survey},'' \emph{to appear in the Internation Journal of Satellite
  Communications and Networking}, 2016.

\bibitem{Biason2016}
A.~Biason, A.~Dittadi, and M.~Zorzi, ``{Spreading and Repetitions in Satellite
  MAC Protocols},'' in \emph{Proceedings of the IEEE International Conference
  on Communications (ICC)}, Kuala Lumpur, Malaysia, May 2016, pp. 1--6.

\bibitem{Brennan1959}
D.~G. Brennan, ``{Linear Diversity Combining Techniques},'' \emph{Proceedings
  of the IRE}, vol.~47, no.~6, pp. 1075--1102, June 1959.

\bibitem{Jakes1974}
W.~C. Jakes, \emph{{Microwave Mobile Communication}}.\hskip 1em plus 0.5em
  minus 0.4em\relax Wiley-IEEE Press, 1974.

\bibitem{Thomas_2000}
G.~Thomas, ``{Capacity of the Wireless Packet Collision Channel With
  Feedback},'' \emph{IEEE Transactions on Information Theory}, vol.~46, no.~3,
  pp. 1141--1144, May 2000.

\bibitem{Winters1984}
J.~H. Winters, ``{Optimum Combining in Digital Mobile Radio with Cochannel
  interference},'' \emph{IEEE Journal on Selected Areas in Communications},
  vol. SAC-2, no.~4, pp. 528--539, July 1984.

\bibitem{Clazzer2013}
F.~Clazzer and C.~Kissling, ``{Optimum Header Positioning in Successive
  Interference Cancellation (SIC) based ALOHA},'' in \emph{Proceedings of the
  IEEE International Conference on Communications (ICC)}, Budapest, Hungary,
  June 2013, pp. 2869--2874.

\bibitem{Clazzer2017}
F.~Clazzer, F.~L\'azaro, G.~Liva, and M.~Marchese, ``{Detection and Combining
  Techniques for Asynchronous Random Access with Time Diversity},'' in
  \emph{Accepted for pubblication at the 11th International ITG Conference on
  Systems, Communications and Coding}, Available at:
  http://arxiv.org/abs/1604.06221, 2016.

\bibitem{Ivanonv_2015_Letter}
M.~Ivanov, F.~Br\"{a}nnstr\"{o}m, A.~Graell~i Amat, and P.~Popovski, ``{Error
  Floor Analysis of Coded Slotted ALOHA Over Packet Erasure Channels},''
  \emph{IEEE Communications Letters}, vol.~19, no.~3, pp. 419--422, March 2015.

\bibitem{Chiani_2010}
M.~Chiani, ``{Noncoherent Frame Synchronization},'' \emph{IEEE Transactions on
  Communications}, vol.~58, no.~5, pp. 1536--1545, May 2010.

\bibitem{Polydoros1984}
A.~Polydoros and C.~Weber, ``{A Unified Apporach to Serial Search
  Spread-Spectrum Code Acquisition Part I: General Theory},'' \emph{IEEE
  Transactions on Communications}, vol.~32, no.~5, pp. 542--549, May 1984.

\bibitem{Ferguson1977}
M.~J. Ferguson, ``{An Approximate Analysis of Delay for Fixed and Variable
  Length Packets in an Unslotted ALOHA Channel},'' \emph{IEEE Transactions on
  Communications}, vol. COM-25, no.~7, pp. 644--654, July 1977.

\bibitem{Bellini1980}
S.~Bellini and F.~Borgonovo, ``{On the Throughput of an ALOHA Channel with
  Variable Length Packets},'' \emph{IEEE Transactions on Communications}, vol.
  COM-28, no.~11, pp. 1932--1935, November 1980.

\bibitem{Okada1978}
H.~Okada, Y.~Nakanishi, and Y.~Igarashi, ``{Analysis of Framed ALOHA Channel in
  Satellite Packet Switching Networks},'' in \emph{Proceedings of the 4th
  International Conference on Computer Communication}, September 1978, pp.
  617--622.

\bibitem{Wieselthier1989}
J.~E. Wieselthier, A.~Ephremides, and L.~A. Michaels, ``{An Exact Analysis and
  Performance Evaluation of Framed ALOHA with Capture},'' \emph{IEEE
  Transactions on Communications}, vol.~37, no.~2, pp. 125--137, February 1989.

\bibitem{Jensen}
J.~L. Jensen, ``{Sur les fonctions convexes et les inegalites entre les valeurs
  moyennes},'' \emph{Acta mathematica}, vol.~30, pp. 175--193, 1905.

\bibitem{N1984}
C.~Namislo, ``{Analysis of Mobile Radio Slotted ALOHA Networks},'' \emph{IEEE
  Journal on Selected Areas in Communications}, vol.~2, no.~4, pp. 583--588,
  July 1984.

\bibitem{Z1997}
M.~Zorzi, ``{Capture Probabilities in Random-Access Mobile Communications in
  the Presence of Rician Fading},'' \emph{IEEE Transactions on Vehicular
  Technology}, vol.~46, no.~1, pp. 96--101, February 1997.

\bibitem{ZZ2012}
A.~Zanella and M.~Zorzi, ``{Theoretical Analysis of the Capture Probability in
  Wireless Systems with Multiple Packet Reception Capabilities},'' \emph{IEEE
  Transactions on Communications}, vol.~60, no.~4, pp. 1058--1071, April 2012.

\bibitem{NEW2007}
G.~D. Nguyen, A.~Ephremides, and J.~E. Wieselthier, ``{On Capture in
  Random-Access Systems},'' in \emph{Proceedings of the IEEE International
  Symposium on Information Theory (ISIT)}, Seattle, WA, USA, July 2006, pp.
  2072--2076.

\bibitem{SMP2014}
C.~Stefanovic, M.~Momoda, and P.~Popovski, ``{Exploiting Capture Effect in
  Frameless ALOHA for Massive Wireless Random Access},'' in \emph{Proceedings
  of IEEE Wireless Communications and Networking Conference (WCNC)}, Istanbul,
  Turkey, May 2014, pp. 1762--1767.

\bibitem{Stefanovic2012}
C.~Stefanovic, P.~Popovski, and D.~Vukobratovic, ``{Frameless ALOHA Protocol
  for Wireless Networks},'' \emph{IEEE Communications Letters}, vol.~16,
  no.~12, pp. 2087--2090, December 2012.

\bibitem{RU2007}
T.~Richardson and R.~Urbanke, \emph{{M}odern {C}oding {T}heory}.\hskip 1em plus
  0.5em minus 0.4em\relax Cambridge University Press, 2007.

\bibitem{diffEvol1997}
R.~Storn and K.~Price, ``{Differential Evolution - A Simple and Efficient
  Heuristic for global Optimization over Continuous Spaces},'' \emph{Journal of
  Global Optimization}, vol.~11, no.~4, pp. 341--359, Dec. 1997.

\bibitem{Zorzi1994_diversity}
M.~Zorzi, ``{Mobile Radio Slotted ALOHA with Capture, Diversity and
  Retransmission Control in the Presence of Shadowing},'' \emph{Wireless
  Networks}, vol.~4, pp. 379--388, August 1994.

\bibitem{OnOff2003}
E.~Perron, M.~Rezaeian, and A.~Grant, ``{The On-Off Fading Channel},'' in
  \emph{Proceedings of the IEEE International Symposium on Information Theory
  (ISIT)}, Yokohama, Japan, July 2003, pp. 244--248.

\bibitem{slomson1991introduction}
A.~B. Slomson, \emph{{An Introduction to Combinatorics}}.\hskip 1em plus 0.5em
  minus 0.4em\relax Chapman and Hall, 1991.

\bibitem{slepian1973noiseless}
D.~Slepian and J.~Wolf, ``{Noiseless Coding of Correlated Information
  Sources},'' \emph{IEEE Transactions on Information Theory}, vol.~19, no.~4,
  pp. 471--480, July 1973.

\bibitem{dana2006capacity}
A.~F. Dana, R.~Gowaikar, R.~Palanki, B.~Hassibi, and M.~Effros, ``{Capacity of
  Wireless Erasure Networks},'' \emph{IEEE Transactions on Information Theory},
  vol.~52, no.~3, pp. 789--804, March 2006.

\bibitem{ho2006random}
T.~Ho, M.~M{\'e}dard, R.~Koetter, D.~R. Karger, M.~Effros, J.~Shi, and
  B.~Leong, ``{A Random Linear Network Coding Approach to Multicast},''
  \emph{IEEE Transactions on Information Theory}, vol.~52, no.~10, pp.
  4413--4430, October 2006.

\bibitem{Munari15_Massap}
A.~Munari, F.~Clazzer, and G.~Liva, ``{Multi-Receiver Aloha - a Survey and New
  Results},'' in \emph{Proceedings of the IEEE International Conference on
  Communications (ICC), Workshop on Massive Uncoordinated Access Protocols
  (MASSAP)}, London (UK), June 2015, pp. 2108--2114.

\bibitem{Landsberg_1893}
G.~Landsberg, ``{Ueber eine Anzahlbestimmung und eine damit zusammenhaengende
  Reihe},'' \emph{Journal fuer die reine und angewandte Mathematik}, vol. 111,
  pp. 87--88, 1893.

\bibitem{Kolchin_RG_1999}
V.~F. Kolchin, \emph{{Random Graphs}}.\hskip 1em plus 0.5em minus 0.4em\relax
  Cambridge University Press, 1999.

\end{thebibliography}
\thispagestyle{empty}


\renewcommand{\chaptermark}[1]{\markboth{\chaptername\ \thechapter.\ #1}{}}

\end{document}